\def\truecondition{1}                                                                                                   %
\def\falsecondition{0}
  
\documentclass[12pt,journal,onecolumn]{IEEEtran}
\let\reportversion=\truecondition

\title{\bf
Convergence of Nonlinear Observers on $\RR ^n$
with a Riemannian Metric (Part III)
}
\author{Ricardo G. Sanfelice\thanks{R. G. Sanfelice is with
the Department of Electrical and Computer Engineering, University of California, Santa Cruz, CA 95064, USA.
       Email: {\tt\small ricardo@ucsc.edu}
}\hskip 0.5em
and
Laurent Praly\thanks{L. Praly is with CAS, ParisTech, Ecole des Mines, 35 rue Saint Honor\'{e}, 77305, 
        Fontainebleau, France.
       Email: {\tt\small Laurent.Praly@ensmp.fr}
}
}
\date{last correction: March 17, 2023, compilation: \today}

\usepackage{amsmath}
\usepackage{latexsym}
\usepackage{amsfonts}
\usepackage{xcolor}
\usepackage{graphicx}
\usepackage{hyperref}
\usepackage{ifthen}

\usepackage{ulem} 

\definecolor{my-violet}{rgb}{0.6,0,0.6}
\definecolor{my-dark-violet}{rgb}{0.42,0.21,0.07}
\definecolor{my-green}{rgb}{0.3,0.9,0.0}
\definecolor{my-dark-green}{rgb}{0.26, 0.78, 0.68}
\definecolor{my-brown}{rgb}{0.6,0.4,0.2}
\definecolor{my-magenta}{rgb}{0.7,0.6,0.3}
\definecolor{MyMagenta}{rgb}{0.8,0,0.8}
\definecolor{my-light-blue}{rgb}{0,0.5,1}
\definecolor{coloremove}{rgb}{.0,0,.5}

\newboolean{Report}
\ifx\reportversion\falsecondition
\setboolean{Report}{false}
\else
\setboolean{Report}{true}
\fi

\newcommand{\IfReport}[2]{%
\colorlet{savecolor}{.}%
\ifthenelse{\boolean{Report}}{{%
\begingroup%
\color{my-violet}#1%
\endgroup%
}%
}{%
#2
}%
\color{savecolor}%
}

\def\stopsuppress{stopsuppress}
\let\stoparchive=\stopsuppress

\IfReport{%
\gdef\glos#1{{\color{my-violet}%
\hskip 2 pt
\null \put(13,4){\oval(24,10)}\put(0,0){%
\hbox to 24pt{\hss{G#1}\hss}
}\hskip 24pt{}}}%
\gdef\startarchive{\begingroup\color{my-violet}}
\gdef\stoparchive{\endgroup}
}{%
\gdef\glos#1{}
\gdef\startarchive{\suppress}
}

\newcommand{\IfTwoCol}[2]{%
\colorlet{savecolor}{.}%
\ifthenelse{\boolean{Report}}{%
\begingroup%
\color{my-violet}#2%
\endgroup%
}{%
#1
}%
\color{savecolor}%
}

\newcommand{\complement}{supplementary material in Appendix~}

\def\suppress{\futurelet\nexttoken\doSuppress}%
\newcommand{\doSuppress}{%
    \ifx \nexttoken \stopsuppress
        \let\do=\relax
    \else
        \let\do=\suppress
    \fi
    \GobbleOneTokenAndCall \do
}
\newcommand{\GobbleOneTokenAndCall}[1]{%
    \afterassignment#1%
    \let\gobbledToken= %
}
\def\startmodif{%
\colorlet{savecolor}{.}%
\begingroup%
\color{blue}%
}

\def\stopmodif{%
\endgroup%
\color{savecolor}%
}

\definecolor{mygreen}{RGB}{0,128,0}

\newcommand{\ricardo}[1]{{\color{my-dark-violet}#1}}

\newcounter{compteurtmp}
\newtheorem{workahead}{} 

\def\repeattext#1#2{%
\setcounter{compteurtmp}{0}%
\loop\ifnum\thecompteurtmp<#2%
\addtocounter{compteurtmp}{1}%
#1%
\repeat}
\DeclareMathAlphabet{\EuScript}{U}{eus}{m}{n}
\DeclareMathAlphabet{\EuScriptbold}{U}{eus}{b}{n}
\SetMathAlphabet{\EuScript}{bold}{U}{eus}{b}{n}
\DeclareMathAlphabet{\mathbfcal}{OMS}{cmsy}{b}{n}
\DeclareMathAlphabet{\mathfrakbf}{U}{euf}{b}{n}
\usepackage{aurical}
\newcommand{\bfFontlukas}{\usefont{T1}{LukasSvatba}{bx}{n}}
\def\cursive#1{\mbox{\Fontlukas{#1}}}
\def\bfcursive#1{\mbox{\bfFontlukas{#1}}}

\def\indxra{{\alpha }}
\def\indxrb{{\beta }}
\def\indxrc{{\gamma }}
\def\indxrd{{\delta }}
\def\indxre{{\epsilon}}
\def\indxrf{{\eta}}
\def\indxrg{{\mu }}
\def\indxrh{{\nu }}

\def\indxa{{a}}
\def\indxb{{b}}
\def\indxc{{c}}
\def\indxd{{d}}
\def\indxe{{e}}
\def\indxf{{f}}

\def\indyi{i}
\def\indyj{j}
\def\indyk{k}
\def\indyl{l}
\def\indym{m}

\DeclareMathAlphabet{\mathcursive}{U}{esstixcal}{m}{n}
\DeclareMathAlphabet{\bfmathcursive}{U}{esstixcal}{bx}{n}
\def\changex{{\mathchoice%
{{\mbox{\large$\scriptstyle \mathcursive{C}$}}}%
{{\mbox{\large$\scriptstyle \mathcursive{C}$}}}%
{{\mbox{$\scriptstyle \mathcursive{C}$}}}%
{{\mbox{$\scriptscriptstyle \mathcursive{C}$}}}%
}}
\def\changey{{\mathchoice%
{{\mbox{\large$\scriptstyle \mathcursive{D}$}}}%
{{\mbox{\large$\scriptstyle \mathcursive{D}$}}}%
{{\mbox{$\scriptstyle \mathcursive{D}$}}}%
{{\mbox{$\scriptscriptstyle \mathcursive{D}$}}}%
}}
\def\bfchangey{{\mathchoice%
{{\mbox{\large$\scriptstyle \mathcursive{D}\mkern-15.5mu\mathcursive{D}\mkern-15.5mu\mathcursive{D}\mkern-15.5mu \mathcursive{D}$}}}%
{{\mbox{\large$\scriptstyle \mathcursive{D}\mkern-15.5mu\mathcursive{D}\mkern-15.5mu\mathcursive{D}\mkern-15.5mu \mathcursive{D}$}}}%
{{\mbox{$\scriptstyle \bfmathcursive{D}$}}}%
{{\mbox{$\scriptscriptstyle \bfmathcursive{D}$}}}%
}}
\def\changexr{{\mathchoice%
{{\mbox{\large$\scriptstyle \mathcursive{E}$}}}%
{{\mbox{\large$\scriptstyle \mathcursive{E}$}}}%
{{\mbox{$\scriptstyle \mathcursive{E}$}}}%
{{\mbox{$\scriptscriptstyle \mathcursive{E}$}}}%
}}

\def\changexa{\changex_a}
\def\changeya{\changey_a}

\def\coordxd{\mathcal{M}}
\def\coordxm{\phi} 
\def\coordxp{x}
\def\coordx{(\coordxp,\coordxd,\coordxm)}
\def\barcoordxd{\bar{\mathcal{M}}}
\def\barcoordxm{\bar \phi}
\def\barcoordxp{\bar x}
\def\barcoordx{(\barcoordxp,\barcoordxd,\barcoordxm)}

\def\coordyd{\mathcal{N}}
\def\coordym{\chi} 
\def\coordyp{y}
\def\coordy{(\coordyp,\coordyd,\coordym)}
\def\barcoordyd{\bar{\mathcal{N}}}
\def\barcoordym{\bar \chi}
\def\barcoordym{\bar \chi}
\def\barcoordyp{\bar y}
\def\barcoordy{(\barcoordyp,\barcoordyd,\barcoordym)}

\def\coordxrd{\mathcal{O}}
\def\coordxrm{\psi} 
\def\coordxrp{\xrond}
\def\coordxr{(\coordxrp,\coordxrd,\coordxrm)}

\def\coordyxrd{\coordxd}
\def\coordyxrm{\coordxm_{\mathcal{N}}} 
\def\coordyxrp{(\coordyp,\coordxrp)}
\def\coordyxr{(\coordyxrp,\coordyxrd,\coordyxrm)}

\def\hhperp{\theta}
\def\bfhhperp{\mbox{\boldmath{$\theta$}}}
\def\hhperpinv{\vartheta}
\def\bfhhperpinv{\mbox{\boldmath{$\vartheta$}}}

\def\coordxid{\mathcal{P}}
\def\coordxim{\omega}
\def\coordxip{\xi}
\def\coordxi{(\coordxip,\coordxid,\coordxim)}

\def\coordxed{\coordxd_e}
\def\coordxem{\coordxm_e}
\def\coordxep{\coordxp_e}
\def\coordxe{(\coordxep,\coordxed,\coordxem)}

\def\coordxxed{\coordxd\times\coordxed}
\def\coordxxem{(\coordxm,\coordxem)}
\def\coordxxep{(\coordxp,\coordxep)}
\def\coordxxe{(\coordxxep,\coordxxed,\coordxxem)}

\def\barcoordxap{\bar \coordxp_a}

\def\coordyad{\coordyd_a}
\def\coordyam{\coordym_a} 
\def\coordyap{\coordyp_a}
\def\coordya{(\coordyap,\coordyad,\coordyam)}
\def\barcoordyap{\bar \coordyp_a}
\def\barcoordya{(\barcoordyap,\coordyad,\coordyam)}

\def\Close{\mathcal{C}}
\def\dy{e}

\def\Did{{\mathfrak{D}^+}}
\def\bfd{{\mbox{\boldmath{$d$}}}} 
\def\Distrib{\cursive{D}}
\def\bfDistrib{\bfcursive{D}}

\def\entree{u}

\def\bff{{\mathchoice
{{\mbox{\normalsize\textit{\textbf{f}}}}}%
{{\mbox{\small\textit{\textbf{f}}}}}%
{{\mbox{\scriptsize\textit{\textbf{f}}}}}%
{{\mbox{\tiny\textit{\textbf{f}}}}}%
}}
\def\bfF{{\mbox{\boldmath{$F$}}}} 

\def\fa{{f_a}}

\def\grad{g}
\def\bfgrad{{\mbox{\boldmath{$g$}}}} 
\def\bfgamma{\mbox{\boldmath{$\gamma $}}}
\def\gammay{\delta }
\def\bfgammay{\mbox{\boldmath{$\delta $}}}
\def\Gammay{\Delta }
\def\bfdh{{\mbox{\boldmath{$dh$}}}} 

\def\bfh{{\mbox{\boldmath{$h$}}}} 

\def\ha{h_a}
\def\barha{{\bar h}_a}
\def\bfha{{\mbox{\boldmath{$h_a$}}}} 
\def\oh{\phi}

\def\oha{\oh_a}

\def\Hess{H}
\def\bfHess{{\mbox{\boldmath{$H$}}}} 
\def\immer{\mkern 1mu {\mathcursive{i}}}
\def\bfimmer{{\mathcursive{i}\mkern-6.2mu\mathcursive{i}\mkern-6.2mu\mathcursive{i}\mkern-6.2mu\mathcursive{i}\mkern-6.2mu\mathcursive{i}}}
\def\kunbf{{k}}
\def\Ouv{\Omega }

\def\euP{\EuScript{P}}
\def\scripteuP{\mbox{\tiny $\euP$}}
\def\bfP{{\mathchoice
{{\mbox{\normalsize\textit{\textbf{P}}}}}%
{{\mbox{\normalsize\textit{\textbf{P}}}}}%
{{\mbox{\scriptsize\textit{\textbf{P}}}}}%
{{\mbox{\tiny\textit{\textbf{P}}}}}%
}}
\def\Py{Q}
\def\bfPy{{\mbox{\boldmath{$Q$}}}}
\def\Pxi{R}
\def\bfPxi{{\mbox{\boldmath{$R$}}}} 
\def\Pa{{P_a}}
\def\bfPa{{\mbox{\boldmath{$P_a$}}}} 

\def\bfvarphi{\mbox{\boldmath{$\varphi$}}}

\def\bfpi{\mbox{\boldmath{$\pi $}}}
\def\wpunbf{\wp}
\newcommand{\qlower}{q}
\newcommand{\reals}{\mathbb{R}}
\def\RR{\mathbb{R}}
\def\bfRR{
\setbox0=\hbox{$\RR$}%
\longueur =  0.999\wd0 
\setbox0=\hbox{$\RR\hskip -\longueur\RR$} 
\RR\hskip -\longueur \RR
\hskip -\wd0 
\longueur =  0.04\ht0
\raise\longueur\box0
}
\newcommand{\RRgeq}{[0,+\infty)}
\def\submer{{\mathcursive{s}}}
\def\bfsubmer{{\mathcursive{s}\mkern-7.5mu\mathcursive{s}\mkern-7.5mu\mathcursive{s}\mkern-7.5mu\mathcursive{s}\mkern-7.5mu\mathcursive{s}}}
\def\SS{\mathbb{S}}
\def\bfSS{%
\setbox0=\hbox{$\SS$}%
\longueur =  0.999\wd0 
\setbox0=\hbox{$\SS\hskip-\longueur\SS$} 
\SS\hskip -\longueur \SS
\hskip -\wd0 
\longueur =  0.04\ht0
\raise\longueur\box0
}
\def\secff{I\!\!I\!}
\def\bfsecff{\mbox{\boldmath{$I\!\!I\!$}}} 

\def\euX{\EuScript{X}}

\def\bfXrond{{\mathbfcal{Z}}}
\def\xrond{z}
\def\bfxrond{{\mbox{\boldmath{$z$}}}} 
\def\barxrond{{\bar{\xrond} }}

\def\bfx{{\mbox{\boldmath{$x$}}}} 
\def\hatx{\hat{x}}
\def\bfhatx{{\mbox{\boldmath{$\hat x$}}}} 
\def\bfX{{\mbox{\boldmath{$X$}}}} 

\def\bfhatX{{\mbox{\boldmath{$\hat X$}}}} 

\def\xa{x_a}

\def\bfxa{{\mbox{\boldmath{$x_a$}}}} 
\def\barxa{{\bar x}_a}

\def\bfXi{\mbox{\boldmath{$\Xi$}}}

\def\bfXi{\mbox{\boldmath{$\Xi$}}}
\newcommand{\xdummy}{\xi}

\def\xe{x_e}
\def\barxe{{\bar x}_e}
\def\hatxe{{\hat x}_e}
\def\bfxe {\mbox{\boldmath$x_e$}}

\def\ya{{y_a}}

\def\bfy{{\mbox{\boldmath{$y$}}}} 
\def\yrond{{\mathchoice%
{{\mbox{$\scriptstyle \mathcal{Y}$}}}%
{{\mbox{$\scriptstyle \mathcal{Y}$}}}%
{{\mbox{$\scriptscriptstyle \mathcal{Y}$}}}%
{{\mbox{\tiny$\scriptscriptstyle \mathcal{Y}$}}}%
}}

\def\tangent{{\mathchoice%
{\textsf{\normalsize tan}}%
{\textsf{\small tan}}%
{\textsf{\hskip -0.5pt \scriptsize t\hskip -0.5pt a\hskip -0.7pt n}}%
{\textsf{\tiny tan}}%
}}

\def\ortho{{\mathchoice%
{\textsf{\normalsize ort}}%
{\textsf{\small ort}}%
{\textsf{\hskip -0.5pt \scriptsize o\hskip -0.5pt r\hskip -0.7pt t}}%
{\textsf{\tiny ort}}%
}}

\def\comp{{\mathchoice%
{\textsf{\normalsize com}}%
{\textsf{\small com}}%
{\textsf{\hskip -0.5pt \scriptsize c\hskip -0.5pt o\hskip -0.7pt m}}%
{\textsf{\tiny com}}%
}}

\def\dot#1{\renewcommand{\arraystretch}{0.5}%
\begin{array}[b]{@{}c@{}}.\\\renewcommand{\arraystretch}{1}#1\end{array}%
\renewcommand{\arraystretch}{1}}
\def\vardot#1#2{%
\setbox0=\hbox{$#1$}\setbox1=\hbox{$#2$}%
\ht1=\wd0\wd0=\wd1%
\setbox0=\hbox{$\dot{\box0}$}%
\wd0=\ht1\box0}

\catcode`\@=11
\font \police line10 at 1\@ptsize pt
\def\rightarrowfill{%
\setbox0=\hbox{{\police\char"2D}}\ht0=0pt
$\m@th
\leaders\vrule\hfill
\mkern-10mu\raise0.7pt\box0
$}
\def\underrightarrowfill#1{%
\mathop{\vbox{\ialign{##\crcr\crcr%
$\hfil\displaystyle{#1}\hfil$%
\crcr\noalign{\kern 0.7ex\nointerlineskip}%
\rightarrowfill\crcr}}}\limits}%
\def\overerrightarrowfill#1{%
\mathop{\vbox{\ialign{##\crcr\crcr%
\rightarrowfill%
\crcr\noalign{\kern 0.97ex\nointerlineskip}%
$\hfil\displaystyle{#1}\hfil$%
\crcr}}}\limits}%
\catcode`\@=12

\catcode`\@=11
\def\downparenfill{$\m@th\braceld\leaders\vrule\hfill\bracerd$}
\def\overparen#1{\mathop{\vbox{\ialign{##\crcr\crcr \noalign{\kern0.4ex}
\downparenfill\crcr\noalign{\kern0.4ex\nointerlineskip}
$\hfil\displaystyle{#1}\hfil$\crcr}}}\limits}
\catcode`\@=12

\newlength{\longueur}

\catcode`\@=11
\def\contenulabel{\protected@edef\@currentlabel}
\catcode`\@=12

\newtheorem{theorem}{Theorem}[section] 
\newtheorem{lemma}[theorem]{Lemma}   
\newtheorem{example}[theorem]{Example}    
\newtheorem{proposition}[theorem]{Proposition}    
\newtheorem{remark}[theorem]{Remark}   
\newtheorem{definition}[theorem]{Definition}
\newtheorem{corollary}[theorem]{Corollary}

\newcounter{hypothesis}

\def\sousection#1{%
\refstepcounter{subsection}%
\subsection*{\thesubsection\   #1}%
\addcontentsline{toc}{subsection}{\protect\numberline{\thesubsection}#1}%
}

\def\soussousection#1{%
\refstepcounter{subsubsection}%
\subsection*{\thesubsubsection\   #1}%
\addcontentsline{toc}{subsubsection}{\protect\numberline{\thesubsubsection}#1}%
}

\usepackage{setspace}
\begin{document}

\maketitle

\noindent
\IfReport{%
\tableofcontents
\newpage
}{\vspace{-0.15in}}


\begin{abstract}
This paper is the third and final component of a three-part effort on
observers contracting a Riemannian distance between the state of the system and its estimate.  
In Part I,
we showed that such a contraction property holds if the system dynamics and the Riemannian
metric satisfy two key conditions:
a differential detectability property and a geodesic monotonicity property.
With
the former condition being the focus of Part II, in this Part III, we study the latter condition
in relationship to the
nullity of the second fundamental form of the output function. 
We formulate sufficient and necessary conditions
for it to hold.
We establish a
link between it
and the infinite gain margin property,
and we provide a systematic way for constructing a metric 
satisfying this condition. Finally, we illustrate cases where
both
conditions
hold%
\startarchive
\   and propose ways to facilitate the satisfaction of these two conditions together%
\stoparchive %
.

\startarchive
\startmodif
{\it Modified or new material relative to the
version submitted for review is colored in blue. } 
\stopmodif
\startarchive
Text in magenta is material that is only in the report version of the journal submission associated to this work.
\stoparchive %
\stoparchive %
\end{abstract}

\section[Introduction]{{Introduction}}
\label{sec:Introduction}

\subsection{Background}
We consider nonlinear systems on $\bfRR^n$ of the form
\\[0.5em]\null \hfill $\displaystyle 
\dot{\bfx}  \;=\;  \bff(\bfx)\  ,\quad 
\bfy  \;=\;  \bfh(\bfx),
$\refstepcounter{equation}\label{eqn:Plant1}\hfill$(\theequation)$\\[0.5em]
where $\bfx$ represents the state living
\startmodif
in
\stopmodif
$\bfRR^n$,
$\bfy: \bfRR^n \to \bfRR^n$ represents the
measured output living in $\bfRR^p$ and
$\bff:\bfRR^n \to \bfRR^n$ and $\bfh : \bfRR^n \to \bfRR^p$ are functions.

\startmodif
For this class of systems,
we continue the study, started in \cite{57} and \cite{127}, of designing
\begin{enumerate}
\item
a state observer, namely, a dynamical system
\\[0.5em]\null \hfill $\displaystyle
\dot \bfhatx \;=\;   \bfF({\bfhatx },\bfh(\bfx))
\  ,
$\refstepcounter{equation}\label{eqn:ObserverSystem}\hfill$(\theequation)$\\[0.5em]
with a state $\bfhatx$ living in the same manifold as the system state $\bfx$ 
to be estimated, such that the zero estimation error set
\\[0.5em]\null \hfill $\displaystyle
\mathcal{A}= \left\{(\bfx,\bfhatx)\in\bfRR^n\times\bfRR^n:\,  \bfx=\bfhatx\right\}
$\refstepcounter{equation}\label{LP233}\hfill$(\theequation)$\\[-0.3em]
\item
is forward invariant,
\item
solutions to \eqref{eqn:Plant1}-\eqref{eqn:ObserverSystem} converge to it -- a property that is guaranteed when a
Riemannian distance between true and estimated 
state strictly decreases,
\item
and has an infinite gain margin (see Definition \ref{def2}).
\end{enumerate}
\stopmodif


\startmodif
There is a large corpus of contributions dedicated to this problem in the literature. The case when the distance is Euclidean,
\stopmodif
in appropriate coordinates, has been deeply investigated, 
giving rise to the well-known Luenberger observer \cite{Luenberger.64.TME},
Kalman filter \cite{Kalman.Bucy.61.JBE}, and high-gain observer \cite{Khalil.Praly.13.IJRNC}.
The case where the distance is
derived from a Riemannian metric given by the dynamics or by the  manifold the state belongs to
was studied in \cite{Aghannan.Rouchon.03,Bonnabel.07,Bonnabel.10}.
The design procedure proposed
there exploits  properties of the given
metric to establish  local convergence of the distance to zero, specifically, via an appropriate choice of coordinates or
modification of the metric.
\startmodif
In this paper, which continues from \cite{57} and
\cite{127},
the Riemannian metric is not given but properly
chosen as part of the design of the observer.
\stopmodif

\subsection{Motivation}
The choice of the Riemannian metric mentioned above is dictated by the following result
reported in Theorem 3.3 and Lemma 3.6 in \cite{57}
(see also \cite{Sanfelice.Praly.13.ArXiv}). We state it slightly differently but keep the original numbering 
of the conditions.
\IfReport{%
\startmodif
A proof of this version is similar to the one of Theorem \ref{thm5} in the \complement \ref{complement29}.
\stopmodif
}{%
}%
The symbols and notions -- e.g.,
\textit{complete},
\textit{Riemannian metric},
$\bfd_1^2\wpunbf$,
\textit{geodesic},
and 
\textit{Riemannian distance} --
are defined in  Appendix~\ref{sec:Glossary}.

\par\vspace{1em}
\startmodif
\begin{theorem}
\label{thm1part1}
Given $C^3$ functions $\bff$ and $\bfh$, suppose there exists a complete $C^3$
Riemannian metric%
\glos{\ref{item:RiemannianMetric}}\glos{\ref{Glos8}}\hskip 3pt
$\,  \bfP$ and 
%
an open subset $\Ouv$ of $\bfRR^n$
such that
\begin{enumerate}
\item[ A2~:]
There exist a continuous function
%
$\rho : \Ouv \to \RRgeq $
and a strictly positive real number $q$
satisfying\footnote{
\color{.}
Component-wise the inequality (\ref{3}) is
\IfReport{%
\\[0.7em]\null \hfill $\displaystyle 
\sum_{\indxc}\left[
\frac{\partial P_{\indxa\indxb}}{\partial x_\indxc}(x) f_\indxc(x)
+
P_{\indxa\indxc}(x) \frac{\partial f_\indxc}{\partial x_\indxb}(x)
+
P_{\indxb\indxc}(x) \frac{\partial f_\indxc}{\partial x_\indxa}(x)
\right]
\leq \rho (x) \sum_\indyi\frac{\partial h_\indyi}{\partial x_\indxa}(x)\frac{\partial h_\indyi}{\partial x_\indxb}(x)
- q P_{\indxa\indxb}(x)
\  ,
$\hfill \null \\[0.3em]
}{%
\\[0.7em]$\displaystyle 
\sum_{\indxc}\left[
\frac{\partial P_{\indxa\indxb}}{\partial x_\indxc}(x) f_\indxc(x)
+
P_{\indxa\indxc}(x) \frac{\partial f_\indxc}{\partial x_\indxb}(x)
+
P_{\indxb\indxc}(x) \frac{\partial f_\indxc}{\partial x_\indxa}(x)
\right]
$\hfill\null\\\null\hfill$\displaystyle
\leq \rho (x) \sum_\indyi\frac{\partial h_\indyi}{\partial x_\indxa}(x)\frac{\partial h_\indyi}{\partial x_\indxb}(x)
- q P_{\indxa\indxb}(x)
\  ,
$\\[-0.3em]
}
and the observer equation (\ref{eqn:GeodesicObserverVectorField}) is
\\[0.3em]\null \hfill $\displaystyle 
\dot{\hatx }_\indxa = f_\indxa(\hatx) - k_E(\hatx)
\sum_{\indxb} [P(\hatx)^{-1}]_{\indxa\indxb}
\sum_\indyi\frac{\partial h_\indyi}{\partial x_\indxb}(\hat x)
\frac{\partial \wp}{\partial y_{1\indyi}}(h(\hatx), y)
\  .
$\hfill \null 
}
\IfReport{%
\\[0.7em]\null \hfill
}{%
\\[0.7em]\null \hskip -\labelwidth
}
$\displaystyle 
\begin{array}[t]{@{}c@{}c@{}c@{}l@{}}
\mathcal{L}_\bff \bfP(\bfx)
&
\,  \leq \,   \rho (\bfx)\,  
&
\bfdh (\bfx)\otimes \bfdh(\bfx)
&
\,  - \,  \qlower\,  \bfP(\bfx)
%
\quad \forall \bfx\in\Ouv
\,.
\\
\mbox{\glos{\ref{Glos4}}}
&
&
\mbox{\glos{\ref{Glos1}}}
&
\end{array}%
$\refstepcounter{equation}\label{3}%
\hfill$(\theequation)$
\item[ A3~:]
There exists a $C^3$ function
$
\wpunbf : (\bfy_1,\bfy_2)\in \bfRR^p\times\bfRR^p  \mapsto \wpunbf(\bfy_1,\bfy_2)\in\RRgeq
$
satisfying\glos{\ref{Glos1}}
\begin{equation}
\label{LP65}
\wpunbf (\bfy,\bfy)= 0
\  ,\quad
\bfd_1^2\wpunbf(\bfy,\bfy) > 0
%
\qquad \forall \bfy\in\bfh(\Ouv)
\end{equation}
and, for any  geodesic
\IfReport{\glos{\ref{Glos7}}}{\!\!\!} $\bfgamma ^*$, taking values
%
in $\Ouv$
and minimizing on the maximal interval $(s_1,s_2)$,
we have%
\IfTwoCol{%
\\[0.7em]$\displaystyle 
\frac{d}{ds}\left\{\wpunbf  (\bfh(\bfgamma ^*(s)),\bfh(\bfgamma ^*(s_3)))\right\}\: >\: 0
\qquad \forall s\in (s_3,s_4),
$\hfill\null\\[0.3em]\null\qquad $\displaystyle
\forall s_3,s_4 \in (s_1,s_2): 
$\refstepcounter{equation}\label{LP99}\hfill$(\theequation)$
\\[0.3em]\null\hfill$\displaystyle
s_3<s_4 \quad \& \quad 
\bfh(\bfgamma ^*(s_3)) \neq \bfh(\bfgamma ^*(s_4))
\  .
$
\\[-1em]
}{%
\\[0.7em]$\displaystyle 
\frac{d}{ds}\left\{\wpunbf  (\bfh(\bfgamma ^*(s)),\bfh(\bfgamma ^*(s_3)))\right\}\; >\; 0
\qquad \forall s\in (s_3,s_4),
$\refstepcounter{equation}\label{LP99}\hfill$(\theequation)$
\\[0.3em]\null\hfill$\displaystyle
\forall s_3,s_4 \in (s_1,s_2):\  
s_3<s_4 \quad \& \quad 
\bfh(\bfgamma ^*(s_3)) \neq \bfh(\bfgamma ^*(s_4))
\  .
$
\\[1em]
}
\end{enumerate}
\color{.}
Under these conditions, for any strictly positive real number $E$
and any closed subset $\mathcal{C}$ of $\Ouv$ with a nonempty interior,
there exists a continuous function 
$\kunbf _E^* :
\mathcal{C}\to \RR_{>0} $ such that
\begin{list}{}{%
\parskip 0pt plus 0pt minus 0pt%
\topsep 0.5ex plus 0pt minus 0pt%
\parsep 0pt plus 0pt minus 0pt%
\partopsep 0pt plus 0pt minus 0pt%
\itemsep 0.5ex plus 0pt minus 0pt
\settowidth{\labelwidth}{--}%
\setlength{\labelsep}{0.5em}%
\setlength{\leftmargin}{\labelwidth}%
\addtolength{\leftmargin}{\labelsep}%
}
\item[--]
for any continuous function 
$\kunbf _E:\mathcal{C} \to \RR$
satisfying
$$
\kunbf _E(\bfhatx )\geq \kunbf _E^*(\bfhatx )
\qquad \forall \bfhatx \in \mathcal{C}
\  ,
$$
\item[--]
for the observer given by
\begin{equation}
\label{eqn:GeodesicObserverVectorField}
\dot{\bfhatx }\;=\; \bfF({\bfhatx },y)\; 
:=
\; \bff({\bfhatx }) \;-\;
\kunbf _E(\bfhatx )\,  \bfgrad_\bfP [\wpunbf \circ\bfh](\bfhatx ,\bfy)
\  ,
\end{equation}
where, for each $\bfy$, $\bfhatx\mapsto \bfgrad_\bfP [\wpunbf \circ\bfh](\bfhatx ,\bfy)$ is the Riemannian gradient%
\glos{\ref{Glos2}}
(with respect to $\bfhatx$) of the function $\bfhatx\mapsto \wpunbf (\bfh(\bfhatx),\bfy)$, 
\item[--]
and, for all $x$ and $\hat x $ in $\mathcal{C}$
satisfying
\begin{equation}
\label{eqn:Basin}
d(\bfhatx ,\bfx) <  E\  ,
\end{equation}
where $d$ denotes the Riemannian distance
\IfReport{\glos{\ref{Glos7}}}{\!\!} induced by $\bfP$,
and linked by
a minimizing normalized geodesic $\gamma ^*$ satisfying
$$
\bfx=\bfgamma ^*(0)
\  ,\quad 
\bfhatx =\bfgamma ^*(\hat s)
\  ,\quad 
\gamma ^*(s)\in \mathcal{C}
\quad \forall s\in [0,\hat s]
\  ,
$$
\end{list}
we have\footnote{
\IfReport{%
$\Did d(\bfhatx ,\bfx)$ is the upper right Dini derivative along the solution, i.e.\\[0.7em]%
}{%
}
$\displaystyle 
\Did d(\bfhatx ,\bfx)\;=\; 
\limsup_{t\searrow 0} \frac{d(\bfhatX ((\bfhatx ,\bfx),t),\bfX(\bfx,t))-d(\bfhatx ,\bfx)}{t}
$
}
\begin{equation}
\label{LP63}
\Did d(\bfhatx ,\bfx)
\; \leq \;  \displaystyle -
\frac{\qlower}{4} 
\,  d(\bfhatx ,\bfx)
\  .
\end{equation}
\end{theorem}

The consequence of (\ref{LP63}) in Theorem~\ref{thm1part1} is that, as long as the assumptions are satisfied, the distance between the
true state $\bfx$ and the estimated state $\bfhatx$ is exponentially decreasing. Among the assumptions is the 
fact that
\stopmodif
the Riemannian metric $\bfP$
must satisfy two key conditions that are of complete different nature.

The first condition, referred to as Condition A2, named 
{\it strong differential detectability with respect to the metric $\bfP$}, is related to
detectability of \eqref{eqn:Plant1}, and, as such pertains to control theory.
It involves the right-hand side $\bff$,
the output map $\bfh$, and the 
Riemannian metric $\bfP$ to be chosen.
\startmodif
Geometrically, it says that the flow generated by \eqref{eqn:Plant1} contracts along directions that are tangent to  
the level sets of $h$.
\stopmodif
In \cite{57}, we show that a weak form of this differential detectability property is necessary for the existence of 
an observer with state $\bfhatx$ in the same space as the system state $\bfx$ and
making the set $\mathcal{A}$ in (\ref{LP233}) invariant and a Riemannian distance between true state and its estimate to decrease exponentially
when evaluated along solutions. We show also that uniform detectability of the linearization of \eqref{eqn:Plant1} along each of its solutions
is necessary for Condition A2 to hold.
In \cite{127}, we present techniques for the design of the Riemannian metric $\bfP$ for given functions $\bff$ and $\bfh$ so that 
Condition A2 holds. 
We show that such a design is possible when \eqref{eqn:Plant1} satisfies any of the following properties:
\begin{itemize}
\item[i)] Strongly infinitesimally reconstructible (see \cite[Definition 3.1]{127}) in the sense that each time-varying linear system
resulting from the linearization along a
solution to the system \eqref{eqn:Plant1} satisfies a uniform reconstructibility property;
\item[ii)]
Strongly differentially observable, in the sense that the state to output 
derivatives
\startmodif
mapping
\stopmodif
is an injective immersion (see \cite[Proposition 4.4]{127});
\startarchive
or
\item[iii)]
An Euler-Lagrange system, the Lagrangian of which is quadratic in the generalized velocities
(see \cite[Section V]{127}).
\stoparchive
\end{itemize}

The second condition, referred to as Condition A3, 
says roughly that, if,
\startmodif
along a geodesic,
\stopmodif
the distance between the true state $\bfx$
 and its estimate $\bfhatx$  reduces then the same holds between the corresponding
true (measured) output $\bfy$ and its estimate $\bfh(\bfhatx)$.
Following \cite[Definition 6.2.3]{Rapcsak.97}, a function $\bfh$ satisfying Condition A3 is known as being
{\it geodesically monotone}. This property involves 
the output map $\bfh$ and the 
Riemannian metric $\bfP$, but not $\bff$.
In \cite[Proposition A3]{57},
we established that this property implies that the level sets of $\bfh$ are strongly convex, which is a 
property that is typically exploited in optimization theory; see, e.g., \cite{Udriste.94}. Actually it is needed only to allow $E$ in (\ref{eqn:Basin})
to be arbitrary -- in this way, making the result semiglobal. Indeed, in \cite{127} we show that, without it,
Condition A2 alone guarantees the existence of a locally
(i.e., $E$ is imposed and small enough)
convergent observer and a locally convergent reduced order 
observer. See \cite[Propositions 2.4 and 2.8]{127}.


\subsection{Contributions}
\label{sec20}
\startmodif
Parallel to \cite{57}, dedicated to the study and design of a metric satisfying the strong 
differential detectability property of Condition A2, we 
devote this paper to the geodesic monotonicity property in Condition A3.
Our contributions are as follows:
\begin{enumerate}
\item 
In Section~\ref{sec:WhyA3}, we show that Condition A3 is equivalent to the infinite gain margin property (see Definition \ref{def2}) when the correction (or innovation) term in the observer is of gradient type
as in (\ref{eqn:GeodesicObserverVectorField}).
\item 
In Section \ref{sec:A3sufficient} and Proposition~\ref{prop12}, we show that
Condition A3 holds if $\bfh$ is a (geodesically)  affine function.
\item
In Section \ref{sec:necessary}, we give necessary conditions for Condition A3 to hold. In particular, we reveal
the key role played by the second fundamental form of the function $\bfh$ (see Definition \ref{def3}),
and the fact that $\bfh$ is a Riemannian 
submersion (see Definition \ref{def4}), with totally geodesic (see Definition \ref{def1})
level sets and an integrable orthogonal distribution (see Definition \ref{def5}).
\item
In Corollary \ref{cor1}, we propose a
test to check if Condition A3 holds.
The conditions to check depend on symbolic 
computations involving $\bfh$, $\bfP$, and their differentials. 
\item
In Theorem~\ref{prop17}, we present a systematic way to construct a metric $\bfP$
satisfying Condition A3.
\item
In Section \ref{sec:A2andA3}, we illustrate, via examples, situations in which both Conditions A2 and A3 hold. 
\startarchive
\item
Finally, in Section \ref{sec14}, we propose two general and promising ways of facilitating the satisfaction of 
Conditions A2 and A3 simultaneously,
via an immersion into an input-dependent system in Section~\ref{sec15} and via dynamic extension in Section~\ref{sec16}.
\stoparchive %
\end{enumerate}
\stopmodif

Because of space limitations,
details behind routine (but
\startmodif
sometimes
\stopmodif
lengthy) computations involved in the proofs are omitted. They can be found in 
\cite{Sanfelice-Praly.III-long}, along with additional material not referred to in this paper.

\startmodif
The reading of this paper requires the knowledge of well established concepts and results from
Riemannian geometry, in particular, on Riemannian submersions
and on optimization on Riemannian manifolds. \cite{Vilms,Nore,Garcia-Kupeli,Udriste.94}
are relevant references on such topics.
\stopmodif

As a difference to our previous work,
coordinates play a significant role in the
solution we have found to design a Riemannian metric $\bfP$ satisfying Condition A3.
Our first step is to introduce our  notation involving coordinates and  related basic assumptions.
\startarchive
See also the glossary in Appendix \ref{sec:Glossary}.
\stoparchive %

\section{Preliminaries and Notation}

\label{sec17}
%
Symbols in bold style represent coordinate-free objects. In particular, $\bfx$ and $\bfy$ are points 
in a manifold, $\bff$ is a vector field on a tangent bundle,
$\bfh$ is a function between manifolds, $\bfP$ is a symmetric $2$-covariant tensor, etc.

Once coordinates, defined below,
$x$ and  $y$, with letters in normal style type,
have been chosen for $\bfx$ and $\bfy$, we can express the corresponding objects:
$f(x)$ for the value of $\bff$ at the point $\bfx$, $h(x)$ for
$\bfh(\bfx)$, $P(x)$ for $\bfP(\bfx)$, etc.

The writing of the system in (\ref{eqn:Plant1}) and of the observer in (\ref{eqn:GeodesicObserverVectorField})
with the state, the output, and the functions in bold style means that the coordinates used therein play no role,
namely, the expression of both the plant and the observer dynamics are coordinate free.
However, the use of
normal style type
for $\kunbf_E$ and $\wpunbf$ in the observer (\ref{eqn:GeodesicObserverVectorField}),
and for $d$
in \eqref{eqn:Basin} and in \eqref{LP63}
is to indicate that a change of coordinates in $\RRgeq$ is not allowed, i.e.,
these are scalar invariant functions taking values
\startmodif
in
\stopmodif
$\RRgeq$.

As a general rule,
when not used as indices, the symbols
\startmodif
$x$, $P$, $d$, $\gamma $, $\coordxm$, $\changex$ $\ldots$ are used for the 
$\bfx$-manifold $\bfRR^n$, while the
symbols $y$, $\Py$, $\dy$, $\delta $, $\coordym$, $\changey$, $\ldots$, following in the alphabetical order,
are used for the $\bfy$-manifold $\bfRR^p$.
For example, the expression of the Riemannian norm of the velocity of a path and of the distance between two points are denoted
$\frac{d\gamma}{ds}(s)^\top P(\gamma (s)) \frac{d\gamma}{ds}(s)$ and $d(x_a,x_b)$ in the $\bfx$-manifold, 
while, in the $\bfy$-manifold, they are denoted
$\frac{d\delta }{ds}(s)^\top Q(\delta  (s))\frac{d\delta }{ds}(s)$ and $e(y_a,y_b)$, respectively. 
\stopmodif

As usual (see \cite[p. 40]{Lee.13})
we call {\it coordinate chart} a triple, respectively, $\coordx$ for the $\bfx$-manifold
$\bfRR^n$ and $\coordy$ for the $\bfy$-manifold $\bfRR^p$, such that
$\coordxd $ and $\coordyd $, called {\it coordinate domains}, are open subsets of, respectively, $\bfRR^n$ and $\bfRR^p$ and
$\coordxm:\coordxd\to \RR^n$ and $\coordym:\coordyd\to \RR^n$ are 
homeomorphisms, called {\it coordinate maps}, satisfying
$$
x\;=\; \coordxm(\bfx)\quad \forall \bfx\in\coordxd
\quad ,\qquad 
y\;=\; \coordym(\bfy)\quad \forall \bfy\in\coordyd
\  .
$$
where $x=(x_\indxa, x_\indxb,\dots)$ in $\RR^n$ and $y=(y_\indyi,y_\indyj,\ldots)$ in  $\RR^p$
are called {\it local coordinates}. 
Roman letters $\indxa$, $\indxb$, \ldots, are used as indices for $x$ and run 
over the range $\{1,2,\ldots,n\}$ and roman letters $
\indyi$, $\indyj$, \ldots, are used as indices for $y$ and run over the range $\{1,2,\ldots, p\}$.

Denoting the family of (as many times as necessary) continuously differentiable functions 
$C^s$, the coordinate charts
$\coordx$ and $\coordy$ are assumed to assure a $C^s$ structure,
 in the sense that, for any two coordinate charts
$(x_1,\coordxd _1,\coordxm_1)$ and $(x_2,\coordxd _2,\coordxm_2)$, $\coordxm_1\circ\coordxm_2^{-1}$ is a $C^s$ diffeomorphism.
A coordinate chart $\coordx$ is said to be {\it a coordinate chart around $\bfx_0$} if $\bfx_0$ belongs to $\coordxd$. 

As an illustration of 
these definitions, given coordinate charts $\coordx$ around $\bfx_0$ and $\coordy$ around 
$\bfh(\bfx_0)$, with $\bfh(\coordxd)$ contained in $\coordyd$, the expression $h$ of the function $\bfh$ is
$$
h(x)\;=\; \coordym(\bfh(\coordxm^{-1}(x)))
\qquad \forall x\in \coordxm(\coordxd).
$$

Given coordinate charts $\coordx$ and $\coordy$, and $C^s$ diffeomorphisms $\changex:\coordxm(\coordxd)\to \RR^n$ and  $\changey:\coordym(\coordyd)\to \RR^p$, 
we obtain new coordinates charts 
$\barcoordx$ and $\barcoordy$, where
$$
\begin{array}{rcl@{\  ,\quad }rcl@{\  ,\quad }rcl}
\barcoordxp&=& \changex (\coordxp)
&
\barcoordxd&=& \coordxd
&
\barcoordxm&=& \changex \circ \coordxm
\ ,
\\[0.5em]
\barcoordyp&=& \changey (\coordyp)
&
\barcoordyd&=& \coordyd
&
\barcoordym&=& \changey \circ \coordym
\  .
\end{array}
$$
Then we have, for example, the following relationships between the expressions of the vector field $\bff$, the function $\bfh$, and the 
symmetric $2$-covariant tensor $\bfP$:%
\IfTwoCol{%
\begin{equation}
\label{LP186}
\renewcommand{\arraystretch}{1.3}
\begin{array}{rcl}
\bar f(\changex(x))&=& \frac{\partial \changex}{\partial x}(x) \,  f(x)
\  ,
\\
\bar h(\changex(x))&=& \changey(h(x))
\  ,
\\
\frac{\partial \changex}{\partial x}(x)^\top \bar P(\changex(x))\,  \frac{\partial \changex}{\partial x}(x)&=& P(x) \ ,
\end{array}
\end{equation}
}{%
\begin{equation}
\label{LP186}
\bar f(\changex(x))\,=\, \frac{\partial \changex}{\partial x}(x) \,  f(x)
\  ,\quad 
\bar h(\changex(x))\,=\, \changey(h(x))
\  ,\quad 
\frac{\partial \changex}{\partial x}(x)^\top \bar P(\changex(x))\,  \frac{\partial \changex}{\partial x}(x)\,=\, P(x) \ ,
\end{equation}
}%
where $(f,h,P)$ and $(\bar f,\bar h ,\bar P)$ are the expressions of $(\bff,\bfh,\bfP)$ in the corresponding 
coordinates $(\coordxp,\coordyp)$ and $(\barcoordxp,\barcoordyp)$, respectively. 

More insight can be gained in the context of observers
when $y$, coordinates for $\bfh(\bfx)$, can be used as part of coordinates 
for $\bfx$. This motivates the following assumption.
\begin{hypothesis}
\label{H1}
The function
$\bfh:\bfRR^n\to \bfRR^p$ is a submersion\footnote{%
This means, the set
\  $\displaystyle 
\Ouv =  \left\{\bfx\in\bfRR^n\,  :\: \mbox{\rm Rank}\left(\bfdh(\bfx)\right)=p
\right\}
$,\  
is open,
where $\bfdh$ is the differential of $\bfh$.
\IfReport{%
See \glos{\ref{Glos1}}.
}{
}%
} on a set $\Ouv$.
\end{hypothesis}
When this assumption holds, we have the following result. See \cite[Theorem I.2.1(2)]{Sakai.96}.
\par\vspace{0.5em}

\begin{theorem}[Local Submersion Theorem]
\label{thm4}
If $\bfh:\bfRR^n\to \bfRR^p$ is a submersion on a set $\Ouv$ then
$\bfh(\Ouv)$ is an open set and, for any point $\bfx_0$ of $\Ouv $, there exist
 a coordinate chart $\coordx$ around $\bfx_0$, 
a coordinate chart $\coordy$ around $\bfh(\bfx_0)$, with $\coordyd$ containing
$\bfh(\coordxd)$, and a submersion
$h^\comp :\coordxm(\coordxd)\to \RR^{n-p}$ on $\coordxm(\coordxd)$ such that 
$\changex =(h,h^\comp ):\coordxm(\coordxd)\to h(\coordxm(\coordxd))\times \RR^{n-p}$
is a $C^s$ diffeomorphism. Consequently, $\coordyxr$,
with $\coordyxrm=\changex\circ\coordxm$, 
is a coordinate chart\footnote{
The subscript $\coordyd$ in $\coordyxrm$ is introduced to emphasize that this particular coordinate chart 
involves, in its construction, a 
coordinate chart for the $\bfy$-manifold $\bfRR^p$.
}
around $\bfx_0$.
\end{theorem}
\par\vspace{0.5em}
In this statement,
$\xrond=(\xrond_\indxra,\xrond_\indxrb,\ldots)$ is in the open set
$h^\comp (\coordxm(\coordxd))$.
The greek letters $\indxra$, $\indxrb$, \ldots, used as indices, run
over the range $\{1,2,\ldots,n-p\}$.
A coordinate chart $\coordyxr$ is, by nature, paired with the coordinate chart $\coordy$, with the same $y$ 
and, without loss of generality,
\IfTwoCol{%
\begin{eqnarray*}
\displaystyle 
\bfh(\coordyxrd)\;=\; \coordyd\ , \ \ y\;=\; \coordym(\bfh(\coordyxrm^{-1}(y,\xrond)))
\quad \forall (y,\xrond)\in \coordyxrm(\coordyxrd)
\  .
\end{eqnarray*}
}{%
\begin{eqnarray*}
&\displaystyle 
\bfh(\coordyxrd)\;=\; \coordyd
\  ,\\&\displaystyle 
y\;=\; \coordym(\bfh(\coordyxrm^{-1}(y,\xrond)))
\qquad \forall (y,\xrond)\in \coordyxrm(\coordyxrd)
\  .
\end{eqnarray*}
}%
When $(y,\xrond)$ are used as coordinates for $\bfx$, we decompose the expression 
$P$ of $\bfP$ as
\begin{equation}
\label{LP185}
P(y,\xrond)\;=\; \left(\begin{array}{cc}
P_{yy}(y,\xrond)
&
P_{y\xrond}(y,\xrond)
\\
P_{\xrond y}(y,\xrond)
&
P_{\xrond\xrond}(y,\xrond)
\end{array}\right)
\  .
\end{equation}
In such a case, changes of coordinates take the particular form
$$
(\bar y,\bar \xrond)\;=\; (\changey (y),\changexr (y,\xrond))
$$
where $\changey:\RR^p\to \RR^p$ is a $C^s$ diffeomorphism and
$\changexr :\RR^p\times \RR^{n-p}\to \RR^{n-p}$ is  such that $(\changey,\changexr )$
is a $C^s$ diffeomorphism.
\startarchive%
See \complement \ref{complement45}.%
\stoparchive%

Throughout the paper, we use objects from differential geometry.
\IfReport{%
These objects are discussed in the glossary in Appendix \ref{sec:Glossary} 
for which \glos{x} are pointers.%
}{%
Some of these objects are explained in the glossary in Appendix \ref{sec:Glossary}.
}

\section{On Condition A3}
\label{sec13}

\subsection{%
Is Condition A3 necessary?
}
\label{sec:WhyA3}
To answer this question, we invoke the infinite gain margin property, which, in the context of our proposed observer, is 
defined as follows.

\begin{definition}[Infinite gain margin {\protect \cite[Definition 2.8]{57}}]
\label{def2}
Let $\bfP$ be a Riemannian metric on $\bfRR ^n$ and
\startmodif
%
$\Ouv$ an open subset of $\bfRR^n$.
\stopmodif
An observer 
\begin{equation}
\label{LP58}
\dot{\bfhatx }\;=\; \bff(\bfhatx ) - \mathfrakbf{C}(\bfhatx ,\bfy)
\  ,
\end{equation}
where $\mathfrakbf{C}$ is a correction term, is said to have an infinite gain margin on
\startmodif
%
$\Ouv$
\stopmodif
with respect to $\bfP$ if, for any geodesic $\bfgamma ^*$ 
taking values in
\startmodif
%
$\Ouv$
\stopmodif
and minimizing on the maximal interval
$(s_1,s_2)$, we have%
\IfTwoCol{%
\\[0.7em]
$
\renewcommand{\arraystretch}{1.5}
\begin{array}{@{}l@{\  }c@{}}
\mbox{either }&
\displaystyle 
\frac{d\bfgamma  ^*}{ds}(s_4)^\top 
\bfP(\bfgamma   ^*(s_4) ) \,  \mathfrakbf{C}(\gamma   ^*(s_4) ,\bfh(\bfgamma   ^*(s_3) ))\; > \; 0
\\
\mbox{or } &
\displaystyle
\mathfrakbf{C}(\bfgamma   ^*(s_4) ,\bfh(\bfgamma   ^*(s_3) ))\; = \; 0
\end{array}
$\refstepcounter{equation}\label{LP61}\hfill$(\theequation)$
\\[0.5em]\null\hfill$\displaystyle
\qquad \forall s_3,s_4\in (s_1,s_2):\,  s_3 < s_4\ ,
$\\[-0.5em]
}{
\begin{equation}
\label{LP61}
\renewcommand{\arraystretch}{1.5}
\begin{array}{@{}l@{\quad  }c@{}}
either&
\displaystyle 
\frac{d\bfgamma  ^*}{ds}(s_4)^\top 
\bfP(\bfgamma   ^*(s_4) ) \,  \mathfrakbf{C}(\bfgamma   ^*(s_4) ,\bfh(\bfgamma   ^*(s_3) ))\; > \; 0
\\
or &
\displaystyle
\mathfrakbf{C}(\bfgamma   ^*(s_4) ,\bfh(\bfgamma   ^*(s_3) ))\; = \; 0
\end{array}
\qquad \forall s_3,s_4\in (s_1,s_2):\,  s_3 < s_4\ .
\end{equation}
}%
\end{definition}

From the first
order variation formula (see, for instance, 
\cite[Theorem 6.14]{Spivak.79} or
\cite[Theorem 5.7]{Isac.Nemeth.08}),
and properties of Riemannian distances and geodesics,
the upper right-hand Dini derivative
of the distance between
$\bfx = \bfgamma^*(s_3)$ and $\bfhatx = \bfgamma^*(s_4)$
satisfies
\IfTwoCol{%
\\[0.7em]$\displaystyle 
\Did d(\bfhatx, \bfx) \; \leq \;
-\frac{d\bfgamma ^*}{ds}(s_4)^\top \bfP(\bfgamma ^*(s_4))  \,
\mathfrakbf{C}(\bfgamma   ^*(s_4) ,
\bfh(\bfgamma   ^*(s_3) )
 ))
$\\[0.3em]\null \quad  $\displaystyle
+\left[
\frac{d\bfgamma ^*}{ds}(s_4)^\top \bfP(\bfgamma ^*(s_4))  \,   \bff(\bfgamma   ^*(s_4) ) 
\right.
$\refstepcounter{equation}\label{LP7}\hfill$(\theequation)$\\ \null \hfill  $\displaystyle 
\left.
-
\frac{d\bfgamma ^*}{ds}(s_3)^\top \bfP(\bfgamma ^*(s_3))  \,   \bff(\bfgamma   ^*(s_3) ) 
\right]
\  .
$\\[0.7em]
}{%
\\[1em]$\displaystyle 
\Did d(\bfhatx, \bfx) \; \leq \; 
\left[
\frac{d\bfgamma ^*}{ds}(s_4)^\top \bfP(\bfgamma ^*(s_4))  \,   \bff(\bfgamma   ^*(s_4) ) 
-
\frac{d\bfgamma ^*}{ds}(s_3)^\top \bfP(\bfgamma ^*(s_3))  \,   \bff(\bfgamma   ^*(s_3) ) 
\right]
$\refstepcounter{equation}\label{LP7}\hfill$(\theequation)$\hfill \null \\\null \hfill $\displaystyle 
-\frac{d\bfgamma ^*}{ds}(s_4)^\top \bfP(\bfgamma ^*(s_4))  \,   \mathfrakbf{C}(\bfgamma   ^*(s_4) ,
\bfh(\bfgamma   ^*(s_3) )
 )) 
\  .
$\\[1em]
}%
Hence, when (\ref{LP61}) holds, 
the correction term $\mathfrak{C}$ always contributes to the decrease of the distance between
$\bfhatx $ and $\bfx$. If the contribution of $\mathfrak{C}$ were  to be negative,
then the desired decrease of the distance would have to be provided by 
the dynamics of the system dictated by $\bff$
-- namely, by the term around brackets in \eqref{LP7}.

\begin{remark}
\label{rem1}
Although observers without infinite gain margin do exist, those with infinite gain margin are  quite common.
They are
guaranteed to exist for any system as in (\ref{eqn:Plant1}) belonging to the following family:
\begin{description}
\item[\normalfont\textit{
\startmodif
``Euclidean family'':
\stopmodif
}]
\textit{%
\startmodif
There exists a coordinate chart $\coordx$ for which the expression $h$ of the output is
\begin{equation}
\label{LP235}
h(x)\;=\; Hx
\end{equation}
and the expression $P$ of a 
Riemannian metric $\bfP$ satisfying Condition A2 is a constant matrix.%
\stopmodif
}
\end{description}
Indeed, when $h$ satisfies this property,
(\ref{LP61}) simplifies to\footnote{%
The fact that $P$ is constant allows to express any minimal geodesic $\gamma^*$
between $x$ and $\hat x$
as a straight line connecting $x = \gamma^*(s_3)$ to $\hat x = \gamma^*(s_4)$.
}
$$
(\hat x- x)^\top P\,  \mathfrakbf{C}(\hat x , H x)\; >\; 0
\qquad \forall (\hatx,x) \in  \coordxm(\coordxd)^2:\hatx  \neq x 
$$
and it suffices to pick
$$
\mathfrakbf{C}(\hat x , Hx)\;=\; \kunbf_E(\hat x) P^{-1} H^\top (H\hatx - Hx)
\  ,
$$
\!\! where 
$\kunbf _E:\bfRR ^n\to \RR$
is a continuous function.
\end{remark}

In the particular case where the correction term in the observer is of gradient type
as in \eqref{eqn:GeodesicObserverVectorField},
the infinite gain margin property is equivalent to Condition A3. 
\startmodif
This equivalence results from
\stopmodif
the following
consequence of \eqref{eqn:GeodesicObserverVectorField}: 
\IfTwoCol{%
\\[0.7em]$\displaystyle 
\frac{d}{ds}\wp (h(\gamma   ^*(s)) ,h(\gamma   ^*(s_3) ))
$\hfill\null\\\null\hfill$
\begin{array}{@{}cl@{}}
=&\displaystyle 
\frac{\partial \wp }{\partial y_1}(h(\gamma ^*(s)),h(\gamma ^*(s_3)))
\frac{\partial h}{\partial x}(\gamma ^*(s))\frac{d\gamma ^*}{ds}(s)
\  ,
\\[0.5em]
 = &
\frac{1}{k_E(x)}\,\displaystyle 
\startmodif
\frac{d\gamma  ^*}{ds}(s)^\top
P(\gamma   ^*(s) ) 
\mathfrak{C}(\gamma   ^*(s) ,h(\gamma   ^*(s_3) ))^\top
\stopmodif
\,  ,
\end{array}\\[0.7em]
$
}{%
\begin{eqnarray*}
\frac{d}{ds}\wp (h(\gamma   ^*(s)) ,h(\gamma   ^*(s_3) ))&=&
\frac{\partial \wp }{\partial y_1}(h(\gamma ^*(s)),h(\gamma ^*(s_3)))
\frac{\partial h}{\partial x}(\gamma ^*(s))\frac{d\gamma ^*}{ds}(s)
\  ,
\\
& = &
 \startmodif
\frac{1}{k_E(x)}\,\displaystyle 
\frac{d\gamma  ^*}{ds}(s)^\top
P(\gamma   ^*(s) ) 
\mathfrak{C}(\gamma   ^*(s) ,h(\gamma   ^*(s_3) ))^\top
\stopmodif
\  ,
\end{eqnarray*}
}%
written with the coordinate chart $\coordx$ where
$$
x\;=\; \gamma   ^*(s_3)
\quad ,\qquad 
\hat x\;=\; \gamma   ^*(s)
\  .
$$
We conclude that, if we want an observer of gradient type with
a correction term contributing to the decrease of the distance between
the state and its estimate, Condition A3 must hold.

\subsection{A Sufficient Condition for the Satisfaction of Condition A3}
\label{sec:A3sufficient}
\startmodif
When the output function $\bfh$ is given, the function $\wp$ and the metric $\bfP$ are the only objects remaining at our disposal
to  satisfy Condition A3.
\stopmodif
\subsubsection{Choice of the Function $\wp$}
\IfTwoCol{%
We need the function $\wpunbf$ to quantify
the ``gap'' between $\bfh(\bfhatx)$ and $\bfh(\bfx)$.
This motivates us to equip the $\bfy$-manifold $\bfRR^p$ with a $C^s$ complete Riemannian metric $\bfPy$ and obtain a distance,
}{%
To satisfy condition (\ref{LP65}) we can view the function $\wpunbf$ as a gauge function (in the sense
of \cite{Massera-Schaaffer}). It quantifies the ``gap'' between $\bfh(\bfhatx)$ and $\bfh(\bfx)$.
This motivates us for considering a Riemannian distance in the $\bfy$-manifold $\bfRR^p$.
Precisely we equip $\bfRR^p$ with a $C^s$ complete Riemannian metric $\bfPy$. This metric gives rise to a 
distance,
}%
which we denote $\dy $. 
Then, we choose the function $\wp$ as
\begin{equation}
\label{LP109}
\wpunbf(\bfy_1,\bfy_2)\;=\; \dy (\bfy_1,\bfy_2)^2\  .
\end{equation}
To have the required smoothness property on $\wp$, we need an extra property for $\bfPy$, 
as the following lemma states.

\begin{lemma}
\label{lem4}
Assume the complete metric $\bfPy$ is such that any piece of geodesic $\bfgammay $ is minimizing.
Then the function $\wp$ defined in (\ref{LP109}) is
$C^s$ and, for any $\bfy_1$, any coordinate chart $\coordy$ around $\bfy_1$,
and any $\bfy_2$ in $\coordyd$
that is linked to $\bfy_1$ by a 
(minimizing) geodesic $\bfgammay ^*$  with values in $\coordyd$
on $(s_1,s_2)$,
\startmodif
i.e. we have
\stopmodif
\IfReport{%
$$
\gammay ^*(s_1)=y_1
\quad ,\qquad   \gammay ^*(s_2)=y_2
\quad ,\qquad 
\gammay ^*(s) \in \coordym(\coordyd)
\quad   
\forall s \in (s_1,s_2)
\  ,
$$
}{%
$$
\gammay ^*(s_1)=y_1
\  ,\quad  \gammay ^*(s_2)=y_2
\  ,\quad 
\gammay ^*(s) \in \coordym(\coordyd)
\   
\forall s \in (s_1,s_2)
\,  ,
$$
}%
we obtain
\begin{equation}
\label{LP110}
\renewcommand{\arraystretch}{2.3}
\begin{array}{c}
\displaystyle 
\left.\frac{\partial ^2\wp}{\partial y_1^2}(y_1,y_2)\right|_{y_2=y_1}\;=\; 2\,  \Py (y_1)
\  ,
\\
\displaystyle 
\frac{\partial \dy^2 }{\partial y_1} (y_1,y_2)^\top\;=\; 2\,  
\frac{\Py (y_1)\frac{d\gammay ^*}{ds}(s_1) \dy (y_1,y_2)}{\sqrt{\frac{d\gammay ^*}{ds}( s_1) ^\top \Py (\gammay ^*( s_1))\frac{d\gammay ^*}{ds}( s_1)}}
\  .
\end{array}
\end{equation}
\end{lemma}
\IfReport{%
\begin{proof}
The condition implies that whatever $\bfy_1$ is in the $\bfy$-manifold $\bfRR^p$, no geodesic emanating from this point has a cut 
point. Then the claim follows from \cite[Proposition III.4.8]{Sakai.96}.
\startarchive%
See \complement \ref{complement26} for more details.
\stoparchive%
\end{proof}
}{%
\begin{proof}
The claim follows from \cite[Proposition III.4.8]{Sakai.96} since, whatever $\bfy_1$ is in the
$\bfy$-manifold $\bfRR^p$, no geodesic emanating from this point has a cut 
point. 
\startarchive%
See \complement \ref{complement26} for more details.
\stoparchive%
\end{proof}
}
\IfReport{%
In light of the statements above, some of the results in this paper assume the following
property.
}{%
With what precedes, we consider the following assumption.
}%
\begin{hypothesis}
\label{H2}
The $\bfy$-manifold $\bfRR^p$ is equipped with a complete metric $\bfPy $ such that
any corresponding
geodesic $\bfgammay $ is minimizing on $\bfRR$ and the function $\wpunbf$ is the square of the corresponding 
distance.
\end{hypothesis}

The simplest way to satisfy the condition in Lemma \ref{lem4} is to choose the metric $\bfPy $ flat.
Precisely, with
\begin{list}{}{%
\parskip 0pt plus 0pt minus 0pt%
\topsep 0pt plus 0pt minus 0pt%
\parsep 0pt plus 0pt minus 0pt%
\partopsep 0pt plus 0pt minus 0pt%
\itemsep 0pt plus 0pt minus 0pt%
\settowidth{\labelwidth}{\quad --}%
\setlength{\labelsep}{0.5em}%
\setlength{\leftmargin}{\labelwidth}%
\addtolength{\leftmargin}{\labelsep}%
}
\item[--] $\bfRR^p$ equipped  with a global coordinate chart
$(\bar y,\bfRR^p,i_d)$ with the identity matrix $I_p$ as expression of a metric;
\item[--] $\bfchangey$ an homeomorphism from $\bfRR^p$ onto $\bfRR^p$ ,
\end{list}
\IfTwoCol{%
for any 
$\bfy_0$ in $\bfRR^p$ and any coordinate chart $\coordy$ around $\bfy_0$, the expressions $\changey$ 
and $\Py$ of $\bfchangey$ and $\bfPy$ satisfy (see (\ref{LP186})):
$$
\bar y= \changey( y)
\  ,\quad 
\Py(y)=  \frac{\partial \changey}{\partial y}(y)^\top I_p \frac{\partial \changey}{\partial y}(y)
\qquad \forall y\in\coordym(\coordyd)\  .
$$
}{
we choose $\bfPy$ as the pullback via $\bfchangey$ of the metric $I_p$. This means that, for any 
$\bfy_0$ in $\bfRR^p$ and any coordinate chart $\coordy$ around $\bfy_0$, the expressions $\changey$ 
and $\Py$ of $\bfchangey$ and $\bfPy$ satisfy (see (\ref{LP186})):
$$
\bar y\;=\; \changey( y)
\quad ,\qquad 
\Py(y)\;=\;  \frac{\partial \changey}{\partial y}(y)^\top I_p \frac{\partial \changey}{\partial y}(y)
\qquad \forall y\in\coordym(\coordyd)\  .
$$
}%
This implies that the distance $e$ is Euclidean when $\bfy$ is expressed
with the specific coordinates $\bar y$ which may not necessarily be the physical
quantities provided by the sensors.

\subsubsection{Choice of the Metric $\bfP$}
With the function $\wpunbf$ being the square of the distance $e$ in the $\bfy$-manifold $\bfRR^p$
as in \eqref{LP109},
\startmodif
in view of the interpretation of Condition A3 in terms of geodesic monotonicity,
\stopmodif
the simplest case for this condition to hold is when the image under $\bfh$ of a geodesic in the $\bfx$-manifold  is a 
geodesic in the $\bfy$-manifold.

\begin{proposition}
\label{prop12}
Suppose Assumptions~\ref{H1} and \ref{H2} hold. 
Then, Condition A3 is satisfied if any
geodesic $\bfgamma ^*$, in the $\bfx$-manifold $\bfRR^n$, that takes values in the
\startmodif
%
open set $\Ouv\subset\bfRR^n$
\stopmodif
on a maximal interval $(s_1,s_2)$, is such that
$s\in(s_1,s_2) \mapsto \bfh\left(\bfgamma^*(s)\right)$
is either constant or a geodesic in the $\bfy$-manifold $\bfRR^p$.
\end{proposition}

\begin{remark}
\label{rem2}
The
\startmodif
``Euclidean family''
\stopmodif
introduced in Remark \ref{rem1} gives the simplest case we can think of for
this property to hold. Indeed with choosing a constant matrix for the
expression of $\bfPy$, we take advantage of the property that the image by a linear function of a straight line is a 
straight line.
\end{remark}

\IfTwoCol{%
\begin{proof}
\startmodif
By assumption, for any geodesic
$\bfgamma ^*$ taking values in
%
$\Ouv$
and minimizing on the maximal interval $(s_1,s_2)$, the function $s\in(s_1,s_2) \mapsto \bfh\left(\bfgamma^*(s)\right)$ is either constant or a 
geodesic
$s\in(s_1,s_2) \mapsto \gammay^*(s)$ in the $\bfy$-manifold $\bfRR^p$. In the former case, we have
$$
\bfh(\bfgamma ^*(s_3)) = \bfh(\bfgamma ^*(s_4))\qquad \forall s_3,s_4 \in (s_1,s_2):\  
s_3<s_4\  .
$$
In the latter case,
we have
\\[0.7em]\null \hfill $\displaystyle 
\frac{d \bfh\circ\bfgamma ^*}{d s}(s)= \frac{d\gammay^*}{ds}(s) \neq  0 \quad  \forall s\in (s_3,s_4)
$\hfill\null\\[0.2em]
and
\\[-0.4em]\null\hfill$\displaystyle
\bfh(\bfgamma ^*(s_3)) \neq \bfh(\bfgamma ^*(s_4))
$\quad \null \hfill  \null \\[0.3em]\null\hfill$\displaystyle   
\forall s_3,s_4 \in (s_1,s_2):\:   
s_3<s_4 
\  .
$\\[0.7em]
On  another hand, with Assumption~\ref{H2}, any geodesic in the $\bfy$-manifold is minimizing, so
(\ref{LP110}) gives,
\stopmodif
\\[0.7em]$\displaystyle 
\frac{d}{ds}\left\{\wpunbf(\bfh(\bfgamma ^*(s) ,\bfh(\bfgamma ^*(s_3))\right\}
\; =\; \displaystyle 
2\,  \dy (\bfh(\bfgamma ^*(s) ,\bfh(\bfgamma ^*(s_3))\; \times
$\hfill \null \\[0.3em]\null\hfill $\displaystyle 
\times
\sqrt{
\frac{d \bfh\circ\bfgamma ^*}{d s}(s)^\top \bfPy (\bfgamma ^*(s)) \frac{d \bfh\circ\bfgamma ^*}{d s}(s)
}
$\qquad \null\\[0.5em]\null\hfill$\displaystyle
\forall s\in (s_3, s_2)\  ,\quad \forall s_3 \in (s_1,s_2)\  .
$\\[0.7em]
where, $\dy$ being a distance for $\bfRR^p$, 
the right-hand side is strictly positive. Hence, (\ref{LP99}) holds.
\end{proof}
}{
\startmodif
\begin{proof}
By assumption, for any  geodesic
$\bfgamma ^*$ taking values in
%
$\Ouv$
and minimizing on the maximal interval $(s_1,s_2)$, the function $s\in(s_1,s_2) \mapsto \bfh\left(\bfgamma^*(s)\right)$ is either constant or a 
geodesic
$s\in(s_1,s_2) \mapsto \gammay^*(s)$ in the $\bfy$-manifold $\bfRR^p$. In the former case, we have
$$
\bfh(\bfgamma ^*(s_3)) = \bfh(\bfgamma ^*(s_4))\qquad \forall s_3,s_4 \in (s_1,s_2):\  
s_3<s_4\  .
$$
In the latter case, we have
$$
\frac{d \bfh\circ\bfgamma ^*}{d s}(s)= \frac{d\gammay^*}{ds}(s) \neq  0
\quad \forall s\in (s_3,s_4)
\quad \  \& \quad \  
\bfh(\bfgamma ^*(s_3)) \neq \bfh(\bfgamma ^*(s_4))
\qquad \forall s_3,s_4 \in (s_1,s_2)\!:\,   
s_3<s_4 
\:  .
$$
On  another hand, with Assumption~\ref{H2}, any geodesic in the $\bfy$-manifold is minimizing, so
(\ref{LP110}) gives,
\\[0.7em]$\displaystyle 
\frac{d}{ds}\left\{\wpunbf(\bfh(\bfgamma ^*(s) ,\bfh(\bfgamma ^*(s_3))\right\}
\; =\; \displaystyle 
2\,  \dy (\bfh(\bfgamma ^*(s) ,\bfh(\bfgamma ^*(s_3)) \sqrt{
\frac{d \bfh\circ\bfgamma ^*}{d s}(s)^\top \bfPy (\bfgamma ^*(s)) \frac{d \bfh\circ\bfgamma ^*}{d s}(s)
}
$\qquad \null\\[0.5em]\null\hfill$\displaystyle
\forall s\in (s_3, s_2)\  ,\quad \forall s_3 \in (s_1,s_2)\  ,
$\\[0.7em]
where, $\dy$ being a distance for $\bfRR^p$, 
the right-hand side is strictly positive. Hence, (\ref{LP99}) holds.
\end{proof}
\stopmodif
}%

\par\vspace{1em}
To make the condition in Proposition \ref{prop12} more explicit, we note that
$\bfh(\bfgamma)$ is a geodesic in the $\bfy$-manifold $\bfRR^p$ if and only if $\bfh(\bfgamma)$ satisfies
the geodesic equation and therefore, thanks to \cite[Theorem 9.12]{Lee.13} for example, if and only if, the following holds:
\begin{list}{}{%
\parskip 0pt plus 0pt minus 0pt%
\topsep 0pt plus 0pt minus 0pt%
\parsep 0pt plus 0pt minus 0pt%
\partopsep 0pt plus 0pt minus 0pt%
\itemsep 0pt plus 0pt minus 0pt%
\settowidth{\labelwidth}{--}%
\setlength{\labelsep}{0.5em}%
\setlength{\leftmargin}{\labelwidth}%
\addtolength{\leftmargin}{\labelsep}%
}
\item[--]
\startmodif
for any $\bfx_0$ in
%
$\Ouv$
\stopmodif
and any pair of coordinate charts
$\coordx$ around $\bfx_0$
and  $\coordy$ around $\bfh(\bfx_0)$, with ``objects'' expressed in these coordinates and
with $\mbox{\large $\Gamma$}^\indxa$ and $\mbox{\large$\Gammay$}^\indyi$ denoting the respective Christoffel 
symbol matrices,
namely
\glos{\ref{item:RiemannianMetric}}
$$
\mbox{\large $\Gamma$}^\indxa\;=\; ( \Gamma _{\indxb\indxc}^\indxa)
\quad ,\qquad 
\mbox{\large$\Gammay$}^\indyi\;=\; (\Gammay_{\indyj\indyk}^\indyi)
\  ,
$$
\item[--]
for any geodesic $\gamma ^*$,
with values in $\coordxm(\coordxd)$ and their image by $\bfh$ in $\coordym(\coordyd)$,
and minimizing on the maximal interval $(s_1,s_2)$,
\end{list}
the geodesic equation in $\coordym(\coordyd)$ and in $\coordxm(\coordxd)$ are, respectively,
\IfTwoCol{%
for all $s$ in $(s_1, s_2)$,
\\[0.6em]$\displaystyle 
0\;=\; \frac{d^2 h_\indyi\circ \gamma^*}{ds^2}(s)
$\hfill\null\\[0.3em]\null\hfill$\displaystyle
\;+\; 
\frac{d\gamma ^* }{ds}(s)^\top
\frac{\partial h}{\partial x}(\gamma ^* (s))^\top
\mbox{\large$\Gammay$}^\indyi(h(\gamma ^* (s)))
\frac{\partial h}{\partial x}(\gamma ^* (s))
\frac{d\gamma ^* }{ds}(s)
$\hfill\null\\[0.3em]\null\hfill$\displaystyle
\forall \indyi \in \{1,2,\ldots, p\}
\  ,
$\\[0.6em]$\displaystyle 
0\;=\; \frac{d^2 \gamma ^* _\indxc}{ds^2}(s)\;+\;
\frac{d\gamma ^* }{ds}(s)^\top
\mbox{\large$\Gamma$}^\indxc(\gamma ^* (s))
\frac{d\gamma ^* }{ds}(s)
$\hfill\null\\[0.3em]\null\hfill$\displaystyle
\forall \indxc \in \{1,2,\ldots,n\}
\  .
$\\%
}{
$$
\renewcommand{\arraystretch}{2}
\begin{array}{@{}r@{\; }c@{\; }ll@{\quad }}
\displaystyle 
\frac{d^2 h_\indyi\circ \gamma^*}{ds^2}(s)\,+\,
\frac{d\gamma ^* }{ds}(s)^\top
\frac{\partial h}{\partial x}(\gamma ^* (s))^\top
\mbox{\large$\Gammay$}^\indyi(h(\gamma ^* (s)))
\frac{\partial h}{\partial x}(\gamma ^* (s))
\frac{d\gamma ^* }{ds}(s)
&=&0
&
\forall \indyi \in \{1,2,\ldots, p\}
\, ,
\\\displaystyle 
\frac{d^2 \gamma ^* _\indxc}{ds^2}(s)\,+\,
\frac{d\gamma ^* }{ds}(s)^\top
\mbox{\large$\Gamma$}^\indxc(\gamma ^* (s))
\frac{d\gamma ^* }{ds}(s)
&=&0
&
\forall \indxc \in \{1,2,\ldots,n\}
\, ,
\\\multicolumn{4}{r}{
\forall s \in (s_1, s_2)
\  .
}
\end{array}
$$
}%

Then, with the uniqueness of the solution to the geodesic equation, we have the following result. 
Its proof can be found in \cite{Vilms} or \cite{Udriste.94}.

\begin{lemma}[{\cite[Proposition 1.5]{Vilms}
or \cite[Theorem 6.4.1]{Udriste.94}}]
\label{lem7}
\startmodif
Let $\Ouv$ be an open subset of $\bfRR^n$.
\stopmodif
Suppose Assumption~\ref{H1} holds.
The following two properties are equivalent\footnote{%
According to \cite{Innami.82} (but not \cite{Vilms}) a function
$\bfh$ satisfying the property in item~\ref{point7}
is said to be affine on $\Ouv$.
In the case where $p=1$ and the metric $\bfPy$ is flat, $\bfh$ is said linear affine in \cite[p. 88 and 
following pages]{Udriste.94} where its necessary and sufficient conditions are presented.
}
\begin{enumerate}
\item
\label{point7}
Any
geodesic $\bfgamma ^*$, in the $\bfx$-manifold $\bfRR^n$ that
takes values in $\Ouv$ on a maximal interval $(s_1,s_2)$, is such that
$s\in (s_1,s_2) \mapsto \bfh\left(\bfgamma^*(s)\right)$ is either constant or a geodesic in the $\bfy$-manifold $\bfRR^p$.
\item
For any point $\bfx_0$ in $\Ouv$, any coordinate chart $\coordx$
around $\bfx_0$ and any coordinate chart $\coordy$ around 
$\bfh(\bfx_0)$, with $\bfh(\coordxd)$ contained in $\coordyd$,
we have,
\IfTwoCol{%
for all $x$ in $\coordxm(\coordxd \cap \Ouv)$,
\\[0.5em]\null \hskip -2em$\displaystyle 
0\;=\; \frac{\partial ^2h_\indyi}{\partial x_\indxa\partial x_\indxb}(x)
$\refstepcounter{equation}\label{LP145}\hfill$(\theequation)$
\\[0.3em]\null \hfill $\displaystyle
-
\sum_{\indxc}
\Gamma _{\indxa\indxb}^\indxc(x)\frac{\partial h_\indyi}{\partial x_\indxc}(x)
+
\sum_{\indyj,\indyk}
\Gammay_{\indyj\indyk}^\indyi(h(x))
\frac{\partial h_\indyj}{\partial x_\indxa}(x)
\frac{\partial h_\indyk}{\partial x_\indxb}(x)
$\\
}{%
\begin{equation}
\label{LP145}
\displaystyle 
\frac{\partial ^2h_\indyi}{\partial x_\indxa\partial x_\indxb}(x)
-
\sum_{\indxc}
\Gamma _{\indxa\indxb}^\indxc(x)\frac{\partial h_\indyi}{\partial x_\indxc}(x)
+
\sum_{\indyj,\indyk}
\Gammay_{\indyj\indyk}^\indyi(h(x))
\frac{\partial h_\indyj}{\partial x_\indxa}(x)
\frac{\partial h_\indyk}{\partial x_\indxb}(x)
\;=\; 0
\qquad \forall x\in \coordxm(\coordxd \cap \Ouv)
\  .
\end{equation}%
}%
\end{enumerate}
\end{lemma}

Lemma~\ref{lem7} motivates the following definition.

\begin{definition}[{\cite[p. 123]{Eells-Sampson}}]
\label{def3}
We call second fundamental form  $\bfsecff _\bfP \bfh$ of $\bfh$ the object defined as follows:
\IfTwoCol{%
For any point $\bfx_0$, any coordinate chart $\coordx$
around $\bfx_0$ and any coordinate chart $\coordy$ around 
$\bfh(\bfx_0)$, with $\bfh(\coordxd)$ contained in $\coordyd$, the expression of $\bfsecff _\bfP \bfh$ is
\\[0.7em]
$\displaystyle 
\secff _P h_{\indxa\indxb}^\indyi(x)
\;=\; \frac{\partial ^2h_\indyi}{\partial x_\indxa\partial x_\indxb}(x)
$\refstepcounter{equation}\label{eqn:2ndFundFormh}\hfill$(\theequation)$
\\[0.5em]\null \hfill $\displaystyle
-
\sum_\indxc
\Gamma _{\indxa\indxb}^\indxc(x)\frac{\partial h_\indyi}{\partial x_\indxc}(x)
+
\sum_{\indyj,\indyk}\Gammay_{\indyj\indyk}^\indyi(h(x))
\frac{\partial h_\indyj}{\partial x_\indxa}(x)
\frac{\partial h_\indyk}{\partial x_\indxb}(x)
\  .
$
\\%
}{%
for any $\bfx_0$ and
for any pair of coordinate charts
$\coordx$ around $\bfx_0$ and $\coordy$ around 
$\bfh(\bfx_0)$, the expression of $\bfsecff _\bfP \bfh$ is
\begin{equation}\label{eqn:2ndFundFormh}
\secff _P h_{\indxa\indxb}^\indyi(x)
\;=\; 
\frac{\partial ^2h_\indyi}{\partial x_\indxa\partial x_\indxb}(x)
-
\sum_\indxc
\Gamma _{\indxa\indxb}^\indxc(x)\frac{\partial h_\indyi}{\partial x_\indxc}(x)
+
\sum_{\indyj,\indyk}
\Gammay_{\indyj\indyk}^\indyi(h(x))
\frac{\partial h_\indyj}{\partial x_\indxa}(x)
\frac{\partial h_\indyk}{\partial x_\indxb}(x)
\  .
\end{equation}%
}%
\end{definition}

A coordinate-free version of this definition can be found in
\cite[Definition I.1.4.1]{Nore} or \cite[Definition 3.1.1]{Garcia-Kupeli}. See also
\cite[Definition 8.1]{Saccon-Hauser-Aguiar}.

\IfTwoCol{%
}{%
Our expression via coordinates is coordinate independent since,  with computations involving components,
we can prove the following property about $\bfsecff_\bfP\bfh$.
}%
\startarchive%
(see \complement \ref{complement36}).%
\stoparchive%

\begin{lemma}
\label{lem1}
The second fundamental form $\bfsecff_\bfP\bfh$ is a bilinear map of a pair of vector fields on the $\bfx$-manifold $\bfRR^n$ into a vector field on 
the $\bfy$-manifold $\bfRR^p$.\! It is a 2-covariant to\! 1-contravariant tensor, i.e. by
changing to coordinates
$$
\bar x = \changex(x)\quad ,\qquad 
\bar y\;=\; \changey (y)
\  ,
$$
where $\changex$ and $\changey$ are $C^s$ diffeomorphisms, the expression $\secff _P \overline{h}$ of the second 
fundamental form  in the new coordinates satisfies
\IfReport{%
\begin{equation}
\label{LP146}
\sum_{\indxc,\indxd}
\frac{\partial \changex_\indxc}{\partial x_\indxa}(x)
\frac{\partial \changex_\indxd}{\partial x_\indxb}(x)
\secff _P \overline{h}_{\indxc\indxd}^\indyk(\bar x)\;=\; 
\sum_\indyi
\frac{\partial \changey_\indyk}{\partial y_\indyi}(h(x))
\secff _P h_{\indxa\indxb}^\indyi(x)
\  .
\end{equation}
}{%
\\[0.7em]$\displaystyle 
\sum_{\indxc,\indxd}
\frac{\partial \changex_\indxc}{\partial x_\indxa}(x)
\frac{\partial \changex_\indxd}{\partial x_\indxb}(x)
\secff _P \overline{h}_{\indxc\indxd}^\indyk(\bar x)\;=\; 
\sum_\indyi
\frac{\partial \changey_\indyk}{\partial y_\indyi}(h(x))
\secff _P h_{\indxa\indxb}^\indyi(x)
\  .
$\hfill \null \\[-0.7em]\refstepcounter{equation}\label{LP146}\null \hfill$(\theequation)$
}%
\end{lemma}%
\IfReport{%
}{%

A proof of Lemma~\ref{lem1} can be found in Appendix~\ref{complement36}. 
}

%
We can rephrase the way to guarantee Condition A3 stated in Proposition~\ref{prop12} as follows.

\begin{corollary}
\label{cor1}
Suppose Assumptions~\ref{H1} and \ref{H2} hold. Then, Condition A3 is satisfied 
if the second fundamental form of $\bfh$ is zero
\startmodif
on the open subset $\Ouv$.
\stopmodif
\end{corollary}

Consequently, (\ref{LP145}) provides a test for Condition A3 that involves $h$, $P$, and their first derivatives.

\vspace{-0.1in}
\subsection{Necessity of the nullity of $\bfsecff_\bfP\bfh$}
\label{sec:necessary}
\startmodif
With Corollary \ref{cor1}, it is tempting to forget about Condition A3 and base our design of the Riemannian 
metric $\bfP$ on guaranteeing that the second fundamental form of $\bfh$ is zero.

Next, we investigate such an approach.
\stopmodif
Namely, given
$\bfh:\bfRR^n\to \bfRR^p$, a submersion on a set $\Ouv$, as in Assumption~\ref{H1}
and
$\bfPy$, a metric equipping the $\bfy$-manifold $\RR^p$ and satisfying Assumption~\ref{H2},
let the function $\wpunbf$ be the square of the distance $e$ in the $\bfy$-manifold $\RR^p$, as written in 
(\ref{LP109}).
We are interested in the following question:
\begin{description}
\item[\normalfont(Q1)]
{\it Is the  second fundamental form of $\bfh$ being null, i.e. the property in 
(\ref{LP145}), a necessary for Condition A3 to hold?} 
\end{description}

We address this question, by grouping, as in \cite[\S 3]{Vilms}, the equations  in (\ref{LP145}) into
three blocks. We employ the following definition; see \cite[\& C p.127]{Eells-Sampson},
\cite[p. 77]{Vilms}, or \cite[p. 205]{ONeill.83}.

\begin{definition}
\label{def5}
The tangent space of the level sets of $\bfh $, denoted $\bfDistrib ^\tangent$, is called the 
tangent distribution. It satisfies
$$
\bfDistrib ^\tangent(\bfx)\;=\; \left\{v^\tangent:\,  \bfdh (\bfx )v^\tangent=0\right\}\ ,
$$
and does not depend on $\bfP$. Its elements $v^\tangent$ are called tangent vectors.

The $\bfP $-orthogonal complement to the tangent distribution, denoted $\bfDistrib _\bfP ^\ortho  $, is called
the orthogonal distribution. It satisfies
$$
\bfDistrib _\bfP ^\ortho  (\bfx )\;=\; \left\{v^\ortho  :\,  v^{\tangent\top}  \bfP (\bfx ) v^\ortho  \;=\; 0\quad  \forall v^\tangent\in 
\bfDistrib ^\tangent(\bfx )\right\}
$$
and does depend on $\bfP$. Its elements $v^\ortho$ are called orthogonal vectors.
\end{definition}

What we refer to as {\it orthogonal distribution} is usually called {\it horizontal distribution}, and 
variations of the letter $h$ are used to denote it.
But since we use the letter $h$ to denote the output function in this paper,
we employ the term {\it orthogonal} and  use the symbol $\ortho$.
For consistency, we use the symbol $\tangent$ and call {\it tangent distribution} what is usually called {\it vertical distribution}.

Ignoring the $\bfx$-dependence, 
basic linear algebra leads to the following result.
\begin{lemma}
\label{lem6}
For any pair of coordinate charts $\coordx$ and $\coordy$,
by letting $h$, $P$ and $\Distrib _P ^\ortho  (x)$ be the corresponding expressions of 
$\bfh $, $\bfP $ and $\bfDistrib _\bfP ^\ortho  (\bfx)$, 
the orthogonal distribution  $\Distrib _P ^\ortho  (x)$ is spanned by the columns
of the gradient of $h$, i.e. by the columns of
$P(x)^{-1}\frac{\partial h}{\partial x}(x)^\top $.
\end{lemma}
\IfReport{%
\begin{proof}
\startmodif
Let $v^\ortho  $ be an arbitrary orthogonal vector
and denote
$$
w^\ortho  \!\!=\!\!P(x)^{-1}\frac{\partial h}{\partial x}(x)^\top  \left(\frac{\partial h}{\partial x}(x) 
P(x)^{-1}\frac{\partial h}{\partial x}(x)^\top \right)^{-1}\frac{\partial h}{\partial x}(x) v^\ortho  
\  .
$$
Let $v^\tangent$ be an arbitrary tangent vector. It satisfies
$$
\frac{\partial h}{\partial x}(x)v^\tangent\;=\; 0
$$
and therefore
$$
v^{\tangent\top}  P(x)\,  w^\ortho  
\;=\; 
0
\  .
$$
This says that $w^\ortho  $ is an orthogonal vector. We also have
$$
\frac{\partial h}{\partial x}(x) [v^\ortho  -w^\ortho  ]
\;=\; 0
\  .
$$
So $v^\ortho  -w^\ortho  $ is a tangent vector. But $v^\ortho  $, as an orthogonal vector, cannot be the sum of 
an orthogonal vector and a tangent one. This implies
$$
v^\ortho  \;=\; w^\ortho  
\  .
$$
This says that $v^\ortho  $ is in the span of $P(x)^{-1}\frac{\partial h}{\partial x}(x)^\top  $.
\stopmodif
\end{proof}
}{%
}

Given a pair of coordinate charts $\coordx$ and $\coordy$, $\Distrib _P ^\ortho  (x)$ and
$\Distrib ^\tangent  (x)$ are
complementary linear subspaces
of the tangent space at $x$ of the $\bfx$-manifold $\bfRR^n$. 
As a consequence, any vector $v$ in this tangent space can be decomposed as
$$
v\;=\; v^\ortho + v^\tangent
$$
with $v^\ortho$ in $\Distrib _P ^\ortho  (x)$ and $v^\tangent$ in
$\Distrib ^\tangent  (x)$. This property allows us to decompose  (\ref{LP145})
in the following three blocks of equations, for each $\indyi$ in $\{1,2,\ldots,p\}$
\IfTwoCol{%
and each $x$ in $\coordxm(\coordxd\cap \Ouv)$:
\begin{eqnarray}
\label{LP151}
{v^{\tangent}}^\top \secff _P h^\indyi(x) v^\tangent &\hskip -0.5em =&\hskip -0.5em  0
\quad 
\forall v^\tangent \in \Distrib ^\tangent(x)\  
\:  ,
\\
\label{LP162}
{v^\ortho }^\top\secff _P h^\indyi(x) v^\ortho &\hskip -0.5em =&\hskip -0.5em  0
\quad 
\forall v^\ortho \in \Distrib _P^\ortho (x)
\:  ,\\
\label{LP161}
{v^\ortho }^\top \secff _P h^\indyi(x) v^\tangent&\hskip -0.5em =&\hskip -0.5em  0
\quad 
\forall (v^\tangent ,v^\ortho )\in \Distrib ^\tangent(x)\times \Distrib _P^\ortho (x)
\:  .\quad\   \null 
\end{eqnarray}%
}{%
\begin{eqnarray}
\label{LP151}
{v^{\tangent}}^\top \secff _P h^\indyi(x) v^\tangent &\hskip -0.5em =&\hskip -0.5em  0
\qquad 
\forall v^\tangent \in \Distrib ^\tangent(x)\  ,\quad \forall x \in \coordxm(\coordxd\cap \Ouv)
\:  ,
\\
\label{LP162}
{v^\ortho }^\top\secff _P h^\indyi(x) v^\ortho &\hskip -0.5em =&\hskip -0.5em  0
\qquad 
\forall v^\ortho \in \Distrib _P^\ortho (x)
\  ,\quad \forall x\in \coordxm(\coordxd\cap \Ouv)
\:  ,
\\
\label{LP161}
{v^\ortho }^\top \secff _P h^\indyi(x) v^\tangent&\hskip -0.5em =&\hskip -0.5em  0
\qquad 
\forall (v^\tangent ,v^\ortho )\in \Distrib ^\tangent(x)\times \Distrib _P^\ortho (x)
\  ,\quad \forall x\in \coordxm(\coordxd\cap \Ouv)
\:  ,\quad \null 
\end{eqnarray}
where $\secff _P h$ is the expression, given in \eqref{eqn:2ndFundFormh}, of the second fundamental form of 
the function $\bfh$ supposed to satisfy Assumption~\ref{H1}.
}%
With the above, we can rephrase question (Q1) as follows:
\begin{description}
\item[\normalfont(Q1')]
{\it 
Does Condition A3 imply that equations (\ref{LP151}), (\ref{LP162}), and (\ref{LP161}) are satisfied?
} 
\end{description}
Our answer builds from the study in \cite{Vilms,Nore,Garcia-Kupeli}. For the sake of 
completeness, we rewrite in our setting some of the  results therein.

\subsubsection{About Necessity of (\ref{LP151})}
\begin{definition}
Given $\bfy_0$  in the $\bfy$-manifold $\bfRR^p$, the set
$$
\mathfrak{H}(\bfy_0)\;=\; \left\{\bfx\in\bfRR^n:\,  
\bfh(\bfx)=\bfy_0\right\}
\  .
$$
is called the $\bfy_0$-level set of $\bfh$.
\end{definition}

\begin{definition}
\label{def1}
Given an open subset $\Ouv$ of $\bfRR^n$ and a point $\bfx_0$ in $\Ouv$,
the $\bfh(\bfx_0)$-level set $\mathfrak{H}(\bfh(\bfx_0))$ of $\bfh$ is said to be
totally geodesic on $\Ouv$ if
any geodesic $\bfgamma$ taking values in $\Ouv$ on the maximal interval $(s_1,s_2)$ and satisfying
$$
\bfd \bfh(\bfgamma(s_3))\frac{d\bfgamma}{ds}(s_3)=0
\  ,
$$
for some $s_3$ in $(s_1,s_2)$, satisfies
$$
\bfh(\bfgamma(s))=\bfh(\bfgamma(s_3))
\qquad \forall  s \in (s_1,s_2)\  .
$$
An equivalent definition is that, for any $\bfx_0$ in $\Ouv$, there exists a coordinate chart $\coordx$
around $\bfx_0$ such that we have
\IfReport{%
\begin{equation}
\label{LP204}
{v^\tangent}^\top \Hess _P h (x) \,  v^\tangent\;=\;  0
\qquad 
\forall v^\tangent \in \Distrib ^\tangent (x)\  ,\quad \forall x \in \coordxm(\coordxd\cap \Ouv)
\  ,
\end{equation}
}{%
\\[0.7em]$\displaystyle 
{v^\tangent}^\top \! \Hess _P h (x) \,  v^\tangent\,=\, 0
\qquad 
\forall v^\tangent \! \in \Distrib ^\tangent (x)
\  ,\   \forall x \in \coordxm(\coordxd\cap \Ouv)
\:  ,
$\\\refstepcounter{equation}\label{LP204}\hfill$(\theequation)$\\
}%
with $\Hess _P h$ the expression of the Riemannian Hessian of $h$%
\glos{\ref{Glos3}}.
\end{definition}

We have the following result. Its proof can be found in
\cite[Lemma 3.2(i)]{Vilms} or \cite[Proposition A.2.2 and A.3.1.a]{57}.

\begin{lemma}
\label{lem12}
\startmodif
Let $\Ouv$ be an open set coming from Assumption~\ref{H1} and suppose Condition A3
holds with such choice of $\Ouv$. Then,
\stopmodif
for any $\bfx_0$ in $\Ouv$ and any coordinate chart $\coordx$ around $\bfx_0$, we have
(\ref{LP151}) 
or, equivalently, for any $\bfx_0$ in $\Ouv$, the $\bfh(\bfx_0)$-level set $\mathfrak{H}(\bfh(\bfx_0))$ is totally geodesic
on $\Ouv$.
\end{lemma}

Hence, the answer to Question (Q1') is ``yes'' as far as (\ref{LP151}) is 
concerned.

\begin{example}\bgroup\normalfont
\label{ex9}
Consider the harmonic oscillator with unknown frequency. Its dynamics are given as
\begin{equation}
\label{LP132}
\dot y = \xrond_\indxra
\quad ,\qquad 
\dot \xrond_\indxra\;=\; -y \xrond_\indxrb
\quad ,\qquad 
\dot \xrond_\indxrb\;=\; 0
\end{equation}
Given $\varepsilon > 0$, we consider the invariant open set\footnote{%
The system is not observable if $y=\xrond_\indxra=0$ or $\xrond_\indxrb=0$.
}
\IfTwoCol{%
\\[0.7em]$\displaystyle 
\Ouv_\varepsilon \;=\;
$\refstepcounter{equation}\label{LP205}\hfill$(\theequation)$
\\[0.3em]\null\hfill$\displaystyle
\left\{
\vrule height 0.5em depth 0.5em width 0pt
(y,\xrond_\indxra,\xrond_\indxrb) \in \RR^3
:\,  
\varepsilon < \xrond_\indxrb y^2 + \xrond_\indxra^2 < \frac{1}{\varepsilon }
\; ,\  
\varepsilon <  \xrond_\indxrb < \frac{1}{\varepsilon }
\right\}
\:  .
$\\
}{%
\begin{equation}
\label{LP205}
\Ouv_\varepsilon \;=\; \left\{
\vrule height 0.5em depth 0.5em width 0pt
(y,\xrond_\indxra,\xrond_\indxrb) \in \RR^3
:\,  
\varepsilon < \xrond_\indxrb y^2 + \xrond_\indxra^2 < \frac{1}{\varepsilon }
\; ,\  
\varepsilon <  \xrond_\indxrb < \frac{1}{\varepsilon }
\right\}
\  .
\end{equation}
}

In \cite[Example 4.5]{127}, we have obtained the following metric satisfying Condition A2:
\IfTwoCol{%
\\[0.7em]$\displaystyle 
P(y,\xrond_\indxra,\xrond_\indxrb)\;=\; 
$\refstepcounter{equation}\label{LP136}\hfill$(\theequation)$
\\\null\hfill$\displaystyle
\left(\begin{array}{cccc}
1 & 0 & -\xrond _\indxrb  & 0 
\\
0 & 1 & 0 & -\xrond _\indxrb 
\\
0 & 0 & -y& -\xrond _\indxra 
\end{array}\right)
\euP
\left(\begin{array}{ccc}
1 & 0 & 0 
\\
0 & 1 & 0 
\\
-\xrond _\indxrb  & 0 & -y
\\
0 & -\xrond _\indxrb  & -\xrond _\indxra 
\end{array}\right),
$\\[0.7em]
where $\euP $ is a positive definite symmetric $(4,4)$ matrix.
}{%
\begin{equation}
\label{LP136}
P(y,\xrond_\indxra,\xrond_\indxrb)\!=\! 
\left(\begin{array}{cccc}
1 & 0 & -\xrond _\indxrb  & 0 
\\
0 & 1 & 0 & -\xrond _\indxrb 
\\
0 & 0 & -y& -\xrond _\indxra 
\end{array}\right)
\euP
\left(\begin{array}{ccc}
1 & 0 & 0 
\\
0 & 1 & 0 
\\
-\xrond _\indxrb  & 0 & -y
\\
0 & -\xrond _\indxrb  & -\xrond _\indxra 
\end{array}\right),
\end{equation}
where $\euP $ remains to be designed as a positive 
definite symmetric $(4,4)$ matrix.
}%
Symbolic computations give
$
\Gamma _{\indxra\indxra}^y\;=\; \Gamma _{\indxra\indxrb}^y\;=\; 0
$
and, up to some nonzero  factor $k$,
\startarchive%
(see \complement \ref{complement24}),%
\stoparchive%
\IfTwoCol{%
\\[0.7em]$\displaystyle
k \Gamma _{\indxrb\indxrb}^y =  
$\hfill \null \\[0.5em]\null \hfill $\displaystyle
-y\left[
\xrond _\indxrb ^2
\det\left(E_2^\top \euP E_2\right)
+
2\xrond _\indxrb \det\left(E_4^\top \euP E_2\right)
+
\det\left(E_4^\top \euP E_4\right)
\right]
$\hfill \null \\[0.5em] $\displaystyle
-
\xrond _\indxra \! \left[
\vrule height 0.8em depth 0.8em width 0pt
\xrond _\indxrb ^2
\det\left(E_1^\top \euP E_2\right)
+
\xrond _\indxrb  \left[
\det\left(E_2^\top \euP E_3\right)
+
\det\left(E_1^\top \euP E_4\right)
\right]
\right.
$\hfill \null \\[0.1em]\null \hfill $\displaystyle
\left.
\vrule height 0.8em depth 0.8em width 0pt
+\,  
\det\left(E_3^\top \euP E_4\right)
\right]
$\\[1em]
}{%
\\[1em]$\displaystyle
k\,  \Gamma _{\indxrb\indxrb}^y\;=\; 
-y\left[
\xrond _\indxrb ^2
\textsf{det}\left(E_2^\top \euP E_2\right)
\;+\; 
2\xrond _\indxrb \textsf{det}\left(E_4^\top \euP E_2\right)
\;+\; 
\textsf{det}\left(E_4^\top \euP E_4\right)
\right]
$\hfill \null \\[0.5em]\null \hfill $\displaystyle
\;-\; 
\xrond _\indxra \left[
\xrond _\indxrb ^2
\textsf{det}\left(E_1^\top \euP E_2\right)
\;+\;
\xrond _\indxrb  \left[
\textsf{det}\left(E_2^\top \euP E_3\right)
+
\textsf{det}\left(E_1^\top \euP E_4\right)
\right]
\;+\; 
\textsf{det}\left(E_3^\top \euP E_4\right)
\right]
$\\[0.7em]%
}%
where
\IfTwoCol{%
$E_i$ is the $4\times 3 $ matrix
made of the $3 \times 3$ identity matrix with an additional raw of zeros inserted as its $i$-th row.
}{%
$$
E_1\;=\; \left(\begin{array}{@{}ccc@{}}
0 & 0 & 0
\\
1 & 0 & 0
\\
0 & 1 & 0
\\
0 & 0 & 1
\end{array}\right)
\  ,\quad
E_2\;=\; \left(\begin{array}{@{}ccc@{}}
1 & 0 & 0
\\
0 & 0 & 0
\\
0 & 1 & 0
\\
0 & 0 & 1
\end{array}\right)
\  ,\quad
E_3\;=\; \left(\begin{array}{@{}ccc@{}}
1 & 0 & 0
\\
0 & 1 & 0
\\
0 & 0 & 0
\\
0 & 0 & 1
\end{array}\right)
\  ,\quad 
E_4\;=\; \left(\begin{array}{@{}ccc@{}}
1 & 0 & 0
\\
0 & 1 & 0
\\
0 & 0 & 1
\\
0 & 0 & 0
\end{array}\right)
$$%
}%
Hence (\ref{LP151}) and therefore Condition A3 do not hold.

\IfReport{%
}{%
Similarly, it can be shown that Condition A3 does not hold for the other metric, satisfying Condition A2, obtained in 
\cite[Example 3.7]{127}.
}%

%
 \startarchive %
Also, in \cite[Example 3.7]{127}, we have obtained the following metric satisfying Condition A2:
\begin{equation}
\label{LP135}
P(y,\xrond_\indxra,\xrond_\indxrb)\;=\; 
\left(\begin{array}{c@{\  ,\  }c@{\  ,\quad  }c}
\displaystyle 
\frac{\lambda ^2 +2\xrond _\indxrb  }{\lambda (\lambda ^2+4\xrond _\indxrb )}
&\displaystyle 
\star
&
\displaystyle 
\star\\[1.56em]
\displaystyle 
-\frac{1  }{ (\lambda ^2+4\xrond _\indxrb )}
&\displaystyle 
\frac{2}{\lambda (\lambda ^2+4\xrond _\indxrb )}
&
\displaystyle 
\star\\[1.56em]
\displaystyle 
\frac{-
\lambda ^3
y
+
(\lambda ^2-4\xrond _\indxrb )
\xrond _\indxra }{\lambda ^2(\lambda ^2+4\xrond _\indxrb )^2}
&
\displaystyle 
\frac{ (3\lambda ^2+4\xrond _\indxrb )
y
-
4\lambda 
\xrond _\indxra 
}{\lambda ^2(\lambda ^2+4\xrond _\indxrb )^2}
&
\diamond
\end{array}\right)
\end{equation}
where the various $\star$ should be replaced by their symmetric 
values and
$$
\diamond=
\frac{6\lambda ^4
+12\lambda ^2\xrond _\indxrb 
+16\xrond _\indxrb ^2
}{\lambda ^3(\lambda ^2+4\xrond _\indxrb )^3}y ^2
-
\frac{
4(5\lambda ^2+4\xrond _\indxrb )
}{\lambda ^2(\lambda ^2+4\xrond _\indxrb )^3}
y\xrond _\indxra 
\;+\; 
\frac{
4(5\lambda ^2+4\xrond _\indxrb )
}{\lambda ^3(\lambda ^2+4\xrond _\indxrb )^3}
\xrond _\indxra ^2 
\  ,
$$
with $\lambda$ strictly positive. It can be checked
(see \complement \ref{complement23})
%
that 
$\Gamma _{\indxra\indxra}^y$ is zero and that $\Gamma _{\indxra\indxrb}^y$ is nonzero.
Then, again, (\ref{LP151}) and therefore Condition A3 do not hold.
 \stoparchive %

Fortunately, with the techniques presented below, we shall be able to obtain, in Example \ref{ex8}, a metric satisfying Conditions A2 and A3.
\egroup\end{example}

\subsubsection{About (\ref{LP162})}
We shall not try to establish that (\ref{LP162}) is necessary for Condition A3 to hold. Instead, we show below that, maybe 
after modifying appropriately $\bfP$, we 
can always guarantee that (\ref{LP162}) holds. To do so, we start by providing an expression of (\ref{LP162}).

Let $\coordx$ be an arbitrary coordinate chart around some $\bfx_0$ in $\Ouv$. With Lemma \ref{lem6},
an expression, in these coordinates, 
of the restriction, denoted $\bfsecff _\bfP \bfh^{\ortho,\ortho}$, to the orthogonal distribution $\bfDistrib _\bfP^\ortho$
of the second fundamental form  of $\bfh$
is, for each $\indyi$ in $\{1,2,\ldots,p\}$
\IfTwoCol{%
 and each $x$ in $\coordxm(\coordxd\cap \Ouv)$,
\begin{equation}
\label{eqn:SecondFundamentalForm}
[\secff _P h^{\ortho,\ortho}]^\indyi(x)\;=\;
\frac{\partial h}{\partial x}(x)P(x)^{-1}
\secff _P h^\indyi(x)
P(x)^{-1}\frac{\partial h}{\partial x}(x)^\top
\  .
\end{equation}%
}{
\begin{equation}
\label{eqn:SecondFundamentalForm}
[\secff _P h^{\ortho,\ortho}]^\indyi(x)\;=\;
\frac{\partial h}{\partial x}(x)P(x)^{-1}
\secff _P h^\indyi(x)
P(x)^{-1}\frac{\partial h}{\partial x}(x)^\top
\qquad 
\forall x \in \coordxm(\coordxd\cap \Ouv)
\  .
\end{equation}%
}%

By definition,  equation (\ref{LP162}) is equivalent to the nullity of
$[\secff _P h^{\ortho,\ortho}]^{\indyi} (x) $, for each $\indyi$ in $\{1,2,\ldots,p\}$ and each $x$ in $\coordxm(\coordxd\cap \Ouv)$.
Via computations using the components of the second fundamental form above, and by expansion 
using the identity
$$
[P^{-1}]_{\indyi\indxra}\;=\; - \sum_\indyl [P^{-1}]_{\indyi\indyl}
\sum_{\indxrf}P_{\indyl\indxrf}[P_{\xrond\xrond}^{-1}]_{\indxrf\indxra}
\  ,
$$
we can establish the following result.

\begin{lemma}~
\label{lem8}
For any $\bfx_0$ in $\Ouv$ and any pair of coordinate charts $\coordyxr$ around $\bfx_0$ and $\coordy$ around $\bfh(\bfx_0)$,
the  expression of $\bfsecff _\bfP \bfh^{\ortho,\ortho}$ in \eqref{eqn:SecondFundamentalForm} is
\IfTwoCol{%
\\[0.7em]$\displaystyle 
2 [\secff _P h^{\ortho,\ortho}]_{\indyj\indyk}^{\indyi}\;=\; 
$\\[0.3em]\null \qquad $\displaystyle 
\sum_\indyl \Py ^{-1}_{\indyi\indyl}
\left(
\frac{\partial \Py _{\indyl\indyk}}{\partial y_\indyj}
+
\frac{\partial  \Py _{\indyl\indyj}}{\partial y_\indyk}
-
\frac{\partial  \Py _{\indyj\indyk}}{\partial y_\indyl}
\right)
$\hfill\null\\\null\hfill$\displaystyle
-
\sum_\indyl 
 [P_y^{-1}]_{\indyi\indyl}
\left(
\frac{\partial [P_y]_{\indyl\indyk}}{\partial y_\indyj}
+
\frac{\partial [P_y]_{\indyl\indyj}}{\partial y_\indyk}
-
\frac{\partial [P_y]_{\indyj\indyk}}{\partial y_\indyl}
\right)
$\\[0.3em]\null \quad  $\displaystyle 
+
\sum_\indyl 
 [P_y^{-1}]_{\indyi\indyl}
\sum_{\indxra,\indxrb}
[P_{\xrond\xrond}^{-1}]_{\indxra\indxrb} \times
$\hfill\null\\[0.3em]\null\hfill$\displaystyle
\times
\left(
\frac{\partial [P_y]_{\indyl\indyk}}{\partial \xrond_\indxra}
P_{\indxrb\indyj}
+
\frac{\partial [P_y]_{\indyl\indyj}}{\partial \xrond_\indxra}
P_{\indxrb\indyk}
-
\frac{\partial [P _y]_{\indyj\indyk}}{\partial \xrond_\indxra}
P_{\indxrb\indyl} 
\right)
\  ,
$\\[0.7em]%
}{%
\\[1em]$\displaystyle 
2 [\secff _P h^{\ortho,\ortho}]_{\indyj\indyk}^{\indyi}\;=\; 
\sum_\indyl \Py ^{-1}_{\indyi\indyl}
\left(
\frac{\partial \Py _{\indyl\indyk}}{\partial y_\indyj}
+
\frac{\partial  \Py _{\indyl\indyj}}{\partial y_\indyk}
-
\frac{\partial  \Py _{\indyj\indyk}}{\partial y_\indyl}
\right)
-
\sum_\indyl 
 [P_y^{-1}]_{\indyi\indyl}
\left(
\frac{\partial [P_y]_{\indyl\indyk}}{\partial y_\indyj}
+
\frac{\partial [P_y]_{\indyl\indyj}}{\partial y_\indyk}
-
\frac{\partial [P_y]_{\indyj\indyk}}{\partial y_\indyl}
\right)
$\hfill \null \\\null \hfill $\displaystyle 
+
\sum_\indyl 
 [P_y^{-1}]_{\indyi\indyl}
\sum_{\indxra,\indxrb}
[P_{\xrond\xrond}^{-1}]_{\indxra\indxrb}
\left(
\frac{\partial [P_y]_{\indyl\indyk}}{\partial \xrond_\indxra}
P_{\indxrb\indyj}
+
\frac{\partial [P_y]_{\indyl\indyj}}{\partial \xrond_\indxra}
P_{\indxrb\indyk}
-
\frac{\partial [P _y]_{\indyj\indyk}}{\partial \xrond_\indxra}
P_{\indxrb\indyl} 
\right)
\  ,
$\\[1em]%
}%
where we have denoted
\begin{equation}
\label{LP199}
P_y(y,\xrond) = P_{yy}(y,\xrond) - P_{y\xrond}(y,\xrond)P_{\xrond\xrond}(y,\xrond)^{-1}P_{\xrond y}(y,\xrond)
\;  .
\end{equation}
\end{lemma}

\IfReport{%
\startarchive%
A proof is given in the \complement \ref{complement37}.%
\stoparchive%
}{
A proof of Lemma~\ref{lem8} can be found in \cite{Sanfelice-Praly.III-long}.
}

It follows from this expression that, if $P_y$ in (\ref{LP199}) does not depend on $\xrond$ and is chosen such that it has the same Christoffel 
symbols as those of $\Py$, then condition (\ref{LP162}) holds. These two conditions are trivially
satisfied if we simply have
\IfTwoCol{%
for all $(y,\xrond)$ in $\coordyxrm(\coordyxrd)$,
$$
\Py(y) = P_y(y) =  P_{yy}(y,\xrond) - P_{y\xrond}(y,\xrond)P_{\xrond\xrond}(y,\xrond)^{-1}P_{\xrond y}(y,\xrond)
\  .
$$
}{%
$$
\Py(y)\;=\; P_y(y)\;=\;  P_{yy}(y,\xrond)\;-\; P_{y\xrond}(y,\xrond)P_{\xrond\xrond}(y,\xrond)^{-1}P_{\xrond y}(y,\xrond)
\qquad \forall (y,\xrond)\in \coordyxrm(\coordyxrd)
\  .
$$%
}%
If, instead of the coordinates $(y,\xrond)$, we use coordinates $x$, this equation is
\begin{equation}
\label{LP160}
\left(\frac{\partial h}{\partial x}(x)P(x)^{-1} \frac{\partial h}{\partial x}(x)^\top\right)^{-1}
\;=\; \Py (h(x))
\qquad \forall x\in \coordxm(\coordxd)
\end{equation}
where $h$ and $\Py $ are expressed with the same coordinates.

\begin{definition}[{\cite[axiom S2]{ONeill.66}}]
\label{def4}
A submersion $\bfh$ satisfying (\ref{LP160}) is called a Riemannian submersion.
An equivalent definition is that $\bfh$ is a submersion preserving
length of orthogonal vectors, i.e., we have
\IfTwoCol{%
for all $v^\ortho$ in $\Distrib _P^\ortho(x)$,
$$
v^{\ortho\top}  \frac{\partial h}{\partial x}(x)^\top \Py (h(x))\frac{\partial h}{\partial x}(x) v^\ortho  
= 
v^{\ortho\top}  P(x) v^\ortho
\, .
$$
}{%
$$
v^{\ortho\top}  \frac{\partial h}{\partial x}(x)^\top \Py (h(x))\frac{\partial h}{\partial x}(x) v^\ortho  
\;=\; 
v^{\ortho\top}  P(x) v^\ortho  
\qquad \forall v^\ortho\in \Distrib _P^\ortho(x)
\  .
$$
}
\end{definition}

\IfTwoCol{%
}{%
An interesting consequence of $\bfh$ being a Riemannian submersion is as follows.

\begin{lemma}[{\cite[Lemma 9.44]{Besse.78}}]
\label{lem11}
Assume $\bfh$ is a Riemannian submersion on $\Ouv$.
Any  geodesic
$\bfgamma $, taking values in $\Ouv$ on the maximal interval $(s_1,s_2)$,
for which there exists $s_3$ in $(s_1,s_2)$ such that we have
$$
\frac{d\bfgamma}{ds}(s_3) \in \bfDistrib _\bfP ^\ortho (\bfgamma(s_3))
$$
satisfies
$$
\frac{d\bfgamma}{ds}(s) \in \bfDistrib _\bfP ^\ortho (\bfgamma(s))
\qquad \forall s\in (s_1,s_2)
$$
and $\bfh(\bfgamma)$ is a geodesic in the $\bfy$-manifold $\bfRR^p$.
\end{lemma}
\startarchive%
A proof is given in the \complement \ref{complement47}.%
\par\vspace{0.5em}
\stoparchive%
}

Now, as indicated above, instead of showing the necessity of (\ref{LP162}) for Condition A3 to hold, we answer 
the following question:
\begin{itemize}
\item[(Q2)]
{\it 
If we are given a metric $\bfP$ satisfying (\ref{LP151}) and (\ref{LP161}) but neither (\ref{LP162}) nor (\ref{LP160}),
can we modify it to satisfy the three conditions?}
\end{itemize}
\par\vspace{0.5em}
To answer this question, we propose the modification $\bfP_{mod}$ of the metric $\bfP$, the expression of 
which is, for the coordinate chart $\coordx$, given by
\IfTwoCol{%
\\[0.7em]$\displaystyle 
P_{mod}(x)\;=\; P(x) + 
$\refstepcounter{equation}\label{LP163}\hfill$(\theequation)$
\\[0.3em]\null\hfill$\displaystyle
\frac{\partial h}{\partial x}(x)^\top  \!\!
\left[\Py (h(x)) -
\left(\frac{\partial h}{\partial x}(x) P(x)^{-1}\frac{\partial h}{\partial x}(x)^\top
\right)^{-1}
\right]
\frac{\partial h}{\partial x}(x)
\  ,
$\\[0.7em]
for all $x$ in $\coordxm(\coordxd)$.
}{%
\begin{equation}
\label{LP163}
P_{mod}(x)\;=\; P(x) + \frac{\partial h}{\partial x}(x)^\top  
\left[\Py (h(x)) -
\left(\frac{\partial h}{\partial x}(x) P(x)^{-1}\frac{\partial h}{\partial x}(x)^\top
\right)^{-1}
\right]
\frac{\partial h}{\partial x}(x)
\qquad \forall x\in \coordxm(\coordxd)
\  .
\end{equation}%
}%
\startmodif
This is a positive definite matrix\footnote{%
\startmodif
For any full row rank matrix $H$ and symmetric positive definite matrices $P$ and $Q$,
the matrix
$P+H^\top\!(Q-(HP^{-1}\!H^\top)^{-1})H$ is positive definite.
\stopmodif
}.
\stopmodif
This definition of $\bfP_{mod}$ via its expression with coordinates gives a covariant $2$-tensor that
is invariant under a change of coordinates for $\bfy$ (and $\bfh$). Namely, with $\bar P_{mod}$ defined as
\IfTwoCol{%
\\[0.7em]$\displaystyle 
\bar P_{mod}(\bar x)
\;=\; 
\bar P(\bar x) +
$\hfill\null\\\null\hfill$\displaystyle
\frac{\partial \bar h}{\partial \bar x}(\bar x)^\top \!\!
\left[\bar \Py (\bar h(\bar x))
-
\left(\frac{\partial \bar h}{\partial \bar x}(\bar x) \bar P(\bar x)^{-1}\frac{\partial \bar h}{\partial \bar 
x}(\bar x)^\top\right)^{-1}
\right]
\frac{\partial \bar h}{\partial \bar x}(\bar x)
\  ,
$\\[0.7em]
where we have (see (\ref{LP186}))
\begin{eqnarray*}
&\displaystyle \bar x = \changex(x)\quad ,\qquad
\bar h(\changex(x))= \changey (h(x))
\  ,
\\
&\displaystyle 
\frac{\partial \changex}{\partial x}(x)^\top \bar P(\changex(x)) \frac{\partial \changex}{\partial x}(x)
\;=\; P(x)
\  ,
\\
&\displaystyle 
\frac{\partial \changey}{\partial y}(y)^\top \bar \Py (\changey (y)) \frac{\partial \changey}{\partial y}(y)
\;=\; \Py (y)
\  ,
\end{eqnarray*}
}{%
$$
\bar P_{mod}(\bar x)
\;=\; 
\bar P(\bar x) +
\frac{\partial \bar h}{\partial \bar x}(\bar x)^\top 
\left[\bar \Py (\bar h(\bar x))
-
\left(\frac{\partial \bar h}{\partial \bar x}(\bar x) \bar P(\bar x)^{-1}\frac{\partial \bar h}{\partial \bar 
x}(\bar x)^\top\right)^{-1}
\right]
\frac{\partial \bar h}{\partial \bar x}(\bar x)
\  ,
$$
where we have (see (\ref{LP186}))
\begin{eqnarray*}
&\displaystyle 
\bar x \;=\; \changex(x)
\  ,\quad 
\bar h(\changex(x))\;=\; \changey (h(x))\  ,
\\
&\displaystyle 
\frac{\partial \changex}{\partial x}(x)^\top \bar P(\changex(x)) \frac{\partial \changex}{\partial x}(x)
\;=\; P(x)
\  ,\quad 
\frac{\partial \changey}{\partial y}(y)^\top \bar \Py (\changey (y)) \frac{\partial \changey}{\partial y}(y)
\;=\; \Py (y)
\  ,
\end{eqnarray*}
}%
we obtain
\startarchive%
(see \complement \ref{complement38})%
\stoparchive
\begin{equation}
\label{LP165}
\frac{\partial \changex}{\partial x}(x)^\top
\bar P_{mod}(\bar x)
\frac{\partial \changex}{\partial x}(x)
\;=\; 
P_{mod}(x)
\  .
\end{equation}
The metric $\bfP_{mod}$ has the following properties.

\begin{lemma}
\label{lem9}~
Given the metric $\bfP$, 
let $\bfP_{mod}$ be a metric, the expression of which, with the coordinate chart $\coordx$, is
as in \eqref{LP163}.  The following holds:
\begin{enumerate}
\item
\label{point3}
$\bfP_{mod}$ satisfies condition (\ref{LP160});
i.e., $\bfh$ is a Riemannian submersion with $\bfP_{mod}$.
\item
\label{point4}
Condition \eqref{LP151}
holds for $\bfP$ if and only if it
holds for $\bfP_{mod}$.
\item
\label{point8}
\null \hfill 
$
\bfDistrib _\bfP^\ortho (\bfx) = \bfDistrib _{\bfP_{mod}} ^\ortho (\bfx)\qquad \forall \bfx\in \Ouv \  .
$
\refstepcounter{equation}\label{LP232}\hfill$(\theequation)$
\end{enumerate}
\end{lemma}
\begin{proof}
Item~\ref{point3} follows from the expression of $\bfP_{mod}$ with coordinates $(y,\xrond)$
which is%
\IfTwoCol{%
\\[0.7em]$\displaystyle 
P_{mod}(y,\xrond)
\;=\; 
$\refstepcounter{equation}\label{LP216}\hfill$(\theequation)$
\\[0.5em]\null\hfill$\displaystyle
\left(\begin{array}{cc}
\Py (y) + P_{y\xrond}(y,\xrond)P_{\xrond\xrond}(y,\xrond)^{-1}P_{\xrond y}(y,\xrond)
&
P_{y\xrond}(y,\xrond)
\\
P_{\xrond y}(y,\xrond) & 
P_{\xrond\xrond}(y,\xrond)
\end{array}
\right)
\,  .
$\\
}{%
\begin{equation}
\label{LP216}
P_{mod}(y,\xrond)
\;=\; 
\left(\begin{array}{cc}
\Py (y) + P_{y\xrond}(y,\xrond)P_{\xrond\xrond}(y,\xrond)^{-1}P_{\xrond y}(y,\xrond)
&
P_{y\xrond}(y,\xrond)
\\
P_{\xrond y}(y,\xrond) & 
P_{\xrond\xrond}(y,\xrond)
\end{array}
\right)
\  .
\end{equation}%
}%

For item~\ref{point4},
\startarchive%
(a direct proof is given in the \complement \ref{complement28})%
\stoparchive
we note that condition (\ref{LP151}) is equivalent to (\ref{LP204}) where
the Hessian of $\bfh$ is related to its gradient by (see (\ref{LP156}))
$$
\Hess _P h (x) \;=\; \frac{1}{2}\,  \mathcal{L}_{\grad _P h} P(x)
\qquad \forall x\in \coordxm(\coordxd)
\  .
$$
So the claim follows from the fact that the product rule for Lie differentiation is formally identical with 
the product rule of ordinary differentiation. Indeed the Lie differentiation of (\ref{LP163}) gives
a matrix $M(x)$ satisfying
\IfTwoCol{%
\\[0.7em]$\displaystyle 
 \mathcal{L}_{\grad _{P_{mod}} h} P_{mod}(x)
=
\mathcal{L}_{\grad _P h} P(x)
$\refstepcounter{equation}\label{LP100}\hfill$(\theequation)$
\\[0.5em]\null\hfill\null \hfill $\displaystyle
\vrule height 0pt depth 1.5em width 0pt
+ \;  M(x) \frac{\partial h}{\partial x}(x) + \frac{\partial h}{\partial x}(x)^\top 
M(x)^\top
\qquad \forall x\in \coordxm(\coordxd)
\  .
$
}{%
\begin{equation}
\label{LP100}
 \mathcal{L}_{\grad _{P_{mod}} h} P_{mod}(x)
=
\mathcal{L}_{\grad _P h} P(x)+ M(x) \frac{\partial h}{\partial x}(x) + \frac{\partial h}{\partial x}(x)^\top 
M(x)^\top
\qquad \forall x\in \coordxm(\coordxd)
\  .
\end{equation}%
}%

For item \ref{point8}, from (\ref{LP163}) we obtain that, for any coordinate chart $\coordx$,
\IfTwoCol{%
\\[0.5em]$\displaystyle 
P_{mod}(x)
 P(x)^{-1}\frac{\partial h}{\partial x}(x)^\top
$\hfill\null\\[0.3em]\null \hfill $\displaystyle
=
\frac{\partial h}{\partial x}(x)^\top
\Py (h(x)) 
\left(\frac{\partial h}{\partial x}(x) P(x)^{-1}\frac{\partial h}{\partial x}(x)^\top
\right)
\  .
$\\[0.7em]
}{%
$$
P_{mod}(x)
 P(x)^{-1}\frac{\partial h}{\partial x}(x)^\top
\;=\; 
\frac{\partial h}{\partial x}(x)^\top
\Py (h(x)) 
\left(\frac{\partial h}{\partial x}(x) P(x)^{-1}\frac{\partial h}{\partial x}(x)^\top
\right)
\  .
$$%
}%
Since $\Py (h(x))$ and $ 
\left(\frac{\partial h}{\partial x}(x) P(x)^{-1}\frac{\partial h}{\partial x}(x)^\top
\right)$ are invertible matrices, this establishes that
the columns of
$P(x)^{-1}\frac{\partial h}{\partial x}(x)^\top $ span the same vector space as the columns of
$P_{mod}(x)^{-1}\frac{\partial h}{\partial x}(x)^\top $.
With Lemma \ref{lem6}, this establishes that
the orthogonal distributions  $\Distrib _P ^\ortho  (x)$
and $\Distrib _{P_{mod}} ^\ortho  (x)$ are identical. Since the coordinate chart $\coordx$ is arbitrary, we have (\ref{LP232}).
\end{proof}

With Lemma \ref{lem9} we have answered positively Question (Q2) but only partially.
Indeed, we have not established that
(\ref{LP161}) holds for $\bfP_{mod}$, when (\ref{LP161}) holds for $\bfP$.
As we show next, working directly with $\bfP_{mod}$ allows to establish this property.

\subsubsection{About \eqref{LP161}}
\label{sec18}
Postponing the study of the necessity of (\ref{LP161}) to the next paragraph,
here we study what it implies.

\begin{lemma}[{\cite[Lemma 3.2(ii)]{Vilms}, \cite[Proposition I.5.4]{Nore}}]
\label{lem3}
Assume $\bfh$ is a Riemannian submersion on $\Ouv$.
If, for any point $\bfx_0$ in $\Ouv$, there exists a coordinate chart $\coordx$ around $\bfx_0$
such that (\ref{LP161}) holds then the distribution $\bfDistrib _\bfP ^\ortho  $ is integrable everywhere 
locally on $\Ouv$, i.e.,
for any $\bfx_0$ in $\Ouv$, and for any pair of coordinate charts $\coordx$ around $\bfx_0$
and $\coordy$ around $\bfh(\bfx_0)$
there exists
a $C^s$ function $ h^\ortho :\coordxm(\coordxd)\to \RR^{n-p}$ satisfying
\begin{equation}
\label{LP149}
\frac{\partial  h^\ortho}{\partial x}(x)
P(x)^{-1}\frac{\partial  h}{\partial x}(x)^\top
\;=\; 0
\qquad \forall  x\in \coordxm(\coordxd)
\  ,
\end{equation}
with $ h$ being the expression of $\bfh$ with the coordinates $\coordyp$,
and such that the function
$$
x\mapsto  \hhperp (x)\;=\; ( h(x), h^\ortho (x))
$$
is a diffeomorphism. 
\end{lemma}

\IfReport{%
\startarchive%
A proof is given in the \complement \ref{complement40}.%
\stoparchive%
}{
A proof of Lemma~\ref{lem3} can also be found in \cite{Sanfelice-Praly.III-long}.
}

\par\vspace{0.5em}
In this statement, thanks to Lemma \ref{lem9}, we can omit the assumption that $\bfh$ is a Riemannian 
submersion if we replace $\bfP$ by $\bfP_{mod}$ and (\ref{LP161}) holds for $\bfP_{mod}$.

\par\vspace{0.5em}
Compared with the claim in the Local Submersion Theorem \ref{thm4}, 
the novelty here is in the fact that the function 
$h^\ortho$ satisfies (\ref{LP149}). This is very useful. Indeed we have the following result.

\begin{lemma}
\label{lem2}
Assume $\bfh$ is a Riemannian submersion on $\Ouv$
and, for any $\bfx_0$ in $\Ouv$, and for any pair of coordinate charts $\coordx$ around $\bfx_0$
and $\coordy$ around $\bfh(\bfx_0)$, there exists
a $C^s$ function $ h^\ortho:\coordxm(\coordxd)\to \RR^{n-p}$ satisfying the properties listed in Lemma 
\ref{lem3}. Under these conditions, the expression $\overline{P}$ of $\bfP$ with the coordinates
\begin{equation}
\label{LP155}
( y,\xrond)\;=\; ( h(x), h^\ortho(x))
\end{equation}
is in the following block diagonal form:
\\[0.7em]\null \hfill $\displaystyle 
\overline{P}( y,\xrond)\;=\; 
\left(\begin{array}{cc}
\overline{P}_y( y,\xrond) & 0
\\
0 &  \overline{P}_\xrond ( y,\xrond)
\end{array}\right)
$\refstepcounter{equation}\label{LP154}\hfill$(\theequation)$
\end{lemma}
\begin{proof}
Let $P$ be the expression of $\bfP$ in the coordinates $x$. From (\ref{LP186}), its 
expression $\overline{P}$ in the coordinates $( y,\xrond)$ satisfies
\IfTwoCol{%
for all $x$ in $\coordxm(\coordxd)$,
\\[0.7em]\vbox{\noindent$\displaystyle 
P(x)\;=\; 
$\hfill\null\\[-0.5em]\null\hfill$\displaystyle
\left(\begin{array}{cc@{}}
\displaystyle 
\frac{\partial  h}{\partial x}(x)^\top
&\displaystyle 
\frac{\partial  h^\ortho}{\partial x}(x)^\top
\end{array}\right)
\overline{P}( h(x), h^\ortho(x))
\left(\begin{array}{c}
\displaystyle 
\frac{\partial  h}{\partial x}(x)
\\[0.9em]\displaystyle 
\frac{\partial  h^\ortho}{\partial x}(x)
\end{array}\right)
\  .
$}\\[0.7em]
}{%
$$
P(x)\;=\; 
\left(\begin{array}{cc@{}}
\displaystyle 
\frac{\partial  h}{\partial x}(x)^\top
&\displaystyle 
\frac{\partial  h^\ortho}{\partial x}(x)^\top
\end{array}\right)
\overline{P}( h(x), h^\ortho(x))
\left(\begin{array}{c}
\displaystyle 
\frac{\partial  h}{\partial x}(x)
\\[0.9em]\displaystyle 
\frac{\partial  h^\ortho}{\partial x}(x)
\end{array}\right)
\qquad \forall x\in \coordxm(\coordxd)
\  .
$$%
}%
Post-multiplying by
$P(x)^{-1}
\left(\begin{array}{@{}c@{\ }c@{}}
\displaystyle 
\frac{\partial  h}{\partial x}(x)^\top
&\displaystyle 
\frac{\partial  h^\ortho}{\partial x}(x)^\top
\hskip -3pt
\end{array}\right)$ and exploiting the invertibility of 
$
\left(\begin{array}{@{}c@{\ }c@{}}
\displaystyle 
\frac{\partial  h}{\partial x}(x)^\top
&\displaystyle 
\frac{\partial  h^\ortho}{\partial x}(x)^\top
\hskip -3pt
\end{array}\right)
$,
this gives
\IfTwoCol{%
\\[0.7em]$\displaystyle 
I_n\; =\; \overline{P}( h(x), h^\ortho(x))\times
$\hfill\null\\[0.5em]\null\hfill$\displaystyle
\times
\left(\begin{array}{@{\,  }c@{\,  }}
\displaystyle 
\frac{\partial  h}{\partial x}(x)
\\[0.9em]\displaystyle 
\frac{\partial  h^\ortho}{\partial x}(x)
\end{array}\right)P(x)^{-1}
\left(\begin{array}{@{\,  }cc@{\,  }}
\displaystyle 
\frac{\partial  h}{\partial x}(x)^\top
&\displaystyle 
\frac{\partial  h^\ortho}{\partial x}(x)^\top
\end{array}\right)  ,
$\\[0.9em] $\displaystyle
\hphantom{I_n\;} 
=\; \overline{P}( h(x), h^\ortho(x))\times
$\refstepcounter{equation}\label{LP152}\hfill$(\theequation)$
\\[0.7em]\null\hfill$\displaystyle
\times\!
\left(\begin{array}{@{\,  }cc@{\,  }}
\displaystyle 
\frac{\partial  h}{\partial x}(x) P(x)^{-1}\frac{\partial  h}{\partial x}(x)^\top
&\displaystyle 
\frac{\partial  h}{\partial x}(x) P(x)^{-1}\frac{\partial  h^\ortho}{\partial x}(x)^\top
\\[0.9em]\displaystyle 
\frac{\partial  h^\ortho}{\partial x}(x)P(x)^{-1}\frac{\partial  h}{\partial x}(x)^\top
&\displaystyle 
\frac{\partial  h^\ortho}{\partial x}(x)P(x)^{-1}\frac{\partial  h^\ortho}{\partial x}(x)^\top
\end{array}\right)  .
$\\[0.7em]
}{%
\begin{eqnarray}
\nonumber
I_n&=&
\overline{P}( h(x), h^\ortho(x))
\left(\begin{array}{c}
\displaystyle 
\frac{\partial  h}{\partial x}(x)
\\[0.9em]\displaystyle 
\frac{\partial  h^\ortho}{\partial x}(x)
\end{array}\right)P(x)^{-1}
\left(\begin{array}{cc}
\displaystyle 
\frac{\partial  h}{\partial x}(x)^\top
&\displaystyle 
\frac{\partial  h^\ortho}{\partial x}(x)^\top
\end{array}\right)
\  ,
\\\label{LP152}
&=&
\overline{P}( h(x), h^\ortho(x))
\left(\begin{array}{cc}
\displaystyle 
\frac{\partial  h}{\partial x}(x) P(x)^{-1}\frac{\partial  h}{\partial x}(x)^\top
&\displaystyle 
\frac{\partial  h}{\partial x}(x) P(x)^{-1}\frac{\partial  h^\ortho}{\partial x}(x)^\top
\\[0.9em]\displaystyle 
\frac{\partial  h^\ortho}{\partial x}(x)P(x)^{-1}\frac{\partial  h}{\partial x}(x)^\top
&\displaystyle 
\frac{\partial  h^\ortho}{\partial x}(x)P(x)^{-1}\frac{\partial  h^\ortho}{\partial x}(x)^\top
\end{array}\right)
\  .
\end{eqnarray}
}%
But, when (\ref{LP149}) holds, the last matrix on the right-hand side is block diagonal. Since this matrix is 
the inverse of $\overline{P}$, $\overline{P}$ is also block 
diagonal.
\end{proof}

\subsubsection{About Question (Q1')}
Up to now, we have established that, if Condition A3 holds, then (\ref{LP151}) and (\ref{LP162}) hold, perhaps after changing $\bfP$ 
into $\bfP_{mod}$. On the other hand, we know with Lemmas \ref{lem9} and \ref{lem3} that, if (\ref{LP161}) 
holds for $\bfP_{mod}$ then the
orthogonal distributions $\bfDistrib _\bfP ^\ortho  $ and $\bfDistrib _{\bfP_{mod}} ^\ortho  $
are integrable. It turns out that, conversely, if, in
addition to Condition A3 we have this integrability property, then (\ref{LP161}) holds.

\begin{proposition}
\label{prop13}
If Condition A3 holds with a metric $\bfP$ such that the orthogonal distribution $\bfDistrib _\bfP^\ortho$ is integrable, 
then the second fundamental form of $\bfh$ for the metric $\bfP_{mod}$ is zero
on $\Ouv$.
\end{proposition}

\begin{proof}
It follows from Lemmas \ref{lem3} and \ref{lem2} that, for any $\bfx_0$ in $\Ouv$, there exists a coordinate 
chart $\coordyxr$
(see (\ref{LP155})) such that the expressions $\overline{P}$ and $\overline{P}_{mod}$
of $\bfP$ and $\bfP_{mod}$, respectively, are
(see (\ref{LP216}) and (\ref{LP154}))
\IfTwoCol{%
\begin{eqnarray*}
\overline{P}\coordyxrp 
&=&
\left(\begin{array}{cc}
\overline{P}_y\coordyxrp  & 0
\\
0 &  \overline{P}_\xrond \coordyxrp 
\end{array}\right)
\  ,
\\
\overline{P}_{mod}\coordyxrp 
&=&
\left(\begin{array}{cc}
\Py(y)  & 0
\\
0 & \overline{P}_\xrond \coordyxrp 
\end{array}\right)
\  .
\end{eqnarray*}
}{%
$$
\overline{P}\coordyxrp 
\;=\; 
\left(\begin{array}{cc}
 \overline{P}_y\coordyxrp  & 0
\\
0 &  \overline{P}_\xrond \coordyxrp 
\end{array}\right)
\quad ,\qquad 
\overline{P}_{mod}\coordyxrp 
\;=\; 
\left(\begin{array}{cc}
\Py(y)  & 0
\\
0 &  \overline{P}_\xrond \coordyxrp 
\end{array}\right)
\  .
$$%
}%
On the other hand, in these specific coordinates, (\ref{LP151}), implied by Condition A3 (see Lemma \ref{lem12}), is equivalent to
\IfTwoCol{%
$$
\sum_\indyj[ \overline{P}_y^{-1}]_{\indyi\indyj}\frac{\partial 
[\overline{P}_\xrond]_{\indxra\indxrb}}{\partial y_\indyj} = \overline{\Gamma} _{\indxra\indxrb}^\indyi= 0
\  .
$$
This implies
$
\frac{\partial [\overline{P}_\xrond]_{\indxra\indxrb}}{\partial y_\indyk}
$
is zero.
}{%
$$
\overline{\Gamma} _{\indxra\indxrb}^\indyi\;=\; 0
$$
and the block diagonal form of $\overline{P}$ gives
$$
\overline{\Gamma} _{\indxra\indxrb}^\indyi
\;=\; 
\sum_\indyj[ \overline{P}_y^{-1}]_{\indyi\indyj}
\frac{\partial [\overline{P}_\xrond]_{\indxra\indxrb}}{\partial y_\indyj}
\  .
$$
This yields
$$
\frac{\partial [\overline{P}_\xrond]_{\indxra\indxrb}}{\partial y_\indyk}
\;=\; \sum_{\indyi} [\overline{P}_y]_{\indyk\indyi}\Gamma _{\indxra\indxrb}^\indyi\;=\; 0\ .
$$
}
Hence $\overline{P}_\xrond$ does not depend on $y$ and its associated
Christoffel symbols are
\IfTwoCol{%
\\[0.7em]$\displaystyle 
[\overline{\Gamma}_{mod}] _{\indyj\indyk}^\indyi \coordyxrp 
$\hfill\null\\\null\hfill$\displaystyle
=\; \frac{1}{2}\sum_\indyl[\Py (y)^{-1}]_{\indyi\indyl}
\left(
\frac{\partial \Py _{\indyl\indyj}}{\partial y_\indyk}(y)
+
\frac{\partial \Py _{\indyl\indyk}}{\partial y_\indyj}(y)
-
\frac{\partial \Py _{\indyj\indyk}}{\partial y_\indyl}(y)
\right)\  ,
$\hfill\null\\\null\hfill$\displaystyle
\;=\;  \Gammay_{\indyj\indyk}^\indyi(y)
\  ,
$\quad \null 
$$
[\overline{\Gamma}_{mod}]  _{\indyj\indxra}^\indyi\coordyxrp
\;=\; 0\quad , \qquad
[\overline{\Gamma}_{mod}]  _{\indxra\indxrb}^\indyi \coordyxrp
\;=\;  0
\  .
$$%
}{%
$$
\renewcommand{\arraystretch}{2}
\begin{array}{rcl}
[\overline{\Gamma}_{mod}] _{\indyj\indyk}^\indyi \coordyxrp
&=&\displaystyle 
\frac{1}{2}\sum_\indyl[\Py (y)^{-1}]_{\indyi\indyl}
\left(
\frac{\partial \Py _{\indyl\indyj}}{\partial y_\indyk}(y)
+
\frac{\partial \Py _{\indyl\indyk}}{\partial y_\indyj}(y)
-
\frac{\partial \Py _{\indyj\indyk}}{\partial y_\indyl}(y)
\right)
\;=\; \overline{\Gammay}_{\indyj\indyk}^\indyi(y)
\  ,
\\{}
[\overline{\Gamma}_{mod}]  _{\indyj\indxra}^\indyi\coordyxrp
&=&0
\  ,
\\{}
[\overline{\Gamma}_{mod}]  _{\indxra\indxrb}^\indyi \coordyxrp
&=&0
\  .
\end{array}
$$%
}%
The result follows from the fact that the nullity of the second fundamental form does not depend on the 
coordinates and that we have
\\[0.7em]\null \qquad $\displaystyle 
\overline{\secff}_{P_{mod}} h_{\indyj\indyk}^\indyi(y,\xrond )
\;=\; \displaystyle 
-[\overline{\Gamma}_{mod}]  _{\indyj\indyk}^\indyi(y,\xrond )+
\Gammay_{\indyj\indyk}^\indyi (y)
\  ,
$\hfill\null\\[0.3em]\null\qquad $\displaystyle
\overline{\secff}_{P_{mod}}h_{\indyj\indxra}^\indyi(y,\xrond )
\;=\; 
-[\overline{\Gamma}_{mod}]  _{\indyj\indxra}^\indyi(y,\xrond )
\  ,
$\refstepcounter{equation}\label{LP148}\hfill$(\theequation)$
\\[0.3em]\null\qquad$\displaystyle
\overline{\secff}_{P_{mod}} h_{\indxra\indxrb}^\indyi(y,\xrond )
\;=\; -[\overline{\Gamma}_{mod}]  _{\indxra\indxrb}^\indyi(y,\xrond )
\  .
\vrule depth 1em height 0pt width 0pt
$%
\end{proof}

This statement is not satisfactory because it requires the extra condition of integrability of the 
orthogonal distribution $\bfDistrib _\bfP^\ortho$. Whether this integrability is implied by Condition A3 is,
for us, an open problem.
Fortunately, as shown below, when the dimension $p$ of the $\bfy$-manifold is $1$, the integrability condition is not needed in the statement of 
Proposition \ref{prop13}. So, in this case,
if Condition A3 
holds, the second fundamental form of $\bfh$ is zero for the metric $\bfP_{mod}$. Namely, our sufficient condition 
is necessary but, perhaps after modifying $\bfP$ into $\bfP_{mod}$. Here we recover in some way \cite[Proposition 
A.3.2.b]{57}.

What we wrote above about the peculiarity of the case $p=1$ is a consequence of the fact that
the assumption of Lemma \ref{lem2} is always satisfied. Indeed, we have the following result.

\begin{lemma}
\label{lem14}
Suppose Assumption~\ref{H1} is satisfied.
If the dimension $p$ of the $\bfy$-manifold $\bfRR^p$ is $1$, and (\ref{LP160}) and Condition A3
\startmodif
hold,
\stopmodif
then
$\bfh$ is a Riemannian submersion, the second fundamental form of which is zero
on $\Ouv$.
\end{lemma}

\begin{proof}
Let  $\bfx_0$ be any point in $\Ouv$ and $\coordx$ around $\bfx_0$ and $\coordy$ around
$\bfh(\bfx_0)$ be any coordinate charts. With all ``objects'' expressed with these coordinates,
let also $\euX(x,t)$ denote the solution at time $t$ of
$$
\dot x = P(x)^{-1}\frac{\partial h}{\partial x}(x)^\top
\  ,
$$
passing through $x$ in $\coordxd$ at time $0$. Because the function $t\mapsto h(\euX(x,t))$ is strictly increasing, there exists an open neighborhood $\cursive{M}$\, of $\coordxm(\bfx_0)$ such that,
for any $x$ in $\cursive{M}$ , there exists a (unique)
$\tau (x)$ satisfying
\IfTwoCol{%
$
h(\euX(x,\tau (x)))\;=\; h(\coordxm(\bfx_0))$
}{%
$$
h(\euX(x,\tau (x)))\;=\; h(\coordxm(\bfx_0))
\  .
$$
}%
Let $h^\comp$  be the submersion given by the Local Submersion Theorem \ref{thm4}
in Section \ref{sec17}. We define a function $h^\ortho:\cursive{M}\,  \to \RR^{n-p}$  as
\IfReport{%
\begin{equation}
\label{LP190}
h^\ortho(x)\;=\; h^\comp (\euX(x,\tau (x)))
\  .
\end{equation}
It satisfies
\startarchive%
(see \complement \ref{complement46})
\stoparchive%
$$
\frac{\partial h^\ortho}{\partial x}(x)
P(x)^{-1}\frac{\partial h}{\partial x}(x)^\top\;=\; 0
\qquad \forall x \in \cursive{M}\  .
$$
Then, with Lemma \ref{lem2}, the coordinates
$$
( y,\xrond)\;=\; (h(x),h^\ortho(x))
\  ,
$$
}{%
\\\null \hfill $\displaystyle 
h^\ortho(x)\;=\; h^\comp (\euX(x,\tau (x)))
\  .
$\refstepcounter{equation}\label{LP190}\hfill$(\theequation)$
\\
It satisfies
$$
\frac{\partial h^\ortho}{\partial x}(x)
P(x)^{-1}\frac{\partial h}{\partial x}(x)^\top\;=\; 0
\qquad \forall x \in \cursive{M}\  .
$$
Then, with Lemma \ref{lem2}, the coordinates
\\[0.5em]\null \hfill $\displaystyle 
( y,\xrond)= (h(x),h^\ortho(x))
$\hfill \null \\[0.5em]
}%
defined on $\cursive{M}$  ,
are such that the expression $\bar P$ of $\bfP$
has the following block diagonal form
(see \cite[p. 57 \S 19]{Eisenhart.25}  or \cite[Theorem 2.6]{127})
$$
P( y,\xrond)\;=\; 
\left(\begin{array}{cc}
P_y( y,\xrond) & 0
\\
0 & P_\xrond ( y,\xrond)
\end{array}\right)
\  .
$$ 
From here we proceed as in the proof of Proposition \ref{prop13}.
\end{proof}

Because Condition A3 and the nullity of the second fundamental form of $\bfh$ for $\bfP_{mod}$ are equivalent 
when $p=1$, it would be interesting to know if we can always ``massage'' the measurements to go
to this case. In other words, if we have a $p$-dimensional output function for which the observer problem can be 
solved, does there exists a $1$-dimensional function of this output function with which we can still solve the 
observer problem. In the linear case, this result is known as Heymann's Lemma%
\footnote{
We are very grateful to Alessandro Astolfi, from  Imperial College in London, for pointing out this lemma.
}
\cite{Heymann,Hautus}.

\startarchive
\begin{example}~
\label{ex:Lagrangian}
\bgroup\normalfont
As in \cite[Section V]{127}, we consider Lagrangian systems.
The state $x$
is made of $y$, the generalized position of dimension $p$ supposed to be measured,
and $\xrond$ the generalized velocities. Its dynamics are
\begin{equation}
\label{LP49}
\dot y _\indyi\;=\; \xrond_\indyi
\quad ,\qquad 
\dot \xrond_\indxra  \;=\; -\sum_{\indxrc ,\indxrd }\mathfrak{C} _{\indxrc\indxrd}^\indxra (y)\,  \xrond_\indxrc  
\xrond_\indxrd  \;+\; S_\indxra  (y,t)
\end{equation}
where
$S$ is a source term, a known time-varying vector field on $\reals^p$, and
$\mathfrak{C}_{\indxrc  \indxrd  }^\indxra  $ are the Christoffel symbols associated 
to the metric $g$ attached to the Lagrangian system.

We know from \cite[Section V]{127} that Condition A2 holds for 
the following modification of the Sasaki metric (see \cite[(3.5)]{Sasaki} or
\cite[page 55]{Sakai.96} or \cite[Section 1.K]{Besse.78}):
$$
P(y,\xrond)\;=\; \left(\begin{array}{cc}
P_{yy}(y,\xrond)& P_{y\xrond }(y,\xrond)
\\
P_{\xrond y}(y,\xrond) & P_{\xrond\xrond}(y,\xrond)
\end{array}\right) 
\  ,
$$
where the  entries of the $(p,p)$ blocks
$P_{yy}$, $P_{y\xrond}$,  $P_{\xrond y }$, and $P_{\xrond\xrond}$ are,
respectively, $P_{\indyi\indyj}$, $P_{\indyi\indxrb }$, $P_{\indxra  \indyj}$, and $P_{\indxra  \indxrb }$,
defined as
\IfReport{%
(see more details in the \complement \ref{complement25})%
}{%
}
\IfTwoCol{%
\begin{eqnarray}
\nonumber
P_{\indyi\indyj}(y,\xrond ) &\hskip -0.5em =&\hskip -0.5em 
ag_{\indyi\indyj}(y)
\\\nonumber  
 &\hskip -0.5em &\hskip -1em 
-c\left(
\sum_{\indxrg,\indxrf} g_{j{\indxrg }}(y)\mathfrak{C} _{{\indxrf }i}^{\indxrg }(y) \xrond _{\indxrf } 
+
\sum_{\indxrh,\indxre}g_{i{\indxrh }}(y)\mathfrak{C} _{{\indxre }j}^{\indxrh }(y) \xrond _{\indxre } 
\right)
\\\nonumber  
 &\hskip -0.5em &\hskip -1em 
 + b\sum_{\indxrg,\indxrh,\indxre,\indxrf}g_{{\indxrg}{\indxrh}}(y)
\mathfrak{C} _{{\indxrf }i}^{\indxrg}(y)
\mathfrak{C} _{{\indxre }j}^{\indxrh}(y) \xrond _{\indxrf } \xrond _{\indxre }
\  ,
\\
\label{LP126}
P_{\indyi\indxrb }(y,\xrond )&\hskip -0.5em =&\hskip -0.5em 
-c g_{\indyi\indxrb }(y)
+b\sum_{\indxrf,\indxrg}g_{\indxrb \indxrg}(y)\mathfrak{C} _{{\indxrf }i}^{\indxrg}(y) \xrond _{\indxrf }
\  ,
\\[-0.2em]
\nonumber
P_{\indxra  \indyj}(y,\xrond ) &\hskip -0.5em =&\hskip -0.5em 
-cg_{\indxra  j}(y)
+b \sum_{\indxre,\indxrh}g_{\indxra  {\indxrh}}(y)\mathfrak{C} _{{\indxre }j}^{\indxrh }(y) \xrond _{\indxre }
\  ,
\\[-0.5em]
\nonumber
P_{\indxra  \indxrb }(y,\xrond )&\hskip -0.5em =&\hskip -0.5em 
bg_{\indxra  \indxrb }(y)
\  ,
\end{eqnarray}%
}{%
\begin{eqnarray}
\nonumber
P_{\indyi\indyj}(y,\xrond ) &=&
ag_{\indyi\indyj}(y)
-c\left(
\sum_{\indxrg,\indxrf} g_{j{\indxrg }}(y)\mathfrak{C} _{{\indxrf }i}^{\indxrg }(y) \xrond _{\indxrf } 
+
\sum_{\indxrh,\indxre}g_{i{\indxrh }}(y)\mathfrak{C} _{{\indxre }j}^{\indxrh }(y) \xrond _{\indxre } 
\right)
\\\nonumber  
&&%
\setbox0=\hbox{$%
\null \hskip 1cm P_{\indyi\indyj}(y,\xrond )%
\hskip 2\arraycolsep%
=
\hskip 2\arraycolsep \null $}%
\setlength{\longueur}{\linewidth}
\addtolength{\longueur}{-\wd0}
\setbox0=\hbox{$%
\displaystyle + b\sum_{\indxrg,\indxrh,\indxre,\indxrf}g_{{\indxrg}{\indxrh}}(y)
\mathfrak{C} _{{\indxrf }i}^{\indxrg}(y)
\mathfrak{C} _{{\indxre }j}^{\indxrh}(y) \xrond _{\indxrf } \xrond _{\indxre }
\  ,\hskip 1cm\null 
$}
\addtolength{\longueur}{-\wd0}
\hskip \longueur \box0
\\[0.5em]
\label{LP126}
P_{\indyi\indxrb }(y,\xrond )&=&
-c g_{\indyi\indxrb }(y)
+b\sum_{\indxrf,\indxrg}g_{\indxrb \indxrg}(y)\mathfrak{C} _{{\indxrf }i}^{\indxrg}(y) \xrond _{\indxrf }
\  ,
\\[0.5em]\nonumber
P_{\indxra  \indyj}(y,\xrond ) & =&
-cg_{\indxra  j}(y)
+b \sum_{\indxre,\indxrh}g_{\indxra  {\indxrh}}(y)\mathfrak{C} _{{\indxre }j}^{\indxrh }(y) \xrond _{\indxre }
\  ,\\[0.5em]\nonumber
P_{\indxra  \indxrb }(y,\xrond )&=&
bg_{\indxra  \indxrb }(y)
\  ,
\end{eqnarray}%
}%
where $a$, $b$ and $c$ are strictly positive numbers such that $c^2 < ab$.

It is established in  \cite[Proposition 1.102, page 47]{Besse.78} and \cite[Problem II.10 page 79]{Sakai.96}
that  condition (\ref{LP151}) is satisfied.

Regarding condition (\ref{LP162}), we know it is satisfied if
$$
P_y(y,\xrond)\;=\; P_{yy}(y,\xrond)-
P_{y\xrond}(y,\xrond)P_{\xrond\xrond}(y,\xrond)^{-1} P_{\xrond y}(y,\xrond)
$$
does not depend on $y$ and we choose $\Py = P_y$ to define the distance leading to the function $\wpunbf$.
It is satisfied since we get
\IfReport{%
(see more details in the \complement \ref{complement39})%
}{%
}%
\begin{equation}
\label{LP167}
P_y\;=\; \left(a-\frac{c^2}{b}\right) g
\end{equation}

So, from the proof of Proposition \ref{prop13}, and Proposition \ref{prop12},
if the corresponding orthogonal distribution
$\bfDistrib _\bfP^\ortho$ is integrable,
Condition A3 holds. Unfortunately, as is proved in \cite[Theorem 1]{Sasaki} and can be checked by direct 
computations with components%
\IfReport{%
\ (see \complement \ref{complement27})%
}{%
}%
, the integrability condition holds if and only if 
the Riemannian metric $g$ is flat.
\egroup\end{example}
\stoparchive

\subsection{Construction of the metric $\bfP$}
\IfReport{%
Having a metric $\bfP$ making $\bfh$ a Riemannian submersion, the second fundamental form of which is zero is 
sufficient for Condition A3 to hold. It is also necessary when the dimension $p$ of $\bfy$ is $1$. So, beyond the
test given by (\ref{LP145}),
it is of prime importance to have a design tool for getting such a metric $\bfP$.
This design
}{%
A design of an appropriate metric $\bfP$
}%
will be provided by
another necessary condition for having the second fundamental form zero when $\bfh$ is a 
Riemannian submersion.


\begin{lemma}
\label{lem13}
If $\bfP$ \startmodif and $\bfPy$ are \stopmodif complete and $\bfh$ is a Riemannian submersion on $\bfRR^n$, 
the second fundamental form of which is zero on $\bfRR^n$, then $\bfh$ is surjective
and there exist
an $n-p$ dimensional $C^s$ manifold $\bfXrond_h $, a surjective submersion
$\bfh^\ortho:\bfRR^n\to \bfXrond_h $ and \startmodif
a
\stopmodif complete metric $\bfPxi$ on $\bfXrond_h$
such that $\bfhhperp=(\bfh,\bfh^\ortho): \bfRR^n \to \bfRR^p \times \bfXrond_h$ is
\IfReport{%
a $C^s$ diffeomorphism and
$\bfP$ is the pullback via $\bfhhperp$ of the product metric denoted $\bfPy \oplus \bfPxi$. 
Namely, if $(x,\RR^n,\phi)$ and $(y,\RR^p,\chi)$ are globally defined coordinate charts,
then, for any  coordinate chart $\coordxr $ for $\bfXrond_h$ used to express $\bfh^\ortho$ and $\bfPxi$, we have
\begin{equation}
\label{LP166}
P(x)\;=\;
\frac{\partial h}{\partial x}(x)^\top  \Py  ( h(x)) \,  \frac{\partial h}{\partial x}(x)
+
\frac{\partial h^\ortho}{\partial x}(x)^\top \Pxi ( h^\ortho(x))\,  \frac{\partial h^\ortho}{\partial x}(x)
\qquad 
\forall x\in  [h^\ortho]^{-1}(\coordxrd)
\  ,
\end{equation}%
}{
a $C^s$ diffeomorphism. Moreover, if $(x,\RR^n,\phi)$ and $(y,\RR^p,\chi)$ are globally defined coordinate charts,
then, for any  coordinate chart $\coordxr $ for $\bfXrond_h$ used to express $\bfh^\ortho$ and $\bfPxi$, we 
have 
for all $x$ in $[h^\ortho]^{-1}(\coordxrd)$,
\\[0.7em]$\displaystyle 
P(x)\;=\;
\frac{\partial h}{\partial x}(x)^\top  \Py  ( h(x)) \,  \frac{\partial h}{\partial x}(x)
$\refstepcounter{equation}\label{LP166}\hfill$(\theequation)$
\\\null\hfill$\displaystyle
+
\frac{\partial h^\ortho}{\partial x}(x)^\top \Pxi ( h^\ortho(x))\,  \frac{\partial h^\ortho}{\partial x}(x)
\  ,
$\\[0.7em]
}%
where $P$, $h$, $h^\ortho$, $\hhperp$, $\Py$, and $\Pxi$ denote expressions in the above mentioned 
coordinates.
\end{lemma}

Lemma~\ref{lem13} is a specific version of the more general \cite[Corollary 3.7]{Vilms}.
\IfReport{%
For the sake of completeness we give a direct proof in Appendix \ref{sec19}
\startmodif
but under more restrictive assumptions of boundedness of the function $\bfy\mapsto \bfPy(\bfy)$ and 
commutation of particular vector fields spanning $\bfDistrib _\bfP ^\ortho$ instead of involutivity.
\stopmodif
}{
A proof can be found in \cite{Sanfelice-Praly.III-long}
\startmodif
under more restrictive assumptions of boundedness of the function $\bfy\mapsto \bfPy(\bfy)$ and, instead of involutivity,  
commutation of particular vector fields spanning $\bfDistrib _\bfP ^\ortho$.
\stopmodif
}

\startmodif
Actually (\ref{LP166}) provides a procedure to construct a metric $\bfP$ making the second fundamental form
of $\bfh$ zero and, consequently, satisfying Condition A3.
The following theorem presents this construction and the forthcoming Example~\ref{ex7} illustrates it.
\stopmodif

\begin{theorem}
\label{prop17}
Suppose Assumption~\ref{H1} holds.
Assume there exist
\begin{list}{}{%
\parskip 0pt plus 0pt minus 0pt%
\topsep 0.5ex plus 0pt minus 0pt%
\parsep 0pt plus 0pt minus 0pt%
\partopsep 0pt plus 0pt minus 0pt%
\itemsep 0pt plus 0pt minus 0pt%
\settowidth{\labelwidth}{\quad ii)}%
\setlength{\labelsep}{0.5em}%
\setlength{\leftmargin}{\labelwidth}%
\addtolength{\leftmargin}{\labelsep}%
}
\item[i)]
a metric $\bfPy$ for $\bfRR^p$ satisfying 
Assumption~\ref{H2};
\item[ii)]
with $q\geq 0$,
an $n-p+q$-dimensional $C^s$ manifold $\bfXi$ equipped with a 
metric $\bfPxi$;
\item[iii)]
a $C^s$ function $\bfh^\ortho:\Ouv\to \bfXi$ is ,  with rank $n-p$ on $\Ouv$
such that  $\bfhhperp=(\bfh,\bfh^\ortho)$ has rank $n$ on $\Ouv$.
\end{list}
Then,
the metric $\bfP$ defined on $\Ouv$ as the pull back via $\bfhhperp$ of the product metric $\bfPy \oplus \bfPxi$
(see its expression with coordinates in (\ref{LP166}))
is such that
$\bfh$ is a Riemannian submersion with a second fundamental form that is zero on $\Ouv$. Furthermore,
Condition A3 holds when $\wpunbf$ is the square of the 
distance given by $\bfPy$.
\end{theorem}


\begin{proof}
It follows from our assumptions and the Rank Theorem that the restriction of $\bfh^\ortho$ to $\Ouv$ is a subimmersion, i.e.,
for each $\bfx_0$ in $\Ouv$, there exist an open neighborhood ${\cursive{M}\mkern 4mu _0}$ of $\bfx_0$, a $C^s$ manifold
$\bfXrond_0$ of dimension $n-p$, a submersion $\bfsubmer_0:{\cursive{M}\mkern 4mu _0}\to \bfXrond_0$ and an immersion
$\bfimmer_0:\bfXrond_0\to\bfXi$ satisfying
$$
\bfh^\ortho(\bfx)\;=\;\bfimmer_0  \left(\bfsubmer_0(\bfx)\right)
\qquad \forall \bfx\in {\cursive{M}\mkern 4mu _0}
$$
The index $0$ is used here to insist on the fact that all the corresponding objects are $\bfx_0$ dependent.

Let $\coordxi$ be a coordinate chart around $\bfh^\ortho(\bfx_0)$ in $\bfXi$,
$\coordxr$ be a coordinate chart around $\bfsubmer_0(\bfx_0)$ in $\bfXrond_0$
and $\coordx$ be a coordinate chart around $\bfx_0$ in $\Ouv$ with
$$
\coordxd \subset {\cursive{M}\mkern 4mu _0}
\quad ,\qquad 
\bfsubmer_0(\coordxd )\subset \coordxrd
\quad ,\qquad 
\bfimmer_0  (\coordxrd) \subset \coordxid
\  .
$$
We have
$$
\coordxip\;=\; h^\ortho(\coordxp)\;=\; \immer_0  (\coordxrp)
\  ,\qquad 
\coordxrp\;=\; \submer_0 (\coordxp)
\qquad 
\forall \coordxp\in \coordxm(\coordxd)
\  .
$$
With the function $(\bfh,\bfh^\ortho)$ having rank $n$, we get
\IfTwoCol{
$$
\renewcommand{\arraystretch}{2}
n\;=\; 
\textsf{Rank}\left(
\begin{array}{cc}
I_p & 0
\\
0 & \displaystyle \frac{\partial \immer_0  }{\partial \coordxrp}(\submer_0(\coordxp))
\end{array}\right)
\left(\begin{array}{c}
\displaystyle \frac{\partial h}{\partial\coordxp}(\coordxp)
\\\displaystyle 
\frac{\partial \submer_0}{\partial \coordxp}(\coordxp)
\end{array}\right)
$$
}
{
$$
\renewcommand{\arraystretch}{2}
n\;=\; \textsf{Rank}\left(
\begin{array}{c}
\displaystyle \frac{\partial h}{\partial\coordxp}(\coordxp)
\\\displaystyle 
\frac{\partial h^\ortho}{\partial \coordxp}(\coordxp)
\end{array}\right)
\;=\; \textsf{Rank}\left(
\begin{array}{cc}
I_p & 0
\\
0 & \displaystyle \frac{\partial \immer_0  }{\partial \coordxrp}(\submer_0(\coordxp))
\end{array}\right)
\left(\begin{array}{c}
\displaystyle \frac{\partial h}{\partial\coordxp}(\coordxp)
\\\displaystyle 
\frac{\partial \submer_0}{\partial \coordxp}(\coordxp)
\end{array}\right)
$$
}%
where $\frac{\partial \immer_0  }{\partial \coordxrp}(\submer_0(\coordxp))$ is an $((n-p+q),(n-p))$ matrix of rank 
$n-p$.
This implies the square matrix 
\IfReport{%
$$
\renewcommand{\arraystretch}{2}
\left(\begin{array}{c}
\displaystyle \frac{\partial h}{\partial\coordxp}(\coordxp)
\\\displaystyle 
\frac{\partial \submer_0}{\partial \coordxp}(\coordxp)
\end{array}\right)
$$
}{%
$$
\left(\begin{array}{@{\,  }cc@{\,  }}
\displaystyle \frac{\partial h}{\partial\coordxp}(\coordxp)^\top
&\displaystyle 
\frac{\partial \submer_0}{\partial \coordxp}(\coordxp)^\top
\end{array}\right)
$$
}%
is invertible. It follows that
$x \mapsto (y,\coordxrp)= (h(x),\submer_0(\coordxp))
$
is a diffeomorphism and $(y,\coordxrp)$ can be used as coordinates for $\bfx$ in a neighborhood of $\bfx_0$.
According to (\ref{LP166}), the expression $P$ in 
these coordinates of the metric $\bfP$ is
\IfTwoCol{
\\[0.5em]$\displaystyle 
P(y,\coordxrp)\;=\;
$\hfill\null\\[0.3em]\null\hfill$\displaystyle
\left(
\begin{array}{cc}
I_p & 0
\\
0 & \displaystyle \frac{\partial \immer_0 }{\partial \coordxrp}(\coordxrp)^\top
\end{array}
\right)\!
\left(\begin{array}{cc}
\Py(y) & 0
\\
0 & \Pxi(\immer_0 (\coordxrp))
\end{array}\right) \!
\left(
\begin{array}{cc}
I_p & 0
\\
0 & \displaystyle \frac{\partial \immer_0 }{\partial \coordxrp}(\coordxrp)
\end{array}
\right)
$\hfill\null\\[0.2em]\null\hfill$\displaystyle
=
\left(\begin{array}{cc}
\Py(y)  & 0
\\
0 & P_\coordxrp(\coordxrp)
\end{array}\right)
$\\
where
\\\null \hfill $\displaystyle 
P_\coordxrp(\coordxrp)
\;=\; 
\frac{\partial \immer_0 }{\partial \coordxrp}(\coordxrp)^\top
\Pxi(\immer_0 (\coordxrp))
\frac{\partial \immer_0  }{\partial \coordxrp}(\coordxrp)
\  .
$\hfill \null \\[0.7em]
}
{
\begin{eqnarray*}
P(y,\coordxrp)
&=&
\left(
\begin{array}{cc}
I_p & 0
\\
0 & \displaystyle \frac{\partial \immer_0 }{\partial \coordxrp}(\coordxrp)^\top
\end{array}
\right)
\left(\begin{array}{cc}
\Py(y) & 0
\\
0 & \Pxi(\immer_0 (\coordxrp))
\end{array}\right)
\left(
\begin{array}{cc}
I_p & 0
\\
0 & \displaystyle \frac{\partial \immer_0 }{\partial \coordxrp}(\coordxrp)
\end{array}
\right)
\\
&=&
\left(\begin{array}{cc}
\Py(y)  & 0
\\
0 & P_\coordxrp(\coordxrp)
\end{array}\right)
\end{eqnarray*}
where
$$
P_\coordxrp(\coordxrp)
\;=\; 
\frac{\partial \immer_0 }{\partial \coordxrp}(\coordxrp)^\top
\Pxi(\immer_0 (\coordxrp))
\frac{\partial \immer_0  }{\partial \coordxrp}(\coordxrp)
\  .
$$
}%
From here, the proof can be concluded as in the proof of Proposition \ref{prop13}.
\end{proof}

\begin{remark}~
\normalfont
\begin{list}{}{%
\parskip 0pt plus 0pt minus 0pt%
\topsep 0pt plus 0pt minus 0pt
\parsep 0pt plus 0pt minus 0pt%
\partopsep 0pt plus 0pt minus 0pt%
\itemsep 0pt plus 0pt minus 0pt
\settowidth{\labelwidth}{\:  1)\:}%
\setlength{\labelsep}{0.5em}%
\setlength{\leftmargin}{\labelwidth}%
\addtolength{\leftmargin}{\labelsep}%
}
\item[1)]
The restriction that $\bfh^\ortho$ has rank $n-p$ is crucial. A counterexample is given by the metric 
(\ref{LP136}), which does not make the level sets of the output function totally geodesic. Indeed,
in \eqref{LP136} we have that
$n=3$ and $p=1$, but $\bfh^\ortho$ has (generic) rank $3$.
\item[2)]
Formula (\ref{LP166}) is remarkable
\startmodif
because of the decomposition of $P$ as a sum. In the upcoming Section~\ref{sec21}, we observe that it
\stopmodif 
simplifies the verification of Condition A2.
\item[3)]
The family of metrics given by (\ref{LP166}) would exactly correspond to the one of those making the second fundamental 
form $\bfh$ zero if we were not imposing the extra condition that $\bfh$ is a Riemannian submersion.
\item[4)]
The metric $\bfP$ given by Theorem~\ref{prop17} is defined only on $\Ouv$ and we do not claim it is complete.
\item[5)]
Once the manifold $\bfXi$ is chosen, the existence of the function $\bfh^\ortho$ satisfying the 
conditions is not guaranteed. Indeed, as a 
consequence of Lemma \ref{lem13}, $\bfh$ and $\bfXi$ cannot be arbitrary. For example $\bfh$ must be surjective and its level sets must be 
diffeomorphic to each other.
\startmodif
Also, $\bfXi$ may not be minimal in terms of dimension and there should exist
an immersion between
$\bfXrond_h$ and $\bfXi$. We illustrate this point in the 
following example.
\stopmodif
\end{list}
\end{remark}

\begin{example}\bgroup\normalfont
\label{ex7}
Let the $\bfx$-manifold be $\bfRR^2$ equipped with globally defined coordinates $(x_1,x_2)$. Let $\bfRR$ be 
the $\bfy$-manifold equipped with a globally defined coordinate $y$. This yields $n=2$ and $p=1$.
The output function $\bfh$ we consider is, when expressed in these coordinates,
$$
y\;=\; h(x) := x_1^2 + x_2^2
\  .
$$
It is a submersion on $ \Ouv:=  \bfRR^2\setminus\{0\}$. The level sets of this output function
are diffeomorphic with the unit circle $\SS^1$.

To design a metric $\bfP$ satisfying Condition A3, we follow the lines of
Theorem~\ref{prop17}. We could select $\bfXi$ as a connected 1-dimensional manifold,
this implying
$q = 0$ in Theorem~\ref{prop17}.
\IfReport{
However,
as shown in the \complement \ref{complement49},
with such a choice, 
if, as required by the assumptions
in Theorem~\ref{prop17}, there exists
a function $\bfh^\ortho: \RR^2\setminus\{0\} \to \bfXi$
such that the rank of $\bfhhperp=(\bfh,\bfh^\ortho)$ is $2$ on $\RR^2\setminus\{0\}$, 
then the manifold $\bfXi$ is necessarily diffeomorphic to $\bfSS^1$.
We are reluctant about choosing 
$\bfXi$ as $\SS^1$ because there is no global chart for this manifold, making tricky the numerical implementation 
in applications.
}{%
}
Instead, we select $\bfXi$ as $\bfRR^2\setminus\{0\}$,
in which $\SS^1$ is embedded. Then, $q$ in Theorem~\ref{prop17}, is equal to $1$ and
$\bfh^\ortho:\Ouv \to \bfRR^2\setminus\{0\}$
is to be chosen with rank $1$.
\startmodif
Then, we choose coordinates for $\bfx$ and $\bfy$. For $\bfx$, we keep
those defined above, namely, $(x_1,x_2)$. To get an extra degree of freedom, for $\bfy$ we change, via a $C^s$
function $\changey:\RR\to \RR$ with nonvanishing derivative, the 
original coordinate $y$ in
\stopmodif
$$
\bar y \;=\; \changey(y)
\  .
$$
Let also $h^\ortho=(h_\indxra^\ortho,h_\indxrb^\ortho)$ be the 
expression of the function $\bfh^\ortho$ we are interested in. For the needs of
Theorem~\ref{prop17}, by letting
$$
\bar h(x) \;=\; \changey(h(x))
\  ,
$$
since $n = 2$ and $p =1$, we want, for each $x \in \RR^2\setminus\{0\}$,
\IfTwoCol{
\\[0.6em]\null \hfill $\displaystyle 
\renewcommand{\arraystretch}{1.5}
\textsf{Rank}\left(\begin{array}{@{\,  }c@{\,  }}
\frac{\partial \bar h}{\partial x}(x)
\\
\frac{\partial h^\ortho}{\partial x}(x)
\end{array}\right)
\;=\; 2 \  , \quad 
\renewcommand{\arraystretch}{1.5}
\textsf{Rank}\left(
\frac{\partial h^\ortho}{\partial x}(x)
\right)
\;=\; 1
$\hfill \null \\[0.6em]
}{%
\begin{eqnarray*}
\renewcommand{\arraystretch}{1.5}
\textsf{Rank}\left(\begin{array}{@{\,  }c@{\,  }}
\frac{\partial \bar h}{\partial x}(x)
\\
\frac{\partial h^\ortho}{\partial x}(x)
\end{array}\right)
&=&
\renewcommand{\arraystretch}{1.5}
\textsf{Rank}\left(\begin{array}{@{\,  }cc@{\,  }}
\changey^\prime (x_1^2+x_2^2)x_1 & \changey^\prime (x_1^2+x_2^2)x_2
\\
\frac{\partial h_\indxra^\ortho}{\partial x_1}(x_1,x_2)
&
\frac{\partial h_\indxra^\ortho}{\partial x_2}(x_1,x_2)
\\
\frac{\partial h_\indxrb^\ortho}{\partial x_1}(x_1,x_2)
&
\frac{\partial h_\indxrb^\ortho}{\partial x_2}(x_1,x_2)
\end{array}\right)\;=\; 2
\  ,
\\
\renewcommand{\arraystretch}{1.5}
\textsf{Rank}\left(
\frac{\partial h^\ortho}{\partial x}(x)
\right)
&=&
\renewcommand{\arraystretch}{1.5}
\textsf{Rank}\left(\begin{array}{@{\,  }cc@{\,  }}
\frac{\partial h_\indxra^\ortho}{\partial x_1}(x_1,x_2)
&
\frac{\partial h_\indxra^\ortho}{\partial x_2}(x_1,x_2)
\\
\frac{\partial h_\indxrb^\ortho}{\partial x_1}(x_1,x_2)
&
\frac{\partial h_\indxrb^\ortho}{\partial x_2}(x_1,x_2)
\end{array}\right)\;=\; 1
\end{eqnarray*}
where $\changey^\prime$ is the derivative of $\changey$.
}%
A  solution to these equations is
\begin{eqnarray*}
h_\indxra^\ortho(x_1,x_2)&=&  \frac{k_\indxra(x_1,x_2)}
{\sqrt{k_\indxra(x_1,x_2)^2 + k_\indxrb(x_1,x_2)^2}}
\  ,
\\ 
h_\indxrb^\ortho(x_1,x_2)&=& \frac{k_\indxrb(x_1,x_2)}
{\sqrt{k_\indxra(x_1,x_2)^2 + k_\indxrb(x_1,x_2)^2}}
\end{eqnarray*}
where $(k_1,k_2): \RR^2 \setminus\{0\} \to \RR^2 \setminus\{0\}$ are $C^s$. The rank 1 condition is 
satisfied because of the normalization. To meet the rank 2 
condition, we
\IfTwoCol{
require, for all $(x_1,x_2)$ in $\RR^2\setminus\{0\}$,
\begin{equation}
\label{LP219}
x_1 \left(k_\indxrb \frac{\partial k_\indxra }{\partial x_2}- k_\indxra \frac{\partial 
k_\indxrb}{\partial x_2}\right)
-
x_2 \left(k_\indxrb \frac{\partial k_\indxra }{\partial x_1}- k_\indxra \frac{\partial 
k_\indxrb}{\partial x_1}\right)\;\neq \; 0
\end{equation}
}{%
must satisfy
\begin{equation}
\label{LP219}
x_1 \left[k_\indxrb \frac{\partial k_\indxra }{\partial x_2}- k_\indxra \frac{\partial 
k_\indxrb}{\partial x_2}\right]
-
x_2 \left[k_\indxrb \frac{\partial k_\indxra }{\partial x_1}- k_\indxra \frac{\partial 
k_\indxrb}{\partial x_1}\right]\;\neq \; 0
\qquad \forall (x_1,x_2)\in \RR^2\setminus\{0\}
\end{equation}
}%
In this way  the image of $\bfh^\ortho$ is indeed the unit circle, not as an ``abstract'' 
manifold, but as an immersed submanifold of $\RR^2$.
Then, following Theorem~\ref{prop17} and according to (\ref{LP166}), a metric satisfying Condition A3
on $\RR^2 \setminus\{0\}$ is, with $\Py $ a $C^s$ function with strictly positive values and
$\left(\begin{array}{cc}
\Pxi_{\indxra\indxra} & \Pxi_{\indxra\indxrb} \\ \Pxi_{\indxra\indxrb} & \Pxi_{\indxrb\indxrb}
\end{array}\right)$ a $C^s$ function with positive definite values,
\IfTwoCol{%
\\[0.5em]$\displaystyle 
P(x_1,x_2) = \changey^\prime (x_1^2+x_2^2)^2
\left(\renewcommand{\arraystretch}{1.3}
\begin{array}{@{\,  }c@{\,  }}
\displaystyle 
 x_1
\\
x_2
\end{array}
\right) \Py (\changey (x_1^2+x_2^2) ) 
\left(\begin{array}{@{\,  }cc@{\,  }}
\displaystyle 
 x_1
&
 x_2
\end{array}
\right) 
$\hfill\null\\\null \hfill $\displaystyle
+
\frac{\partial h^\ortho}{\partial x}(x)^\top \Pxi(h^\ortho(x)) \frac{\partial h^\ortho}{\partial x}(x)
$
\\$\displaystyle 
\hphantom{P(x_1,x_2)}=
\changey^\prime (x_1^2+x_2^2)^2
\left(\renewcommand{\arraystretch}{1.3}
\begin{array}{@{\,  }c@{\,  }}
\displaystyle 
 x_1
\\
x_2
\end{array}
\right) \Py (\changey (x_1^2+x_2^2) ) 
\left(\begin{array}{@{\,  }cc@{\,  }}
\displaystyle 
 x_1
&
 x_2
\end{array}
\right) 
$\hfill \null \\[0.3em]\null \hfill $\displaystyle 
+ \mathcursive{K}(x)\widetilde \Pxi (k_\indxra(x),k_\indxrb(x))\mathcursive{K}(x)^\top
$\\[0.3em]
where
\IfReport{
\\[-0.1em]\null \hfill 
}{%
$\changey^\prime$ is the derivative of $\changey$ and
\\[0.3em]\null \hfill  
}%
$\displaystyle
\mathcursive{K}(x)\;=\; 
\left(\renewcommand{\arraystretch}{2.1}\begin{array}{@{\,  }c@{\,  }}
\displaystyle 
k_\indxrb (x)\frac{\partial k_\indxra}{\partial x_1} (x)
-
 k_\indxra (x)\frac{\partial k_\indxrb}{\partial x_1}(x)
\\
\displaystyle 
k_\indxrb (x)\frac{\partial k_\indxra}{\partial x_2} (x)
-
k_\indxra (x)\frac{\partial k_\indxrb}{\partial x_2}(x)
\end{array}\right)
$\hfill\null\\[0.5em]$\displaystyle
\widetilde \Pxi (k_\indxra,k_\indxrb)\;=\; \frac{
k_\indxra^2\Pxi_{\indxrb\indxrb}(k_\indxra,k_\indxrb)
-
2 k_\indxra k_\indxrb\Pxi_{\indxra\indxrb}(k_\indxra,k_\indxrb)
}
{(k_\indxra^2 + k_\indxrb^2)^3}$\\[-0.4em]
\null\hfill $\displaystyle+\frac{
k_\indxrb^2\Pxi_{\indxra\indxra}(k_\indxra,k_\indxrb)
}
{(k_\indxra^2 + k_\indxrb^2)^3}
\  .
$\\%
}{%
\\[1em]$\displaystyle 
P(x_1,x_2)
$\hfill \null 
\\[1em]\vbox{\noindent\null \quad $
\;=\;
\left(\renewcommand{\arraystretch}{1.3}
\begin{array}{@{\,  }c@{\,  }}
\displaystyle 
\changey^\prime (x_1^2+x_2^2) x_1
\\
\changey^\prime (x_1^2+x_2^2) x_2
\end{array}
\right) \Py (\changey (x_1^2+x_2^2) ) 
\left(\begin{array}{@{\,  }cc@{\,  }}
\displaystyle 
\changey^\prime (x_1^2+x_2^2) x_1
&
\changey^\prime (x_1^2+x_2^2) x_2
\end{array}
\right) 
$\hfill \null \\\null \hfill $
+
\left(\renewcommand{\arraystretch}{2.1}
\begin{array}{@{\,  }cc@{\,  }}
\displaystyle 
\frac{\partial h^\ortho _\indxra}{\partial x_1} 
& \displaystyle
\frac{\partial h^\ortho _\indxrb}{\partial x_1}
\\\displaystyle 
\frac{\partial h^\ortho _\indxra}{\partial x_2} 
& \displaystyle
\frac{\partial h^\ortho _\indxrb}{\partial x_2}
\end{array}
\right)
\left(\renewcommand{\arraystretch}{1.3}
\begin{array}{@{\,  }cc@{\,  }}
\displaystyle
\Pxi_{\indxra\indxra}( h^\ortho(x_1,x_2))
& \displaystyle
\Pxi_{\indxra\indxrb}( h^\ortho(x_1,x_2))
\\
\displaystyle
\Pxi_{\indxra\indxrb}( h^\ortho(x_1,x_2))
& \displaystyle
\Pxi_{\indxrb\indxrb}( h^\ortho(x_1,x_2))
\end{array}\right)
\left(\renewcommand{\arraystretch}{2.1}
\begin{array}{@{\, }cc@{\,  }}
\displaystyle 
\frac{\partial h^\ortho _\indxra}{\partial x_1} 
& \displaystyle
\frac{\partial h^\ortho _\indxra}{\partial x_2} 
\\
\displaystyle 
\frac{\partial h^\ortho _\indxrb}{\partial x_1}
& \displaystyle
\frac{\partial h^\ortho _\indxrb}{\partial x_2}
\end{array}
\right)
$}
\\[1em]\vbox{\noindent\null \quad $\displaystyle 
\;=\;
\left(\begin{array}{@{\,  }c@{\,  }}
x_1
\\
x_2
\end{array}\right) \changey^\prime(x^\top x)^2 
\Py(\changey(x^\top x)) 
\left(\begin{array}{@{\,  }cc@{\,  }}
x_1
&
x_2
\end{array}\right)
$\hfill \null \\[0.7em]\null \qquad \qquad $
+
\left(\renewcommand{\arraystretch}{2.1}\begin{array}{@{\,  }c@{\,  }}
\displaystyle 
k_\indxrb (x)\frac{\partial k_\indxra}{\partial x_1} (x)
-
 k_\indxra (x)\frac{\partial k_\indxrb}{\partial x_1}(x)
\\
\displaystyle 
k_\indxrb (x)\frac{\partial k_\indxra}{\partial x_2} (x)
-
k_\indxra (x)\frac{\partial k_\indxrb}{\partial x_2}(x)
\end{array}\right)
\times
$\hfill \null \\[0.7em]\null \qquad \qquad \qquad  $\displaystyle
\times\   \widetilde \Pxi (k_\indxra(x),k_\indxrb(x))\  \times
$\hfill \null \\[0.7em]\null \hfill $
\times
\left(\begin{array}{@{\,  }cc@{\,  }}
\displaystyle 
\left[k_\indxrb (x)\frac{\partial k_\indxra}{\partial x_1} (x)
-
 k_\indxra (x)\frac{\partial k_\indxrb}{\partial x_1}(x)\right]&\displaystyle 
\left[k_\indxrb (x)\frac{\partial k_\indxra}{\partial x_2} (x)
-
k_\indxra (x)\frac{\partial k_\indxrb}{\partial x_2}(x)\right]
\end{array}\right)
\  .
$}
\\[1em]
$$
\widetilde \Pxi (k_\indxra,k_\indxrb)\;=\; \frac{
k_\indxra^2\Pxi_{\indxrb\indxrb}(k_\indxra,k_\indxrb)
-
2 k_\indxra k_\indxrb\Pxi_{\indxra\indxrb}(k_\indxra,k_\indxrb)
+
k_\indxrb^2\Pxi_{\indxra\indxra}(k_\indxra,k_\indxrb)
}
{(k_\indxra^2 + k_\indxrb^2)^3}
\  .
$$
}
\egroup
\startmodif
For example, the particular choice
$$
\changey(s)=s 
\: ,\; 
\Py=1
\: ,\; 
k_\indxra=x_1
\: ,\; 
k_\indxrb=x_2
\: ,\; 
\Pxi=I
$$
gives
\\[-0.5em]\null \hfill $\displaystyle 
P(x_1,x_2)\;=\; \left(
\begin{array}{cc}
x_1^2 + \frac{x_2^2}{(x_1^2+x_2^2)^2}
&
x_1x_2 - \frac{x_1x_2}{(x_1^2+x_2^2)^2}
\\
x_1x_2 - \frac{x_1x_2}{(x_1^2+x_2^2)^2}
&
x_2^2 + \frac{x_1^2}{(x_1^2+x_2^2)^2}
\end{array}\right)
$\hfill \null 
\stopmodif
\end{example}

\section{On Simultaneous Satisfaction of Conditions A2 and A3}
\label{sec:A2andA3}
We have observed that Conditions A2 and A3 are of completely different nature.

\startmodif
The next example shows both of these conditions may not always hold simultaneously.

In this section, we investigate ways, from a design standpoint, to guarantee that Condition A3 holds
when Condition A2 is already satisfied, and vice versa.
\stopmodif

\startmodif
\begin{example}\bgroup\normalfont
Consider the system
\begin{equation}
\label{NLLP2}
\dot x_1 \;=\; 2x_2\quad ,\qquad \dot x_2 = 1-x_1
\quad ,\qquad y=x_1^2+x_2^2
\end{equation}
It is differentially observable of order five. Furthermore, there
exists a globally convergent observer with linear dynamics estimating the output and its four derivatives.
We have seen in Example \ref{ex7} how to construct a metric satisfying Condition A3. Also there exists an expression
$P$ of the metric $\bfP$, which is polynomial of degree $2$ in $(x_1,x_2)$, satisfying Condition A2.

For this system (\ref{NLLP2}), the observer 
(\ref{eqn:GeodesicObserverVectorField}) takes the form
\\\null \hfill $\displaystyle 
\left(\begin{array}{@{\,}c@{\,}}
\dot{\hat x}_1 \\ \dot{\hat x}_2
\end{array}\right)
= 
\left(\begin{array}{@{\,}c@{\,}}
2\hat x_2
\\
 1-\hat x_1
\end{array}\right)
+ k P(\hat x_1,\hat x_2)^{-1}
 \left(\begin{array}{@{\,}c@{\,}}
\hat x_1
\\
\hat x_2
\end{array}\right)  (y-\hat y)
\, .
$\refstepcounter{equation}\label{NLLP12}\hfill$(\theequation)$
\\[0.6em]
With the first order variation formula, this observer leads to a strict decrease of the Riemannian 
distance induced by $P$, if, for any normalized geodesic $\gamma ^*$, minimizing on $[s,\hat s]$,
\IfReport{
\\[1em]$\displaystyle 
\frac{d\gamma ^*}{ds}(\hat s)^\top P(\gamma ^*(\hat s))
\left[
\left(\begin{array}{c}
2\gamma ^*_2(\hat s)
\\
 1-\gamma ^*_1(\hat s)
\end{array}\right) - k P(\gamma ^*(\hat s))^{-1}
\gamma ^*(\hat s)  (|\gamma ^*(\hat s)|^2-|\gamma ^*(s)|^2)
\right]
$\refstepcounter{equation}\label{NLLP6}\hfill$(\theequation)$
\\\null\hfill$\displaystyle
\;-\; 
\frac{d\gamma ^*}{ds}( s)^\top P(\gamma ^*( s))\left(\begin{array}{c}
2\gamma ^*_2( s)
\\
 1-\gamma ^*_1( s)
\end{array}\right)
\; <\; 0
$\\[1em]
}{
\\[0.6em]$\displaystyle 
\frac{d\gamma ^*}{ds}(\hat s)^\top P(\gamma ^*(\hat s))\,  \times
$\refstepcounter{equation}\label{NLLP6}\hfill$(\theequation)$
\\[0.3em]\null\hfill$\displaystyle
\times
\left[
\left(\begin{array}{@{\,}c@{\,}}
2\gamma ^*_2(\hat s)
\\
 1-\gamma ^*_1(\hat s)
\end{array}\right) - k P(\gamma ^*(\hat s))^{-1}
\gamma ^*(\hat s)  (|\gamma ^*(\hat s)|^2-|\gamma ^*(s)|^2)
\right]
$\hfill\null\\\null\hfill$\displaystyle
\;-\; 
\frac{d\gamma ^*}{ds}( s)^\top P(\gamma ^*( s))\left(\begin{array}{@{\,}c@{\,}}
2\gamma ^*_2( s)
\\
 1-\gamma ^*_1( s)
\end{array}\right)
\; <\; 0
$\\[0.6em]
}
The correction term contributes strictly to this decrease if Condition A3 holds, i.e.
\begin{equation}
\label{NLLP10}
|\gamma ^*(\hat s)|\neq |\gamma ^*(s)| \,  \Rightarrow\,  
\frac{d\gamma ^*}{ds}(\hat s)^\top \!\!
\gamma ^*(\hat s)   (|\gamma ^*(\hat s)|^2-|\gamma ^*(s)|^2)
 > 0
\,  .
\end{equation}

Unfortunately, there is no complete metric $\bfP$ such that the observer (\ref{NLLP12}) satisfies (\ref{NLLP6}) and 
(\ref{NLLP10}) together.
Indeed we know with  \cite[A.3.1.a]{57} that (\ref{NLLP10}) implies that, for any $x$ different from the 
origin and any unit vector $v$ tangent at $x$ to the circle with radius $|x|$ and centered at the origin,
the geodesic $\gamma ^*$ satisfying
$$
\gamma ^*(0)\;=\; x\  ,\quad 
\frac{d \gamma ^*}{ds}(0)=v
\  ,
$$
remains in that circle. 
Actually, there are two normalized geodesics issued from $x$, say $\gamma ^*_+$ and $\gamma ^*_-$, satisfying
$$
\frac{d\gamma ^*_+}{ds}(0)\;=\; +v
\quad ,\qquad 
\frac{d\gamma ^*_-}{ds}(0)\;=\; -v
$$
which remain in the circle.
The metric being complete 
by assumption, the orbits of these geodesics are the complete circle and there exist $s_+$ and $s_-$ such 
that\footnote{$\gamma ^*_+(s_+)= \gamma ^*_-(s_-)$ is a cut point of $x$. See \cite[Cut Points ch. 
10]{Lee-RM}.}
\begin{equation}
\label{NLLP13}
\gamma ^*_+(s_+)\;=\; \gamma ^*_-(s_-)
\quad ,\qquad 
\frac{d\gamma ^*_+}{ds}(s_+)=-\frac{d\gamma ^*_-}{ds}(s_-)
\end{equation}
and $\gamma ^*_+$, respectively $\gamma ^*_-$, is minimizing on $[0,s_+]$, respectively $[0,s_-]$.
But if (\ref{NLLP6}) holds, we obtain
\IfReport{%
$$ 
\frac{d\gamma ^*_+}{ds}(s_+)^\top P(\gamma ^*_+(s_+))
\left(\begin{array}{c}
2\gamma ^*_{+2}(s_+)
\\
 1-\gamma ^*_{+1}(s_+)
\end{array}\right)
\;<\; 
v^\top P(x)\left(\begin{array}{c}
2 x_2
\\
 1- x_1
\end{array}\right)
$$
}{%
\\[0.7em]$\displaystyle 
\frac{d\gamma ^*_+}{ds}(s_+)^\top P(\gamma ^*_+(s_+))
\left(\begin{array}{c}
2\gamma ^*_{+2}(s_+)
\\
 1-\gamma ^*_{+1}(s_+)
\end{array}\right)
$\hfill\null\\\null\hfill$\displaystyle
\;<\; v^\top P(x)\left(\begin{array}{c}
2 x_2
\\
 1- x_1
\end{array}\right)
$\\
}
and
\IfReport{%
$$
\frac{d\gamma ^*_-}{ds}(s_-)^\top P(\gamma ^*_-(s_-))
\left(\begin{array}{c}
2\gamma ^*_{-2}(s_-)
\\
 1-\gamma ^*_{-1}(s_-)
\end{array}\right)
\;<\;  -v^\top P(x)\left(\begin{array}{c}
2 x_2
\\
 1- x_1
\end{array}\right)
$$
}{%
\\[0.7em]$\displaystyle 
\frac{d\gamma ^*_-}{ds}(s_-)^\top P(\gamma ^*_-(s_-))
\left(\begin{array}{c}
2\gamma ^*_{-2}(s_-)
\\
 1-\gamma ^*_{-1}(s_-)
\end{array}\right)
$\hfill\null\\\null\hfill$\displaystyle
\;<\;  -v^\top P(x)\left(\begin{array}{c}
2 x_2
\\
 1- x_1
\end{array}\right)
$\\[0.7em]
}
With (\ref{NLLP13}), these inequalities cannot hold simultaneously.  On the other hand,
it is possible to satisfy either Condition A2 or Condition A3 by properly choosing
the metric.
\egroup\end{example}
\stopmodif

\subsection{Satisfying Condition A2 first} 
We know with \cite[Proposition 2.4]{127} that a Riemannian metric satisfying Condition A2 gives a locally 
convergent observer. This motivates starting with Condition A2.

In \cite[(47) and Propositions 3.2 and 3.5]{127}
we have given procedures for obtaining metrics satisfying Condition A2.
Then, with a metric constructed via such procedures, it remains to check if Condition A3 holds. 
Such a check consists of testing whether or not the second fundamental form of 
$\bfh$ is zero. For this test to be positive, we must have (\ref{LP162}) which is satisfied if
(\ref{LP160}) holds. We know the latter condition can always be satisfied by modifying the given metric $\bfP$ 
into $\bfP_{mod}$ as given in (\ref{LP163}). 
Fortunately, the satisfaction of Condition A2 is not affected by this modification, 
as the following result shows.

\begin{proposition}
\label{prop: prop3}
Condition A2
holds for $\bfP$ if and only if Condition A2 holds for $\bfP_{mod}$.
\end{proposition}
\begin{proof}
The claim 
is a direct consequence of the identity
\IfTwoCol{%
\\[1em]$\displaystyle 
\frac{\partial }{\partial x}\!\left\{\!
\vrule height 1.2em depth 1.2em width 0pt
{\left[\frac{\partial h}{\partial x}(x) v^\tangent\right]\!}^\top
\! \left[\Py (h(x)) -
{\left(\frac{\partial h}{\partial x}(x) P(x)^{-1}\frac{\partial h}{\partial x}(x)^\top
\!\right)\!}^{-1}
\right]\right.$\\[1em]$\displaystyle 
\qquad \hfill \null \left.
\vrule height 1.2em depth 1.2em width 0pt
\times
\left[\frac{\partial h}{\partial x}(x) v^\tangent\right]\!\right\} =  0 \displaystyle 
\quad\forall v^\tangent \in \Distrib ^\tangent (x)
\  ,\ \forall x\in \coordxm(\coordxd)
$\refstepcounter{equation}\label{trav24}\hfill$(\theequation)$\\[1em]}{%
\\[1em]$\displaystyle 
\frac{\partial }{\partial x}\left\{v^{\tangent\top} \frac{\partial h}{\partial x}(x)^\top 
\left[\Py (h(x)) -
\left(\frac{\partial h}{\partial x}(x) P(x)^{-1}\frac{\partial h}{\partial x}(x)^\top
\right)^{-1}
\right]
\frac{\partial h}{\partial x}(x) v^\tangent\right\}\;=\; 0
$\refstepcounter{equation}\label{trav24}\hfill$(\theequation)$\\\null \hfill $\displaystyle 
\forall v^\tangent \in \Distrib ^\tangent (x)
\  ,\quad \forall x\in \coordxm(\coordxd)
\  ,
$\\[1em]}
being valid for any coordinate chart $\coordx$.
\end{proof}

\startmodif
\begin{example}\textit{(Systems that are strongly differentially observable of order $n$)}\bgroup\normalfont
\label{ex4}
In \cite[\S IV]{127}, we have seen that, when $p=1$ and
$$
\bfimmer_{n}(\bfx)\;=\; \left(\begin{array}{c}
\bfh(\bfx) \\ 
L_\bff \bfh(\bfx) \\ \vdots \\ 
L_\bff^{n-1}\bfh(\bfx)
\end{array}\right)
$$%
is a diffeomorphism from some open set $\Ouv$ to $\RR^n$, the expression in some coordinate chart 
$\coordx$ of a
metric satisfying Condition A2 on $\Ouv$ is
$$
P(x) = \frac{\partial \immer_{n}}{\partial x}(x)^\top \bar P \frac{\partial\immer_{n}}{\partial x}(x) \  ,
$$
where $\bar P$ is a symmetric positive definite matrix to be chosen  (see \cite[Lemma 
4.2]{127}).
Actually,
\IfReport{%
$$
\bar x\;=\; \immer_{n}(x)
$$
}{%
$\bar x= \immer_{n}(x)$
}
are other coordinates for $\bfx$ in which
the expression of the metric $\bfP$
is simply the constant matrix $\bar P$. 
Moreover, the expression of $\bfh$ in the same particular coordinates 
is linear, i.e.
$$
y\;=\; C \bar x
\  ,
$$
with the notation
$$
C\;=\; \left(\begin{array}{ccccc}
1&0&\ldots 
& 0
\end{array}\right)
\  .
$$
Hence, the system belongs to the ``Euclidean family'' and the observer (\ref{eqn:GeodesicObserverVectorField})
takes the form (see \cite{Deza.ea.92.SCL})
$$
\dot{\hat x}\;=\; f(\hat x)-k_E(\hat x)
\frac{\partial \immer_{n}}{\partial x}(\hat x)^{-1} \bar P^{-1} C^T(C \hat x-y)
\  .
$$
\egroup\end{example}
\stopmodif

\subsection{Satisfying Condition A3 First}
\label{sec21}
For Condition A2 to hold, the Riemannian metric $\bfP$ must satisfy the inequality (\ref{3}). Instead,
for A3 to hold, according to Lemma \ref{lem12}, we must have at least the equalities (\ref{LP151}). It may be 
easier to satisfy first the equalities and then the inequalities. Namely,
instead of starting with a metric that satisfies Condition A2, we start with a metric given
by Theorem~\ref{prop17}, which is guaranteed to satisfy Condition A3.  
Then, it remains to define the degrees of freedom $\bfh^\ortho$ and $\bfPxi$ involved in its construction so as 
to satisfy Condition A2.

In this context, the fact that the formula (\ref{LP166}) for $P$ is a sum implies that Condition
A2 takes a particular form. Indeed,  for any $\bfx_0$ in $\bfRR^n$, for any coordinate charts
$(x,\RR^n,\phi)$ around $\bfx_0$ and $(y,\RR^p,\chi)$ around $\bfh(\bfx_0)$ in $\bfRR^p$, for any
$\coordxi $ around $\bfh^\ortho(\bfx_0)$ in $\bfXi$,
for all vectors $v$ satisfying
\begin{equation}
\label{LP236}
\sum_{\indxa}\frac{\partial h_\indyi}{\partial x_\indxa}(x) v_\indxa \;=\; 0
\  ,
\end{equation}
and with the definitions
\IfTwoCol{%
\begin{equation}
\label{LP220}
g_{ \indxrc}(x)\!=\! \sum_{\indxc}\frac{\partial h^\ortho_\indxrc}{\partial x_\indxc}(x)f_\indxc(x)
\ ,\ 
w_{\indxra}(x)\!=\! \sum_{\indxa}\frac{\partial h^\ortho_\indxra}{\partial x_\indxa}(x) v_\indxa 
\  ,
\end{equation}
}{%
\begin{equation}
\label{LP220}
g_{ \indxrc}(x)\;=\; \sum_{\indxc}\frac{\partial h^\ortho_\indxrc}{\partial x_\indxc}(x)f_\indxc(x)
\quad ,\qquad 
w_{\indxra}(x)\;=\; \sum_{\indxa}\frac{\partial h^\ortho_\indxra}{\partial x_\indxa}(x) v_\indxa 
\  ,
\end{equation}
}%
the expression in (\ref{LP166}) of the metric $\bfP$ gives
\IfReport{%
\begin{eqnarray*}
v^\top \mathcal{L}_f P(x) v
&\hskip -0.5em =&\hskip -0.5em \displaystyle 
\frac{\partial }{\partial x}\left\{v^\top P(x) v\right\} f(x)
+
2 v^\top P(x) \frac{\partial f}{\partial x}(x) v
\  ,
\\
&\hskip -0.5em=&\hskip -0.5em\displaystyle 
\sum_{\indxa,\indxb,\indxc,\indxra,\indxrb}\frac{\partial }{\partial x_\indxc}\left\{
v_\indxa\frac{\partial h^\ortho_\indxra}{\partial x_\indxa}(x)
\Pxi_{\indxra\indxrb}(h^\ortho(x))
\frac{\partial h^\ortho_\indxrb}{\partial x_\indxb}(x)v_\indxb
\right\} f_\indxc(x)
\\&\hskip -0.5em&\hskip -0.5em\displaystyle \qquad \qquad \qquad \qquad \qquad \qquad 
+\sum_{\indxa,\indxb,\indxc,\indxra,\indxrb}
2 
v_\indxa\frac{\partial h^\ortho_\indxra}{\partial x_\indxa}(x)
\Pxi_{\indxra\indxrb}(h^\ortho(x))
\frac{\partial h^\ortho_\indxrb}{\partial x_\indxc}(x)
\frac{\partial f_\indxc}{\partial x_\indxb}(x)v_\indxb
\  ,
\\
&\hskip -0.5em=&\hskip -0.5em\displaystyle 
\sum_{\indxa,\indxb,\indxc,\indxra,\indxrb}
v_\indxa \frac{\partial ^2h^\ortho_\indxra}{\partial x_\indxa\partial x_\indxc}(x)f_\indxc(x)
\Pxi_{\indxra\indxrb}(h^\ortho(x))
\frac{\partial h^\ortho_\indxrb}{\partial x_\indxb}(x)
v_\indxb
\\&\hskip -0.5em&\hskip -0.5em\displaystyle \qquad \qquad 
+\sum_{\indxa,\indxb,\indxc,\indxra,\indxrb,\indxrc}v_\indxa \frac{\partial h^\ortho_\indxra}{\partial x_\indxa}(x)
\frac{\partial \Pxi_{\indxra\indxrb}}{\partial \coordxip_\indxrc}(h^\ortho(x))
\frac{\partial h^\ortho_\indxrc}{\partial x_\indxc}(x)f_\indxc(x)
\frac{\partial h^\ortho_\indxrb}{\partial x_\indxb}(x)
v_\indxb
\\&\hskip -0.5em&\hskip -0.5em\displaystyle \qquad \qquad \qquad \qquad 
+ \sum_{\indxa,\indxb,\indxc,\indxra,\indxrb}
v_\indxa \frac{\partial h^\ortho_\indxra}{\partial x_\indxa}(x)
\Pxi_{\indxra\indxrb}(h^\ortho(x))
\frac{\partial ^2h^\ortho_\indxrb}{\partial x_\indxb\partial x_\indxc}(x)
 f_\indxc(x)
v_\indxb
\\&\hskip -0.5em&\hskip -0.5em\displaystyle \qquad \qquad \qquad \qquad\qquad \qquad 
+
\sum_{\indxa,\indxb,\indxc,\indxra,\indxrb}
2 
v_\indxa\frac{\partial h^\ortho_\indxra}{\partial x_\indxa}(x)
\Pxi_{\indxra\indxrb}(h^\ortho(x))
\frac{\partial h^\ortho_\indxrb}{\partial x_\indxc}(x)
\frac{\partial f_\indxc}{\partial x_\indxb}(x)v_\indxb
\  ,
\\
&\hskip -0.5em=&\hskip -0.5em
\displaystyle \sum_{\indxra,\indxrb,\indxrc}
w_{ \indxra}
\frac{\partial \Pxi_{\indxra\indxrb}}{\partial \coordxip_\indxrc}(h^\ortho(x))
\,  g_{ \indxrc}(x)
w_{ \indxrb}
+2\sum_{\indxa,\indxrb,\indxrc}
v_\indxa \frac{\partial g_{\indxrc}}{\partial x_\indxa}(x)\,  
\Pxi_{\indxrc\indxrb}(h^\ortho(x))
w_{\indxrb}
\  .
\end{eqnarray*}%
This expression is to be compared with $w^\top \mathcal{L}_{g} R(h^\ortho) \,  w$ considering
$\sum_{\indxa} v_\indxa \frac{\partial g_{\indxrc}}{\partial x_\indxa}$ is formally equal to
$\sum_\indxra w_\indxra \frac{\partial g_{\indxrc}}{\partial h^\ortho_\indxra}$, because of (\ref{LP220}).
Using the expression for $v^\top \mathcal{L}_f P(x) v$ above
and by invoking the S-Lemma (see \cite{Polik-Terlaky.07}),
}{%
\\[0.3em]$\displaystyle 
v^\top \mathcal{L}_f P(x) v\;=\;
\sum_{\indxra,\indxrb,\indxrc}
w_{ \indxra}
\frac{\partial \Pxi_{\indxra\indxrb}}{\partial \coordxip_\indxrc}(h^\ortho(x))
\,  g_{ \indxrc}(x) w_{ \indxrb}
$\\[-0.5em]\null\hfill
$\displaystyle
+2\sum_{\indxa,\indxrb,\indxrc}
v_\indxa \frac{\partial g_{\indxrc}}{\partial x_\indxa}(x)\,  
\Pxi_{\indxrc\indxrb}(h^\ortho(x))
w_{\indxrb}
\  .
$\\[0.7em]
By invoking the S-Lemma (see \cite{Polik-Terlaky.07}),
}%
we obtain that Condition A2 is satisfied if, when (\ref{LP236}) holds,
we have
\IfTwoCol{%
\\[1em]$\displaystyle 
\displaystyle \sum_{\indxra,\indxrb}
\left(
w_{ \indxra}(x)
\left[\sum_\indxrc\frac{\partial \Pxi_{\indxra\indxrb}}{\partial \coordxip_\indxrc}(h^\ortho(x))g_{ 
\indxrc}(x)\right]
w_{ \indxrb}(x) \right.$
\\\null \hfill$\displaystyle \left.
+2\left[\sum_\indxa
v_\indxa \frac{\partial g_{\indxra}}{\partial x_\indxa}(x)\right]
\Pxi_{\indxra\indxrb}(h^\ortho(x))
w_{\indxrb}(x)\right)
$\refstepcounter{equation}\label{LP218}\hfill$(\theequation)$
\\\null \hfill $\displaystyle \; \leq \; - q \sum_{\indxra,\indxrb}
w_{\indxra} (x)\Pxi_{\indxra\indxrb}(h^\ortho(x)) w_{\indxrb}(x)
$\\
}{%
\\[1em]$\displaystyle 
\displaystyle \sum_{\indxra,\indxrb}
\left(
w_{ \indxra}(x)
\left[\sum_\indxrc\frac{\partial \Pxi_{\indxra\indxrb}}{\partial \coordxip_\indxrc}(h^\ortho(x))g_{ 
\indxrc}(x)\right]
w_{ \indxrb}(x)
+2\left[\sum_\indxa
v_\indxa \frac{\partial g_{\indxra}}{\partial x_\indxa}(x)\right]
\Pxi_{\indxra\indxrb}(h^\ortho(x))
w_{\indxrb}(x)\right)
$\refstepcounter{equation}\label{LP218}\hfill$(\theequation)$
\\\null \hfill $\displaystyle \; \leq \; - q \sum_{\indxra,\indxrb}
w_{\indxra} (x)\Pxi_{\indxra\indxrb}(h^\ortho(x)) w_{\indxrb}(x)
$\\
}%
for some strictly positive $q$.

With the above,
 we have reduced the design of the observer (\ref{eqn:GeodesicObserverVectorField}) to the 
problem of finding functions $h^\ortho$, of rank $n-p$ and such that $(h,h^\ortho)$ has rank $n$, and $\Pxi$ with positive definite values, satisfying
the inequality above.

\startarchive
\begin{example}[A Lagrangian system with one degree of freedom]\bgroup\normalfont
We consider the Lagrangian system of \cite[Example 5.1]{127}. This system is given by
$$%
\dot y=\xrond \  ,\quad 
\dot \xrond  = \xrond ^2
\  .
$$
Following the constructions in Example \ref{ex:Lagrangian}, the corresponding functions are
$$%
g(y)\;=\; \exp(-2y)\  ,\quad \mathfrak{C} = -1
\  ,\quad 
P(y,\xrond )\;=\; \exp(-2y)
\left(\begin{array}{cc}
a+2c  \xrond  + b \xrond ^2
&
-c-b\xrond 
\\
-c-b\xrond 
&
b
\end{array}\right)
\  ,
$$
and the source term $S$ is zero.
Since $p=1$, the integrability condition is not needed and we know from Example \ref{ex:Lagrangian} that Conditions A2 and A3 hold.
However, the metrics $g$ and $P$, 
respectively on $\RR$ and $\RR^2$, are not complete. Ignoring this fact, we have that a normalized
geodesic for $g$ starting from $y_1$ is
$$
\gammay ^*(s)\;=\; y_1-\log(1-s\exp(y_1))
$$
Then, the function $\wp $, defined as the square of the distance, is
$$
\wp (y_1,y_2) = s_2^2
$$
where $s_2$ is solution to
$$
y_1-\log(1-s_2\exp(y_1))=y_2
$$
i.e.
$$
s_2=|\exp(-y_1)-\exp(-y_2)|
\  .
$$
This gives
$$
\wp (y_1,y_2) = [\exp(-y_1)-\exp(-y_2)]^2
\  .
$$
Then, the observer (\ref{eqn:GeodesicObserverVectorField}) is
\begin{eqnarray*}
\dot {\hat y}&=&\hat \xrond +
2 k_E(\hat y,\hat \xrond )\,  \frac{b\exp(\hat y)}{ab-c^2}
[\exp(-\hat y)-\exp(-y)]
\\
\dot{\hat \xrond}&=&\hat \xrond^2+
2k_E(\hat y,\hat \xrond) \,  \frac{(c+b\hat \xrond)\exp(\hat y)}{ab-c^2}
[\exp(-\hat y)-\exp(-y)]
\end{eqnarray*}
\egroup
\end{example}
\stoparchive %

\begin{example}[Systems of dimension two]\bgroup\normalfont
We consider a general system written as
\begin{equation}
\label{LP82}
\dot y\;=\; f_y(y,\xrond )\quad ,\qquad \dot \xrond \;=\; f_\xrond (y,\xrond )
\end{equation}
with $n=2$ and $p=1$.
\startmodif
It follows from Theorem~\ref{prop17}, that Condition A3 is satisfied if $P$ is in the form (see \eqref{LP166})
\IfTwoCol{%
\\[0.7em]$\displaystyle 
\renewcommand{\arraystretch}{1.5}
P(y,\xrond )= \left(\begin{array}{@{}cc@{}}
1 + \frac{\partial h^\ortho}{\partial y}(y,\xrond)^2
&
\frac{\partial h^\ortho}{\partial y}(y,\xrond)
\frac{\partial h^\ortho}{\partial \xrond}(y,\xrond)
\\
\frac{\partial h^\ortho}{\partial y}(y,\xrond)
\frac{\partial h^\ortho}{\partial \xrond}(y,\xrond)
&
\frac{\partial h^\ortho}{\partial \xrond}(y,\xrond)^2
\end{array}\right)
\  ,
$\\[1em]
}{%
$$ 
\renewcommand{\arraystretch}{1.5}
P(y,\xrond )= \left(\begin{array}{@{}c@{}}
1 \\ 0
\end{array}\right)
\left(\begin{array}{@{}cc@{}}
1 & 0
\end{array}\right)
+
\left(
\begin{array}{@{}c@{}} 
\displaystyle 
\frac{\partial h^\ortho}{\partial y}(y,\xrond) 
\\\displaystyle 
\frac{\partial h^\ortho}{\partial \xrond}(y,\xrond)
\end{array}
\right)
\left(
\begin{array}{@{}cc@{}} 
\displaystyle 
\frac{\partial h^\ortho}{\partial y}(y,\xrond) 
&\displaystyle 
\frac{\partial h^\ortho}{\partial \xrond}(y,\xrond)
\end{array}
\right)
\  ,
$$
}%
where $h^\ortho$ is any $C^3$ function with $\frac{\partial h^\ortho}{\partial z}(y,\xrond)$ strictly 
positive for all 
$(y,\xrond)$.
In this case, we choose
$$
\wp (y_1,y_2)\;=\; |y_1-y_2|^2
$$
and Condition A2 holds if we have
\IfTwoCol{%
$$
\frac{\partial }{\partial \xrond}\!
\left\{\!
\frac{\partial h^\ortho}{\partial y}(y,\xrond) f_y(y,\xrond)
+\frac{\partial h^\ortho}{\partial \xrond}(y,\xrond) f_{\xrond}(y,\xrond)\!\right\}
\leq  \! - \frac{\partial h^\ortho}{\partial \xrond}(y,\xrond)^2
$$%
}{%
$$
2\frac{\partial }{\partial \xrond}
\left\{
\frac{\partial h^\ortho}{\partial y}(y,\xrond) f_y(y,\xrond)
+\frac{\partial h^\ortho}{\partial \xrond}(y,\xrond) f_{\xrond}(y,\xrond)\right\}
\; \leq \; - \frac{\partial h^\ortho}{\partial \xrond}(y,\xrond)
\  .
$$
}%
In this case, the observer (\ref{eqn:GeodesicObserverVectorField}) is
\begin{eqnarray*}
\dot {\hat y}&=&f_y(\hat y,\hat \xrond ) - k_E(\hat y,\hat \xrond )
(\hat y-y)
\  ,
\\
\dot{\hat \xrond}&=&f_\xrond(\hat y,\hat \xrond ) +k_E(\hat y,\hat \xrond )
\frac{
\frac{\partial h^\ortho}{\partial y}(y,\xrond) 
}{
\frac{\partial h^\ortho}{\partial \xrond}(y,\xrond)
}
(\hat y-y)
\  .
\end{eqnarray*}
\stopmodif
\egroup\end{example}

\begin{example}
\bgroup\normalfont
\label{ex8}
We consider again the harmonic oscillator with unknown frequency in (\ref{LP132}),
the state of which evolves in the invariant set $\Ouv_\varepsilon $ defined in (\ref{LP205}).
We have seen in Example~\ref{ex9} that the metrics considered thus far
satisfying Condition A2 do not satisfy Condition A3. 
Following the observations at the beginning of this section,
we proceed by constructing $P$ so that Condition A3 holds, and then assess the satisfaction of
Condition A2.

Following Theorem~\ref{prop17}, 
the level sets of the output function being diffeomorphic to $\bfRR^2$, we choose $\bfXi$ as $\bfRR^2$.
Then, following Theorem~\ref{prop17}, a metric satisfying Condition A3 is%
\IfTwoCol{%
, with the notation
$\frac{\partial }{\partial x}=\left(\begin{array}{@{}ccc@{}}
\frac{\partial }{\partial y}
&
\frac{\partial }{\partial \xrond_\indxra}
&
\frac{\partial }{\partial \xrond_\indxrb}
\end{array}\right)$,
\\[0.3em]$\displaystyle 
P(y,\xrond_\indxra,\xrond_\indxrb)\;=\;
\changey^\prime (y)^2 
\left(\begin{array}{@{\,  }c@{\,  }}
1 \\ 0 \\0
\end{array}\right)\Py (\changey (y) ) 
\left(\begin{array}{@{}ccc@{}}
1 & 0 & 0
\end{array}\right)
$\refstepcounter{equation}\label{LP187}\hfill$(\theequation)$\\[-0.3em]\null \hfill   $\displaystyle 
\!\!\!\!+ \frac{\partial
\left(\begin{array}{@{\,  }c@{\,  }}
h^\ortho_\indxra \\ h^\ortho_\indxrb
\end{array}\right)
}{\partial x}(y,\xrond_\indxra,\xrond_\indxrb)^\top
\Pxi(h^\ortho(y,\xrond_\indxra,\xrond_\indxrb))\frac{\partial \left(\begin{array}{@{\,  }c@{\,  }}
h^\ortho_\indxra \\ h^\ortho_\indxrb
\end{array}\right)}{\partial x}(y,\xrond_\indxra,\xrond_\indxrb)
$\\[0.7em]
where
}{%
\\[1em]$\displaystyle 
P(y,\xrond_\indxra,\xrond_\indxrb)\;=\;
\changey^\prime (y)^2 
\left(\begin{array}{@{\,  }c@{\,  }}
1 \\ 0 \\0
\end{array}\right)\Py (\changey (y) ) 
\left(\begin{array}{@{}ccc@{}}
1 & 0 & 0
\end{array}\right)+
$\refstepcounter{equation}\label{LP187}\hfill$(\theequation)$\\[0.3em]\null   $
+
\left(\renewcommand{\arraystretch}{2.2}
\begin{array}{@{\,  }cc@{\,  }}
\displaystyle
\frac{\partial h^\ortho_\indxra}{\partial y}(y,\xrond_\indxra,\xrond_\indxrb)
& \displaystyle
\frac{\partial h^\ortho_\indxrb}{\partial y}(y,\xrond_\indxra,\xrond_\indxrb)
\\
\displaystyle
\frac{\partial h^\ortho_\indxra}{\partial \xrond_\indxra}(y,\xrond_\indxra,\xrond_\indxrb)
& \displaystyle
\frac{\partial h^\ortho_\indxrb}{\partial \xrond_\indxra}(y,\xrond_\indxra,\xrond_\indxrb)
\\
\displaystyle
\frac{\partial h^\ortho_\indxra}{\partial \xrond_\indxrb}(y,\xrond_\indxra,\xrond_\indxrb)
& \displaystyle
\frac{\partial h^\ortho_\indxrb}{\partial \xrond_\indxrb}(y,\xrond_\indxra,\xrond_\indxrb)
\end{array}
\right)
\times
$\hfill \null \\[0.7em]\null \hfill  $\displaystyle 
\times
\left(\renewcommand{\arraystretch}{1.3}
\begin{array}{@{\,  }cc@{\,  }}
\displaystyle \Pxi_{\indxra\indxra}( h^\ortho(y,\xrond_\indxra,\xrond_\indxrb))
& \displaystyle
\Pxi_{\indxra\indxrb}( h^\ortho(y,\xrond_\indxra,\xrond_\indxrb))
\\
\displaystyle
\Pxi_{\indxra\indxrb}( h^\ortho(y,\xrond_\indxra,\xrond_\indxrb))
& \displaystyle
\Pxi_{\indxrb\indxrb}( h^\ortho(y,\xrond_\indxra,\xrond_\indxrb))
\end{array}\right)
\times
$\hfill \null \\[0.7em]\null \hfill $\displaystyle 
\times
\left(\renewcommand{\arraystretch}{2.2}
\begin{array}{@{\,  }ccc@{\,  }}
\displaystyle 
\frac{\partial h^\ortho_\indxra}{\partial y}(y,\xrond_\indxra,\xrond_\indxrb)
& \displaystyle
\frac{\partial h^\ortho_\indxra}{\partial \xrond_\indxra}(y,\xrond_\indxra,\xrond_\indxrb)
& \displaystyle
\frac{\partial h^\ortho_\indxra}{\partial \xrond_\indxrb}(y,\xrond_\indxra,\xrond_\indxrb)
\\
\displaystyle 
\frac{\partial h^\ortho_\indxrb}{\partial y}(y,\xrond_\indxra,\xrond_\indxrb)
& \displaystyle
\frac{\partial h^\ortho_\indxrb}{\partial \xrond_\indxra}(y,\xrond_\indxra,\xrond_\indxrb)
& \displaystyle
\frac{\partial h^\ortho_\indxrb}{\partial \xrond_\indxrb}(y,\xrond_\indxra,\xrond_\indxrb)
\end{array}
\right)
$\\[1em]
where
}%
it remains to choose
\begin{list}{}{%
\parskip 0pt plus 0pt minus 0pt%
\topsep 0pt plus 0pt minus 0pt
\parsep 0pt plus 0pt minus 0pt%
\partopsep 0pt plus 0pt minus 0pt%
\itemsep 0pt plus 0pt minus 0pt
\settowidth{\labelwidth}{\quad --}%
\setlength{\labelsep}{0.5em}%
\setlength{\leftmargin}{\labelwidth}%
\addtolength{\leftmargin}{\labelsep}%
}
\item[--]
$\Py $ as a $C^s$ function with strictly positive values,
\item[--]
$\changey$ and $h^\ortho$ such that $(y,\xrond)\mapsto (\changey(y),h^\ortho (y,\xrond))$ is a $C^s$ diffeomorphism on $\Ouv_\varepsilon $,
\item[--]
and $\Pxi$ as a $C^s$ function with positive definite values
\end{list}
to satisfy Condition A2 or its sufficient condition (\ref{LP218}). 
To help in this task, we remind the reader of the findings in \cite[Example 2.2]{127}. In that example
we show that, with $\coordxi$ as a global coordinate chart for $\bfXi=\bfRR^2$, the arrival set of $\bfh^\ortho$,
and the choices
\IfTwoCol{%
\begin{eqnarray*}
&\displaystyle 
\changey(y) = y
\  ,\quad
Q(y)=c
\  ,\quad 
h_\indxra^\ortho (y,\xrond)=\xrond_\indxra -  y
\  ,
\\
&\displaystyle 
h_\indxrb^\ortho (y,\xrond)=\xrond_\indxrb + \textstyle \frac{1}{2} y^2
\  ,\quad 
\Pxi(\coordxip)\;=\; I_2
\  ,
\end{eqnarray*}
}{%
$$
\changey(y) = y
\  ,\quad
Q(y)=c
\  ,\quad 
h_\indxra^\ortho (y,\xrond)=\xrond_\indxra -  y
\  ,\quad 
h_\indxrb^\ortho (y,\xrond)=\xrond_\indxrb + \textstyle \frac{1}{2} y^2
\  ,\quad 
\Pxi(\coordxip)\;=\; I_2
\  ,
$$}%
where $c$ is a strictly positive real number,
the expression $P$ of the metric obtained in (\ref{LP187}), namely,
$$
P(\coordyp,\coordxrp_\indxra,\coordxrp_\indxrb)
\;=\; \left(\begin{array}{@{}c@{}}
1 \\ 0 \\ 0
\end{array}\right)c
\left(\begin{array}{@{}ccc@{}}
1 & 0 & 0
\end{array}\right)
+
\left(\begin{array}{@{}cc@{}}
-1 & \coordyp
\\
1 & 0
\\
0 & 1
\end{array}\right)
\left(\begin{array}{@{}ccc@{}}
-1 & 1 & 0
\\
\coordyp & 0 & 1
\end{array}\right)
$$
is such that Condition A2 is satisfied but not strictly -- namely, it only certifies weak differential detectability.

From this point we proceed with a ``deformation'' of the metric above to meet both conditions. We choose
\IfTwoCol{%
\\[0.7em]$
h_\indxra^\ortho (y,\xrond) \!=\! \coordxip_\indxra\! =\! \xrond_\indxra  -  y
\  ,\quad 
h_\indxrb^\ortho (y,\xrond) \!=\! \coordxip_\indxrb = \xrond_\indxrb
+ \textstyle \frac{1}{2} \displaystyle y^2 + a b y\xrond_\indxra
\  ,$\\[0.7em]\null \hfill $\displaystyle
\Pxi(\coordxip)\! =\! \left(\begin{array}{cc}
1 &  0 \\   0 &  1 + a \coordxip_\indxra^2 
\end{array}\right)
\  ,
$\hfill \null \\[0.7em]}{%
$$
h_\indxra^\ortho (y,\xrond) = \coordxip_\indxra = \xrond_\indxra  -  y
\  ,\quad 
h_\indxrb^\ortho (y,\xrond)= \coordxip_\indxrb = \xrond_\indxrb
+ \textstyle \frac{1}{2} \displaystyle y^2 + a b y\xrond_\indxra
\  ,\quad 
\Pxi(\coordxip) = \left(\begin{array}{cc}
1 &  0 \\   0 &  1 + a \coordxip_\indxra^2 
\end{array}\right)
\  ,
$$}%
where  $a$ and $b$ are strictly positive real numbers to be chosen, with, a priori, $a$
being small. We express the
inequality (\ref{LP218}) with the coordinates $\coordyxrp$, restricted
to the set $\Ouv_\varepsilon $. Since (\ref{LP220}) reads
\IfTwoCol{%
\\[0.5em]$
\renewcommand{\arraystretch}{1.3}
\begin{array}{rcl}
g_{\indxra}(y,\coordxrp_\indxra,\coordxrp_\indxrb)&=&-y \coordxrp_\indxrb - \coordxrp_\indxra
\;=\; 
-y \coordxrp_\indxrb - \coordxip_\indxra - y
\  ,
\\
g_{\indxrb}(y,\coordxrp_\indxra,\coordxrp_\indxrb)&=& y \coordxrp_\indxra + ab \coordxrp_\indxra^2 - ab y^2 \coordxrp_\indxrb
\  ,
\end{array}$\\[0.3em]
$\begin{array}{rcl}
w_{\indxra}= v_\indxra-v_y
\  , \ \
w_{\indxrb}=v_\indxrb + y v_y + ab \coordxrp_\indxra v_y + ab y 
v_\indxra
\  ,
\end{array}$\\[0.2em]\null \hfill $\displaystyle 
\left(\begin{array}{@{}cc@{}}
v_\indxra & v_\indxrb
\end{array}\right)
\;=\; \left(\begin{array}{@{}cc@{}}
w_\indxra & w_\indxrb
\end{array}\right)
\left(\begin{array}{@{}cc}
1 & - ab y
\\
0 &1
\end{array}\right)
\  ,
$\hfill \null \\[0.3em]
}{%
$$
\renewcommand{\arraystretch}{1.5}
\begin{array}{rcl}
g_{\indxra}(y,\coordxrp_\indxra,\coordxrp_\indxrb)&=&-y \coordxrp_\indxrb - \coordxrp_\indxra
\;=\; 
-y \coordxrp_\indxrb - \coordxip_\indxra - y
\  ,
\\
g_{\indxrb}(y,\coordxrp_\indxra,\coordxrp_\indxrb)&=& y \coordxrp_\indxra + ab \coordxrp_\indxra^2 - ab y^2 \coordxrp_\indxrb
\  ,
\end{array}
\qquad 
\begin{array}{rcl}
w_{\indxra}&=& v_\indxra-v_y
\  ,
\\
w_{\indxrb}&=&v_\indxrb + y v_y + ab \coordxrp_\indxra v_y + ab y 
v_\indxra
\  ,
\end{array}
$$
and we have,
\startmodif
when (\ref{LP236}) holds, i.e. $v_y=0$,
\stopmodif
$$
\left(\begin{array}{@{}cc@{}}
v_\indxra & v_\indxrb
\end{array}\right)
\;=\; \left(\begin{array}{@{}cc@{}}
w_\indxra & w_\indxrb
\end{array}\right)
\left(\begin{array}{@{}cc}
1 & - ab y
\\
0 &1
\end{array}\right)
\  ,
$$
}%
inequality (\ref{LP218}) is
\IfTwoCol{%
\\[0.3em]$\displaystyle 
w_\indxrb^2 2 a \coordxip_\indxra g_\indxra
+
\startmodif
2
\stopmodif
\left(\begin{array}{@{}cc@{}}
w_\indxra & w_\indxrb
\end{array}\right)
\left(\begin{array}{@{}cc}
1 & - ab y
\\
0 &1
\end{array}\right)
\left(\begin{array}{@{}cc}
-1 & y + 2 ab \coordxrp_\indxra
\\
-y & -ab y^2
\end{array}\right)\times\\\null\hfill
\times \left(\begin{array}{cc}
1 &  0 \\   0 &  1 + a \coordxip_\indxra^2 
\end{array}\right)
\left(\begin{array}{@{}c@{}}
w_\indxra \\ w_\indxrb
\end{array}\right)
$\hfill \null \\[0.2em]\null \hfill $\displaystyle 
\leq \; 
-q 
\left(\begin{array}{cc}
w_\indxra & w_\indxrb
\end{array}\right)
\left(\begin{array}{cc}
1 & 0 \\  0 & 1 + a \coordxip_\indxra^2 
\end{array}\right)
\left(\begin{array}{c}
w_\indxra \\ w_\indxrb
\end{array}\right)
$\\[1em]}{%
\\[0.7em]$\displaystyle 
w_\indxrb^2 2 a \coordxip_\indxra g_\indxra
+
\startmodif
2
\stopmodif
\left(\begin{array}{@{}cc@{}}
w_\indxra & w_\indxrb
\end{array}\right)
\left(\begin{array}{@{}cc}
1 & - ab y
\\
0 &1
\end{array}\right)
\left(\begin{array}{@{}cc}
-1 & y + 2 ab \coordxrp_\indxra
\\
-y & -ab y^2
\end{array}\right)
\left(\begin{array}{cc}
1 &  0 \\   0 &  1 + a \coordxip_\indxra^2 
\end{array}\right)
\left(\begin{array}{@{}c@{}}
w_\indxra \\ w_\indxrb
\end{array}\right)
$\hfill \null \\\null \hfill $\displaystyle 
\leq \; 
-q 
\left(\begin{array}{cc}
w_\indxra & w_\indxrb
\end{array}\right)
\left(\begin{array}{cc}
1 & 0 \\  0 & 1 + a \coordxip_\indxra^2 
\end{array}\right)
\left(\begin{array}{c}
w_\indxra \\ w_\indxrb
\end{array}\right)
$\\[0.7em]
}%
for some strictly positive $q$. This inequality can be rewritten
\\[0.7em]
$\displaystyle 
-
w_\indxrb^2 \left[
2a \coordxip_\indxra (y \coordxrp_\indxrb + \coordxip_\indxra + y)
+
\startmodif
2
\stopmodif
a b y^2 (1 + a \coordxip_\indxra^2 )
- q (1 + a \coordxip_\indxra^2 )
\right]
$\hfill \null \\[0.3em] $\displaystyle 
+\;
\startmodif
2
\stopmodif
w_\indxra w_\indxrb
a\left[
 2  b(\coordxip_\indxra+ y)(1 + a \coordxip_\indxra^2 )
 +
 y\coordxip_\indxra^2 
 +
 a b^2  y^3(1 + a \coordxip_\indxra^2 )
\right]
$\qquad \null \\\null \hfill $\displaystyle 
\;-\;  w_\indxra^2
\left[
\startmodif
2 (1-a b  y^2)
\stopmodif
 -q\right]
\; \leq \; 0
$\\[0.7em]
It is satisfied if we have
\IfTwoCol{%
$\startmodif
2 (1-a b  y^2)
\stopmodif
-q\; >\; 0$,
}{%
$$
\startmodif
2 (1-a b  y^2)
\stopmodif
-q\; >\; 0
$$
}%
and
\IfTwoCol{%
\\[0.7em]$\displaystyle 
\startmodif
4
\stopmodif
a^2\left[
 2  b(\coordxip_\indxra+ y)(1 + a \coordxip_\indxra^2 )
 +
 y\coordxip_\indxra^2 
 +
 a b^2  y^3(1 + a \coordxip_\indxra^2 )
\right]^2\; <\; 
$\hfill\null\\[0.5em]$\displaystyle
4\left[
2a \coordxip_\indxra (y \coordxrp_\indxrb + \coordxip_\indxra + y)
+ 
\startmodif
2 
\stopmodif
a b y^2 (1 + a \coordxip_\indxra^2 )
- q (1 + a \coordxip_\indxra^2 )
\right] \times$\\[0.5em]\null\hfill$\times \left[
\startmodif
2 (1-a b  y^2)
\stopmodif
-q\right]
\  ,
$\\[0.7em]
}{%
\\[0.7em]\vbox{\noindent$\displaystyle 
\startmodif
4
\stopmodif
a^2\left[
 2  b(\coordxip_\indxra+ y)(1 + a \coordxip_\indxra^2 )
 +
 y\coordxip_\indxra^2 
 +
 a b^2  y^3(1 + a \coordxip_\indxra^2 )
\right]^2
$\hfill\null\\[0.5em]\null\hfill$\displaystyle
\; <\; 4\left[
2a \coordxip_\indxra (y \coordxrp_\indxrb + \coordxip_\indxra + y)
+
\startmodif
2 
\stopmodif
a b y^2 (1 + a \coordxip_\indxra^2 )
- q (1 + a \coordxip_\indxra^2 )
\right]
\left[
\startmodif
2 (1-a b  y^2)
\stopmodif
-q\right]
\  ,
$}\\[0.7em]
}%
for all $(y,\coordxrp)$ in $\Ouv_\varepsilon $,
\IfTwoCol{%
and therefore satisfying
\begin{equation}\label{eqn:HarmonicNewCoord}
\frac{4}{\varepsilon ^2} > y^2 + \coordxip_\indxra^2 > \frac{\varepsilon ^2}{4}
\quad ,\qquad 
\varepsilon < \coordxrp_\indxrb < \frac{1}{\varepsilon }
\  .
\end{equation}
}{%
i.e., for all $(y,\coordxrp)$ satisfying
$$
\frac{1}{\varepsilon ^2} > y^2 + \coordxrp_\indxra^2 > \varepsilon ^2
\quad ,\qquad 
\varepsilon < \coordxrp_\indxrb < \frac{1}{\varepsilon }
$$
and therefore, for all $(y,\coordxip_\indxra,\coordxrp_\indxrb)$ satisfying
\begin{equation}\label{eqn:HarmonicNewCoord}
\frac{4}{\varepsilon ^2} > y^2 + \coordxip_\indxra^2 > \frac{\varepsilon ^2}{4}
\quad ,\qquad 
\varepsilon < \coordxrp_\indxrb < \frac{1}{\varepsilon }
\  .
\end{equation}
}
We have
\\[0.7em]\vbox{\noindent$\displaystyle 
2a \coordxip_\indxra (y \coordxrp_\indxrb + \coordxip_\indxra + y)
+ 
\startmodif
2 
\stopmodif
a b y^2 (1 + a \coordxip_\indxra^2 )
- q (1 + a \coordxip_\indxra^2 )
$\hfill\null\\[0.6em]\null\hfill$\displaystyle
=\; 
a\left[2 \coordxip_\indxra y \left(\xrond_\indxrb + 1\right)
+ (2-q)\coordxip_\indxra^2 
+
\startmodif
2 
\stopmodif
b y^2 \right]
+
\startmodif
2 
\stopmodif
a^2  b y^2\coordxip_\indxra^2 
- q 
\  .
$}\\[0.7em]
So, by choosing $b$ large enough to satisfy
\IfTwoCol{%
$
\startmodif
2 
\stopmodif
\left(\frac{1}{\varepsilon } +1\right)^2 \leq (2-q) b$,
}{%
$$
\startmodif
2 
\stopmodif
\left(\frac{1}{\varepsilon } +1\right)^2 \leq (2-q) b 
\  ,
$$
}
we obtain successively
\IfTwoCol{%
$$
2 \coordxip_\indxra y \left(\xrond_\indxrb + 1\right)
+ (2-q)\coordxip_\indxra^2 
+
\startmodif
2 
\stopmodif
b y^2 
\; \geq \;
\frac{1}{2}\left[(2-q)\coordxip_\indxra^2 +
\startmodif
2 
\stopmodif
b y^2 \right]
$$
$
2a \coordxip_\indxra (y \coordxrp_\indxrb + \coordxip_\indxra + y)
+
\startmodif
2 
\stopmodif
a b y^2 (1 + a \coordxip_\indxra^2 )
- q (1 + a \coordxip_\indxra^2 )$\\[0.4em]
\null\hfill
$\; \geq \; 
\frac{a}{2}\min\{(2-q),
\startmodif
2 
\stopmodif
b\} \frac{\varepsilon ^2}{4}
+
\startmodif
2 
\stopmodif
a^2  b y^2\coordxip_\indxra^2 
- q 
\  .
$\\[0.6em]
}{%
\begin{eqnarray*}
2 \coordxip_\indxra y \left(\xrond_\indxrb + 1\right)
+ (2-q)\coordxip_\indxra^2 
+
\startmodif
2 
\stopmodif
b y^2 
& \geq &
\frac{1}{2}\left[(2-q)\coordxip_\indxra^2 +
\startmodif
2 
\stopmodif
b y^2 \right]
\  ,
\\
2a \coordxip_\indxra (y \coordxrp_\indxrb + \coordxip_\indxra + y)
+
\startmodif
2 
\stopmodif
a b y^2 (1 + a \coordxip_\indxra^2 )
- q (1 + a \coordxip_\indxra^2 )
& \geq &
\frac{a}{2}\min\{(2-q),
\startmodif
2 
\stopmodif
b\} \frac{\varepsilon ^2}{4}
+
\startmodif
2 
\stopmodif
a^2  b y^2\coordxip_\indxra^2 
- q 
\  .
\end{eqnarray*}
}%
Also, using \eqref{eqn:HarmonicNewCoord}, we have
\IfTwoCol{%
$
\startmodif
2 (1-a b  y^2)
\stopmodif
 -q \; \geq \;
\startmodif
2 -a b  \frac{8}{\varepsilon ^2}
\stopmodif
-q$,
}{%
$$
\startmodif
2 (1-a b  y^2)
\stopmodif
-q \; \geq \;
\startmodif
2 -a b  \frac{8}{\varepsilon ^2}
\stopmodif
-q
$$
}
and
\IfTwoCol{%
\\$
 2  b(\coordxip_\indxra+ y)(1 + a \coordxip_\indxra^2 )
 +
 y\coordxip_\indxra^2 
 +
 a b^2  y^3(1 + a \coordxip_\indxra^2 )$\\[0.5em]
\null\hfill$\leq 
\displaystyle 
 2  b\frac{4}{\varepsilon }\left(1 + a \frac{4}{\varepsilon ^2} \right)
 +
 \frac{8}{\varepsilon ^3} 
 +
 a b^2   \frac{8}{\varepsilon ^3} \left(1 + a \frac{4}{\varepsilon ^2} \right)$\\[0.5em]
\null\hfill$\leq 
  \frac{8}{\varepsilon ^3}\left(
  b
 +
1 
 +
 a b^2   \right)\left(1 + a \frac{4}{\varepsilon ^2} \right)
 \  .
$\\[0.7em]}{%
\begin{eqnarray*}
 2  b(\coordxip_\indxra+ y)(1 + a \coordxip_\indxra^2 )
 +
 y\coordxip_\indxra^2 
 +
 a b^2  y^3(1 + a \coordxip_\indxra^2 )
&\leq &
\displaystyle 
 2  b\frac{4}{\varepsilon }\left(1 + a \frac{4}{\varepsilon ^2} \right)
 +
 \frac{8}{\varepsilon ^3} 
 +
 a b^2   \frac{8}{\varepsilon ^3} \left(1 + a \frac{4}{\varepsilon ^2} \right)
 \\
 &\leq &
  \frac{8}{\varepsilon ^3}\left( b + 1  + a b^2   \right)\left(1 + a \frac{4}{\varepsilon ^2} \right)
 \  .
\end{eqnarray*}}%
Then, a sufficient condition for Condition A2 to hold is
\IfTwoCol{%
\begin{equation}
\label{LP228}
\renewcommand{\arraystretch}{1.5}
\hskip -0.9em\left.
\begin{array}{@{}c@{}}
\displaystyle
\startmodif
2 
\stopmodif
\left(\frac{1}{\varepsilon } +1\right)^2 \leq (2-q) b 
\quad ,\qquad  
\startmodif
2 -a b  \frac{8}{\varepsilon ^2}
\stopmodif
-q\; >\; 0
\\
\multicolumn{1}{@{}l@{}}{ \displaystyle 
\frac{64a^2}{\varepsilon ^6}\left( b + 1  + a b^2   \right)^2\left(1 + a \frac{4}{\varepsilon ^2} \right)^2
\; <}
\\\multicolumn{1}{@{}r@{}}{\displaystyle 
\left(\frac{a}{2}\min\{(2-q),b\} \frac{\varepsilon ^2}{4}- q \right)
\left(
\startmodif
2 -a b  \frac{8}{\varepsilon ^2}
\stopmodif
-q\right)
}
\end{array}\quad   \right\} 
\end{equation}
}{%
\begin{equation}
\label{LP228}
\renewcommand{\arraystretch}{1.5}
\left.
\begin{array}{@{}rcl}
\displaystyle
\startmodif
2 
\stopmodif
\left(\frac{1}{\varepsilon } +1\right)^2 \leq (2-q) b 
&,& \displaystyle 
\startmodif
2 -a b  \frac{8}{\varepsilon ^2}
\stopmodif
-q\; >\; 0
\  ,
\\\displaystyle 
  \frac{64a^2}{\varepsilon ^6}\left(
  b
 +
1 
 +
 a b^2   \right)^2\left(1 + a \frac{4}{\varepsilon ^2} \right)^2
&<& \displaystyle
\left(\frac{a}{2}\min\{(2-q),b\} \frac{\varepsilon ^2}{4}
- q \right)
\left(
\startmodif
2 -a b  \frac{8}{\varepsilon ^2}
\stopmodif
-q\right)
\  .
\end{array}\right\}
\end{equation}
}%
\IfTwoCol{%
With $b$ fixed as
$b= 4\left(\frac{1}{\varepsilon } +1\right)^2$,%
}{%
With $b$ fixed as
$$
b\;=\; 
\startmodif
2
\stopmodif
\left(\frac{1}{\varepsilon } +1\right)^2
\  ,
$$%
}%
since the following inequality is satisfied when $a=0$,
there exists a strictly positive real number $\bar a$ such that, for all $a$ in $[0,\bar a)$, we have
$$
\frac{64a}{\varepsilon ^6}\left(  b + 1  + a b^2   \right)^2\left(1 + a \frac{4}{\varepsilon ^2} \right)^2
\; <\;
\startmodif
\frac{\varepsilon ^2}{8}
\stopmodif
\left( \startmodif
2 -a b  \frac{8}{\varepsilon ^2}
\stopmodif\right)
$$
We fix $a$ in $(0,\bar a)$. By continuity, there exists $q$ satisfying (\ref{LP228}).

We have established the existence of a triplet $(a,b,q)$ such that Conditions A2 and A3 are satisfied on $\Ouv_\varepsilon$ by the metric 
$\bfP$, the expression of which, with the coordinate $\coordyxrp$, is
\IfReport{%
$$ 
P(\coordyp,\coordxrp_\indxra,\coordxrp_\indxrb)
=
\left(\begin{array}{c}
1 \\ 0 \\ 0
\end{array}\right)c
\left(\begin{array}{@{}ccc@{}}
1 & 0 & 0
\end{array}\right)
+
\left(\begin{array}{@{}c@{\ }c@{}}
-1 & \coordyp + ab \coordxrp_\indxra
\\
1 & ab \coordyp
\\
0 & 1
\end{array}\right)\hskip -0.5em
\left(\begin{array}{@{}c@{\  }c@{}}
1 & 0 \\  0 & 1 + a (\coordxrp_\indxra-y)^2 
\end{array}\right)
\left(\begin{array}{@{}c@{\  }c@{\  }c@{}}
-1 & 1 & 0
\\
\coordyp + ab \coordxrp_\indxra& ab \coordyp & 1
\end{array}\right)
\,  ,
$$
}{%
\\[0.7em]$\displaystyle 
P(\coordyp,\coordxrp_\indxra,\coordxrp_\indxrb)
\;=\; 
\left(\begin{array}{c}
1 \\ 0 \\ 0
\end{array}\right)c
\left(\begin{array}{@{}ccc@{}}
1 & 0 & 0
\end{array}\right)
$\hfill\null\\\null\hfill$\displaystyle
+
\left(\begin{array}{@{}c@{\ }c@{}}
-1 & \coordyp + ab \coordxrp_\indxra
\\
1 & ab \coordyp
\\
0 & 1
\end{array}\right)\hskip -0.5em
\left(\begin{array}{@{}c@{\  }c@{}}
1 & 0 \\  0 & 1 + a (\coordxrp_\indxra-y)^2 
\end{array}\right)
\left(\begin{array}{@{}c@{\  }c@{\  }c@{}}
-1 & 1 & 0
\\
\coordyp + ab \coordxrp_\indxra& ab \coordyp & 1
\end{array}\right)
\,  ,
$\\[0.7em]
}
The observer (\ref{eqn:GeodesicObserverVectorField}) 
for the harmonic oscillator with unknown frequency is
\IfTwoCol{%
\\[0.5em]$\displaystyle 
\dot{\overparen{\left(\begin{array}{c}
\hat y
\\
\hat \xrond_\indxra
\\
\hat \xrond_\indxrb
\end{array}
\right)}}
\;=\; 
\left(\begin{array}{c}
\hat \xrond_\indxra
\\
-\hat y \hat \xrond_\indxrb
\\
0
\end{array}
\right)
$\hfill\null\\\null\hfill$\displaystyle
- \frac{k_E(\hat y,\hat \xrond_\indxra,\hat \xrond_\indxrb)}{c}
\left(\begin{array}{ccc}
1 
\\
1 
\\
-\hat y-a b (\hat \xrond_\indxra + \hat y)
\end{array}\right)
(\hat y -y)
\  .
$\\[0.7em]
}{%
$$
\dot{\overparen{\left(\begin{array}{c}
\hat y
\\
\hat \xrond_\indxra
\\
\hat \xrond_\indxrb
\end{array}
\right)}}
\;=\; 
\left(\begin{array}{c}
\hat \xrond_\indxra
\\
-\hat y \hat \xrond_\indxrb
\\
0
\end{array}
\right)
- \frac{k_E(\hat y,\hat \xrond_\indxra,\hat \xrond_\indxrb)}{c}
\left(\begin{array}{ccc}
1 
\\
1 
\\
-\hat y-a b (\hat \xrond_\indxra + \hat y)
\end{array}\right)
(\hat y -y)
\  .
$$
}

As a final remark, we note that the 
expression of the metric with the coordinates $(\coordyp,\coordxip)$ is (by definition)
$$
\bar P(\coordyp,\coordxip_\indxra,\coordxip_\indxrb)\;=\; 
\left(\begin{array}{@{}ccc@{}}
c & 0 & 0
\\
0 & 1 & 0 \\  0 & 0 & 1 + a \coordxip_\indxra^2 
\end{array}\right)
\  .
$$
All the corresponding Christoffel symbols are zero, except
\IfTwoCol{%
$
\bar \Gamma _{\indxrb\indxrb}^\indxra=-\frac{\coordxip_\indxra}{2}$, 
$\bar \Gamma _{\indxra\indxrb}^\indxrb=-\frac{\coordxip_\indxra}{2(a \coordxip_\indxra^2 +1)}
$.
It follows that the component
$
\mathfrak{R} _{\indxra\indxrb\indxrb}^\indxra
=
\frac{\partial \bar \Gamma _{\indxrb\indxrb}^\indxra}{\partial \coordxip_\indxra}
-
\bar \Gamma _{\indxrb\indxrb}^\indxra
\bar \Gamma _{\indxra\indxrb}^\indxrb
= -2a - \frac{\coordxip_\indxra^2}{4(a \coordxip_\indxra^2 +1)}
$
}{%
$$
\bar \Gamma _{\indxrb\indxrb}^\indxra=-\frac{\coordxip_\indxra}{2}
\quad ,\qquad 
\bar \Gamma _{\indxra\indxrb}^\indxrb=-\frac{\coordxip_\indxra}{2(a \coordxip_\indxra^2 +1)}
\  .
$$
It follows that the component $\mathfrak{R} _{\indxra\indxrb\indxrb}^\indxra$
$$
\mathfrak{R} _{\indxra\indxrb\indxrb}^\indxra
\;=\; 
\frac{\partial \bar \Gamma _{\indxrb\indxrb}^\indxra}{\partial \coordxip_\indxra}
-
\bar \Gamma _{\indxrb\indxrb}^\indxra
\bar \Gamma _{\indxra\indxrb}^\indxrb
\;=\; -2a - \frac{\coordxip_\indxra^2}{4(a \coordxip_\indxra^2 +1)}
$$
}
of the Riemann curvature tensor is not zero. So there is no coordinates for which the expression of the 
metric is Euclidean.
\egroup
\end{example}

\startarchive
\subsection{Possible ways to facilitate the satisfaction of Conditions A2 and A3}
\label{sec14}
We propose now two possible ways to facilitate the satisfaction of Conditions A2 and A3. We just give the initial ideas, the 
full analysis remaining to be done.

\subsubsection{Immersion into an input dependent system}
\label{sec15}
We consider again the harmonic oscillator with unknown frequency (\ref{LP132})
evolving in the invariant set $\Ouv_\varepsilon $ defined in (\ref{LP205}).
Its solutions are solutions of the following system, with input $\entree _y$,
\begin{equation}
\label{LP170}
\dot y = \xrond_\indxra
\quad ,\qquad 
\dot{\xrond}_\indxra= 
-\entree _y \xrond_\indxrb
\quad ,\qquad 
\dot{\xrond}_\indxrb = 
[\entree _y  -y]\xrond_\indxra
\end{equation}
when we make the particular choice
$$
\entree _y\;=\; y
$$
for the input. This trivial remark leads us to pay some attention to input-dependent systems. To ease this 
presentation we work within a given coordinate chart $\coordx$ and
we consider the case where everything depends on a, possibly time varying, input vector $\entree  $. 
Namely the system is
\begin{equation}
\label{1}
\dot x \;=\;  f(x,\entree  )
\quad ,\qquad 
y\;=\; h(x)
\end{equation}
and we denote $X(x,t;\entree  )$ its solution. It is important here, for Condition A3, that $h$ does not 
depend on $u$. Let also 
the metric be input-dependent as $(x,\entree  ) \to P(x,\entree  )$. It gives rise to a continuous family of 
Riemannian spaces.

It can be shown 
(see \complement \ref{complement29})
that, in this case,
\begin{list}{}{%
\parskip 0pt plus 0pt minus 0pt%
\topsep 0pt plus 0pt minus 0pt%
\parsep 0pt plus 0pt minus 0pt%
\partopsep 0pt plus 0pt minus 0pt%
\itemsep 0pt plus 0pt minus 0pt%
\settowidth{\labelwidth}{--}%
\setlength{\labelsep}{0.5em}%
\setlength{\leftmargin}{\labelwidth}%
\addtolength{\leftmargin}{\labelsep}%
}
\item[--]
Condition A2 is modified into
\\
There exist 
a continuous function $\rho :\RR ^n\to \RRgeq$
and a strictly positive real number $q$ such that
$$
L_fP(x,\entree  ) + \frac{\partial P}{\partial \entree  }(x,\entree  )\dot \entree  
\; \leq  \; 
\rho (x)\,  \frac{\partial h}{\partial x}(x)^\top
\frac{\partial h}{\partial x}(x)
\;-\; \qlower\,  P(x,\entree  )
\qquad
\forall x\in\coordxd(\Ouv)
\  ;
$$
\item[--]
With Assumption~\ref{H2}, Condition A3 is still implied by the nullity of the second fundamental form of $h$ which is
$$
\secff _P h_{\indxa\indxb}^\indyi(x,\entree )
\;=\; 
\frac{\partial ^2h_\indyi}{\partial x_\indxa\partial x_\indxb}(x)
-\Gamma _{\indxa\indxb}^\indxc(x,\entree)\frac{\partial h_\indyi}{\partial x_\indxc}(x)
+
\Gammay_{\indyj\indyk}^\indyi(h(x))
\frac{\partial h_\indyj}{\partial x_\indxa}(x)
\frac{\partial h_\indyk}{\partial x_\indxb}(x)
\  ;
$$
\item[--]
Theorem \ref{thm1part1} holds with these modifications.
\end{list}

We show the interest of immersing a system into an input-dependent one via an example.
\begin{example}
\label{ex6}
\bgroup
\normalfont
We consider the harmonic oscillator with unknown frequency (\ref{LP132}) immersed into the input-dependent system 
(\ref{LP170}). To make sure that this input-dependent system satisfies Condition A3 we follow
Theorem~\ref{prop17}. We select
\begin{list}{}{%
\parskip 0pt plus 0pt minus 0pt%
\topsep 0pt plus 0pt minus 0pt%
\parsep 0pt plus 0pt minus 0pt%
\partopsep 0pt plus 0pt minus 0pt%
\itemsep 0pt plus 0pt minus 0pt%
\settowidth{\labelwidth}{--}%
\setlength{\labelsep}{0.5em}%
\setlength{\leftmargin}{\labelwidth}%
\addtolength{\leftmargin}{\labelsep}%
}
\item[--]
the set $\bfXi$ as $\bfRR^2$ and equip it with a global coordinate chart $\coordxi$
and a metric $\bfPxi$ the expression of which is
\startmodif
$$
R(\coordxip_\indxra,\coordxip_\indxrb,\entree _y,\entree _\indxra)\;=\; 
\left(\begin{array}{@{\,  }cc@{\,  }}
2a{} +b{}^2 
&
\entree _yb{}
\\
\entree _yb{}
&
2a{} + 2\entree _y^2 + 2 \entree _\indxra^2 - 2 \entree _y \entree _\indxra 
\end{array}\right)
$$
\stopmodif
where $a$ and $b$ are two real numbers, with $a$ strictly positive.
Note that $R$
does not depend on $(\coordxip_\indxra,\coordxip_\indxrb)$ but it depends on $\entree _y$, input of
(\ref{LP170}), and also on $\entree_\indxra$, a new input.
\item[--]
The expression of the function $\bfh^\ortho$ as
$$
h^\ortho(y,\xrond_\indxra,\xrond_\indxrb) 
\;=\; 
\left(\xrond_\indxra - y\,  ,\,  \xrond_\indxrb+\frac{y^2}{2}\right)
\  .
$$
\item[--]
the expression of the metric $\bfPy$ for the $\bfy$-manifold as  $c$, a strictly positive real number to be chosen.
\end{list} 
The construction of the input-dependent system and the metric for the harmonic oscillator
is given in the \complement \ref{complement30}.

From (\ref{LP166}), we get the metric
\startmodif
\begin{equation}
\label{LP172}
P(y,\entree  _y,\entree  _\indxra) = 
\left(\begin{array}{@{}c@{}}
1 \\ 0 \\ 0
\end{array}\right)
c
\left(\begin{array}{@{}ccc@{}}
1 & 0 & 0
\end{array}\right)
+
\left(\begin{array}{@{}cc@{}}
-1 & y 
\\
1 & 0
\\
0 & 1
\end{array}\right)
\left(\begin{array}{@{\,  }cc@{\,  }}
2a{} +b{}^2 
&
\entree _yb{}
\\
\entree _yb{}
&
2a{} + 2\entree _y^2 + 2 \entree _\indxra^2 - 2 \entree _y \entree _\indxra 
\end{array}\right)
\left(\begin{array}{@{}ccc@{}}
-1 & 1 & 0
\\
y & 0 & 1
\end{array}\right)
\end{equation}
\stopmodif
It follows from Theorem~\ref{prop17} that Condition A3 holds since it is not affected by the input-dependence.

Invoking the S-Lemma (see \cite{Polik-Terlaky.07}), Condition A2 holds if we have
\\[1em]\vbox{\noindent$\displaystyle 
\left(\begin{array}{@{}c@{\  }c@{\  }c@{}}
0 & v_\indxra & v_\indxrb
\end{array}\right)\!
\left[
\frac{\partial P}{\partial \entree  _y} \dot \entree  _y
+
\frac{\partial P}{\partial \entree  _\indxra } \dot \entree  _\indxra 
\right]\!
\left(\begin{array}{@{}c@{}}
0 \\ v_\indxra  \\ v_\indxrb
\end{array}\right)
+
2\left(\begin{array}{@{}c@{\  }c@{\  }c@{}}
0 & v_\indxra & v_\indxrb
\end{array}\right)
P(y,\entree  _y,\entree  _\indxra)\left(\begin{array}{@{}cc@{}}
1 & 0
\\
0 & -\entree  _y
\\
\entree  _y-y & 0
\end{array}\right)
\left(\begin{array}{@{}c@{}}
 v_\indxra  \\ v_\indxrb
\end{array}\right)
$\hfill \null \\\null \hfill $\displaystyle 
\leq - q \  
\left(\begin{array}{@{}c@{\  }c@{\  }c@{}}
0 & v_\indxra & v_\indxrb
\end{array}\right)
P
\left(\begin{array}{@{}c@{}}
0 \\ v_\indxra  \\ v_\indxrb
\end{array}\right)
$}\\[1em]
\startmodif
We compute
\\[1em]$\displaystyle 
\left(\begin{array}{@{}ccc@{}}
0 & v_\indxra & v_\indxrb
\end{array}\right)
\left[
\frac{\partial P}{\partial \entree  _y} \dot \entree  _y
+
\frac{\partial P}{\partial \entree  _\indxra } \dot \entree  _\indxra 
\right]
\left(\begin{array}{@{}c@{}}
0 \\ v_\indxra  \\ v_\indxrb
\end{array}\right)
$\hfill \null \\\null \hfill $
\;=\; 
\left(\begin{array}{@{}ccc@{}}
v_\indxra & v_\indxrb
\end{array}\right)
\left[\left(\begin{array}{@{}cc@{}}
0
&
b{}
\\
b{}
&
4\entree _y - 2 \entree  _\indxra 
\end{array}\right) \dot \entree  _y
\;+\; 
\left(\begin{array}{@{}cc@{}}
0
&
0
\\
0
&
4 \entree  _\indxra - 2 \entree _y 
\end{array}\right) \dot \entree  _\indxra 
\right]
\left(\begin{array}{@{}c@{}}
v_\indxra \\ v_\indxrb
\end{array}\right)
$%
\\[1em]$\displaystyle
\left(\begin{array}{@{\,  }cc@{\,  }}
2a{} +b{}^2 
&
\entree _yb{}
\\
\entree _yb{}
&
2a{} + 2\entree _y^2 + 2 \entree _\indxra^2 - 2 \entree _y \entree _\indxra 
\end{array}\right)
\left(\begin{array}{@{}ccc@{}}
-1 & 1 & 0
\\
y & 0 & 1
\end{array}\right)
\left(\begin{array}{@{}cc@{}}
1 & 0
\\
0 & -\entree  _y
\\
\entree  _y-y & 0
\end{array}\right)
$\hfill\null\\\null\hfill$\displaystyle
=\; 
\left(\begin{array}{@{\,  }cc@{\,  }}
2a{} +b{}^2 
&
\entree _yb{}
\\
\entree _yb{}
&
2a{} + 2\entree _y^2 + 2 \entree _\indxra^2 - 2 \entree _y \entree _\indxra 
\end{array}\right)
\left(\begin{array}{@{}cc@{}}
-1 & -\entree  _y
\\
\entree  _y & 0
\end{array}\right)
$\qquad  \null 
\\[1em]\null \hfill $\displaystyle 
= \renewcommand{\arraystretch}{1.5}
\left(\begin{array}{@{}cc@{}}
\left[-(2a{} +b{}^2)+\entree _y^2b{}\right]
&
-\left[2a+b^2\right]\entree  _y
\\
\left[
- \entree _yb{}
+\entree  _y\left(2a{} + 2\entree _y^2 + 2 \entree  _\indxra ^2 - 2 \entree _y \entree  _\indxra  \right)
\right]
&
-\left[\entree  _y ^2b{}\right]
\end{array}\right)
$\\[1em]
\stopmodif
By expanding we get the inequality
\\[1em]\vbox{\noindent\null \quad $\displaystyle 
2v_\indxra^2
\left[-\left(1-\frac{q}{2}\right)(2a{} +b{}^2)+\entree _y^2b{}\right]
$\hfill \null \\[0.7em]\null \qquad \qquad   $\displaystyle
+
2v_\indxra v_\indxrb
\left[
b{}\dot \entree  _y
+
\entree  _y
\left(-b{}^2-[1-q]b{}
+
 2\entree _y^2 +  2 \entree  _\indxra ^2 - 2 \entree _y \entree  _\indxra 
\right)
\right]
$\refstepcounter{equation}\label{LP188}\hfill$(\theequation)$
\\[0.7em]\null \qquad \qquad \qquad \qquad   $\displaystyle
+
v_\indxrb^2
\left[
\left(
4\entree _y - 2 \entree  _\indxra 
\right) \dot \entree  _y
+
\left(4 \entree  _\indxra - 2 \entree _y 
\right) \dot \entree  _\indxra 
- \entree  _y^2 b{}
+ q [2a{} + 2\entree _y^2 + 2 \entree  _\indxra ^2 - 2 \entree _y \entree  _\indxra ]
\right]
$\hfill \null \\[0.7em]\null \qquad \qquad \qquad \qquad  \qquad \qquad  $\displaystyle 
\leq \; 0
$
}\\[1em]
%
The specific expression we have chosen for $\dot\xrond_\indxrb$, in the input-dependent system, plays an
important role here with preventing
the presence of $a{}$  in the cross term $v_\indxra v_\indxrb$.

At this point, we remind the reader that, to match the given harmonic oscillator with unknown frequency, $u_y$ 
is to be $y$, but $u_\indxra$ is still any time function. Inspired by (\ref{LP132}), we have found fruitful 
to choose $u_\indxra$ generated by the following system~:
\begin{equation}
\label{16}
\dot \entree  _y = \entree  _\indxra 
\quad ,\qquad 
\dot{\entree  }_\indxra 
 = -\entree _y \entree  _\indxrb  
\quad ,\qquad 
\dot{\entree  }_\indxrb 
= 0
\end{equation}
with $(\entree  _y ,\entree  _\indxra,\entree  _\indxrb)$ evolving in the invariant set 
$\Ouv_\varepsilon $. Its meaning is that, when we have $u_y=y$, we have at the same time 
$u_\indxra=\xrond_\indxra$ and $u_\indxrb=\xrond_\indxrb$. With the notations
$$
\rho{} \;=\; v_\indxra
\quad ,\qquad 
r{} _y\;=\; \entree  _y v_\indxrb
\quad ,\qquad 
r{} _\indxra\;=\; \entree  _\indxra v_\indxrb
\quad ,\qquad 
k\;=\; a\,  q
\  ,
$$
inequality (\ref{LP188}) becomes
\\[1em]$\displaystyle 
-2 \rho{} ^2
\left[\left(1-\frac{k}{2a}\right)(2a{} +b{}^2)-\entree _y^2b{}\right]
$\hfill \null \\\null \hfill $\displaystyle 
-
2
\left[
r{} _\indxra ^2\left(1-\frac{k}{a}-\frac{k}{u_y^2+u_\indxra^2}\right)
+
2 r{} _\indxra r{} _y \left(\entree  _\indxrb+\frac{k}{2a} -1\right)
+
r{} _y^2\left(\frac{b{}}{2}-\entree  _\indxrb-\frac{k}{a}-\frac{k}{u_y^2+u_\indxra^2}\right)
\right]
$\refstepcounter{equation}\label{LP173}\hfill$(\theequation)$
\\\null \hfill $\displaystyle
+
2\rho{}  
\left[
b{}r _\indxra
-
r _yb{}^2
-
r _yb{}\left(1-\frac{k}{a}\right)
+
 2r{} _y \entree _y^2
+
\startmodif
2 r{} _y \entree  _\indxra ^2
\stopmodif
-
2 \entree _y r{} _\indxra 
\right]
\; \leq \; 0
$\\[1em]
When $\rho{} $ is zero, this reduces to the fact that the quadratic form in $(r_y,r_\indxra)$
$$
r{} _\indxra ^2\left(1-\frac{k}{a}-\frac{k}{u_y^2+u_\indxra^2}\right)
+
2 r{} _\indxra r{} _y \left(\entree  _\indxrb+\frac{k}{2a} -1\right)
+
r{} _y^2\left(b{}-\entree  _\indxrb-\frac{k}{a}-\frac{k}{u_y^2+u_\indxra^2}\right)
$$
is positive definite when $(\entree  _y,\entree  _\indxra ,\entree  _\indxrb )$ is in the set $\Ouv_\varepsilon $ where
$$
\varepsilon ^2\leq u_y^2+u_\indxra^2
\quad ,\qquad 
u_\indxrb\leq \frac{1}{\varepsilon }
\  .
$$
This is the case if we fix $a$ arbitrary, $k$ satisfying~:
$$
k\left(\frac{1}{a}+\frac{1}{\varepsilon ^2}\right)\; <\; 1
$$
and then $b$ large enough. By continuity, for $\rho $ sufficiently small, (\ref{LP173}) holds when $(\entree  _y,\entree  _\indxra ,\entree  _\indxrb )$ is
in $\Ouv_\varepsilon $. So we have a negative definite quadratic form in
$(v_\indxra,\entree  _y v_\indxrb, \entree  _\indxra v_\indxrb)$ and therefore in 
$(v_\indxra,v_\indxrb)$. 

With Conditions A3 and A2 satisfied, the assumptions of the extension of Theorem \ref{thm1part1}
to the input-dependent case are satisfied. The observer (\ref{eqn:GeodesicObserverVectorField}) 
for the input-dependent system (\ref{LP170}) takes the form
\begin{eqnarray}
\nonumber
\dot{\hat y}&=&\hat \xrond_\indxra - \frac{k_E}{c{}}[\hat y -y]
\  ,
\\\label{LP171}
\dot{\hat{\xrond}}_\indxra&=&
-\entree _y \hat{\xrond}_\indxrb
- \frac{k_E}{c{}}[\hat y -y]
\  ,
\\\nonumber
\dot{\hat{\xrond}}_\indxrb &=&
[\entree _y  -\hat y]\hat \xrond_\indxra +  \frac{k_E}{c{}}\hat y[\hat y -y]
\  .
\end{eqnarray}
Actually we are interested in an observer not for (\ref{LP170}) but for the given harmonic oscillator with unknown frequency.
We have noticed that solutions of the latter are solutions of (\ref{LP170})
if we choose
$$
\entree  _y\;=\; y
$$
Then, because of (\ref{16}), we should choose also
$$
\entree  _\indxra \;=\; \xrond_\indxra
\quad ,\qquad 
\entree  _\indxrb \;=\; \xrond_\indxrb
$$
But, if $y$ is known as a measurement, $\xrond_\indxra$ and $\xrond_\indxrb$ are not (this is why we need an 
observer!). So the observer (\ref{LP171})
cannot depend on
$(\entree  _\indxra ,\entree  _\indxrb )$ when $\entree  _y = y$ if we want it to lead to an observer for 
the given system. This is fortunately the case. So all this procedure leads to the following
observer for the harmonic oscillator with unknown frequency
\begin{eqnarray}
\nonumber
\dot{\hat y}&=&\hat \xrond_\indxra - \frac{k_E}{c{}}[\hat y -y]
\  ,
\\
\label{LP206}
\dot{\hat{\xrond}}_\indxra&=&
-y \hat{\xrond}_\indxrb
- \frac{k_E}{c{}}[\hat y -y]
\  ,
\\
\nonumber
\dot{\hat{\xrond}}_\indxrb &=&
\left[\frac{k_E}{c{}}\hat y - \hat \xrond_\indxra\right]
[\hat y-y]
\  .
\end{eqnarray}
Instead of relying on Theorem \ref{thm1part1}, we can prove its convergence by establishing that the flow it 
generates is contracting for the metric $\bfP$ the expression of which is (\ref{LP172}) evaluated at
$(\hat y,y,\xrond_\indxra)$.
See \complement \ref{complement42} for details.
\egroup
\end{example}

\subsubsection{Dynamic extension}
\label{sec16}
Another possibly fruitful route to find metrics satisfying both Conditions A2 and A3 is to 
augment the state space.

We consider an augmented state $\bfxa$ in the augmented space $\RR^m$
and we define an augmented dynamics for $\bfxa$ in such a way that any observer we could design 
for the augmented state $\bfxa$ can be used as an observer for the given state $\bfx$. This is possible if~:
\begin{list}{}{%
\parskip 0pt plus 0pt minus 0pt%
\topsep 0.5ex plus 0pt minus 0pt%
\parsep 0pt plus 0pt minus 0pt%
\partopsep 0pt plus 0pt minus 0pt%
\itemsep 0.5ex plus 0pt minus 0pt
\settowidth{\labelwidth}{\qquad a)}%
\setlength{\labelsep}{0.5em}%
\setlength{\leftmargin}{\labelwidth}%
\addtolength{\leftmargin}{\labelsep}%
}
\item[a)]
we have a $C^s$ function $\bfpi:\bfRR^m \to \bfRR^n$ giving
$$
\bfx=\bfpi(\bfxa)
\  ;
$$
\item[b)]
any solution to (\ref{eqn:Plant1}) 
can be augmented in a solution of the augmented system to be defined;
\item[c)]
the measurement of the given system (\ref{eqn:Plant1}) is a measurement of the augmented one.
\end{list}

For a), we note that, if $\coordx$ and $\coordxe$ are any coordinate charts for $\RR^n$ and 
$\RR^{m-n}$ respectively, then
$\coordxxe$ is a privileged coordinate chart for $\RR^m$. 
So a) is met as long as, for any $\bfxa$
there is a coordinate chart $\coordxxe$ around $\bfxa$. This is realized by keeping track of the canonical 
projection $\pi:\RR^m\to \RR^n$.

For b), we choose to augment a solution of 
(\ref{eqn:Plant1}) with a component $\bfxe  $ which is constantly zero. This means
that, for the augmented dynamics, expressed with the coordinate chart
$\coordxxe$ as
\begin{equation}
\label{LP174}
 \dot{\overparen{\left(\begin{array}{c}
x \\ \xe 
\end{array}\right)}}\;=\; 
\left(\begin{array}{l}
f_x  (x,\xe )
\\
f_{\xe}(x,\xe )
\end{array}\right)
\  ,
\end{equation}
we impose
\begin{equation}
\label{LP175}
 f_x(x,0)\;=\; f(x)
\quad ,\qquad 
f_{\xe}  (x,0)\;=\; 0
\  .
\end{equation}

Finally, for c), the output function $\bfha$ of the augmented system should
be such that, for any $\bfxa$ we have coordinate charts $\coordxxe$ around $\bfxa$
and
$\coordya$ around $\bfha(\bfxa)$ such that the expressions $\ha$ of $\bfha$ and $h$ of $\bfh$
satisfy
\begin{equation}
\label{LP179}
\ha(x,0)\;=\; h(x)
\  .
\end{equation}
At this point it is important to note that dynamic extension is of no help for satisfying Condition A2.
%
Indeed, we have the following extension of \cite[Proposition 4.4]{127})

\begin{proposition}
\label{prop:prop19}
Assume that, with a coordinate chart  $\coordxxe$,
Condition A2 holds for the system (\ref{LP174})
satisfying (\ref{LP175}) and (\ref{LP179}), and with the expression $\Pa$ of some metric.
Then, with the particular coordinate chart $\coordx$, the  system
\begin{equation}
\label{LP183}
\dot x= f(x)
\quad ,\qquad 
y=h(x)
\end{equation}
satisfies Condition A2 for a metric $\bfP$ the expression of which is
\begin{equation}
\label{LP176}
P(x) = 
\left(\begin{array}{@{}cc@{}}
I_n & 0
\end{array}\right)
\Pa  (x,0)
\left(\begin{array}{@{}c@{}}
I_n \\ 0
\end{array}\right) 
\  .
\end{equation}
\end{proposition}
\begin{proof}
Condition A2 for the expression (\ref{LP174}) in the coordinates $(x,\xe)$ of the augmented system is:
\\
There exists a strictly positive real number $q_a$ such that,
for all $(x,\xe,v,v_e)$ satisfying
\begin{equation}
\label{LP181}
\frac{\partial \ha}{\partial x}(x,\xe) v
+ \frac{\partial \ha}{\partial \xe}(x,\xe) v_e
\;=\; 0
\  ,
\end{equation}
we have~:
\\[1em]$\displaystyle 
\frac{1}{2}\frac{\partial }{\partial x}
\left\{
\left(\begin{array}{@{}cc@{}}
v^\top & v_e^\top
\end{array}\right) \Pa (x,\xe)\left(\begin{array}{@{}c@{}}
v \\ v_e
\end{array}\right)\right\}f_x(x,\xe)
+
\frac{1}{2}\frac{\partial }{\partial \xe}
\left\{
\left(\begin{array}{@{}cc@{}}
v^\top & v_e^\top
\end{array}\right)
\Pa (x,\xe)
\left(\begin{array}{@{}c@{}}
v \\ v_e
\end{array}\right)
\right\}
f_{\xe}(x,\xe)
$\refstepcounter{equation}\label{LP182}\hfill$(\theequation)$\\\null \hfill $\displaystyle +
\left(\begin{array}{@{}cc@{}}
v^\top & v_e^\top
\end{array}\right)
\Pa (x,\xe)
\left(\begin{array}{@{}cc@{}}
\frac{\partial f_x}{\partial x}(x,\xe)
&
\frac{\partial f_x}{\partial \xe}(x,\xe)
\\[0.5em]
\frac{\partial f_{\xe}}{\partial x}(x,\xe)
&
\frac{\partial f_{\xe}}{\partial \xe}(x,\xe)
\end{array}\right)
\left(\begin{array}{@{}c@{}}
v \\ v_e
\end{array}\right)
\; \leq  \; -q_a
\,  
\left(\begin{array}{@{}cc@{}}
v^\top & v_e^\top
\end{array}\right)
\Pa (x,\xe)
\left(\begin{array}{@{}c@{}}
v \\ v_e
\end{array}\right)
\  .
$\\[1em]
For $\xe=0$,
(\ref{LP175})  and (\ref{LP179}) give
\begin{equation}
\label{LP180}
\frac{\partial f_x}{\partial x}(x,0)\;=\; \frac{\partial f}{\partial x}(x)
\  ,\quad 
f_{\xe}(x,0)\;=\; 0
\  ,\quad 
\frac{\partial f_{\xe}}{\partial x}(x,0)\;=\; 0
\  ,\quad 
\frac{\partial \ha}{\partial x}(x,0)\;=\; \frac{\partial h}{\partial x}(x)
\  .
\end{equation}
So, when $\xe=0$ and $v_e=0$, (\ref{LP181}) is
$$
\frac{\partial h}{\partial x}(x) v\;=\; 0
$$
and  (\ref{LP182}) is
\\[1em]$\displaystyle 
\frac{1}{2}\frac{\partial }{\partial x}
\left\{
\left(\begin{array}{@{}cc@{}}
v^\top & 0
\end{array}\right)
\Pa (x,0)
\left(\begin{array}{@{}c@{}}
v \\ 0
\end{array}\right)\right\}f(x)
$\hfill \\\null \hfill $\displaystyle +
\left(\begin{array}{@{}cc@{}}
v^\top & 0
\end{array}\right)
\Pa (x,0)
\left(\begin{array}{@{}cc@{}}
\frac{\partial f}{\partial x}(x)
&
\frac{\partial f_x}{\partial \xe}(x,0)
\\[0.5em]
0
&
\frac{\partial f_{\xe}}{\partial \xe}(x,0)
\end{array}\right)
\left(\begin{array}{@{}c@{}}
v \\ 0
\end{array}\right)
\; \leq \;  
-q_a
\,  
\left(\begin{array}{@{}cc@{}}
v^\top & 0
\end{array}\right)
\Pa (x,0)
\left(\begin{array}{@{}c@{}}
v \\ 0
\end{array}\right)
$\\[1em]
With (\ref{LP176}), we have established:\\
There exists a strictly positive real number $q_a$ such that, for all 
$(x,v)$ satisfying
$$
\frac{\partial h}{\partial x}(x) v 
\;=\; 0
\  ,
$$
we have~:
$$
\frac{1}{2}\frac{\partial }{\partial x}
\left\{
v^\top P(x) v \right\}f(x)
\;+\; 
v^\top 
P(x)\frac{\partial f}{\partial x}(x)v 
\; \leq \;  -q_a v^\top P (x) v
\  .
$$
This is Condition A2 for the system (\ref{LP183}).
\end{proof}

We conclude that, if we want an augmented system satisfying Condition A2, we need to have this condition 
already satisfied for the given system and moreover we must have the relation (\ref{LP176}) between the 
metrics. So dynamic extension does not bring any loosening to Condition A2.

This is different for Condition A3. Indeed, we know, from Isometric Embedding Theorems\footnote{
We are very grateful to Vincent Andrieu, from  LAGEP in Lyon, for suggesting this link.
}, that with $m$ sufficiently large, 
there is no loss of generality in imposing $\bfPa$ to be flat, i.e. we have a global coordinate chart
$(\barcoordxap,\RR^m,i_d)$ such that the expression of $\bfPa$ is the identity matrix. In this case, as noticed in Example 
\ref{ex4}, Condition A3 holds if Assumption~\ref{H2} is satisfied and, for any $\bfxa$,
there exists a coordinate chart $\barcoordya$ around $\bfha(\bfxa)$ and a constant matrix $C$
such that the expression $\barha$ of $\bfha$ in these coordinates, i.e.
$$
\barha(\barcoordxap)
\;=\; 
\coordyam(\bfha(i_d^{-1}(\barcoordxap)))
$$
satisfies
$$
\barha(\barcoordxap)\;=\; C\,  \barcoordxap
\qquad \forall \barcoordxap:
\bfha(i_d^{-1}(\barcoordxap))\in \coordyad
\  .
$$

\begin{example}[{See \cite[Section II.A]{Bernard-Praly-Andrieu-CDC15}}]
\bgroup
\normalfont
With choosing $m=4$, the dynamics of an augmented system for the harmonic oscillator with unknown frequency 
takes the form~:
\begin{equation}
\label{LP168}
\renewcommand{\arraystretch}{1.5}
\begin{array}{lcl}
\dot y &=&\displaystyle  \
\xrond_\indxra +  f_y(y,\xrond_\indxra,\xrond_\indxrb,\xe)\xe
\  ,
\\
\dot \xrond_\indxra&=& \displaystyle-y \xrond_\indxrb +  f_\indxra(y,\xrond_\indxra,\xrond_\indxrb,\xe)\xe
\  ,
\\
\dot \xrond_\indxrb&=&\displaystyle  f_\indxrb(y,\xrond_\indxra,\xrond_\indxrb,\xe)\xe
\  ,
\\
\vardot{\xe}{x}&=&\displaystyle  f_{\xe}(y,\xrond_\indxra,\xrond_\indxrb,\xe)\xe
\   ,
\\
\ya&=&\displaystyle  y +  f_h(y,\xrond_\indxra,\xrond_\indxrb,\xe)\xe
\  .
\end{array}
\end{equation}

We look for the functions $f_y$, $f_\indxra$, $f_\indxrb$, $f_{\xe}$ and $f_h$, the constant positive definite symmetric matrix $\bar \Pa$
and diffeomorphisms $\changexa$ and $\changeya$ such that Condition A2 
holds with
$$
\Pa(y,\xrond_\indxra,\xrond_\indxrb,\xe)\;=\; \left(\begin{array}{c}
\displaystyle 
\frac{\partial \changexa}{\partial y}^\top
\\\displaystyle 
\frac{\partial \changexa}{\partial \xrond_\indxra}^\top
\\\displaystyle 
\frac{\partial \changexa}{\partial \xrond_\indxrb}^\top
\\\displaystyle 
\frac{\partial \changexa}{\partial \xe}^\top
\end{array}\right)
\bar \Pa
\left(\begin{array}{cccc}
\displaystyle 
\frac{\partial \changexa}{\partial y}
&\displaystyle 
\frac{\partial \changexa}{\partial \xrond_\indxra}
&\displaystyle 
\frac{\partial \changexa}{\partial \xrond_\indxrb}
&\displaystyle 
\frac{\partial \changexa}{\partial \xe}
\end{array}\right)
$$
$$
\changeya(y + \xe f_{\xe}(y,\xrond_\indxra,\xrond_\indxrb,\xe))\;=\; 
C \changexa(y,\xrond_\indxra,\xrond_\indxrb,\xe)
$$
These two equations means that, with the coordinates
$$
\barcoordxap\;=\; \changexa(\ya,\xrond_\indxra,\xrond_\indxrb,\xe)
\quad ,\qquad 
\barcoordyap\;=\; \changeya(\coordyap)
$$
the expression of the metric $\bfPa$ is the constant matrix $\bar \Pa$ and the expression of the output 
function $\bfha$ is linear. So, as reminded above, with the function $\wpunbf$ chosen as
$$
\wpunbf (y_{a1},y_{a2})\;=\; |y_{a1}-y_{a2}|^2
\  ,
$$
these two equations guarantee Condition A3 holds for the augmented system.

Inspired by the fact that the harmonic oscillator is strongly differentially observable of order $4$, we 
choose $y$ and its first $3$ derivatives as coordinates for $\bfxa$. Precisely we choose the diffeomorphisms 
$
\changexa
$, on the set $\Ouv_\varepsilon $ defined in (\ref{LP205}), and $\changeya$ and the function $f_h$ as
$$
\changexa(y, \xrond_\indxra, \xrond_\indxrb , \xe )=(y, \xrond_\indxra, -y \xrond_\indxrb + \xrond_\indxra \xe ,- \xrond_\indxra \xrond_\indxrb- y\xe)
\  ,\quad 
\changeya(y)\;=\; y
\  ,\quad 
f_h(y,\xrond_\indxra,\xrond_\indxrb,\xe)\;=\; 0
\  .
$$
With these coordinates the dynamics (\ref{LP168}) are
\begin{eqnarray*}
\dot y &=& \xrond_\indxra +  f_y(y,\xrond_\indxra,\xrond_\indxrb,\xe)\xe
\  ,
\\
\vardot{\bar \xrond_\indxra}{x}
&=&  -y \xrond_\indxrb + f_\indxra(y,\xrond_\indxra,\xrond_\indxrb,\xe)\xe
\\
&=&
\bar \xrond_\indxrb
+  [ f_\indxra(y,\xrond_\indxra,\xrond_\indxrb,\xe)-\xrond_\indxra]\xe
\  ,
\\
\vardot{\bar \xrond_\indxrb}{x}&=&-\xrond_\indxra\xrond_\indxrb -
y  f_\indxrb(y,\xrond_\indxra,\xrond_\indxrb,\xe) \xe
- y \xrond_\indxrb \xe
+
\xrond_\indxra \xe ^2
+
\xrond_\indxra  f_{\xe}(y,\xrond_\indxra,\xrond_\indxrb,\xe) \xe
\\
&=&
\barxe
+y[1-  f_\indxrb(y,\xrond_\indxra,\xrond_\indxrb,\xe)-\xrond_\indxrb ] \xe
+
\xrond_\indxra  [\xe+f_{\xe}(y,\xrond_\indxra,\xrond_\indxrb,\xe)] \xe
\  ,
\\
\vardot{\barxe}{x}
&=&
-
[-y \xrond_\indxrb +  f_\indxra(y,\xrond_\indxra,\xrond_\indxrb,\xe)\xe]
\xrond_\indxrb
-
\xrond_\indxra
[f_\indxrb(y,\xrond_\indxra,\xrond_\indxrb,\xe)\xe]
\\&&\qquad \qquad \qquad \qquad \qquad \qquad -
[\xrond_\indxra +  f_y(y,\xrond_\indxra,\xrond_\indxrb,\xe)\xe]\xe
-
y
[f_{\xe}(y,\xrond_\indxra,\xrond_\indxrb,\xe)\xe]
\  .
\end{eqnarray*}
As expected from the particular choice of $\changexa$, by choosing
\begin{equation}
\label{LP229}
\left.
\renewcommand{\arraystretch}{1.5}
\begin{array}{rcl@{\quad ,\qquad }rcl@{\qquad }}
f_y(y,\xrond_\indxra,\xrond_\indxrb,\xe)&=& 0
&
f_\indxra(y,\xrond_\indxra,\xrond_\indxrb,\xe)&=& \xrond_\indxra
\  ,
\\
f_\indxrb(y,\xrond_\indxra,\xrond_\indxrb,\xe)&=&1-\xrond_\indxrb
&
f_{\xe}(y,\xrond_\indxra,\xrond_\indxrb,\xe)&=& -\xe
\  ,
\end{array}
\right\}
\end{equation}
we get the chain of integrators
\begin{eqnarray*}
\dot y &=& \barxrond_\indxra
\  ,
\\
\vardot{\barxrond_\indxra}{x}
&=& \bar \xrond_\indxrb
\  ,
\\
\vardot{\bar \xrond_\indxrb}{x}&=&
\barxe
\end{eqnarray*}
and
$$
\vardot{\barxe}{x}
=
-
[-y \xrond_\indxrb +  \xrond_\indxra\xe]
\xrond_\indxrb
-
\xrond_\indxra
[1-\xrond_\indxrb]\xe
-
\xrond_\indxra \xe
+
y\xe^2
$$
where
$$
\xrond_\indxrb\;=\;
-\frac{y\bar \xrond_\indxrb+\barxrond_\indxra \barxe}{y^2+\barxrond_\indxra^2}
\quad ,\qquad 
\xe\;=\; \frac{\barxrond_\indxra \bar \xrond_\indxrb-y\barxe}{y^2+\barxrond_\indxra^2}
$$
This gives
\begin{eqnarray*}
\vardot{\barxe}{x}
&=& \bar \xrond_\indxrb \frac{y\bar \xrond_\indxrb+\barxrond_\indxra
\barxe
}{y^2+\barxrond_\indxra^2}
-\barxrond_\indxra \left(1+\frac{y\bar \xrond_\indxrb+\barxrond_\indxra
\barxe
}{y^2+\barxrond_\indxra^2}\right) 
\frac{\barxrond_\indxra \bar \xrond_\indxrb-y
\barxe
}{y^2+\barxrond_\indxra^2}
- \barxrond_\indxra \frac{\barxrond_\indxra \bar \xrond_\indxrb-y
\barxe
}{y^2+\barxrond_\indxra^2}
+y \left(\frac{\barxrond_\indxra \bar \xrond_\indxrb-y
\barxe
}{y^2+\barxrond_\indxra^2}\right)^2
\\
&=&
\frac{y\bar \xrond_\indxrb ^2+\barxrond_\indxra \bar \xrond_\indxrb
\barxe
-\barxrond_\indxra^2\bar \xrond_\indxrb+y\barxrond_\indxra
\barxe
}{y^2+\barxrond_\indxra^2}
-\barxrond_\indxra 
+
\frac{
-\barxrond_\indxra [y\bar \xrond_\indxrb+\barxrond_\indxra 
\barxe
][\barxrond_\indxra \bar \xrond_\indxrb-y
\barxe
]
+y [\barxrond_\indxra \bar \xrond_\indxrb-y
\barxe
]^2
}{[y^2+\barxrond_\indxra^2]^2}
\\
&=&
\frac{y\bar \xrond_\indxrb ^2+\barxrond_\indxra \bar \xrond_\indxrb
\barxe
-\barxrond_\indxra^2\bar \xrond_\indxrb+y\barxrond_\indxra
\barxe
}{y^2+\barxrond_\indxra^2}
-\barxrond_\indxra 
+
\frac{[y^2+\barxrond_\indxra^2] y \xe^2 - \barxrond_\indxra^3 \xrond_\indxrb \xe
}{[y^2+\barxrond_\indxra^2]^2}
\\
&=&
\frac{y\bar \xrond_\indxrb ^2+\barxrond_\indxra \bar \xrond_\indxrb
\barxe
-\barxrond_\indxra^2\bar \xrond_\indxrb+y\barxrond_\indxra
\barxe
+
y \xe^2
}{y^2+\barxrond_\indxra^2}
-\barxrond_\indxra 
+
\frac{ - \barxrond_\indxra^3 \xrond_\indxrb \xe
}{[y^2+\barxrond_\indxra^2]^2}
\\
&=&
\bar f_{\xe}(y,\barxrond_\indxra,\barxrond_\indxrb,\barxe)
\end{eqnarray*}
Hence we have the observer form
\begin{eqnarray*}
\dot y &=& \barxrond_\indxra 
\  ,
\\
\vardot{\barxrond_\indxra}{x}&=& \bar \xrond_\indxrb
\  ,
\\
\vardot{\barxrond_\indxrb}{x}&=&
\barxe
\  ,
\\
\vardot{\barxe}{x}
&=&
\bar f_{\xe}(y,\barxrond_\indxra,\barxrond_\indxrb,\barxe)
\end{eqnarray*}
It is known that a high gain observer can be used. It follows that Condition A2 holds with
a metric $\bfPa$ the expressions of which with the coordinates $(y,\xrond_\indxra,\barxrond_\indxrb,\barxe)$
is a constant matrix $\bar \Pa$.

The observer given by (\ref{eqn:GeodesicObserverVectorField}) 
for the augmented system (\ref{LP168}) with (\ref{LP229}), i.e.
$$
\dot{\overparen{\left(\begin{array}{c}
y 
\\
\xrond_\indxra
\\
\xrond_\indxrb
\\
\xe
\end{array}\right)}}
\;=\; 
\left(\begin{array}{c}
\displaystyle  \xrond_\indxra 
\\
\displaystyle-y \xrond_\indxrb +  \xrond_\indxrb\xe
\\
\displaystyle  [1-\xrond_\indxrb]\xe
\\-\displaystyle  \xe^2
\end{array}
\right)
$$
is
$$
\dot{\overparen{\left(\begin{array}{@{\,  }c@{\,  }}
\hat y 
\\
\hat \xrond_\indxra
\\
\hat \xrond_\indxrb
\\
\hatxe
\end{array}\right)}}
=
\left(\begin{array}{@{\,  }c@{\,  }}
\displaystyle  \hat \xrond_\indxra 
\\
\displaystyle-\hat y \hat \xrond_\indxrb +  \hat \xrond_\indxrb\xe
\\
\displaystyle  [1-\hat \xrond_\indxrb]\hatxe
\\-\displaystyle  \hatxe^2
\end{array}
\right)
-\,  k_E(\hat y,\hat \xrond_\indxra,\hat \xrond_\indxrb,\hatxe)
\left(\begin{array}{@{\,  }cccc@{\,  }}
1 & 0 & 0 & 0
\\
0 & 1 & 0 & 0
\\
-\hat \xrond_\indxrb & \hatxe & -\hat y & \hat \xrond_\indxra
\\
-\hatxe & -\hat \xrond_\indxrb & -\hat \xrond_\indxra & -\hat y
\end{array}\right)^{-1}
\bar \Pa^{-1}
\left(\begin{array}{@{\,  }c@{\,  }}
1 \\ 0  \\ 0 \\ 0
\end{array}\right)
[\hat y -y]
$$
Since any solution of the  harmonic oscillator with unknown frequency 
augmented with $\xe=0$ is a solution of the augmented system,
it is an observer when $y$ is the actual measurement.
\egroup
\end{example}

More comments on dynamic extension are given in the \complement \ref{complement31}.

\stoparchive %

\section{Conclusions}
\subsection{Conclusions of this paper}
In \cite{57}, we have established that an observer, the 
correction term of which is based on a gradient of a ``gap'' function between 
measured output and estimated output, converges when
a strong differential detectability condition (Condition A2) holds
and when the output function is geodesic monotone (Condition A3) holds.

In \cite{127}, we have shown how, for a given system  for which all the 
variational systems are reconstructible, we can design a metric 
satisfying Condition A2.

In this paper, we have shown that Condition A3 is strongly linked to the 
nullity of the second fundamental form of the output function. 
Actually these two properties are equivalent when the dimension $p$ of the 
$\bfy$-manifold $\RR^p$ is $1$. When $p$ is larger than $1$, the latter implies 
always the former but, for the converse, we need the extra assumption that the 
orthogonal distribution is involutive. With 
this study we have been able to propose a design tool for obtaining a metric satisfying 
Condition A3.
This tool, described in Theorem~\ref{prop17}, is systematic in the sense that it does not rely on some 
equation or inequality to solve. It is a formula (see (\ref{LP166})) giving the expression of the metric in any given 
coordinate chart.

\startarchive
We have illustrated that two techniques -- immersion in an input-dependent system and dynamic extension --
could be helpful to ease the satisfaction of Conditions A2 and A3. But it remains ``to turn our dream into 
reality''. For example it is known in the high gain observer paradigm that dynamic extension solves the problem when the system is strongly 
differentially observable with an order $m$, strictly larger than $n$. Is it still the case with a weaker 
assumption ? For this topic it could be useful to exploit what is known on immersion and dynamic extension
as reported in \cite{Ticlea-Besancon,Bernard.ea.18}.
\stoparchive %

In our study, we have left open many problems. We point out
\startmodif
only two
\stopmodif
of them:
\begin{enumerate}
\item
\textit{%
Is the involutivity of the orthogonal distribution necessary for  Condition A3 to hold?}
\\
In case of a positive answer, the nullity of the second fundamental form of the output 
function would be equivalent to Condition A3 and, more interestingly, our design procedure would construct
all the metrics satisfying Condition A3.
\item
\textit{%
Under which conditions does a nonlinear counterpart of Heymann's Lemma \cite{Heymann,Hautus} holds?
}
\\
More precisely, the problem to solve is the following:\\
\textit{Assume system (\ref{eqn:Plant1}) satisfies Condition A2 and the level sets of the given output function 
$\bfh$ are totally geodesic. Do there exist  
functions $\mathfrak{f}:\RR^n\times\RR^p\times\RR^p\to \RR^n$ 
and $\mathfrak{h}:\RR^p\to \RR$  satisfying
\IfTwoCol{%
$
\mathfrak{f}(\bfx,\bfh(\bfx),\bfh(\bfx))\;=\; \bff(\bfx)
$
and such that Conditions A2 and A3 hold for the modified system
$
\dot \bfx = \mathfrak{f}(\bfx,\bfh(\bfx),y(t))$, 
$\bfy_{mod}=\mathfrak{h}(\bfh(\bfx))$
}{%
$$
\mathfrak{f}(\bfx,\bfh(\bfx),\bfh(\bfx))= \bff(\bfx)
$$
and such that Conditions A2 and A3 hold for the modified system
$$
\dot \bfx = \mathfrak{f}(\bfx,\bfh(\bfx),\bfy(t))
\quad ,\qquad
\bfy_{mod}=\mathfrak{h}(\bfh(\bfx))
$$
}%
where $t\mapsto \bfy(t)$ is considered as an 
(known) 
input function.}\\
This modified system is obtained with output injection in the function $\mathfrak{f}$ and the reduction to $1$ of the 
number of outputs via the function $\mathfrak{h}$. In case of a positive answer, the extra conditions of
involutivity of the orthogonal distribution would be unnecessary.%
\end{enumerate}

\subsection{Conclusions on our study of convergence of observers with a Riemannian Metric}
\ricardo{HERE}

\startmodif
As written in the introduction, the three papers (\cite{57}, \cite{127}, and this one)
formulate 
sufficient conditions that are as close as possible to necessary conditions for the design of
\begin{enumerate}
\item
an observer; namely, a dynamical system
with a state evolving in the same space as the (true) state of the given system,
\item
with convergence established by the decrease of a Riemannian distance between the estimated and the true state, 
\item
with the set (\ref{LP233}) of points where the estimated state is equal to the true state
being forward invariant,
\item
and with an infinite gain margin.
\end{enumerate}
A key motivation for this effort is assessing if
contraction theory is a fundamental tool for analyzing observer convergence.
We have put this into practice by studying the effect of the flow of the system-observer pair on a Riemannian distance 
between the estimate  
generated by the observer and the system state, knowing that the Euclidean case,
with therefore appropriately 
chosen coordinates, had been dealt with already (see, e.g.,  
\cite{Tsinias.90,Praly.01.NOLCOS.Observers}).

From our study, we conclude that the expected 
 condition of differential detectability (Condition A2), related to a contraction property in the 
tangent space to the level sets of the output function, is not by itself sufficient to obtain a
(at most local) convergent observer. 
It is also required for these level sets to be at least totally geodesic. Our results indicate
that this latter condition,
related to convexity, is likely a consequence of the fact that
an observer is, to some extent, searching for the global minimum of a cost function
that depends on the output error.  
To the best of our knowledge, this condition has not been proposed before,
the reason being perhaps
that, when the coordinates are such that the output function is 
linear and the metric is Euclidean; i.e., when the system is in the ``Euclidean family'' of Remark~\ref{rem1},
the said condition is automatically satisfied.
\stopmodif

Our study remained at a theoretical level and has not addressed real-world applications, mainly due to the difficulty of satisfying Conditions A2 and A3 simultaneously.
These conditions are of completely different nature and,
for the time being,
we do not know of a systematic way for having both satisfied for general systems.
In one way or the other,
other methods  assume
 the knowledge of a family of metrics satisfying the two properties, e.g., there are 
coordinates for which the output function is linear and the pair $(f,h)$ is differentially detectable with 
respect to a constant metric.
Nevertheless, we have reduced the 
problem of simultaneously satisfying Conditions A2 and A3 to
finding functions $h^\ortho$, of rank $n-p$ and such that $(h,h^\ortho)$ has rank $n$, and $\Pxi$ with
positive definite values, satisfying (\ref{LP218}) for some $q>0$. Also, fortunately, as shown in \cite{127},
Condition A2 only is already sufficient to obtain a locally convergent observer.

\startmodif
Ultimately, rather than advocating for a new observer design technique, our work presents substantiated arguments about the  
advantages and disadvantages/limitations of designing observers based on 
contracting a Riemannian distance. In these regards we appropriate the words of D.C. Lewis who,
after proposing in \cite{Lewis.71}
contraction by flows of a Finsler distance (which includes as a special case a 
Riemannian distance) to study the dependence of solutions of dynamical systems on the initial conditions, wrote
in \cite{Lewis.51}
\begin{quote}
\textit{
It thus was felt that the method [\textsf{= contraction}] as so far developed [\textsf{in \cite{Lewis.71}}] was intrinsically too 
crude to yield the desired results in applications.
}
\end{quote}
\stopmodif

\IfReport{%
\newpage
}{%
\vspace{-0.1in}
}

\bibliographystyle{unsrt} 
\bibliography{long,ObserverGeodesic,RGS,RGSweb}

\IfReport{%
\newpage
}{%
\vspace{-0.15in}
}

\section*{\textbf{Appendix}}
\addcontentsline{toc}{section}{Appendix}


\renewcommand{\thesection}{A\arabic{section}}
\renewcommand{\thesubsection}{\thesection-\arabic{subsection}}
\renewcommand{\thesubsubsection}{\thesubsection-\arabic{subsubsection}}
\renewcommand{\theHsection}{A\arabic{section}}
\renewcommand{\theHsubsection}{\thesection-\arabic{subsection}}
\renewcommand{\theHsubsubsection}{\thesubsection-\arabic{subsubsection}}

\setcounter{section}{0}
\section{Glossary}
\label{sec:Glossary}
\IfReport{%
As a complement or maybe an introduction to the following Glossary,
we recommend reading \cite[Section 1]{Figalli-Villani}.

\begin{enumerate}
\parskip 0pt plus 0pt minus 0pt%
\topsep 0pt plus 0pt minus 0pt%
\parsep 0pt plus 0pt minus 0pt%
\partopsep 0pt plus 0pt minus 0pt%
\itemsep 0pt plus 0pt minus 0pt%
\settowidth{\labelwidth}{$\bullet$}%
\setlength{\labelsep}{0.5em}%
\setlength{\leftmargin}{\labelwidth}%
\addtolength{\leftmargin}{\labelsep}%
\item\label{item:RiemannianMetric}
A Riemannian metric $\bfP$ is a symmetric $2$-covariant tensor with positive 
definite values.

The associated Christoffel symbols $\Gamma_{\indxa\indxb}^\indxc$ expressed in coordinates $x$ are
\begin{equation}
\label{eqn:ChristoffelSymbols}
\Gamma_{\indxa\indxb}^\indxc(x)\!=\! \frac{1}{2}\sum_{\indxd}
(P(x)^{-1})_{\indxc\indxd}
\left[\frac{\partial P_{\indxa\indxd}}{\partial x_\indxb}(x)
+
\frac{\partial P_{\indxb\indxd}}{\partial x_\indxa}(x)
-
\frac{\partial P_{\indxa\indxb}}{\partial x_\indxd}(x)
\right]
\  .
\end{equation}

The corresponding geodesic equation is, in its Euler-Lagrange form
\begin{equation}\label{14}
2\frac{d}{ds}\left\{\bfP(\bfgamma (s))\frac{d\bfgamma }{ds}(s)\right\}
\;=\; \left.\bfd_x\left\{\frac{d\bfgamma }{ds}(s)^\top \bfP(\bfx) \frac{d\bfgamma 
}{ds}(s)\right\}\right|_{\bfx=\bfgamma (s)}
\end{equation}
or with coordinates $x$ for $\bfx$
$$
\frac{d^2\gamma _\indxa}{ds^2}(s)\;=\; -\sum_{\indxb,\indxc} \Gamma _{\indxb\indxc}^\indxa (\gamma (s))\frac{d\gamma 
_\indxb}{ds}(s)\frac{d\gamma _\indxc}{ds}(s)
$$

\item
\label{Glos4}
The Lie derivative $\mathcal{L}_\bff \bfP$ of a symmetric $2$-covariant tensor 
$\bfP$ along the vector field $\bff$
 is a symmetric $2$-covariant tensor.
With coordinates $x$, its expression is, for all $v$ in $\RR ^n$,
\begin{eqnarray*}
\displaystyle v^\top \mathcal{\mathcal{L}}_f P(x)\,  v &  = &
v^\top 
\left.\frac{d}{dt}
\left\{\frac{\partial X}{\partial x}(x,t)^\top
P(X(x,t))
\frac{\partial X}{\partial x}(x,t) 
\right\}\right|_{t=0}
v
\\
&=&
\frac{\partial}{ \partial x}\left\{
\vrule height 1em depth 0.5em width 0pt
v^{\top} P(x) \,  v\right\}\,  f(x) \;+\; 2\,  
v^{\top} P(x)\left(\frac{\partial f}{\partial x}(x) \,  v\right)
\end{eqnarray*}

We would like the reader to distinguish the notation
$\mathcal{L}_fP$ for the Lie derivative of a symmetric a symmetric $2$-covariant tensor from
$L_fa $, which is used for the more usual Lie derivative of a function $a $.
%
\item
\label{Glos1}
Given a function $\bfh:\RR^n\to \RR^p$,
\begin{itemize}
\item
$\bfdh$ denotes its differential the expression of which, with the coordinates 
$x$, is
$\frac{\partial h_\indyi}{\partial x_\indxa}(x)$.
With $\otimes$, a tensor product, $\bfdh \otimes \bfdh$ is a symmetric $2$-covariant tensor  
the expression of which, with coordinates $x$ is
$$
\left(dh \otimes dh\right)(x)_{\indxa\indxb}\;=\; \sum_\indyi \frac{\partial h_\indyi}{\partial x_\indxa}(x) 
\frac{\partial h_\indyi}{\partial x_\indxb}(x)
\  .
$$
\item
$\bfd^2\bfh$ denotes the second differential the expression of which, with the coordinates 
$x$, is
$\frac{\partial ^2h_\indyi}{\partial x_\indxa\partial x_\indxb}(x)$.
\end{itemize}
When a function $\wpunbf$ has more than two arguments, a subscript attached to $\wpunbf$
denotes with respect to which argument the differential is computed.
For instance, 
the second differential with respect to the first argument
of the function $(\bfy_1,\bfy_2) \mapsto \wpunbf(\bfy_1,\bfy_2)$
is denoted $\bfd_1^2\wpunbf$.
\item
\label{Glos2}
Given a Riemannian metric $\bfP$ and a function $\bfh$, $\bfgrad_{\bfP} \bfh$ denotes the 
(Riemannian) gradient of $\bfh$. Its expression with coordinates $x$ is
$$
\grad_P h(x)\;=\; P(x)^{-1}\frac{\partial h}{\partial x}(x)^\top
\  .
$$
Given a path $\bfgamma $, we have
$$
\frac{d}{ds}\left\{\bfh(\bfgamma (s))\right\}\;=\; \bfgrad_\bfP \bfh(\bfgamma (s))^\top \bfP(\bfgamma (s))\,  
\frac{d\bfgamma }{ds}(s)
$$
See \cite[\S 1.3]{Udriste.94}.
\item
\label{Glos3}
Given a Riemannian metric $\bfP$ and
a function
 $\bfh$, 
$\bfHess _\bfP \bfh$ denotes the (Riemannian) Hessian of $\bfh$.
It is a 2-covariant to 1-contravariant tensor defined as
\begin{equation}
\label{LP156}
\bfHess _\bfP \bfh (\bfx) \;=\; \frac{1}{2}\,  \mathcal{L}_{\bfgrad _\bfP \bfh} \bfP(\bfx)
\  .
\end{equation}
Its expression with coordinates $x$ is
$$
(\Hess _P h (x))_{\indxa\indxb}=\frac{\partial ^2h}{\partial x_\indxa\partial x_\indxb}(x)
-
\sum_{\indxc}\Gamma _{\indxa\indxb}^\indxc(x)
\frac{\partial h}{\partial x_\indxc}(x)
\  .
$$
Given a geodesic $\bfgamma $, we have (see \cite[Exercise 3.16]{ONeill.83}
or the \complement \ref{complement41})
\begin{equation}
\label{LP78}
\frac{d^2}{ds^2}\left\{
\vrule height 0.5em depth 0.5em width 0pt \bfh(\bfgamma (s))\right\}
\;=\; \frac{d\bfgamma }{ds}(s)^\top \bfHess _\bfP \bfh (\bfgamma (s))\frac{d\bfgamma }{ds}(s)
\end{equation}
%
%
See \cite[\S 1.3]{Udriste.94}.
\item
\label{Glos7}
The Riemannian distance $d(x_1,x_2)$ is 
the minimum of $\left.\vrule height 1em depth 0.5em width 0pt
L(\gamma )\right|_{s_1}^{s_2}$ among all possible piecewise $C^1$ 
paths $\gamma $ between $x_1$ and $x_2$.
A minimizer giving the 
distance is called a
minimizing geodesic and is denoted $\gamma ^*$. 
\item
\label{Glos8}
A topological space equipped with a Riemannian distance is complete
when every geodesic
can be maximally extended to $\RR$.
%
\end{enumerate}
}{%
We give here a complement to the glossary in \cite{127}. 
Also, we recommend reading \cite[Section 1]{Figalli-Villani}.
\begin{enumerate}
\item\label{item:RiemannianMetric}
A Riemannian metric $\bfP$ is a  symmetric covariant $2$-tensor with positive 
definite values.
\item
The length of a $C^1$ path
$\bfgamma $ between points $\bfx_a$ and $\bfx_b$ 
is defined as
$$
\left.\vrule height 1em depth 0.5em width 0pt
L(\bfgamma )\right|_{s_a}^{s_{b}}\;=\; \int_{s_a}^{s_b}
\sqrt{
\frac{d\bfgamma }{ds}(s) ^{\top}\bfP(\bfgamma (s)) \frac{d\bfgamma }{ds}(s)
} \,  ds,
$$
where
$
\bfgamma (s_a)\;=\; \bfx_a$ and $\bfgamma (s_b)\;=\; \bfx_b$.
\item
\label{Glos7}
The Riemannian distance $d(\bfx_a,\bfx_b)$ is 
the minimum of $\left.\vrule height 1em depth 0.5em width 0pt
L(\bfgamma )\right|_{s_a}^{s_b}$ among all possible piecewise $C^1$ 
paths $\bfgamma $ between $\bfx_a$ and $\bfx_b$.
A minimizer giving the 
distance is called a
minimizing geodesic and is denoted $\bfgamma ^*$. 
\item
A Riemannian metric $\bfP$ is said complete
when every geodesic
can be maximally extended to $\RR$.
\item
\label{glossary2}
$\bfd^2\bfh$ denotes the second differential the expression of which, with the coordinates 
$x$, is
$\frac{\partial ^2h_\indyi}{\partial x_\indxa\partial x_\indxb}(x)$.
\item
\label{glossary3}
The second differential with respect to the first argument
of the function $(\bfy_1,\bfy_2) \mapsto \wpunbf(\bfy_1,\bfy_2)$
is denoted $\bfd_1^2\wpunbf$.
\item
\label{glossary4}
$\bfHess _\bfP \bfh$ denotes the (Riemannian) Hessian of $\bfh$.
It is a 2-covariant to 1-contravariant tensor defined as
\begin{equation}
\label{LP156}
\bfHess _\bfP \bfh (\bfx) \;=\; \frac{1}{2}\,  \mathcal{L}_{\bfgrad _\bfP \bfh} \bfP(\bfx)
\  .
\end{equation}
Its expression with coordinates $x$ is
$$
(\Hess _P h (x))_{\indxa\indxb}=\frac{\partial ^2h}{\partial x_\indxa\partial x_\indxb}(x)
-
\sum_{\indxc}\Gamma _{\indxa\indxb}^\indxc(x)
\frac{\partial h}{\partial x_\indxc}(x)
\  .
$$
Given a geodesic $\bfgamma $, we have (see \cite[Exercise 3.16]{ONeill.83})
\begin{equation}
\label{LP78}
\frac{d^2}{ds^2}\left\{
\vrule height 0.5em depth 0.5em width 0pt \bfh(\bfgamma (s))\right\}
\;=\; \frac{d\bfgamma }{ds}(s)^\top \bfHess _\bfP \bfh (\bfgamma (s))\frac{d\bfgamma }{ds}(s)
\end{equation}
\end{enumerate}
}

\section{Proof of Lemma \ref{lem1}}
\label{complement36}%
\IfReport{%
Let $\bar \Gamma $ and $\bar \Gammay$ be the expressions of 
the Christoffel symbols in the new coordinates.
From 
\cite[(3.5.22)]{Lovelock-Rund}, we get
\begin{equation}\label{eqn:ChristoffelNewCoord}
\sum_{\indxd,\indxe}\frac{\partial \changex _\indxd}{\partial x_\indxa}
\bar \Gamma _{\indxd\indxe}^\indxc
\frac{\partial \changex_\indxe}{\partial x_\indxb}
\;=\; 
\sum_{\indxd}\frac{\partial \changex_\indxc}{\partial x_\indxd}
\Gamma _{\indxa\indxb}^\indxd
-
\frac{\partial ^2 \changex_\indxc}{\partial x_\indxa\partial x_\indxb}
\end{equation}
$$
\sum_{\indyl,\indym}\frac{\partial \changey _{\indyl}}{\partial y_\indyi}
\bar \Gammay _{\indyl\indym}^\indyk
\frac{\partial \changey _{\indym}}{\partial y_\indyj}
\;=\; 
\sum_{\indyl}\frac{\partial \changey _{\indyk}}{\partial y_\indyl}
\Gammay _{\indyi\indyj}^\indyl
-
\frac{\partial ^2 \changey _{\indyk}}{\partial y_\indyi\partial y_\indyj}
$$
We  also have
$$
\bar h(\changex(x))\;=\; \changey(h(x))
\  ,
$$
$$
\sum_\indxc\frac{\partial \bar h}{\partial \bar x_\indxc}(\changex(x))
\frac{\partial \changex_\indxc}{\partial x_\indxa}(x)\;=\;
\sum_\indym\frac{\partial \changey }{\partial y_\indym}(h(x))\frac{\partial h_\indym}{\partial x_\indxa}(x)
\  ,
$$
$$
\frac{\partial \bar h}{\partial \bar x_\indxc}(\changex(x))
\;=\;
\sum_{\indxe,\indym}
\frac{\partial \changey }{\partial y_\indym}(h(x))\frac{\partial h_\indym}{\partial x_\indxe}(x)
\left[\frac{\partial \changex}{\partial x}(x)^{-1}\right]_{\indxe\indxc}
\  ,
$$
and\\[1em]
$\displaystyle 
\sum_{\indxc,\indxd}\frac{\partial ^2 \bar h}{\partial \bar x_\indxc\partial \bar x_\indxd}(\changex(x))
\frac{\partial \changex_\indxc}{\partial x_\indxa}(x)
\frac{\partial \changex_\indxd}{\partial x_\indxb}(x)
+\sum_\indxc
\frac{\partial \bar h}{\partial \bar x_\indxc}(\changex(x))
\frac{\partial ^2\changex_\indxc}{\partial x_\indxa\partial x_\indxb}(x)
$\hfill \null \\\null \hfill $\displaystyle 
\;=\; 
\sum_{\indyi,\indyj}\frac{\partial ^2\changey }{\partial y_\indyi\partial y_\indyj}(h(x))
\frac{\partial h_\indyi}{\partial x_\indxa}(x)
\frac{\partial h_\indyj}{\partial x_\indxb}(x)
+
\sum_\indyi\frac{\partial \changey }{\partial y_\indyi}(h(x))
\frac{\partial ^2h_\indyi}{\partial x_\indxa\partial x_\indxb}(x)
\  .
$\\[1em]
Using these expressions, we obtain
$$
\secff _P \overline{h}_{\indxc\indxd}^\indyk(\bar x) \;=\;  \frac{\partial ^2\bar h_\indyk}{\partial \bar x_\indxc\partial \bar x_\indxd}(\bar x)
-
\sum_\indxe\bar \Gamma _{\indxc\indxd}^\indxe(x)\frac{\partial \bar h_\indyk}{\partial \bar x_\indxe}(\bar x)  
+
\sum_{\indyl,\indym}\bar \Gammay_{\indyl\indym}^\indyk(\bar h(\bar x))
\frac{\partial \bar h_\indyl}{\partial \bar x_\indxc}(\bar x)
\frac{\partial \bar h_\indym}{\partial \bar x_\indxd}(\bar x)
\  ,
$$
Substituting it in \eqref{eqn:ChristoffelNewCoord},
we obtain
\\[1em]$\displaystyle 
\sum_{\indxc,\indxd}
\frac{\partial \changex_\indxc}{\partial x_\indxa}(x)
\frac{\partial \changex_\indxd}{\partial x_\indxb}(x)
\secff _P \overline{h}_{\indxc\indxd}^\indyk(\bar x)
$\hfill \null \\\null \qquad $\displaystyle 
\;=\; 
\sum_{\indxc,\indxd}
\frac{\partial \changex_\indxc}{\partial x_\indxa}(x)
\frac{\partial \changex_\indxd}{\partial x_\indxb}(x)
\frac{\partial ^2\bar h_\indyk}{\partial \bar x_\indxc\partial \bar x_\indxd}(\bar x)
$\hfill \null \\\null \hfill $\displaystyle
-
\sum_{\indxc,\indxd}
\frac{\partial \changex_\indxc}{\partial x_\indxa}(x)
\frac{\partial \changex_\indxd}{\partial x_\indxb}(x)
\sum_\indxe\bar \Gamma _{\indxc\indxd}^\indxe(x)\frac{\partial \bar h_\indyk}{\partial \bar x_\indxe}(\bar x)
$\hfill \null \\\null \hfill $\displaystyle 
+ \sum_{\indxc,\indxd}
\frac{\partial \changex_\indxc}{\partial x_\indxa}(x)
\frac{\partial \changex_\indxd}{\partial x_\indxb}(x)
\sum_{\indyl,\indym}\bar \Gammay_{\indyl\indym}^\indyk(\bar h(\bar x))
\frac{\partial \bar h_\indyl}{\partial \bar x_\indxc}(\bar x)
\frac{\partial \bar h_\indym}{\partial \bar x_\indxd}(\bar x)
$\\\null \qquad $\displaystyle 
\;=\; 
\sum_{\indyi,\indyj}\frac{\partial ^2\changey _\indyk}{\partial y_\indyi\partial y_\indyj}(h(x))
\frac{\partial h_\indyi}{\partial x_\indxa}(x)
\frac{\partial h_\indyj}{\partial x_\indxb}(x)
+
\sum_\indyi\frac{\partial \changey _\indyk}{\partial y_\indyi}(h(x))
\frac{\partial ^2h_\indyi}{\partial x_\indxa\partial x_\indxb}(x)
$\hfill \null \\\null \qquad \qquad \qquad    $\displaystyle
-
\sum_{\indxc,\indxe,\indyi}
\frac{\partial \changey _\indyk}{\partial y_\indyi}(h(x))\frac{\partial h_\indyi}{\partial x_\indxe}(x)
\left[\frac{\partial \changex}{\partial x}(x)^{-1}\right]_{\indxe\indxc}
\frac{\partial ^2\changex_\indxc}{\partial x_\indxa\partial x_\indxb}(x)
$\\
\null \qquad \qquad \qquad \qquad \qquad 
$\displaystyle
-
\sum_\indxc
\left[
\sum_{\indxd}\frac{\partial \changex_\indxc}{\partial x_\indxd}
\Gamma _{\indxa\indxb}^\indxd
-
\frac{\partial ^2 \changex_\indxc}{\partial x_\indxa\partial x_\indxb}
\right]
\sum_{\indxe,\indyi}
\frac{\partial \changey _\indyk}{\partial y_\indyi}(h(x))\frac{\partial h_\indyi}{\partial x_\indxe}(x)
\left[\frac{\partial \changex}{\partial x}(x)^{-1}\right]_{\indxe\indxc}
$\\\hfill \null \\\null \qquad \qquad \qquad \qquad \qquad \qquad \qquad  $\displaystyle 
+
\sum_{\indyl,\indym}\bar \Gammay_{\indyl\indym}^\indyk(\bar h(\bar x))
\left[
\sum_\indyi\frac{\partial \changey _\indyl}{\partial y_\indyi}(h(x))\frac{\partial h_\indyi}{\partial x_\indxa}(x)
\right]\left[
\sum_\indyj\frac{\partial \changey _\indym}{\partial y_\indyj}(h(x))\frac{\partial h_\indyj}{\partial x_\indxb}(x)
\right]
$\\\null \qquad $\displaystyle 
\;=\; 
\sum_{\indyi,\indyj}\frac{\partial ^2\changey _\indyk}{\partial y_\indyi\partial y_\indyj}(h(x))
\frac{\partial h_\indyi}{\partial x_\indxa}(x)
\frac{\partial h_\indyj}{\partial x_\indxb}(x)
$\\\null \hfill $\displaystyle
+
\sum_\indyi
\frac{\partial \changey _\indyk}{\partial y_\indyi}(h(x))
\left[
\frac{\partial ^2h_\indyi}{\partial x_\indxa\partial x_\indxb}(x)
-
\sum_{\indxd}
\Gamma _{\indxa\indxb}^\indxd
\frac{\partial h_\indyi}{\partial x_\indxd}(x)
\right]
$\hfill \null \\\null \hfill $\displaystyle 
+\sum_{\indyi,\indyj}
\frac{\partial h_\indyi}{\partial x_\indxa}(x)
\frac{\partial h_\indyj}{\partial x_\indxb}(x)
\left[
\sum_{\indyl}\frac{\partial \changey _{\indyk}}{\partial y_\indyl}
\Gammay _{\indyi\indyj}^\indyl
-
\frac{\partial ^2 \changey _{\indyk}}{\partial y_\indyi\partial y_\indyj}
\right]
$\\[0.7em]%
which reduces to
$$
\sum_\indyi
\frac{\partial \changey _\indyk}{\partial y_\indyi}(h(x))
\left[
\frac{\partial ^2h_\indyi}{\partial x_\indxa\partial x_\indxb}(x)
-
\sum_{\indxd}
\Gamma _{\indxa\indxb}^\indxd
\frac{\partial h_\indyi}{\partial x_\indxd}(x)
+
\sum_{\indyl,\indym}
\frac{\partial h_\indyl}{\partial x_\indxa}(x)
\frac{\partial h_\indym}{\partial x_\indxb}(x)
\Gammay _{\indyl\indym}^\indyi
\right]
$$
and establishes (\ref{LP146}).
%
%
\newpage
}{%
\vskip -0.3em
Let $\bar \Gamma $ and $\bar \Gammay$ be the expressions of 
the Christoffel symbols in the new coordinates.
From 
\cite[(3.5.22)]{Lovelock-Rund}, we get
\begin{equation}\label{eqn:ChristoffelNewCoord}
\sum_{\indxd,\indxe}\frac{\partial \changex _\indxd}{\partial x_\indxa}
\bar \Gamma _{\indxd\indxe}^\indxc
\frac{\partial \changex_\indxe}{\partial x_\indxb}
\;=\; 
\sum_{\indxd}\frac{\partial \changex_\indxc}{\partial x_\indxd}
\Gamma _{\indxa\indxb}^\indxd
-
\frac{\partial ^2 \changex_\indxc}{\partial x_\indxa\partial x_\indxb}
\end{equation}
$$
\sum_{\indyl,\indym}\frac{\partial \changey _{\indyl}}{\partial y_\indyi}
\bar \Gammay _{\indyl\indym}^\indyk
\frac{\partial \changey _{\indym}}{\partial y_\indyj}
\;=\; 
\sum_{\indyl}\frac{\partial \changey _{\indyk}}{\partial y_\indyl}
\Gammay _{\indyi\indyj}^\indyl
-
\frac{\partial ^2 \changey _{\indyk}}{\partial y_\indyi\partial y_\indyj}
$$
We  also have
$
\bar h(\changex(x))\;=\; \changey(h(x))$,
$$
\sum_\indxc\frac{\partial \bar h}{\partial \bar x_\indxc}(\changex(x))
\frac{\partial \changex_\indxc}{\partial x_\indxa}(x)\;=\;
\sum_\indym\frac{\partial \changey }{\partial y_\indym}(h(x))\frac{\partial h_\indym}{\partial x_\indxa}(x)
\  ,
$$
$$
\frac{\partial \bar h}{\partial \bar x_\indxc}(\changex(x))
\;=\;
\sum_{\indxe,\indym}
\frac{\partial \changey }{\partial y_\indym}(h(x))\frac{\partial h_\indym}{\partial x_\indxe}(x)
\left[\frac{\partial \changex}{\partial x}(x)^{-1}\right]_{\indxe\indxc}
\  ,
$$
and\\[0.7em]
$\displaystyle 
\sum_{\indxc,\indxd}\frac{\partial ^2 \bar h}{\partial \bar x_\indxc\partial \bar x_\indxd}(\changex(x))
\frac{\partial \changex_\indxc}{\partial x_\indxa}(x)
\frac{\partial \changex_\indxd}{\partial x_\indxb}(x)$\\\null \hfill
$\displaystyle+\sum_\indxc
\frac{\partial \bar h}{\partial \bar x_\indxc}(\changex(x))
\frac{\partial ^2\changex_\indxc}{\partial x_\indxa\partial x_\indxb}(x)
$\hfill \null \\\null \hfill $\displaystyle 
\;=\; 
\sum_{\indyi,\indyj}\frac{\partial ^2\changey }{\partial y_\indyi\partial y_\indyj}(h(x))
\frac{\partial h_\indyi}{\partial x_\indxa}(x)
\frac{\partial h_\indyj}{\partial x_\indxb}(x)$\\\null \hfill
$\displaystyle
+
\sum_\indyi\frac{\partial \changey }{\partial y_\indyi}(h(x))
\frac{\partial ^2h_\indyi}{\partial x_\indxa\partial x_\indxb}(x)
\  .
$\\[0.7em]
Using these expressions, we obtain
\begin{eqnarray*}
\secff _P \overline{h}_{\indxc\indxd}^\indyk(\bar x) & = & \frac{\partial ^2\bar h_\indyk}{\partial \bar x_\indxc\partial \bar x_\indxd}(\bar x)
-
\sum_\indxe\bar \Gamma _{\indxc\indxd}^\indxe(x)\frac{\partial \bar h_\indyk}{\partial \bar x_\indxe}(\bar x)  \\
& & 
+
\sum_{\indyl,\indym}\bar \Gammay_{\indyl\indym}^\indyk(\bar h(\bar x))
\frac{\partial \bar h_\indyl}{\partial \bar x_\indxc}(\bar x)
\frac{\partial \bar h_\indym}{\partial \bar x_\indxd}(\bar x)
\  ,
\end{eqnarray*}
Substituting it in \eqref{eqn:ChristoffelNewCoord},
we obtain
\\[0.7em]$\displaystyle 
\sum_{\indxc,\indxd}
\frac{\partial \changex_\indxc}{\partial x_\indxa}(x)
\frac{\partial \changex_\indxd}{\partial x_\indxb}(x)
\secff _P \overline{h}_{\indxc\indxd}^\indyk(\bar x)
$\hfill \null \\\null \qquad $\displaystyle 
\;=\; 
\sum_{\indxc,\indxd}
\frac{\partial \changex_\indxc}{\partial x_\indxa}(x)
\frac{\partial \changex_\indxd}{\partial x_\indxb}(x)
\frac{\partial ^2\bar h_\indyk}{\partial \bar x_\indxc\partial \bar x_\indxd}(\bar x)
$\hfill \null \\\null \hfill $\displaystyle
-
\sum_{\indxc,\indxd}
\frac{\partial \changex_\indxc}{\partial x_\indxa}(x)
\frac{\partial \changex_\indxd}{\partial x_\indxb}(x)
\sum_\indxe\bar \Gamma _{\indxc\indxd}^\indxe(x)\frac{\partial \bar h_\indyk}{\partial \bar x_\indxe}(\bar x)
$\hfill \null \\\null \hfill $\displaystyle 
+ \sum_{\indxc,\indxd}
\frac{\partial \changex_\indxc}{\partial x_\indxa}(x)
\frac{\partial \changex_\indxd}{\partial x_\indxb}(x)
\sum_{\indyl,\indym}\bar \Gammay_{\indyl\indym}^\indyk(\bar h(\bar x))
\frac{\partial \bar h_\indyl}{\partial \bar x_\indxc}(\bar x)
\frac{\partial \bar h_\indym}{\partial \bar x_\indxd}(\bar x)
$\\\null \qquad $\displaystyle 
\;=\; 
\sum_{\indyi,\indyj}\frac{\partial ^2\changey _\indyk}{\partial y_\indyi\partial y_\indyj}(h(x))
\frac{\partial h_\indyi}{\partial x_\indxa}(x)
\frac{\partial h_\indyj}{\partial x_\indxb}(x)
$\\\null \qquad $\displaystyle 
+
\sum_\indyi\frac{\partial \changey _\indyk}{\partial y_\indyi}(h(x))
\frac{\partial ^2h_\indyi}{\partial x_\indxa\partial x_\indxb}(x)
$\hfill \null \\\null \qquad \qquad    $\displaystyle
-
\sum_{\indxc,\indxe,\indyi}
\frac{\partial \changey _\indyk}{\partial y_\indyi}(h(x))\frac{\partial h_\indyi}{\partial x_\indxe}(x)
\left[\frac{\partial \changex}{\partial x}(x)^{-1}\right]_{\indxe\indxc}
\frac{\partial ^2\changex_\indxc}{\partial x_\indxa\partial x_\indxb}(x)
$\\
$\displaystyle
-
\sum_\indxc
\left[
\sum_{\indxd}\frac{\partial \changex_\indxc}{\partial x_\indxd}
\Gamma _{\indxa\indxb}^\indxd
-
\frac{\partial ^2 \changex_\indxc}{\partial x_\indxa\partial x_\indxb}
\right]
\sum_{\indxe,\indyi}
\frac{\partial \changey _\indyk}{\partial y_\indyi}(h(x))\frac{\partial h_\indyi}{\partial x_\indxe}(x) $\\\null \hfill $\displaystyle \times
\left[\frac{\partial \changex}{\partial x}(x)^{-1}\right]_{\indxe\indxc}
$\\\hfill \null \\\null \hfill $\displaystyle 
+
\sum_{\indyl,\indym}\bar \Gammay_{\indyl\indym}^\indyk(\bar h(\bar x))
\left[
\sum_\indyi\frac{\partial \changey _\indyl}{\partial y_\indyi}(h(x))\frac{\partial h_\indyi}{\partial x_\indxa}(x)
\right]$\\\hfill \null \\\null \hfill $\displaystyle
\times\left[
\sum_\indyj\frac{\partial \changey _\indym}{\partial y_\indyj}(h(x))\frac{\partial h_\indyj}{\partial x_\indxb}(x)
\right]
$\\\null \qquad $\displaystyle 
\;=\; 
\sum_{\indyi,\indyj}\frac{\partial ^2\changey _\indyk}{\partial y_\indyi\partial y_\indyj}(h(x))
\frac{\partial h_\indyi}{\partial x_\indxa}(x)
\frac{\partial h_\indyj}{\partial x_\indxb}(x)
$\\\null \hfill $\displaystyle
+
\sum_\indyi
\frac{\partial \changey _\indyk}{\partial y_\indyi}(h(x))
\left[
\frac{\partial ^2h_\indyi}{\partial x_\indxa\partial x_\indxb}(x)
-
\sum_{\indxd}
\Gamma _{\indxa\indxb}^\indxd
\frac{\partial h_\indyi}{\partial x_\indxd}(x)
\right]
$\hfill \null \\\null \hfill $\displaystyle 
+\sum_{\indyi,\indyj}
\frac{\partial h_\indyi}{\partial x_\indxa}(x)
\frac{\partial h_\indyj}{\partial x_\indxb}(x)
\left[
\sum_{\indyl}\frac{\partial \changey _{\indyk}}{\partial y_\indyl}
\Gammay _{\indyi\indyj}^\indyl
-
\frac{\partial ^2 \changey _{\indyk}}{\partial y_\indyi\partial y_\indyj}
\right]
$\\[0.7em]%
which reduces to\\[0.7em]%
$\displaystyle 
\sum_\indyi
\frac{\partial \changey _\indyk}{\partial y_\indyi}(h(x))$\hfill \null \\\null \hfill $\displaystyle 
\times
\left[
\frac{\partial ^2h_\indyi}{\partial x_\indxa\partial x_\indxb}(x)
-
\sum_{\indxd}
\Gamma _{\indxa\indxb}^\indxd
\frac{\partial h_\indyi}{\partial x_\indxd}(x)
+
\sum_{\indyl,\indym}
\frac{\partial h_\indyl}{\partial x_\indxa}(x)
\frac{\partial h_\indym}{\partial x_\indxb}(x)
\Gammay _{\indyl\indym}^\indyi
\right]
$\\[0.7em]
and establishes (\ref{LP146}).
%
%
\vspace{-0.2in}
}

\startarchive
\section{Proof of Lemma \ref{lem13}}
\label{sec19}
The proof of Lemma~\ref{lem13} given here is a direct extension, from
the case $p=1$ to the general case, of \cite[Proposition 2.2 and Theorem 8.1]{Fischer.96}.

\startmodif
Fix $\bfx_0$ arbitrarily in $\bfRR^n$.
From the Submersion Level Set Theorem
(see \cite[p.105]{Lee.13}), the $\bfh(\bfx_0)$-level set $\mathfrak{H}(\bfh(\bfx_0))$
is a properly embedded $n-p$ dimensional $C^s$ submanifold. In the following, it is useful to view
$\mathfrak{H}(\bfh(\bfx_0))$ as an ``abstract'' $n-p$ dimensional $C^s$ manifold. To make this very clear we 
denote by
$\bfXrond_h$ this abstract manifold and we denote
$$
\bfvarphi:
\begin{array}[t]{rcl}
\bfXrond_h  & \to & \bfRR^p
\\
\bfxrond & \mapsto & \bfvarphi(\bfxrond)
\end{array}
$$
the corresponding smooth embedding the image of which is $\mathfrak{H}(\bfh(\bfx_0))$.
\stopmodif

Since the $\bfx$-manifold is $\bfRR^n$ and the $\bfy$-manifold is $\bfRR^p$, we 
have a globally defined pair of coordinate charts $(x,\RR^n,\phi)$ and $(y,\RR^p,\chi)$ which, without loss of 
generality, satisfy
\begin{equation}
\label{LP230}
h(x_0)\;=\; 0
\  .
\end{equation}
Let also $w{}_\indyi$, with $\indyi=1,\ldots,p$, be a basis of normalized vectors, for the Euclidean norm of $\RR ^p$ and associated with the coordinates $y$, 
i.e.
\IfTwoCol{%
$y_\indyi=w{}_\indyi^\top y$.
}{%
$$
y_\indyi=w{}_\indyi^\top y
\  .
$$
}%

\underline{Construction of the bijection $\bfhhperpinv$:}\\
\startmodif
We follow the procedure described at \cite[Bottom of page 235]{Lee.13}. For this, we let, for each integer $\indyi$, $g_\indyi$ be the vector field defined by
$$
g_\indyi(x)= P(x)^{-1}\frac{\partial h}{\partial x}(x)^\top \Py(h(x)) w{}_\indyi
= 
\grad_P h (x) \Py(h(x))w{}_\indyi\;  .
$$
They are known to
span the orthogonal distribution $\Distrib _P ^\ortho  (x)$ which, from Lemma \ref{lem3},
 is involutive. But for this procedure we need commuting vector fields. It is known
that  every involutive distribution is locally spanned by independent smooth commuting vector fields
\cite[Top of page 497]{Lee.13}. But 
this is a local result only and we want a global one. So we assume
\begin{itemize}
\item[]
\textit{The vector fields $g_\indyi$ commute.}
\end{itemize}
Assuming also
\begin{itemize}
\item[]
\textit{The function $\bfy\mapsto \bfPy(\bfy)$ is bounded}
\end{itemize}
we get that these vector fields are bounded since they satisfy
\begin{eqnarray}
\nonumber
g_\indyi(x)^\top P(x) g_\indyj(x)&=&
w{}_\indyi^\top\Py(h(x))\frac{\partial h}{\partial x}(x) P(x)^{-1}
P(x)
P(x)^{-1}\frac{\partial h}{\partial x}(x) \Py(h(x))
w{}_\indyj^\top
\\\nonumber
&=&
w{}_\indyi^\top
\Py(h(x))\frac{\partial h}{\partial x}(x)P(x)^{-1} \frac{\partial h}{\partial x}(x)^\top\Py(h(x))
w{}_\indyj
\\\label{LP210}
&=&
w{}_\indyi^\top\Py (h(x)) w{}_\indyj
\  ,
\end{eqnarray}
where, $\bfh$ being a Riemannian submersion, we have used
$$
\frac{\partial h}{\partial x}(x)P(x)^{-1} \frac{\partial h}{\partial x}(x)^\top
\;=\; \Py (h(x))^{-1}
\  .
$$
Also, the metric $\bfP$ being assumed complete, these vector fields are complete\footnote{%
$d(x(t_a),x(t_b))\leq \int_{t_a}^{t_b}
\sqrt{\dot x(t) P(x(t)) \dot x(t)} dt$.
Since $P$ is complete, $d(x(t_a),x(t_b))$ finite implies no explosion.
}.

We denote $X_\indyi (x,{}t)$ their corresponding flows which commute
\stopmodif
(see \cite[Theorem 9.44]{Lee.13}) and we have
\begin{equation}
\label{LP211}
\frac{\partial X_\indyj}{\partial x}(x,t_\indyj)g_\indyi(x)\;=\; g_\indyi(X_\indyj(x,t_\indyj))
\  .
\end{equation}

With the notation
\IfTwoCol{%
$\cursive{t}=(t_1,\ldots,t_p)$,
}{%
$$
\cursive{t}=(t_1,\ldots,t_p)
\  ,
$$
}
we write the composition of the flows as
\IfTwoCol{%
$$
X_g (x,\cursive{t})\!=\! X_p\! \left(\!  \vrule height 1.2em depth 1.2em width 0pt
X_{p-1}\! \left(\vrule height 0.8em depth 0.8em width 0pt
\! \ldots\! \left(\! \vrule height 0.6em depth 0.6em width 0pt
\! X_{2}
\! \left(\! \vrule height 0.4em depth 0.4em width 0pt
 X_1(x,t_1),t_2) \right)\! \ldots\! \right)\! ,t_{p-1}\! \right),t_p\! \right)
\  .
$$
}{%
$$
X_g (x,\cursive{t})\;=\; X_p\left( \vrule height 1.2em depth 1.2em width 0pt
X_{p-1}\left(\vrule height 0.8em depth 0.8em width 0pt
\ldots\left(\vrule height 0.6em depth 0.6em width 0pt
X_{2}
\left(\vrule height 0.4em depth 0.4em width 0pt
X_1(x,t_1),t_2)\right)\ldots\right),t_{p-1}\right),t_p\right)
\  .
$$
}%
(\ref{LP211}) implies
$X_g (x,\cursive{t})$ does not depend on any
permutation $\sigma $ of the integers $1$ to $p$, i.e. we have
\IfTwoCol{%
\\[0em]
$X_g (x,\cursive{t})\;=\; X_{\sigma (p)}\left( \vrule height 1.2em depth 1.2em width 0pt
X_{\sigma (p-1)}\left(\vrule height 0.8em depth 0.8em width 0pt
\ldots\right.\right.$\\[-0.3em]
$\null\hfill\left.\left.\left(\vrule height 0.6em depth 0.6em width 0pt
X_{\sigma (2)}
\left(\vrule height 0.4em depth 0.4em width 0pt
X_{\sigma (1)}(x,t_{\sigma (1)}),t_{\sigma (2)})\right)\ldots
\vrule height 0.6em depth 0.6em width 0pt
\right),t_{\sigma (p-1)}
\vrule height 0.8em depth 0.8em width 0pt
\right),t_{\sigma (p)}
\vrule height 1.2em depth 1.2em width 0pt
\right)
$\\[0em]
}{%
$$
X_g (x,\cursive{t})\;=\; X_{\sigma (p)}\left( \vrule height 1.2em depth 1.2em width 0pt
X_{\sigma (p-1)}\left(\vrule height 0.8em depth 0.8em width 0pt
\ldots\left(\vrule height 0.6em depth 0.6em width 0pt
X_{\sigma (2)}
\left(\vrule height 0.4em depth 0.4em width 0pt
X_{\sigma (1)}(x,t_{\sigma (1)}),t_{\sigma (2)})\right)\ldots\right),t_{\sigma (p-1)}\right),t_{\sigma (p)}\right)
$$
}%
and therefore
\IfTwoCol{%
\begin{equation}
\label{LP225}
\left.
\renewcommand{\arraystretch}{1.5}
\begin{array}{@{}c@{}}
X_g (X_g(x,\cursive{t}),-\cursive{t})\;=\; x
\quad \forall (x,\cursive{t})
\ ,
\\
X_g(x_a,\cursive{t})\;=\; X_g(x_b,\cursive{t})
\  \Rightarrow\  x_a=x_b
\  .
\end{array}
\qquad \right\}
\end{equation}
}{%
\begin{equation}
\label{LP225}
X_g (X_g(x,\cursive{t}),-\cursive{t})\;=\; x
\qquad \forall (x,\cursive{t})
\quad ,\qquad 
X_g(x_a,\cursive{t})\;=\; X_g(x_b,\cursive{t})
\  \Rightarrow\  x_a=x_b
\  .
\end{equation}
}%
\indent
On the other hand, (\ref{LP160}) gives us
\IfTwoCol{%
$$
\frac{\partial }{\partial {}t}\left\{h_\indyj(X_\indyi(x,{}t))\right\}
\;=\; 
\frac{\partial h_\indyj}{\partial x}(X_\indyi(x,{}t))
\frac{\partial X_\indyi}{\partial t}(x,{}t)
\;=\; \delta _{\indyj\indyi}
\  .
$$
}{%
\begin{eqnarray*}
\frac{\partial }{\partial {}t}\left\{h_\indyj(X_\indyi(x,{}t))\right\}
&=&
\frac{\partial h_\indyj}{\partial x}(X_\indyi(x,{}t))
\frac{\partial X_\indyi}{\partial t}(x,{}t)
\\
&=&
w{}_\indyj^\top \frac{\partial h}{\partial x}(X_\indyi(x,{}t)) P(X_\indyi(x,{}t))^{-1}\frac{\partial h}{\partial 
 x}(X_\indyi(x,{}t))^\top
\Py(h(X_\indyi(x,{}t)))w{}_\indyi
\\
&=&\delta _{\indyj\indyi}
\  .
\end{eqnarray*}
}%
This implies
$$
h_\indyj(X_\indyi(x,{}t))\;=\; h_\indyj(x)\;+\; \delta _{\indyj\indyi} t
\  .
$$
Hence, given
\IfTwoCol{%
$y\;=\; (y_1,\ldots,y_p)$,
}{%
$$
y\;=\; (y_1,\ldots,y_p)
\  ,
$$
}%
by letting
\IfTwoCol{%
\begin{eqnarray*}
&\displaystyle 
x_1=X_1(x,y_1)
\  ,\quad 
x_\indyi\;=\; X_\indyi(x_{\indyi-1},y_\indyi)
\quad \indyi\in \{2,\ldots,p\}
\  ,
\\
&\displaystyle 
x_p\;=\; X_p(x_{p-1},y_p)\;=\; X_g(x,y)
\  ,
\end{eqnarray*}
}{%
\begin{eqnarray*}
x_1&=&X_1(x,y_1)
\  ,
\\
x_\indyi&=& X_\indyi(x_{\indyi-1},y_\indyi)
\qquad \indyi\in \{2,\ldots,p\}
\  ,
\\ 
x_p&=& X_p(x_{p-1},y_p)\;=\; X_g(x,y)
\  ,
\end{eqnarray*}
}%
we obtain
$$
\begin{array}{@{}r@{\;}c@{\;}lcl@{}}
h_\indyj(x_1)&=&
 h_1(x)+y_1
&
\mbox{if}
&
\indyj=1
\  ,\\[0.3em]&=&
h_\indyj(x)
&
\mbox{if}
&
1<\indyj\leq p
\  ,
\\[0.7em]
h_\indyj(x_\indyi)&=& 
h_\indyj(x)+y_\indyj(x)
&
\mbox{if}
&
1\leq \indyj\leq \indyi
\  ,\\[0.3em]&=&
h_\indyj(x)
&
\mbox{if}
&
\indyi < \indyj\leq p
\  .\end{array}
$$
\IfReport{%
This yields
\begin{equation}
\label{LP226}
h(X_g(x,y))=h(x_p)=h(x)+y
\qquad \forall (x,y)\in \RR^n\times\RR^p
\end{equation}
and, in particular,
$$
h(X_g(x,-h(x)))=0
\qquad \forall x\in \RR^n
\  .
$$%
}{%
This yields, for all $(x,y)$ in $\RR^n\times\RR^p$,
\begin{equation}
\label{LP226}
h(X_g(x,y))=h(x_p)=h(x)+y\  . 
\end{equation}
}%
%
To go on, we introduce the notation
$$
\bfX_g(\bfx,\bfy)\;=\; 
\coordxm^{-1}(X_g(\coordxm(\bfx),\coordym(\bfy)))
\  .
$$
\IfTwoCol{%
The identity (\ref{LP226})
}{%
The identity (\ref{LP226}), expressed with coordinates, %
}%
implies that $\bfX_g(\bfx,-\bfh(\bfx))$ is in the 
\startmodif
$\bfh(\bfx_0)$-level set $\mathfrak{H}(\bfh(\bfx_0)) $, the image by $\bfvarphi$ of $\bfXrond_h$.
\stopmodif
So we have established that, for any $\bfx$ in $\bfRR^n$, there 
exists $\bfxrond$ in $\bfXrond_h $ which, with (\ref{LP230}), satisfies
\begin{equation}
\label{LP231}
\bfvarphi(\bfxrond)\;=\; \bfX_g(\bfx,-\bfh(\bfx))
\quad ,\qquad 
h(\coordxm(\bfvarphi(\bfxrond)))\;=\; 0
\  .
\end{equation}
Since $\bfvarphi$ is a bijection onto its image, we have defined a $C^s$ function $\bfh^\ortho:\bfRR^n\to 
\bfXrond_h$ satisfying
$$
\bfvarphi(\bfh^\ortho(\bfx))\;=\; \bfX_g(\bfx,-\bfh(\bfx))
\  ,
$$
and the function $\bfhhperp: \bfRR^n \to \bfRR^p\times \bfXrond_h$ defined as
\begin{equation}
\label{LP227}
\bfhhperp(\bfx)\;=\; (\bfh(\bfx),\bfh^\ortho(\bfx))
\end{equation}
is $C^s$. Because of (\ref{LP225}), it is injective. It is also surjective since, for any $(\bfy,\bfxrond)$ 
in $\bfRR^p\times\bfXrond_h$,
by letting
$
\bfx\;=\;  \bfX_g(\bfvarphi(\bfxrond),\bfy)
\  ,
$
we obtain, with (\ref{LP226}) and (\ref{LP231}),
$$
\bfvarphi(\bfxrond)= \bfX_g(\bfx,-\bfh(\bfx))
\  , 
\renewcommand{\arraystretch}{1.5}
\IfReport{
\begin{array}[t]{r@{\,  }c@{\,  }l@{}}
\bfy&=&\coordym^{-1}(h(\coordxm(\bfvarphi(\bfxrond)))+\coordym(\bfy) )
\\
&=& \coordym^{-1}(h(X_g(\coordxm(\bfvarphi(\bfxrond)),\coordym(\bfy))))
\\
&=&\bfh(\bfX_g(\bfvarphi(\bfxrond),\bfy))
\,=\,
\bfh(\bfx)
\  .
\end{array}
}{
\begin{array}[t]{r@{\,  }c@{\,  }l@{}}
\bfy&=&\coordym^{-1}(h(\coordxm(\bfvarphi(\bfxrond)))+\coordym(\bfy) )
\\
&=&\bfh(\bfX_g(\bfvarphi(\bfxrond),\bfy))
\,=\,
\bfh(\bfx)
\  .
\end{array}
}
$$
Hence $\bfhhperp$ is a $C^s$ bijection and so is its inverse $\bfhhperpinv:\bfRR^p\times\bfXrond_h\to \bfRR^n$
which satisfies
\begin{equation}
\label{LP214}
\left.
\renewcommand{\arraystretch}{1.5}
\begin{array}{@{}c@{}}
\bfhhperpinv(\bfy,\bfxrond) \;=\;  \bfX_g(\bfvarphi(\bfxrond),\bfy)
\  ,\\ 
\bfy \;=\;  \bfh(\bfhhperpinv(\bfy,\bfxrond))
\quad ,\qquad  
\bfxrond \;=\;   \bfh^\ortho(\bfhhperpinv(\bfy,\bfxrond))
\  .
\end{array}\quad \right\}\quad 
\end{equation}
Finally, given an arbitrary integer $\indyi$, let $\sigma $ be a permutation satisfying
$$
\sigma (1)\;=\; \indyi
\  ,
$$
$\coordxr$ be an arbitrary coordinate chart in $\bfXrond_h $, and  $\hhperpinv$ and $\varphi$ be the 
expressions  of $\bfhhperpinv$ and $\bfvarphi$.
\startmodif
Exploiting commutation, we have
\stopmodif
$$
\hhperpinv (y,\xrond)\;=\; 
\; X_{\sigma (p)}\left( \vrule height 1.2em depth 1.2em width 0pt
X_{\sigma (p-1)}\left(\vrule height 0.8em depth 0.8em width 0pt
\ldots\left(\vrule height 0.6em depth 0.6em width 0pt
X_{\sigma (2)}
\left(\vrule height 0.4em depth 0.4em width 0pt
X_{\indyi}(\varphi(\xrond),y_\indyi),y_{\sigma (2)})\right)\ldots\right),y_{\sigma (p-1)}\right),y_{\sigma (p)}\right)
\  .$$

%
%
\noindent\underline{$\bfhhperpinv$ is a diffeomorphism:}\\[0.2em]
We obtain
\IfTwoCol{%
\begin{eqnarray*}
\frac{\partial \hhperpinv}{\partial y_\indyi}(y,\coordxrp)
&\hskip -0.5em =&\hskip -0.5em \displaystyle 
\left[\frac{\partial X_{\sigma (p)}}{\partial x}
\frac{\partial X_{\sigma (p-1)}}{\partial x}
\ldots
\frac{\partial X_{\sigma (2)}}{\partial x}
\right] 
\frac{\partial X_{\indyi}}{\partial t}(\varphi(\coordxrp),y_\indyi)
\\
&\hskip -0.5em =&\hskip -0.5em \displaystyle 
\frac{\partial X_g}{\partial x}(\varphi(\coordxrp),y)
\frac{\partial X_{\indyi}}{\partial x}(\varphi(\coordxrp),y)^{-1}
g_\indyi(X_{\indyi}(\varphi(\coordxrp),y_\indyi))
\end{eqnarray*}
}{%
\begin{eqnarray*}
\frac{\partial \hhperpinv}{\partial y_\indyi}(y,\coordxrp)
&=&
\left[\frac{\partial X_{\sigma (p)}}{\partial x}
\frac{\partial X_{\sigma (p-1)}}{\partial x}
\ldots
\frac{\partial X_{\sigma (2)}}{\partial x}
\left(X_{\indyi}(\varphi(\xrond),y_\indyi),y_{\sigma (2)})\right)
\right]
\frac{\partial X_{\indyi}}{\partial t}(\varphi(\coordxrp),y_\indyi)
\\
&=&
\frac{\partial X_g}{\partial x}(\varphi(\coordxrp),y)
\frac{\partial X_{\indyi}}{\partial x}(\varphi(\coordxrp),y)^{-1}
g_\indyi(X_{\indyi}(\varphi(\coordxrp),y_\indyi))
\end{eqnarray*}
}%
But with (\ref{LP211}), we have
$$
\frac{\partial X_{\indyi}}{\partial x}(\varphi(\coordxrp),y)^{-1}
g_\indyi(X_{\indyi}(\varphi(\coordxrp),y_\indyi))\;=\; 
g_\indyi(\varphi(\coordxrp))
$$
and
\IfReport{%
\begin{eqnarray*}
\frac{\partial X_g}{\partial x}(\varphi(\coordxrp),y)
g_\indyi(\varphi(\coordxrp))
&\hskip -0.5em =&\hskip -0.5em \displaystyle  
\left[\frac{\partial X_{\sigma (p)}}{\partial x}
\frac{\partial X_{\sigma (p-1)}}{\partial x}
\ldots
\frac{\partial X_{\sigma (2)}}{\partial x}\right]
\frac{\partial X_{\sigma (1)}}{\partial x}(\varphi(\coordxrp),y_{\sigma (1)})\,   g_{\sigma (1)}(\varphi(\coordxrp))
\\[0.5em]
&\hskip -0.5em =&\hskip -0.5em  \displaystyle 
\left[\frac{\partial X_{\sigma (p)}}{\partial x}
\frac{\partial X_{\sigma (p-1)}}{\partial x}
\ldots
\frac{\partial X_{\sigma (2)}}{\partial x}
\right]\: 
g_{\sigma (1)}\!\left(
\vrule height 0.4em depth 0.4em width 0pt
X_{\sigma (1)}(\varphi(\coordxrp),y_\indyi)\right)
\\[0.5em]
&\hskip -0.5em =&\hskip -0.5em \displaystyle
\left[\frac{\partial X_{\sigma (p)}}{\partial x}
\frac{\partial X_{\sigma (p-1)}}{\partial x}
\ldots
\frac{\partial X_{\sigma (3)}}{\partial x}
\right]
g_{\sigma (1)}\!\left(
\vrule height 0.5em depth 0.5em width 0pt
X_{\sigma (2)}\left(
\vrule height 0.4em depth 0.4em width 0pt
X_{\sigma (1)}(\varphi(\coordxrp),y_{\sigma (1)}),y_{\sigma (2)}\right)
\right)
\\[-0.5em]&\hskip -0.5em &\hskip -0.5em \displaystyle 
\null \  \vdots
\\
&\hskip -0.5em =&\hskip -0.5em \displaystyle g_{\sigma (1)}(\hhperpinv(y,\coordxrp))
\; =\;  
g_{\indyi}(\hhperpinv(y,\coordxrp))
\end{eqnarray*}
}{%
\\[0.7em]$\displaystyle 
\frac{\partial X_g}{\partial x}(\varphi(\coordxrp),y)
g_\indyi(\varphi(\coordxrp))
$\hfill \null \\[0.5em]$\displaystyle 
=\;  
\left[\frac{\partial X_{\sigma (p)}}{\partial x}
\frac{\partial X_{\sigma (p-1)}}{\partial x}
\ldots
\frac{\partial X_{\sigma (2)}}{\partial x}\right]\times
$\hfill\null\\\null\hfill$\displaystyle
\times \frac{\partial X_{\sigma (1)}}{\partial x}(\varphi(\coordxrp),y_{\sigma (1)})\,   g_{\sigma (1)}(\varphi(\coordxrp))
$\\[0.5em]$\displaystyle
=\; \displaystyle 
\left[\frac{\partial X_{\sigma (p)}}{\partial x}
\frac{\partial X_{\sigma (p-1)}}{\partial x}
\ldots
\frac{\partial X_{\sigma (2)}}{\partial x}
\right]\: 
g_{\sigma (1)}\!\left(
\vrule height 0.4em depth 0.4em width 0pt
X_{\sigma (1)}(\varphi(\coordxrp),y_\indyi)\right)
$\hfill\null\\[0.5em]$\displaystyle
=\; \displaystyle
\left[\frac{\partial X_{\sigma (p)}}{\partial x}
\frac{\partial X_{\sigma (p-1)}}{\partial x}
\ldots
\frac{\partial X_{\sigma (3)}}{\partial x}
\right]\times
$\hfill\null\\\null\hfill$\displaystyle
\times
g_{\sigma (1)}\!\left(
\vrule height 0.5em depth 0.5em width 0pt
X_{\sigma (2)}\left(
\vrule height 0.4em depth 0.4em width 0pt
X_{\sigma (1)}(\varphi(\coordxrp),y_{\sigma (1)}),y_{\sigma (2)}\right)
\right)
$\\[-0.5em]$\displaystyle
\null \  \vdots
$\hfill\null\\[0.5em]$\displaystyle
=\;  
g_{\sigma (1)}(\hhperpinv(y,\coordxrp))
=\;  
g_{\indyi}(\hhperpinv(y,\coordxrp))
$\\[0.7em]
}%
All this gives the expressions
\begin{eqnarray}
\label{LP212}
\frac{\partial \hhperpinv}{\partial y_\indyi}(y,\coordxrp)
&=&
\frac{\partial X_g}{\partial x}(\varphi(\coordxrp),y)
g_\indyi(\varphi(\coordxrp))
\\\label{LP213}
&=&
g_{\indyi}(\hhperpinv(y,\coordxrp))
\  .
\end{eqnarray}
We obtain also
\begin{eqnarray*}
\frac{\partial \hhperpinv}{\partial \coordxrp}(y,\coordxrp)
&=&\frac{\partial X_{\sigma (p)}}{\partial x}
\frac{\partial X_{\sigma (p-1)}}{\partial x}
\ldots
\frac{\partial X_{\sigma (1)}}{\partial x}(\varphi(\coordxrp),y)
\frac{\partial \varphi}{\partial \coordxrp}(\coordxrp)
\\
&=&
\frac{\partial X_g}{\partial x}(\varphi(\coordxrp),y)
\frac{\partial \varphi}{\partial \coordxrp}(\coordxrp)
\end{eqnarray*}
Assume there exists a nonzero vector $(u,w)$ in $\RR^p\times\RR^{n-p}$ satisfying
$$
\left(\begin{array}{cc}
\displaystyle 
\frac{\partial \hhperpinv}{\partial y}(y,\coordxrp)
&\displaystyle 
\frac{\partial \hhperpinv}{\partial \coordxrp}(y,\coordxrp)
\end{array}\right)
\left(\begin{array}{c}
u \\ w
\end{array}\right)\;=\; 0
$$
With (\ref{LP212}), this is equivalent to
\IfTwoCol{%
$$ 
\frac{\partial X_g}{\partial x}(\varphi(\coordxrp),y)
\left[\left(\begin{array}{@{\,  }ccc@{\,  }}
\displaystyle 
g_1(\varphi(\coordxrp))
&\!\!\!
\ldots\!\!\!
&\displaystyle 
g_p(\varphi(\coordxrp))
\end{array}\right)
u 
\;+\; 
\frac{\partial \varphi}{\partial \coordxrp}(\coordxrp) w
\right]\!=\! 0
$$
}{%
$$ 
\frac{\partial X_g}{\partial x}(\varphi(\coordxrp),y)
\left[\left(\begin{array}{@{\,  }ccc@{\,  }}
\displaystyle 
g_1(\varphi(\coordxrp))
&
\ldots
&\displaystyle 
g_p(\varphi(\coordxrp))
\end{array}\right)
u 
\;+\; 
\frac{\partial \varphi}{\partial \coordxrp}(\coordxrp) w
\right]\;=\; 0
$$
}%
Then, since $\frac{\partial X_g}{\partial x}(\varphi(\coordxrp),y)$ is a product of invertible matrices, this 
is further equivalent to
$$
\left(\begin{array}{@{\,  }ccc@{\,  }}
\displaystyle 
g_1(\varphi(\coordxrp))
&
\ldots
&\displaystyle 
g_p(\varphi(\coordxrp))
\end{array}\right)
u 
\;+\; 
\frac{\partial \varphi}{\partial \coordxrp}(\coordxrp) w
\;=\; 0
$$
Here we note that
the vectors $g_\indyi$ are independent in $\Distrib _P ^\ortho (\varphi(\coordxrp))$ and, $\bfvarphi$ being an embedding, for any nonzero vector $w$ in the tangent 
space to $\bfXrond_h $ at $\bfxrond$, the vector
$$
v^\tangent\;=\; \frac{\partial \varphi}{\partial \coordxrp}(\coordxrp)\,  w
$$
is a nonzero vector in $\Distrib^\tangent (\varphi(\coordxrp))$.
So there is no such nonzero vector $(u,w)$. This proves that the matrix
$\left(\begin{array}{cc}
\displaystyle 
\frac{\partial \hhperpinv}{\partial y}(y,\coordxrp)
&\displaystyle 
\frac{\partial \hhperpinv}{\partial \coordxrp}(y,\coordxrp)
\end{array}\right)$
is invertible and therefore that $\hhperpinv$ is a diffeomorphism.
These arguments establish also that we have
\begin{equation}
\label{LP215}
\frac{\partial \hhperpinv }{\partial \coordxrp }(y,\coordxrp)^\top
P(\hhperpinv (y,\coordxrp))
\frac{\partial \hhperpinv }{\partial \coordxrp }(y,\coordxrp)\; >\; 0
\end{equation}
%
%
\noindent
\underline{Construction of the metrics $\bfPy$ and $\bfPxi$:}\\
Let $\mathfrak{P}$ be the expression of the pull back by $\bfhhperpinv$ of $\bfP$. It satisfies
$$
\mathfrak{P}({}y,\coordxrp)\;=\; 
\left(
\renewcommand{\arraystretch}{1.9}
\begin{array}{c}
\displaystyle \frac{\partial \hhperpinv}{\partial {}y}(y, \coordxrp)^\top
\\
\displaystyle \frac{\partial \hhperpinv}{\partial  \coordxrp}(y, \coordxrp)^\top
\end{array}\right)
P(\hhperpinv(y,\coordxrp))
\left(\begin{array}{cc}
\displaystyle \frac{\partial \hhperpinv}{\partial {}y}(y, \coordxrp)
&
\displaystyle \frac{\partial \hhperpinv}{\partial  \coordxrp}(y, \coordxrp)
\end{array}\right)
$$
To simplify this expression, we note that we have (see (\ref{LP210}))
\IfTwoCol{%
\\[0.7em]$\displaystyle 
\frac{\partial \hhperpinv }{\partial {}y_\indyi}(y,\coordxrp)^\top
P(\hhperpinv (y,\coordxrp))
\frac{\partial \hhperpinv }{\partial {}y_\indyj}(y,\coordxrp)
$\hfill\null\\\null\hfill$\displaystyle
=\; 
g_\indyi(\hhperpinv (y,\coordxrp))
P(\hhperpinv (y,\coordxrp))g_\indyj(\hhperpinv (y,\coordxrp))
\;=\; \Py_{\indyi\indyj}(h(\hhperpinv (y,\coordxrp)))
$\\[0.7em]
}{%
\begin{eqnarray*}
\frac{\partial \hhperpinv }{\partial {}y_\indyi}(y,\coordxrp)^\top
P(\hhperpinv (y,\coordxrp))
\frac{\partial \hhperpinv }{\partial {}y_\indyj}(y,\coordxrp)&=&
g_\indyi(\hhperpinv (y,\coordxrp))
P(\hhperpinv (y,\coordxrp))
g_\indyj(\hhperpinv (y,\coordxrp))
\\
&=&\Py_{\indyi\indyj}(h(\hhperpinv (y,\coordxrp)))
\end{eqnarray*}
}%
and also, using (\ref{LP214}) for the last identity,
\IfTwoCol{%
\\[0.7em]$\displaystyle 
\frac{\partial \hhperpinv }{\partial {}y_\indyi}(y,\coordxrp)^\top
P(\hhperpinv (y,\coordxrp))
\frac{\partial \hhperpinv }{\partial \coordxrp }(y,\coordxrp)
$\hfill\null\\\null\hfill$\displaystyle
\renewcommand{\arraystretch}{1.5}
\begin{array}{@{}cl@{}}
=&\displaystyle 
w{}_\indyi^\top
\Py(h(\hhperpinv (y,\coordxrp)))  \frac{\partial h}{\partial x}(\hhperpinv (y,\coordxrp))\frac{\partial \hhperpinv }{\partial \xrond}(y,\coordxrp)
\\
=&\displaystyle 
w{}_\indyi^\top\Py(h(\hhperpinv (y,\coordxrp)))
\frac{\partial }{\partial \coordxrp}\left\{h(\hhperpinv (y,\coordxrp))\right\}
\\
=&\displaystyle 
w{}_\indyi^\top\Py(h(\hhperpinv (y,\coordxrp)))
\frac{\partial }{\partial \coordxrp}\left\{y\right\}
\; =\;0\  .
\end{array}
$\\[0.7em]
This yields
\\[0.5em]\null \hfill $\displaystyle 
\mathfrak{P}({}y,\coordxrp)\;=\; \left(
\begin{array}{cc}
\Py(y) & 0 \\ 0 & \Pxi ({}y,\coordxrp)
\end{array}\right)
$\hfill \null \\[0.3em]
where
$$
\Pxi ({}y,\coordxrp)\;=\; 
\frac{\partial \hhperpinv }{\partial \coordxrp}(y,\coordxrp)^\top
P(\hhperpinv (y,\coordxrp))
\frac{\partial \hhperpinv }{\partial \coordxrp}(y,\coordxrp)
$$
}{%
\begin{eqnarray*}
\frac{\partial \hhperpinv }{\partial {}y_\indyi}(y,\coordxrp)^\top
P(\hhperpinv (y,\coordxrp))
\frac{\partial \hhperpinv }{\partial \coordxrp }(y,\coordxrp)&=&
w{}_\indyi^\top
\Py(h(\hhperpinv (y,\coordxrp)))\frac{\partial h}{\partial x}(\hhperpinv (y,\coordxrp))\frac{\partial \hhperpinv }{\partial \xrond}(y,\coordxrp)
\\
&=&
w{}_\indyi^\top\Py(h(\hhperpinv (y,\coordxrp)))
\frac{\partial }{\partial \coordxrp}\left\{h(\hhperpinv (y,\coordxrp))\right\}
\\
&=&
w{}_\indyi^\top\Py(h(\hhperpinv (y,\coordxrp)))
\frac{\partial }{\partial \coordxrp}\left\{y\right\}
\\
&=&0
\end{eqnarray*}
This yields
$$
\mathfrak{P}({}y,\coordxrp)\;=\; \left(
\begin{array}{cc}
\Py(y) & 0 \\ 0 & \Pxi ({}y,\coordxrp)
\end{array}\right)
$$
where
$$
\Pxi ({}y,\coordxrp)\;=\; 
\frac{\partial \hhperpinv }{\partial \coordxrp}(y,\coordxrp)^\top
P(\hhperpinv (y,\coordxrp))
\frac{\partial \hhperpinv }{\partial \coordxrp}(y,\coordxrp)
$$
}%
With (\ref{LP213}), we have
\\[1em]$\displaystyle 
\frac{\partial R}{\partial {}y_\indyi}({}y,\coordxrp)
$\hfill \null 
\\\null \quad  $\displaystyle 
\;=\;
\frac{\partial ^2\hhperpinv }{\partial \coordxrp\partial y_\indyi}(y,\coordxrp)^\top
P(\hhperpinv (y,\coordxrp))\frac{\partial \hhperpinv }{\partial \coordxrp}(y,\coordxrp)
$\hfill \null \\\null \hfill $\displaystyle +
\frac{\partial \hhperpinv }{\partial \coordxrp}(y,\coordxrp)^\top
\left[\frac{\partial P}{\partial x}(x)\frac{\partial \hhperpinv }{\partial y_\indyi}(y,\coordxrp))\right]
\frac{\partial \hhperpinv }{\partial \coordxrp}(y,\coordxrp)
$\hfill \null \\\null \hfill $\displaystyle +
\frac{\partial \hhperpinv }{\partial \coordxrp}(y,\coordxrp)^\top P(\hhperpinv (y,\coordxrp))
\frac{\partial ^2\hhperpinv }{\partial \coordxrp \partial y_\indyi}(y,\coordxrp)
$\\\null \quad  $\displaystyle 
\;=\;
\frac{\partial \hhperpinv }{\partial \coordxrp}(y,\coordxrp)^\top
\frac{\partial g_\indyi}{\partial x}(\hhperpinv(y,\coordxrp)^\top
P(\hhperpinv (y,\coordxrp))\frac{\partial \hhperpinv }{\partial \coordxrp}(y,\coordxrp)
$\hfill \null \\\null \hfill $\displaystyle +
\frac{\partial \hhperpinv }{\partial \coordxrp}(y,\coordxrp)^\top
\left[\frac{\partial P}{\partial x}(x)g_\indyi(\hhperpinv (y,\coordxrp))\right]
\frac{\partial \hhperpinv }{\partial \coordxrp}(y,\coordxrp)
$\hfill \null \\\null \hfill $\displaystyle +
\frac{\partial \hhperpinv }{\partial \coordxrp}(y,\coordxrp)^\top P(\hhperpinv (y,\coordxrp))
\frac{\partial g_\indyi}{\partial x}(\hhperpinv(y,\coordxrp)\frac{\partial \hhperpinv }{\partial \coordxrp }(y,\coordxrp)
$\\\null \quad  $\displaystyle 
\;=\;
\frac{\partial \hhperpinv }{\partial \coordxrp}(y,\coordxrp)^\top
\mathcal{L}_{g_\indyi}P(\hhperpinv (y,\coordxrp))
\frac{\partial \hhperpinv }{\partial \coordxrp}(y,\coordxrp)\  ,$\\[1em]
where
\IfTwoCol{%
\\[1em]$\displaystyle 
\mathcal{L}_{g_\indyi}P(x)_{\indxa\indxb}
$\null \hfill \\\null \quad $\displaystyle 
\;=\;
\sum_{\indxc,\indyj}\frac{\partial P_{\indxa\indxb}}{\partial x_\indxc}(x)
\grad_P[h_\indyj]_\indxc (x) \Py_{\indyj\indyi}(h(x))
$\hfill \null \\ $\displaystyle
+
\sum_{\indxc,\indyj}
\left[
P_{\indxa\indxc}(x)\frac{\partial \grad_P[h_\indyj]_\indxc}{\partial x_\indxb}(x)
+
P_{\indxb\indxc}(x)\frac{\partial \grad_P[h_\indyj]_\indxc}{\partial x_\indxa}(x)
\right]\Py_{\indyj\indyi}(h(x))
$\hfill \null \\$\displaystyle
+
\sum_{\indxc,\indyj,\indyk}
\left[
P_{\indxa\indxc}(x) \frac{\partial h_\indyk}{\partial x_\indxb}(x)
+
P_{\indxb\indxc}(x) \frac{\partial h_\indyk}{\partial x_\indxa}(x)
\right]\times
$\hfill \null \\[-0.5em]\null \hfill $\displaystyle
\times
\grad_P[h_\indyj]_\indxc(x) \frac{\partial \Py_{\indyj\indyi}}{\partial y}(h(x))
$\\\null \quad $\displaystyle 
\;=\;
2 \sum_{\indyj}\Hess_P[h_\indyj]_{\indxa\indxb}(x) \Py_{\indyj\indyi}(h(x))
$\hfill \null \\ $\displaystyle
+
\sum_{\indxc,\indyj,\indyk}
\left[
P_{\indxa\indxc}(x) \frac{\partial h_\indyk}{\partial x_\indxb}(x)
+
P_{\indxb\indxc}(x) \frac{\partial h_\indyk}{\partial x_\indxa}(x)
\right]\times
$\hfill \null \\[-0.5em]\null \hfill $\displaystyle
\times
\grad_P[h_\indyj]_\indxc(x) \frac{\partial \Py_{\indyj\indyi}}{\partial y}(h(x))
\  .$\\[0.5em]
}{%
\\[1em]$\displaystyle 
\mathcal{L}_{g_\indyi}P(x)_{\indxa\indxb}
$\null \hfill \\\null \quad $\displaystyle 
\;=\;
\sum_{\indxc,\indyj}\frac{\partial P_{\indxa\indxb}}{\partial x_\indxc}(x)
\grad_P[h_\indyj]_\indxc (x) \Py_{\indyj\indyi}(h(x))
$\hfill \null \\\null \hfill $\displaystyle
+
\sum_{\indxc,\indyj}
\left[
P_{\indxa\indxc}(x)\frac{\partial \grad_P[h_\indyj]_\indxc}{\partial x_\indxb}(x) \Py_{\indyj\indyi}(h(x))+
P_{\indxb\indxc}(x)\frac{\partial \grad_P[h_\indyj]_\indxc}{\partial x_\indxa}(x) \Py_{\indyj\indyi}(h(x))
\right]
$\hfill \null \\\null \hfill $\displaystyle
+
\sum_{\indxc,\indyj,\indyk}
\left[
P_{\indxa\indxc}(x)\grad_P[h_\indyj]_\indxc(x) \frac{\partial \Py_{\indyj\indyi}}{\partial y}(h(x))
\frac{\partial h_\indyk}{\partial x_\indxb}(x)
+
P_{\indxb\indxc}(x)\grad_P[h_\indyj]_\indxc(x) \frac{\partial \Py_{\indyj\indyi}}{\partial y}(h(x))
\frac{\partial h_\indyk}{\partial x_\indxb}(x)
\right]
$\\\null \quad $\displaystyle 
\;=\;
2 \sum_{\indyj}\Hess_P[h_\indyj]_{\indxa\indxb}(x) \Py_{\indyj\indyi}(h(x))
$\hfill \null \\\null \hfill $\displaystyle
+
\sum_{\indxc,\indyj,\indyk}
\left[
P_{\indxa\indxc}(x)\grad_P[h_\indyj]_\indxc(x) \frac{\partial \Py_{\indyj\indyi}}{\partial y}(h(x))
\frac{\partial h_\indyk}{\partial x_\indxb}(x)
+
P_{\indxb\indxc}(x)\grad_P[h_\indyj]_\indxc(x) \frac{\partial \Py_{\indyj\indyi}}{\partial y}(h(x))
\frac{\partial h_\indyk}{\partial x_\indxb}(x)
\right]
\  .
$\\[1em]
}%
Since from Lemma \ref{lem12},
the level sets of $\bfh$ are totally geodesic, it follows from
\cite[Proposition A.2.2]{57} that there exist
continuous  functions $k_{\indyj\indyk\indxa}$ such that
we have
$$
\Hess_P[h_\indyj]_{\indxa\indxb}(x) 
\;=\; \sum_{\indyk}k_{\indyj\indyk\indxa}(x) 
\frac{\partial h_\indyk}{\partial x_\indxb}(x) 
+
k_{\indyj\indyk\indxb}(x) 
\frac{\partial h_\indyk}{\partial x_\indxa}(x)
\  .
$$
Consequently, there exist continuous functions $\ell_{\indyi\indyk\indxa}$ satisfying
\IfReport{%
$$
\frac{\partial R}{\partial {}y_\indyi}({}y,\coordxrp) =
\sum_{\indxa,\indxb,\indyk}
\frac{\partial \hhperpinv _\indxa}{\partial \coordxrp}(y,\coordxrp)^\top
\!
\left[
\ell_{\indyi\indyk\indxa}(\hhperpinv (y,\coordxrp))\frac{\partial h_\indyk}{\partial x_\indxb}(\hhperpinv (y,\coordxrp))
+\ell_{\indyi\indyk\indxb}(\hhperpinv (y,\coordxrp))
\frac{\partial h_\indyk}{\partial x_\indxa}(\hhperpinv (y,\coordxrp))^\top 
\right]\frac{\partial \hhperpinv _\indxb}{\partial \coordxrp}(y,\coordxrp)
$$
}{%
\\[0.7em]$\displaystyle 
\frac{\partial R}{\partial {}y_\indyi}({}y,\coordxrp)\; =\; 
\sum_{\indxa,\indxb,\indyk}
\frac{\partial \hhperpinv _\indxa}{\partial \coordxrp}(y,\coordxrp)^\top
\left[
\ell_{\indyi\indyk\indxa}(\hhperpinv (y,\coordxrp))\frac{\partial h_\indyk}{\partial x_\indxb}(\hhperpinv (y,\coordxrp))
\right.
$\hfill \null \\\null \hfill $\displaystyle
\left.+\ell_{\indyi\indyk\indxb}(\hhperpinv (y,\coordxrp))
\frac{\partial h_\indyk}{\partial x_\indxa}(\hhperpinv (y,\coordxrp))^\top 
\right]\frac{\partial \hhperpinv _\indxb}{\partial \coordxrp}(y,\coordxrp)
$\\[0.7em]
}
Then with
$$
\sum_{\indxa}
\frac{\partial \hhperpinv _\indxa}{\partial \coordxrp}(y,\coordxrp)^\top
\frac{\partial h_\indyk}{\partial x_\indxa}(\hhperpinv (y,\coordxrp))^\top 
\;=\; 
\frac{\partial }{\partial \coordxrp}\left\{h_\indyk(\hhperpinv (y,\coordxrp))\right\}^\top
\  ,
$$
and (\ref{LP214}), we obtain that $\frac{\partial R}{\partial {}y_\indyi}({}y,\coordxrp)\;$ is zero and 
therefore that $R$ does not depend on $y$.

Finally we note that $R$ is 
the expression of a covariant $2$ tensor with positive definite values. Indeed we have (\ref{LP215}) and if
$\bar \coordxrp$ are other coordinates for $\bfxrond$, i.e. 
\IfTwoCol{%
$\bar \coordxrp = \changexr(\coordxrp)$
where $\changexr$ is a diffeomorphism, then with the definitions
\\[0.5em]\null \hfill $\displaystyle 
\bar \hhperpinv (y,\bar \coordxrp)\,=\, \hhperpinv(y,\coordxrp)
\ ,\ 
\bar \Pxi (\bar \coordxrp)\,=\, 
\frac{\partial \bar \hhperpinv }{\partial \bar \coordxrp}(y,\bar \coordxrp )^\top
P(\bar \hhperpinv (y,\bar \coordxrp ))
\frac{\partial \bar \hhperpinv }{\partial \bar \coordxrp}(y,\bar \coordxrp )
$\hfill \null \\
we obtain
\IfTwoCol{%
$\frac{\partial  \hhperpinv }{\partial  \coordxrp}(y, \coordxrp )= 
\frac{\partial \bar \hhperpinv }{\partial \bar \coordxrp}(y,\coordxrp )
\frac{\partial \changexr}{\partial \coordxrp}(\coordxrp)
$.
}{
\\[0.3em]\null \hfill $\displaystyle 
\frac{\partial  \hhperpinv }{\partial  \coordxrp}(y, \coordxrp )\;=\; 
\frac{\partial \bar \hhperpinv }{\partial \bar \coordxrp}(y,\coordxrp )
\frac{\partial \changexr}{\partial \coordxrp}(\coordxrp)
$\hfill \null \\
}
This yields
\begin{eqnarray*}
\frac{\partial \changexr}{\partial \coordxrp}(\coordxrp)^\top
\!\bar \Pxi (\bar \coordxrp)
\frac{\partial \changexr}{\partial \coordxrp}(\coordxrp)
&\!\!\!\!=\!\!\!\!&
\frac{\partial  \hhperpinv }{\partial  \coordxrp}(y, \coordxrp )^\top
\!P(\bar \hhperpinv (y,\bar \coordxrp ))
\frac{\partial  \hhperpinv }{\partial  \coordxrp}(y, \coordxrp )
\\
&\!\!\!\!=\!\!\!\!&
\frac{\partial  \hhperpinv }{\partial  \coordxrp}(y, \coordxrp )^\top
\!P( \hhperpinv (y, \coordxrp ))
\frac{\partial  \hhperpinv }{\partial  \coordxrp}(y, \coordxrp )
\!=\! \Pxi(\coordxrp )
\:  .
\end{eqnarray*}
}{%
$$
\bar \coordxrp\;=\;\changexr(\coordxrp)
$$
where $\changexr$ is a diffeomorphism, then with the definitions
$$
\bar \hhperpinv (y,\bar \coordxrp)\;=\; \hhperpinv(y,\coordxrp)
\quad ,\qquad 
\bar \Pxi (\bar \coordxrp)\;=\; 
\frac{\partial \bar \hhperpinv }{\partial \bar \coordxrp}(y,\bar \coordxrp )^\top
P(\bar \hhperpinv (y,\bar \coordxrp ))
\frac{\partial \bar \hhperpinv }{\partial \bar \coordxrp}(y,\bar \coordxrp )
$$
we obtain
$$
\frac{\partial  \hhperpinv }{\partial  \coordxrp}(y, \coordxrp )\;=\; 
\frac{\partial \bar \hhperpinv }{\partial \bar \coordxrp}(y,\coordxrp )
\frac{\partial \changexr}{\partial \coordxrp}(\coordxrp)
$$
This yields
\begin{eqnarray*}
\frac{\partial \changexr}{\partial \coordxrp}(\coordxrp)^\top
\bar \Pxi (\bar \coordxrp)
\frac{\partial \changexr}{\partial \coordxrp}(\coordxrp)
&=&
\frac{\partial  \hhperpinv }{\partial  \coordxrp}(y, \coordxrp )^\top
P(\bar \hhperpinv (y,\bar \coordxrp ))
\frac{\partial  \hhperpinv }{\partial  \coordxrp}(y, \coordxrp )
\\
&=&
\frac{\partial  \hhperpinv }{\partial  \coordxrp}(y, \coordxrp )^\top
P( \hhperpinv (y, \coordxrp ))
\frac{\partial  \hhperpinv }{\partial  \coordxrp}(y, \coordxrp )
\;=\; \Pxi(\coordxrp )
\  .
\end{eqnarray*}
}%

\vspace{-0.15in}
\noindent
\underline{Properties of $\bfh^\ortho$:}\\
So far, we have established that, if
\startmodif
\begin{itemize}
\item[]
$\bfP$ is complete,
\item[]
$\bfh$ is a Riemannian submersion on $\bfRR^n$,
\item[]
the vector fields $g_\indyi$ commute,
\item[]
the function $\bfy\mapsto \bfPy(\bfy)$ is bounded,
\end{itemize}%
\stopmodif\noindent%
there exists a globally defined diffeomorphism
$\bfhhperpinv $ which
is an  isometry between the Riemannian product
of $(\bfRR,\bfPy) \times (\bfXrond_h ,\bfPxi)$ and $(\bfRR^n,\bfP)$.
This implies
that the inverse $\bfhhperp$ of $\bfhhperpinv$ is an isometry between
$(\bfRR^n,\bfP)$ and $(\bfRR,\bfPy) \times (\bfXrond_h ,\bfPxi)$. Moreover its component $\bfh^\ortho$ (see (\ref{LP227})) is a surjective submersion. It is also a Riemannian submersion because of the construction of the metric 
$\bfPxi$ for  $\bfXrond_h$.
Moreover, with (\ref{LP214}) and (\ref{LP213}), we obtain, for any coordinate chart,
\begin{eqnarray*}
\frac{\partial h^\ortho}{\partial x}(\hhperpinv(y,\xrond))
\frac{\partial \hhperpinv}{\partial y}(y,\xrond)&=&0
\\
\frac{\partial h^\ortho}{\partial x}(x)\left(
\begin{array}{@{\,  }ccc@{\,  }}
g_1(x) & \dots & g_p(x)\end{array}\right)
&=&0\qquad \forall x
\end{eqnarray*}
So the space tangent at $\bfx$ to the level sets of $\bfh^\ortho$ is $\bfDistrib _\bfP ^\ortho(\bfx)$.
With Lemma \ref{lem11}
it follows that, for any $\bfx_0$, in the $\bfx$-manifold $\bfRR^n$, the $\bfh^\ortho(\bfx_0)$-level set of 
$\bfh^\ortho$
\IfReport{
$$
\mathfrak{H}^\ortho (\bfh^\ortho(\bfx_0))\;=\; \left\{\bfx\in\bfRR^n:\,  
\bfh^\ortho(\bfx)=\bfh^\ortho(\bfx_0))\right\}
$$
}
{
$
\mathfrak{H}^\ortho (\bfh^\ortho(\bfx_0))\;=\; \left\{\bfx\in\bfRR^n:\,  
\bfh^\ortho(\bfx)=\bfh^\ortho(\bfx_0))\right\}
$
}
is totally geodesic. Then, with the definition of the metric $\bfPy$, any geodesic 
$\bfgammay:(s_1,s_2)\to \mathfrak{H}^\ortho (\bfx_0)$ for the 
metric $\bfPy$ equipping $\mathfrak{H}^\ortho (\bfx_0)$ gives rise to a geodesic
$s\in (s_1,s_2)\mapsto \bfgamma (s)=(\bfgammay(s),\bfh^\ortho(\bfx_0))$ for the metric $\bfP$ equipping 
$\bfRR^n$. If $\bfP$ is complete the latter can be 
extended to $(-\infty ,+\infty )$ and, because $\mathfrak{H}^\ortho (\bfx_0)$ is totally geodesic, the $\bfxrond$ 
component of this extension remains constant and its $\bfy$ component is a geodesic in this 
$\bfh^\ortho(\bfx_0)$-level set
$\mathfrak{H}^\ortho (\bfh^\ortho(\bfx_0))$ for the 
metric $\bfPy$. With the Hopf-Rinow Theorem 
\cite[Theorem II.1.1]{Sakai.96}  this implies $\bfPy$ is complete. The same argument apply for the metric 
$\bfPxi$.
\null 
\vskip -2em

\stoparchive


\startarchive
\newpage
\section{Supplementary Material}

\sousection{Observer with a gradient like correction term}
\soussousection{Motivation for a gradient like correction term}
\label{complement10}
Our postulate is that we want an observer in a
Kalman form 
$$
\dot{\hatx}\;=\; f(x)\;-\; \mathfrak{C}(\hatx,y)
$$
such that, given a Riemannian metric with expression $P$, the correction term $\mathfrak{C}$
contributes to the decrease of the distance between $\hat x$ and $x$. This means that we have
\begin{equation}
\label{LP61repeat}
\frac{d\gamma  ^*}{ds}(s)^\top 
P(\gamma   ^*(s) ) \,  \mathfrak{C}(\gamma   ^*(s) ,h(\gamma   ^*(0) ))\; > \; 0
\qquad \forall s:\,  0 < s\ .
\end{equation}
Since the correction must vanish when $h(\hatx)=y$, i.e.
\begin{equation}
\label{1b}
\mathfrak{C}(x,h(x))\;=\; 0\  ,
\end{equation}
by defining
$$
x=\gamma ^*(0)
\quad ,\qquad 
v\;=\; \frac{d\gamma ^*}{ds}(0)
\  ,
$$
(\ref{LP61repeat}) gives, for all $(x,v)$, 
\\[1em]\vbox{\noindent\null\qquad $\displaystyle 
\lim_{s\to 0} 
\frac{1}{s}
\left[\frac{d\gamma  ^*}{ds}(s)^{\top}
P(\gamma   ^*(s) ) \,  \mathfrak{C}(\gamma   ^*(s) ,h(\gamma   ^*(0) ))
-
\frac{d\gamma  ^*}{ds}(0)^{\top}
P(\gamma   ^*(0) ) \,  \mathfrak{C}(\gamma   ^*(0) ,h(\gamma   ^*(0) ))
\right]
$\hfill \null \\[0.5em]\null \hfill $\displaystyle \;=\; 
v^{\top} P(x)\,  \frac{\partial \mathfrak{C}}{\partial x}(x,h(x))\,  v
\; \geq \; 0
\  .
$\qquad \null }\\[1em]
Taking the total derivative of $\mathfrak{C}$ and using (\ref{1b}) we obtain
$$
\frac{\partial \mathfrak{C}}{\partial x}(x,h(x))\;=\;-
\frac{\partial \mathfrak{C}}{\partial y}(x,h(x))\frac{\partial h}{\partial x}(x)
$$
Hence (\ref{LP61repeat}) implies~:
\begin{equation}
\label{LP69}
v^{\top} P(x)\frac{\partial \mathfrak{C}}{\partial y}(x,h(x))\frac{\partial h}{\partial x}(x)\,  v
\; \leq \; 0
\end{equation}
In the case where $\frac{\partial h}{\partial x}(x)$ is full row 
rank, this implies\footnote{
\startmodif
Assume we have $v^{\top} B A v\geq 0$ for all $v$ and where $A$ is a 
full row rank matrix. There exists another full row rank matrix $C$  such 
that the matrix $M=\left(\begin{array}{c}
A \\ C
\end{array}\right)$ is invertible and we have 
$AM^{-1}=\left(\begin{array}{cc}
I & 0
\end{array}\right)$. We decompose $v$ and $B$ as follows
$
Mv=\left(\begin{array}{cc}
v_A\\v_C
\end{array}\right)$ and $
{M^{-1}}^{\top} B=\left(\begin{array}{c}
B_A \\ B_C
\end{array}\right)
$.
Then our assumption can be rewritten as
$
[v_A^{\top} B_A + v_C^{\top} B_C]v_A\geq 0$ for all $v_A$ and $v_C$.
This implies the symmetric part of $B_A$ is non negative and $B_C$ is 
zero.
This gives
$B=M^{\top}\left(\begin{array}{c} B_A \\ 0\end{array}\right)=
A^{\top} B_A
$.
\stopmodif
} the existence of a matrix with non negative symmetric part 
$R(x)$ (not necessarily continuous in $x$)
such that we have~:
$$
P(x)\frac{\partial \mathfrak{C}}{\partial y}(x,h(x))\;=\; -
\frac{\partial h}{\partial x}(x)^\top R(x) 
$$
or, in other words, we must have
$$
\frac{\partial \mathfrak{C}}{\partial y}(x,h(x))\;=\; -\grad_P h  (x)\,  R(x)
\  .
$$
This condition is satisfied when the correction term is
$$
\mathfrak{C}(\hatx ,y)\;=\;  k_E(\hatx )\,  \grad_P h(\hatx )
\frac{\partial \wpunbf }{\partial y_a}(h(\hatx ),y)^\top
$$

\soussousection{Our observer is a gradient algorithm}
\startmodif
In (\ref{eqn:GeodesicObserverVectorField}),
the observer is
\begin{equation}
\label{eqn:GeodesicObserverVectorFieldrepeat}
\dot{\bfhatx }\;=\;  \bff({\bfhatx }) \;-\;
\kunbf _E(\bfhatx )\,  \bfgrad_\bfP [\wpunbf \circ\bfh](\bfhatx ,\bfy)
\  .
\end{equation}
\stopmodif
To show how it can be seen as a gradient algorithm, let 
$\bfx_0$ be any point in the $\bfx$-manifold $\bfRR^n$, $\coordx$ be any coordinate chart around  
$\bfx_0$ and $\coordy$ be any coordinate chart around  
$\bfh(\bfx_0)$. The expression of the observer (\ref{eqn:GeodesicObserverVectorFieldrepeat}) is~:
$$
\dot{\hatx } = f(\hatx) - k_E(\hatx)
P(x)^{-1}\frac{\partial h}{\partial x}(\hat x)^\top
\frac{\partial \wpunbf}{\partial y_1}(h(\hatx), y)^\top
$$
We consider the time-varying diffeomorphism $x\mapsto \changex_t(x)$ defined as~:
$$
\changex_t(x)\;=\; X(x,-t)
$$
where $X(x,t)$ denotes the solution 
at time $t$, issued from $x$, at time $0$, of
$$
\dot x = f(x)
\  .
$$
It satisfies
$$
\frac{\partial \changex_t}{\partial x}(x) f(x)\;=\; f(\changex_t(x))
$$
and gives rise to another but time-varying coordinate chart 
$(\barcoordxp_t,\barcoordxd_t,\barcoordxm_t)$ around $\bfx_0$. With $\coordy$,
and
$$
\bar x_t\;=\; \changex_t(x)
\quad ,\qquad 
\bar k_{tE}(\bar x_t)\;=\; k_E(\changex_t^{-1}(\bar x_t))
\quad ,\qquad 
\bar h_t(\bar x_t)\;=\; h(\changex_t^{-1}(\bar x_t))
$$
and
$$ 
P(x)\;=\; \frac{\partial \changex_t}{\partial x}(x)^\top \bar P_t(\changex_t(x)) \frac{\partial \changex_t}{\partial x}(x)
$$
the other expression of the observer is~:
\begin{eqnarray}
\dot{{\overline{\hatx }}}_t &=&  -  k_E(\hatx )
\frac{\partial \changex_t}{\partial x}(\hatx)
P(\hatx )^{-1}\frac{\partial  h}{\partial x}(\hatx)^\top
\frac{\partial \wpunbf}{\partial y_1}(\bar h(\hatx ), y)^\top
\,
\\
\label{5b}
&=& - \bar k_E(\overline{\hatx }_t)
\bar P_t(\overline{\hatx }_t)^{-1}\frac{\partial \bar h_t}{\partial \bar{x}_t}(\overline{\hatx }_t)^\top
\frac{\partial \wpunbf}{\partial y_1}(\bar h_t(\overline{\hatx }_t), y)^\top
\end{eqnarray}
Of course, because of coordinate independence, in the right hand side, we recognize the partial derivative
$$
\frac{\partial }{\partial \overline{\hatx }_t }
\left\{
\wpunbf (\bar h_t(\overline{\hatx }_t), y)
\right\}
\;=\; 
\frac{\partial \wpunbf}{\partial y_1}(\bar h_t(\overline{\hatx }_t), y)
\frac{\partial \bar h_t}{\partial \bar{x}_t}(\overline{\hatx }_t)
$$
It follows that $\dot{\overline{\hatx }}_t$ is proportional to
the expression in the time varying coordinates of the Riemannian gradient,
with respect to $\bfhatx$, of the function $\bfhatx\mapsto \wpunbf (\bfh(\bfhatx),\bfy)$.

\sousection{More about coordinates}
\label{complement45}
When $x$ are used as coordinates for $\,  \bfx$, it is useful to collect all the 
Christoffel symbols
$\,  \Gamma _{\indxb\indxc}^\indxa$
\glos{\ref{item:RiemannianMetric}},
attached to the expression $P$ of a metric $\bfP$ in the ``matrix''
$$
\mbox{\large $\Gamma $}^\indxa\;=\; 
\left(\Gamma _{\indxb\indxc}^\indxa\right)
$$
This gives for example the compact notation
\begin{equation}
\label{LP158}
u^\top \mbox{\large $\Gamma $}^\indxa v\;=\; \sum_{\indxb,\indxc}\Gamma _{\indxb\indxc}^\indxa u_\indxb 
v_\indxc
\end{equation}

With other coordinates
$$
\bar x\;=\; \changex (x)
\quad ,\qquad 
\bar y \;=\; \changey (y)
\ ,
$$
we have (see \cite[(3.5.22)]{Lovelock-Rund})
$$
\bar h(\changex(x))\;=\; \changey(h(x))
\quad ,\qquad 
\frac{\partial \changex}{\partial x}(x)^\top \bar P(\changex(x))\,  \frac{\partial \changex}{\partial x}(x)\;=\; P(x)
\  ,
$$
\begin{equation}
\label{LP193}
\sum_{\indxa,\indxb}
\frac{\partial \changex _\indxa}{\partial x_\indxd}(x)
\bar \Gamma _{\indxa\indxb}^\indxc(\changex(x))\frac{\partial \changex_\indxb}{\partial x_\indxe}(x)
+ \frac{\partial ^2 \changex_k}{\partial x_\indxd \partial x_\indxe}(x)
\;=\; 
\sum _\indxf\Gamma_{\indxd\indxe}^\indxf(x) \frac{\partial \changex _\indxc}{\partial x_\indxf}(x)
\  .
\end{equation}

When $(y,\xrond)$ are used as coordinates for $\bfx$, we collect all the Christoffel symbols
$\Gamma _{\indyj\indyk}^\indyi$, 
$\Gamma _{\indyj\indxrb}^\indyi$,
$\Gamma _{\indxrb\indxrc}^\indyi$,
$\Gamma _{\indyj\indyk}^\indxra$,
$\Gamma _{\indyj\indxrb}^\indxra$,
x$\Gamma _{\indxrb\indxrc}^\indxra$,
in the two ``matrices''
$$
\mbox{\large $\Gamma $}^\indyi\;=\; 
\left(\begin{array}{cc}
\Gamma _{yy}^\indyi
&
\Gamma _{y\xrond}^\indyi
\\
\Gamma _{\xrond y}^\indyi
&
\Gamma _{\xrond\xrond}^\indyi
\end{array}\right)
\quad ,\qquad 
\mbox{\large $\Gamma $}^\indxra\;=\; 
\left(\begin{array}{cc}
\Gamma _{yy}^\indxra
&
\Gamma _{y\xrond}^\indxra
\\
\Gamma _{\xrond y}^\indxra
&
\Gamma _{\xrond\xrond}^\indxra
\end{array}\right)
\ .
$$
Also a change of coordinates of the form
$$
(\bar y,\bar \xrond)\;=\; (\changey (y),\changexr (y,\xrond))
$$
leads to
\begin{eqnarray*}
\frac{\partial \changey}{\partial y} ^\top \bar P_{yy}  \frac{\partial \changey}{\partial y}
+
\frac{\partial \changexr}{\partial y} ^\top \bar P_{\xrond y}  \frac{\partial \changey}{\partial y}
+
\frac{\partial \changey}{\partial y} ^\top \bar P_{y\xrond}  \frac{\partial \changey}{\partial y}
+
\frac{\partial \changexr}{\partial y} ^\top \bar P_{\xrond\xrond}  \frac{\partial \changexr}{\partial y}
&=&P_{yy}
\\
\frac{\partial \changey}{\partial y} ^\top \bar P_{y\xrond}  \frac{\partial \changexr}{\partial \xrond}
+
\frac{\partial \changexr}{\partial y} ^\top \bar P_{\xrond \xrond}  \frac{\partial \changexr}{\partial \xrond}
&=&
P_{y\xrond}
\\
\frac{\partial \changexr}{\partial \xrond} ^\top \bar P_{\xrond \xrond}  \frac{\partial \changexr}{\partial \xrond}
&=&P_{\xrond\xrond}
\  .
\end{eqnarray*}
and when 
$$
(y,\xrond)\;\mapsto\; (\bar y,\barxrond)\;=\;(y,\changexr(y,\xrond))
$$
\begin{eqnarray*}
\bar \Gamma _{yy}^\indyi
+
\bar \Gamma _{y\xrond}^\indyi
\frac{\partial \changexr}{\partial y}
+
\frac{\partial \changexr}{\partial y}^\top
\bar \Gamma _{\xrond y}^\indyi
+
\frac{\partial \changexr }{\partial y}^\top
\bar \Gamma _{\xrond\xrond}^\indyi
\frac{\partial \changexr }{\partial y}
&=&
\Gamma _{yy}^\indyi
\  ,
\\
\bar \Gamma _{y\xrond}^\indyi
\frac{\partial \changexr}{\partial \xrond}
+
\frac{\partial \changexr }{\partial y}^\top
\bar \Gamma _{\xrond\xrond}^\indyi
\frac{\partial \changexr}{\partial \xrond}
&=&
\Gamma _{y\xrond}^\indyi
\  ,\\
\frac{\partial \changexr }{\partial \xrond}^\top
\bar \Gamma _{\xrond\xrond}^\indyi
\frac{\partial \changexr}{\partial \xrond}
&=&
\Gamma _{\xrond\xrond}^\indyi
\end{eqnarray*}
The proof of the latter is as follows.
A direct decomposition of (\ref{LP193}) gives
\\[1em]$\displaystyle 
\sum_{\indyj\indyk}\frac{\partial y_\indyj}{\partial y_\indyl}
\bar \Gamma _{\indyj\indyk}^\indyi
\frac{\partial y_\indyk}{\partial y_\indym}
+
\sum_{\indyj\indxrc}\frac{\partial y_\indyj}{\partial y_\indyl}
\bar \Gamma _{\indyj\indxrc}^\indyi
\frac{\partial \changexr_\indxrc}{\partial y_\indym}
+
\sum_{\indyj\indxrc}\frac{\partial \changexr_\indxrc}{\partial y_\indyl}
\bar \Gamma _{\indyj\indxrc}^\indyi
\frac{\partial y_\indyj}{\partial y_\indym}
+
\sum_{\indxrc\indxrd}\frac{\partial \changexr _\indxrc}{\partial y_\indyl}
\bar \Gamma _{\indxrc\indxrd}^\indyi
\frac{\partial \changexr_\indxrd}{\partial y_\indym}
$\hfill \null \\\null \hfill $\displaystyle
\;=\; 
\sum_{\indyj}\frac{\partial y_\indyi}{\partial y_\indyj}
\Gamma _{\indyl\indym}^\indyj
+
\sum_{\indxrc}\frac{\partial y_\indyi}{\partial x_\indxrc}
\Gamma _{\indyl\indym}^\indxrc
-
\frac{\partial ^2 y_\indyi}{\partial y_\indyl\partial y_\indym}
$\\[1em]$\displaystyle 
\sum_{\indyj\indyk}\frac{\partial y_\indyj}{\partial y_\indyl}
\bar \Gamma _{\indyj\indyk}^\indyi
\frac{\partial y_\indyk}{\partial \xrond_\indxra}
+
\sum_{\indyj\indxrc}\frac{\partial y_\indyj}{\partial y_\indyl}
\bar \Gamma _{\indyj\indxrc}^\indyi
\frac{\partial \changexr_\indxrc}{\partial \xrond_\indxra}
+
\sum_{\indyj\indxrc}\frac{\partial \changexr_\indxrc}{\partial y_\indyl}
\bar \Gamma _{\indyj\indxrc}^\indyi
\frac{\partial y_\indyj}{\partial \xrond_\indxra}
+
\sum_{\indxrc\indxrd}\frac{\partial \changexr _\indxrc}{\partial y_\indyl}
\bar \Gamma _{\indxrc\indxrd}^\indyi
\frac{\partial \changexr_\indxrd}{\partial \xrond_\indxra}
$\hfill \null \\\null \hfill $\displaystyle
\;=\; 
\sum_{\indyj}\frac{\partial y_\indyi}{\partial y_\indyj}
\Gamma _{\indyl\indxra}^\indyj
+
\sum_{\indxrc}\frac{\partial y_\indyi}{\partial x_\indxrc}
\Gamma _{\indyl\indxra}^\indxrc
-
\frac{\partial ^2 y_\indyi}{\partial y_\indyl\partial \xrond_\indxra}
$\\[1em]$\displaystyle 
\sum_{\indyj\indyk}\frac{\partial y_\indyj}{\partial \xrond_\indxra}
\bar \Gamma _{\indyj\indyk}^\indyi
\frac{\partial y_\indyk}{\partial \xrond_\indxrb}
+
\sum_{\indyj\indxrc}\frac{\partial y_\indyj}{\partial \xrond_\indxra}
\bar \Gamma _{\indyj\indxrc}^\indyi
\frac{\partial \changexr_\indxrc}{\partial \xrond_\indxrb}
+
\sum_{\indyj\indxrc}\frac{\partial \changexr_\indxrc}{\partial \xrond_\indxra}
\bar \Gamma _{\indyj\indxrc}^\indyi
\frac{\partial y_\indyj}{\partial \xrond_\indxrb}
+
\sum_{\indxrc\indxrd}\frac{\partial \changexr _\indxrc}{\partial \xrond_\indxra}
\bar \Gamma _{\indxrc\indxrd}^\indyi
\frac{\partial \changexr_\indxrd}{\partial \xrond_\indxrb}
$\hfill \null \\\null \hfill $\displaystyle
\;=\; 
\sum_{\indyj}\frac{\partial y_\indyi}{\partial y_\indyj}
\Gamma _{\indxra\indxrb}^\indyj
+
\sum_{\indxrc}\frac{\partial y_\indyi}{\partial x_\indxrc}
\Gamma _{\indyl\indym}^\indxrc
-
\frac{\partial ^2 y_\indyi}{\partial \xrond_\indxra\partial \xrond_\indxrb}
$\\[1em]
This simplifies in
$$
\bar \Gamma _{\indyl\indym}^\indyi
+
\sum_{\indxrc}
\bar \Gamma _{\indyl\indxrc}^\indyi
\frac{\partial \changexr_\indxrc}{\partial y_\indym}
+
\sum_{\indxrc}\frac{\partial \changexr_\indxrc}{\partial y_\indyl}
\bar \Gamma _{\indym\indxrc}^\indyi
+
\sum_{\indxrc\indxrd}\frac{\partial \changexr _\indxrc}{\partial y_\indyl}
\bar \Gamma _{\indxrc\indxrd}^\indyi
\frac{\partial \changexr_\indxrd}{\partial y_\indym}
\;=\; 
\Gamma _{\indyl\indym}^\indyi
$$
$$
\sum_{\indxrc}
\bar \Gamma _{\indyl\indxrc}^\indyi
\frac{\partial \changexr_\indxrc}{\partial \xrond_\indxra}
+
\sum_{\indxrc\indxrd}\frac{\partial \changexr _\indxrc}{\partial y_\indyl}
\bar \Gamma _{\indxrc\indxrd}^\indyi
\frac{\partial \changexr_\indxrd}{\partial \xrond_\indxra}
\;=\; 
\Gamma _{\indyl\indxra}^\indyi
$$
$$
\sum_{\indxrc\indxrd}\frac{\partial \changexr _\indxrc}{\partial \xrond_\indxra}
\bar \Gamma _{\indxrc\indxrd}^\indyi
\frac{\partial \changexr_\indxrd}{\partial \xrond_\indxrb}
\;=\; 
\Gamma _{\indxra\indxrb}^\indyi
$$
%
%

%
%
\sousection{Proof of (\ref{LP110})}
\label{complement26}
Let $\dy (y,y_1)$ be the distance between $y_1$ and $y$. We keep $y_1$ fixed and let 
$y$ vary as
$y_1 + {\varsigma  _1} {u_1}
$ with $u_1$ some normalized vector.
Then the first variation formula, gives us
$$
\lim_{\varsigma  _1\to 0} \frac{\dy (y+{\varsigma  _1} {u_1},y_1)^2-\dy (y,y_1)^2}{\varsigma  _1}\;=\; 
2 \dy (y,y_1)
\  u_1^T \Py (y_1) \frac{d\gammay^* }{ds}(s)
\qquad \forall u_1
\  ,
$$
\startmodif
where $\gammay^* $ is a minimizing normalized geodesic from $y_1$ to $y$ satisfying
\stopmodif
\begin{eqnarray*}
&\displaystyle \gammay ^*(0)=y_1
\quad ,\qquad \gammay ^*(s)=y
\  ,
\\
&\displaystyle  
\frac{d\gammay ^*}{ds}( s) ^\top \Py (\gammay ^*( s))\frac{d\gammay ^*}{ds}( s) =1
\qquad 
\forall s 
\  .
\end{eqnarray*}%
Hence the second identity in (\ref{LP110}) is established.

For the first one, we note that the geodesic equation gives the ``acceleration'' $c_1$ of a geodesic issued from
$y_1$ in the direction $v_1$ as
\begin{equation}
\label{LP194}
2\sum_{\indyk}\Py _{\indyi\indyk} (y_1)[c_1]_{\indyk}\;=\;  -2 \sum_\indyk\left[\frac{\partial \Py_{\indyi\indyk}}{\partial y}(y_1) v_1\right] [v_1]_\indyk + 
\sum_{\indyk,\indyl}[v_1]_\indyk \frac{\partial \Py _{\indyk\indyl}}{\partial y_\indyi}(y_1) [v_1]_\indyl
\end{equation}
with the compact notation
$$
\left[\frac{\partial \Py_{\indyi\indyk}}{\partial y}(y_1) v_1\right]
\;=\; 
\sum_\indyj
\frac{\partial \Py_{\indyi\indyk}}{\partial y_\indyj}(y_1) [v_1]_\indyj
$$
Then Taylor expansions give
\begin{eqnarray*}
\gammay ^*(s)&=& y_1 + s v_1 + \frac{s^2}{2} c_1+O(s^3)
\\
\frac{d\gammay ^*}{ds}(s)&=& v_1 + s c_1+O(s^2)
\\
\Py _{\indyi\indyk}(\gammay ^*(s))&=& \Py _{\indyi\indyk}(y_1) + s
\left[\frac{\partial \Py _{\indyi\indyk}}{\partial y}(y_1) v_1\right]+O(s^2)\end{eqnarray*}
Hence, with $s_2$ satisfying
$$
y_2\;=\; y_1 + s_2 v_1 +O(s_2^2)
\  ,
$$
we obtain
\begin{eqnarray*}
\dy (y_1,y_2) &\hskip -0.6em=&\hskip -0.6em\int_0^{s_2}
\sqrt{
\left(v_1 + s c_1\right)^\top
\left(\Py (y_1) + s \left[\frac{\partial \Py }{\partial y}(y_1) v_1\right]\right)
\left(v_1 + s c_1\right)
+O(s^2)
} ds
\\
&\hskip -0.6em=&\hskip -0.6em
\sqrt{v_1 ^\top \Py (y_1) v_1 }
\int_0^{s_2}\!
\left[1+s
\frac{
c_1^\top \Py (y_1) v_1 
+
v_1^\top\left[\frac{\partial \Py }{\partial y}(y_1) v_1\right] v_1
+
v_1 ^\top \Py (y_1) c_1
}{
2\sqrt{v_1 ^\top \Py (y_1) v_1 }
}+O(s^2)\right]\!ds
\  .
\end{eqnarray*}
But (\ref{LP194}) implies~:
\\[1em]$\displaystyle 
\sum_{\indyi,\indyk}[v_1]_\indyi\left[\frac{\partial \Py _{\indyi\indyk}}{\partial y}(y_1) v_1\right] 
[v_1]_\indyk
+
2\sum_{\indyi,\indyj}[v_1]_\indyi \Py _{\indyi\indyk}(y_1) [c_1]_\indyk
$\hfill \null \\[0.5em]\null \hfill $
\begin{array}{@{}cl@{}}
=&\displaystyle 
\sum_{\indyi,\indyk}[v_1]_\indyi\left[\frac{\partial \Py _{\indyi\indyk}}{\partial y}(y_1) v_1\right] 
[v_1]_\indyk
-2 \sum_{\indyi,\indyk}[v_1]_\indyi\left[\frac{\partial \Py_{\indyi\indyk}}{\partial y}(y_1) v_1\right] [v_1]_\indyk + 
\sum_{\indyi,\indyk,\indyl}[v_1]_\indyi [v_1]_\indyk \frac{\partial \Py _{\indyk\indyl}}{\partial y_\indyi}(y_1) [v_1]_\indyl
\\[0.6em]
=&\displaystyle 
-\sum_{\indyi,\indyj,\indyk}[v_1]_\indyi\frac{\partial \Py _{\indyi\indyk}}{\partial y_\indyj}(y_1) [v_1]_\indyj
[v_1]_\indyk
+
\sum_{\indyi,\indyk,\indyl}[v_1]_\indyi [v_1]_\indyk \frac{\partial \Py _{\indyk\indyl}}{\partial y_\indyi}(y_1) [v_1]_\indyl
\\[0.6em]
=&\displaystyle 0
\end{array}$\\[1em]
So we have finally
\begin{eqnarray*}
\dy (y_1,y_2) ^2&=&[s_2v_1] ^\top \Py (y_1) [s_2v_1] + O(s_2^4)
\\
&=& [y_2-y_1] ^\top \Py (y_1) [y_2-y_1] + O(|y_2-y_1|^4)
\end{eqnarray*}
Hence~:
$$
\frac{\partial \dy^2}{\partial y_1}(y_1,y_1)\;=\; 0
\quad ,\qquad 
\frac{\partial ^2\dy^2}{\partial y_1^2}(y_1,y_1)\;=\; 2 \Py (y_1)
\  .
$$
\sousection{{Condition A3 does not hold for the metric (\ref{LP135})}}
\label{complement23}
In dimension $3$ with $x_1$ being the output,
the output level sets are totally geodesic if the metric satisfies~:
$$
\left(
\begin{array}{@{}c@{\qquad }c@{\qquad }c@{}}
\displaystyle 
2\frac{\partial P_{12}}{\partial x_2}
-
\frac{\partial P_{22}}{\partial x_1}
&\displaystyle 
\frac{\partial P_{22}}{\partial x_2}
&\displaystyle 
2\frac{\partial P_{23}}{\partial x_2}
-
\frac{\partial P_{22}}{\partial x_3}
\\[1em]\displaystyle 
\frac{\partial P_{12}}{\partial x_3}
+
\frac{\partial P_{13}}{\partial x_2}
-
\frac{\partial P_{23}}{\partial x_1}
&\displaystyle 
\frac{\partial P_{22}}{\partial x_3}
&\displaystyle 
\frac{\partial P_{33}}{\partial x_2}
\\[1em]\displaystyle 
2\frac{\partial P_{13}}{\partial x_3}
-
\frac{\partial P_{33}}{\partial x_1}
&\displaystyle 
2\frac{\partial P_{23}}{\partial x_3}
-
\frac{\partial P_{33}}{\partial x_2}
&\displaystyle 
\frac{\partial P_{33}}{\partial x_3}
\end{array}
\right)
\left(\begin{array}{c}
P_{22}P_{33}-P_{23}^2
\\[1em]\displaystyle 
P_{23}P_{13}-P_{33}P_{12}
\\[1em]\displaystyle 
P_{23}P_{12}-P_{22}P_{13}
\end{array}\right)
\;=\; 0\  .
$$
\startmodif
We obtain for the metric (\ref{LP135})
\stopmodif
\begin{eqnarray*}
2\frac{\partial P_{12}}{\partial x_2} - \frac{\partial P_{22}}{\partial x_1}
&=&0
\\
\frac{\partial P_{22}}{\partial x_2}
&=&0
\\
2\frac{\partial P_{23}}{\partial x_2} - \frac{\partial P_{22}}{\partial x_3}
&=&
2\frac{ -4 }{\lambda (\lambda ^2+4x_3)^2}
+\frac{8}{\lambda (\lambda ^2+4x_3)^2}
\;=\; 0
\end{eqnarray*}
So the first line is zero.

For the second line, we have
\begin{eqnarray*}
\frac{\partial P_{12}}{\partial x_3} + \frac{\partial P_{13}}{\partial x_2}
-\frac{\partial P_{23}}{\partial x_1}
&=&
\frac{4 }{ (\lambda ^2+4x_3)^2}
+ \frac{-(\lambda ^2-4x_3)}{\lambda ^2(\lambda ^2+4x_3)^2}
-\frac{ (3\lambda ^2+4x_3)}{\lambda ^2(\lambda ^2+4x_3)^2}
\;=\; 0
\\
\frac{\partial P_{22}}{\partial x_3}
&=&
-\frac{8}{\lambda (\lambda ^2+4x_3)^2}
\\
\frac{\partial P_{33}}{\partial x_2}
&=&
-
\frac{4(5\lambda ^2+4x_3)}{\lambda ^2(\lambda ^2+4x_3)^3}x_1
+ 
\frac{8(5\lambda ^2+4x_3)}{\lambda ^3(\lambda ^2+4x_3)^3}x_2
\\
&=&
\frac{
4(5\lambda ^2+4x_3)(-\lambda x_1
+8x_2)
}{\lambda ^3(\lambda ^2+4x_3)^3}
\end{eqnarray*}
$$
P_{23}P_{13}-P_{33}P_{12}
\begin{array}[t]{cl@{}}
=&\displaystyle 
\frac{ (3\lambda ^2+4x_3) x_1 - 4\lambda  x_2
}{\lambda ^2(\lambda ^2+4x_3)^2}
\; 
\frac{-\lambda ^3x_1+(\lambda ^2-4x_3)x_2
}{\lambda ^2(\lambda ^2+4x_3)^2}
\\\multicolumn{2}{@{}r@{}}{\displaystyle 
+
\frac{1 }{ (\lambda ^2+4x_3)}
\;
\frac{
[6\lambda ^4+12\lambda ^2x_3+16x_3^2]x_1 ^2
-
4(5\lambda ^2+4x_3)\lambda x_1x_2
+
4(5\lambda ^2+4x_3)x_2^2
}{\lambda ^3(\lambda ^2+4x_3)^3}
}
\\[1.5em]
=&\displaystyle 
\frac{
- \lambda ^3 (3\lambda ^2+4x_3)  x_1^2
+
[
(3\lambda ^2+4x_3)(\lambda ^2-4x_3)
+ 4 \lambda ^4
]  x_1x_2
-
4 (\lambda ^2-4x_3)\lambda x_2^2
}{\lambda ^4(\lambda ^2+4x_3)^3}
\\\multicolumn{2}{@{}r@{}}{\displaystyle 
+
\frac{
[6\lambda ^4+12\lambda ^2x_3+16x_3^2]\lambda x_1 ^2
-
4(5\lambda ^2+4x_3)\lambda ^2 x_1x_2
+
4(5\lambda ^2+4x_3)\lambda x_2^2
}{\lambda ^4(\lambda ^2+4x_3)^4}
}
\\[1.5em]
=&\displaystyle 
\frac{
[3\lambda ^4+8\lambda ^2x_3+16x_3^2]\lambda x_1 ^2
-
[
-13 \lambda  ^4
-24 \lambda ^2 x_3
-16 x_3^2
] x_1 x_2
+
4(4\lambda ^2+8x_3)\lambda x_2^2
}{\lambda ^4(\lambda ^2+4x_3)^4}
\end{array}
$$
$$
P_{23}P_{12}-P_{22}P_{13} 
\begin{array}[t]{cl@{}}
=&\displaystyle 
-\frac{ (3\lambda ^2+4x_3)
x_1
-
4\lambda 
x_2
}{\lambda ^2(\lambda ^2+4x_3)^2}
\; 
\frac{1 }{ (\lambda ^2+4x_3)}
+
\frac{2  }{ \lambda (\lambda ^2+4x_3)}
\; 
\frac{-
\lambda ^3
x_1
+
(\lambda ^2-4x_3)
x_2}{\lambda ^2(\lambda ^2+4x_3)^2}
\\[1.5em]
=&\displaystyle 
\frac{ -(5\lambda ^2+4x_3)\lambda x_1
+2 (3\lambda ^2-4x_3) x_2
}{\lambda ^3(\lambda ^2+4x_3)^3}
\end{array}
$$
So we get for the second line~:
\\[1em]$\displaystyle 
\frac{\partial P_{22}}{\partial x_3}
[P_{23}P_{13}-P_{33}P_{12}]
+
\frac{\partial P_{33}}{\partial x_2}
[P_{23}P_{12}-P_{22}P_{13}]\;=\; 
$\hfill \null \\\null \hfill $\displaystyle 
-\frac{8}{\lambda (\lambda ^2+4x_3)^2}
\frac{
[3\lambda ^4+8\lambda ^2x_3+16x_3^2]\lambda x_1 ^2
-
[
-13 \lambda  ^4
-24 \lambda ^2 x_3
-16 x_3^2
] x_1 x_2
+
4(4\lambda ^2+8x_3)\lambda x_2^2
}{\lambda ^4(\lambda ^2+4x_3)^4}
$\hfill \null \\\null \hfill $\displaystyle 
+\frac{
4(5\lambda ^2+4x_3)(-\lambda x_1
+8x_2)
}{\lambda ^3(\lambda ^2+4x_3)^3}
\frac{ -(5\lambda ^2+4x_3)\lambda x_1
+2 (3\lambda ^2-4x_3) x_2
}{\lambda ^3(\lambda ^2+4x_3)^3}
$\\[1em]
When $x_1=0$, this reduces to
$$
-\frac{8}{\lambda (\lambda ^2+4x_3)^2}
\; \frac{
4(4\lambda ^2+8x_3)\lambda x_2^2
}{\lambda ^4(\lambda ^2+4x_3)^4}
+\frac{
4(5\lambda ^2+4x_3)8x_2
}{\lambda ^3(\lambda ^2+4x_3)^3}
\; \frac{ 
2 (3\lambda ^2-4x_3) x_2
}{\lambda ^3(\lambda ^2+4x_3)^3}\; \neq \; 0
$$
For the third line, we have
\begin{eqnarray*}
2\frac{\partial P_{13}}{\partial x_3} - \frac{\partial P_{33}}{\partial x_1}
&=&
2
\frac{-4x_2}{\lambda ^2(\lambda ^2+4x_3)^2}
-2
\frac{8[-\lambda ^3x_1+(\lambda ^2-4x_3)x_2]}{\lambda ^2(\lambda ^2+4x_3)^3}
\;=\; 8\frac{
2\lambda ^3 x_1
-[3\lambda ^2 - 4 x_3] x_2
}{\lambda ^2(\lambda ^2+4x_3)^3}
\\
2\frac{\partial P_{23}}{\partial x_3} - \frac{\partial P_{33}}{\partial x_2}
&=&
2\frac{ 4x_1}{\lambda ^2(\lambda ^2+4x_3)^2}
-
2
\frac{ 8[(3\lambda ^2+4x_3)x_1-4\lambda x_2]
}{\lambda ^2(\lambda ^2+4x_3)^3}
\;=\; 
8
\frac{
-2[\lambda ^2 +2x_3]x_1
+ 8 \lambda  x_2
}{\lambda ^2(\lambda ^2+4x_3)^3}
\end{eqnarray*}
We conclude that the level sets of the output functions are not 
totally geodesic.
%
%

%
%
\sousection{{Condition A3 does not hold for the metric (\ref{LP136})}}~\\
\label{complement24}
\\[1em]$\displaystyle
k\,  \Gamma _{33}^1\;=\; 
-y\left[
\xrond _\indxrb ^2
\left|\begin{array}{@{}cccc@{}}
\scripteuP_{11} & . & \scripteuP_{13} & \scripteuP_{14}
\\
 .  & . & . & . 
\\
\scripteuP_{13} & . & \scripteuP_{33} & \scripteuP_{34}
\\
\scripteuP_{14} & . & \scripteuP_{34} & \scripteuP_{44}
\end{array}\right|
\;+\; 
2\xrond _\indxrb \left|\begin{array}{@{}cccc@{}}
\scripteuP_{11} & . & \scripteuP_{13} & \scripteuP_{14}
\\
\scripteuP_{12} & . & \scripteuP_{23} & \scripteuP_{24}
\\
\scripteuP_{13} & . & \scripteuP_{33} & \scripteuP_{34}
\\
 . & . & . & . 
\end{array}\right|
\;+\; 
\left|\begin{array}{@{}cccc@{}}
\scripteuP_{11} & \scripteuP_{12} & \scripteuP_{13} & . 
\\
\scripteuP_{12} & \scripteuP_{22} & \scripteuP_{23} & . 
\\
\scripteuP_{13} & \scripteuP_{23} & \scripteuP_{33} & . 
\\
 . & . & . & . 
\end{array}\right|
\right]
$\hfill \null \\\null \hfill $\displaystyle
\;-\; 
\xrond _\indxra \left[
\xrond _\indxrb ^2
\left|\begin{array}{@{}cccc@{}}
 . & . & . & . 
\\
\scripteuP_{12} & . & \scripteuP_{23} & \scripteuP_{24}
\\
\scripteuP_{13} & . & \scripteuP_{33} & \scripteuP_{34}
\\
\scripteuP_{14} & . & \scripteuP_{34} & \scripteuP_{44}
\end{array}\right|
\;+\;
\xrond _\indxrb  \left(
\left|\begin{array}{@{}cccc@{}}
\scripteuP_{11} & \scripteuP_{12} & . & \scripteuP_{14}
\\
 . & . & . & . 
\\
\scripteuP_{13} & \scripteuP_{23} & . & \scripteuP_{34}
\\
\scripteuP_{14} & \scripteuP_{24} & . & \scripteuP_{44}
\end{array}\right|
+
\left|\begin{array}{@{}cccc@{}}
 . & . & . & . 
\\
\scripteuP_{12} & \scripteuP_{22} & \scripteuP_{23} & . 
\\
\scripteuP_{13} & \scripteuP_{23} & \scripteuP_{33} & . 
\\
\scripteuP_{14} & \scripteuP_{24} & \scripteuP_{34} & . 
\end{array}\right|
\right)
\;+\; 
\left|\begin{array}{@{}cccc@{}}
\scripteuP_{11} & \scripteuP_{12} & \scripteuP_{13} & . 
\\
\scripteuP_{12} & \scripteuP_{22} & \scripteuP_{23} & . 
\\
 . & . & . & . 
\\
\scripteuP_{14} & \scripteuP_{24} & \scripteuP_{34} & .
\end{array}\right|
\right]
$\\[1em]
\sousection{Proof of Lemma \ref{lem8}}
\label{complement37}
Let $\coordyxr$ and $\coordy$ be an arbitrary coordinate chart pair.
We have the expressions (see (\ref{LP185}), (\ref{LP148}) and (\ref{LP199}))~:
$$
\frac{\partial h}{\partial x}\;=\; \left(\begin{array}{cc}
I_p & 0
\end{array}\right)
\  ,\quad 
\secff _P h^\indyi\;=\; 
\left(\begin{array}{cc}
\Gammay _{yy}^\indyi
-\Gamma _{yy}^\indyi
&
-\Gamma _{y\xrond}^\indyi
\\
- \Gamma _{\xrond y}^\indyi
&
-\Gamma _{\xrond\xrond}^\indyi
\end{array}\right)\  ,
$$
$$
P^{-1}\;=\; 
\left(\begin{array}{cc}
P_y^{-1}
&
-P_y^{-1}P_{y\xrond}P_{\xrond\xrond}^{-1}
\\
-P_{\xrond\xrond}^{-1}P_{\xrond y}
P_y^{-1}
&
\left(P_{\xrond\xrond}-P_{\xrond y}P_{yy}^{-1}P_{y \xrond}\right)^{-1}
\end{array}\right) 
\  .
$$
and
\begin{eqnarray}
\nonumber
[\secff _P h^{\ortho,\ortho}]^\indyi&=&
\left(\begin{array}{cc}
I_p
&
-
P_{y\xrond}P_{\xrond\xrond}^{-1}
\end{array}\right) 
\left(\begin{array}{cc}
\Gammay _{yy}^\indyi
-\Gamma _{yy}^\indyi
&
-\Gamma _{y\xrond}^\indyi
\\
- \Gamma _{\xrond y}^\indyi
&
-\Gamma _{\xrond\xrond}^\indyi
\end{array}\right)
\left(\begin{array}{cc}
I_p
\\
-P_{\xrond\xrond}^{-1}P_{\xrond y}
\end{array}\right) 
\\\label{LP192}
&=&
\mbox{\large$\Gammay$}^\indyi
-
\Gamma _{yy}^\indyi
\;+\; 
P_{y\xrond}P_{\xrond\xrond}^{-1}
\Gamma _{\xrond y}^\indyi
\;+\; 
\Gamma _{y\xrond}^\indyi
P_{\xrond\xrond}^{-1}P_{\xrond y}
\;-\; 
P_{y\xrond}P_{\xrond\xrond}^{-1}
\Gamma _{\xrond\xrond}^\indyi
P_{\xrond\xrond}^{-1}P_{\xrond y}
\  .
\end{eqnarray}
We expand the latter in
\\[1em]$\displaystyle 
2\left[\secff _P h^{\ortho,\ortho}\right]_{\indyj\indyk}^\indyi
$\hfill \null \\\null \   $\displaystyle  =\; 
\sum_\indyl \Py ^{-1}_{\indyi\indyl}
\left(
\frac{\partial \Py _{\indyl\indyk}}{\partial y_\indyj}
+
\frac{\partial  \Py _{\indyl\indyj}}{\partial y_\indyk}
-
\frac{\partial  \Py _{\indyj\indyk}}{\partial y_\indyl}
\right)
$\hfill \null \\\null \   $\displaystyle  \hphantom{=\; }
-
\sum_\indyl[P^{-1}]_{\indyi\indyl}
\left(
\frac{\partial P_{\indyl\indyk}}{\partial y_\indyj}
+
\frac{\partial P_{\indyl\indyj}}{\partial y_\indyk}
-
\frac{\partial P_{\indyj\indyk}}{\partial y_\indyl}
\right)
-
\sum_\indxra [P^{-1}]_{\indyi\indxra}
\left(
\frac{\partial P_{\indxra\indyk}}{\partial y_\indyj}
+
\frac{\partial P_{\indxra\indyj}}{\partial y_\indyk}
-
\frac{\partial P_{\indyj\indyk}}{\partial \xrond_\indxra}
\right)
$\hfill \null \\\null \   $\displaystyle  \hphantom{=\; }
+\sum_{\indxrb,\indxrc}P_{\indyj\indxrb}[P_{\xrond\xrond}^{-1}]_{\indxrb\indxrc}
\left[
\sum_\indyl [P^{-1}]_{\indyi\indyl}
\left(
\frac{\partial P_{\indyl\indyk}}{\partial \xrond_\indxrc}
+
\frac{\partial P_{\indyl\indxrc}}{\partial y_\indyk}
-
\frac{\partial P_{\indxrc\indyk}}{\partial y_\indyl}
\right)
+
\sum_\indxra [P^{-1}]_{\indyi\indxra}
\left(
\frac{\partial P_{\indxra\indyk}}{\partial \xrond_\indxrc}
+
\frac{\partial P_{\indxra\indxrc}}{\partial y_\indyk}
-
\frac{\partial P_{\indxrc\indyk}}{\partial \xrond_\indxra}
\right)
\right]
$\hfill \null \\\null \   $\displaystyle  \hphantom{=\; }
+\sum_{\indxrd,\indxre}
\left[
\sum_\indyl [P^{-1}]_{\indyi\indyl}
\left(
\frac{\partial P_{\indyl\indyj}}{\partial \xrond_\indxrd}
+
\frac{\partial P_{\indyl\indxrd}}{\partial y_\indyj}
-
\frac{\partial P_{\indxrd\indyj}}{\partial y_\indyl}
\right)
+
\sum_\indxra [P^{-1}]_{\indyi\indxra}
\left(
\frac{\partial P_{\indxra\indyj}}{\partial \xrond_\indxrd}
+
\frac{\partial P_{\indxra\indxrd}}{\partial y_\indyj}
-
\frac{\partial P_{\indxrd\indyj}}{\partial \xrond_\indxra}
\right)
\right]
[P_{\xrond\xrond}^{-1}]_{\indxrd\indxre}
P_{\indxre\indyk}
$\hfill \null \\\null \   $\displaystyle  \hphantom{=\; }
-\sum_{\indxrb,\indxrc,\indxrd,\indxre}P_{\indyj\indxrb}[P_{\xrond\xrond}^{-1}]_{\indxrb\indxrc}
\left[
\sum_\indyl [P^{-1}]_{\indyi\indyl}
\left(
\frac{\partial P_{\indyl\indxrd}}{\partial \xrond_\indxrc}
+
\frac{\partial P_{\indyl\indxrc}}{\partial \xrond_\indxrd}
-
\frac{\partial P_{\indxrc\indxrd}}{\partial y_\indyl}
\right)
+
\right.
$\hfill \null \\\null \hfill $\displaystyle 
\left.
+
\sum_\indxra [P^{-1}]_{\indyi\indxra}
\left(
\frac{\partial P_{\indxra\indxrd}}{\partial \xrond_\indxrc}
+
\frac{\partial P_{\indxra\indxrc}}{\partial \xrond_\indxrd}
-
\frac{\partial P_{\indxrc\indxrd}}{\partial \xrond_\indxra}
\right)
\right]
[P_{\xrond\xrond}^{-1}]_{\indxrd\indxre}
P_{\indxre\indyk}
$\\[1em]
The inversion of the matrix $P$
$$
\left(\begin{array}{cc}
[P^{-1}]_{yy}  & [P^{-1}]_{y\xrond}
\\{}
[P^{-1}]_{\xrond y} & [P^{-1}]_{\xrond\xrond}
\end{array}
\right)
\left(\begin{array}{cc}
P_{yy}  & P_{y\xrond}
\\
P_{\xrond y} & P_{\xrond\xrond}
\end{array}
\right)
\;=\; \left(\begin{array}{cc}
I_p & 0 \\0 & I_{n-p}
\end{array}\right)
$$
gives
$$
[P^{-1}]_{yy} P_{y\xrond}\;=\; -[P^{-1}]_{y\xrond}P_{\xrond\xrond}
$$
and therefore
\begin{equation}
\label{LP200}
[P^{-1}]_{\indyi\indxra}\;=\; - \sum_\indyl [P^{-1}]_{\indyi\indyl}
M_{\indxra\indyl}
\  ,
\end{equation}
with the notation
$$
M_{\indxra\indyl}
\;=\; \sum_{\indxrf}P_{\indyl\indxrf}[P_{\xrond\xrond}^{-1}]_{\indxrf\indxra}
$$
This yields
\\[1em]$\displaystyle 
2\left[\secff _P h^{\ortho,\ortho}\right]_{\indyj\indyk}^\indyi
$\hfill \null \\\null \   $\displaystyle  =\; 
\sum_\indyl \Py ^{-1}_{\indyi\indyl}
\left(
\frac{\partial \Py _{\indyl\indyk}}{\partial y_\indyj}
+
\frac{\partial  \Py _{\indyl\indyj}}{\partial y_\indyk}
-
\frac{\partial  \Py _{\indyj\indyk}}{\partial y_\indyl}
\right)
$\hfill \null \\\null \   $\displaystyle  \hphantom{=\; }
-
\sum_\indyl [P^{-1}]_{\indyi\indyl}
\left[
\left(
\frac{\partial P_{\indyl\indyk}}{\partial y_\indyj}
+
\frac{\partial P_{\indyl\indyj}}{\partial y_\indyk}
-
\frac{\partial P_{\indyj\indyk}}{\partial y_\indyl}
\right)
-
\sum_{\indxra}M_{\indxra\indyl}
\left(
\frac{\partial P_{\indxra\indyk}}{\partial y_\indyj}
+
\frac{\partial P_{\indxra\indyj}}{\partial y_\indyk}
-
\frac{\partial P_{\indyj\indyk}}{\partial \xrond_\indxra}
\right)
\right]
$\hfill \null \\\null \   $\displaystyle  \hphantom{=\; }
+\sum_\indyl [P^{-1}]_{\indyi\indyl}
\left[
\sum_{\indxrc}M_{\indxrc\indyj}
\left[
\vphantom{\sum_{\indxrd,\indxre}}
\left(
\frac{\partial P_{\indyl\indyk}}{\partial \xrond_\indxrc}
+
\frac{\partial P_{\indyl\indxrc}}{\partial y_\indyk}
-
\frac{\partial P_{\indxrc\indyk}}{\partial y_\indyl}
\right)
-
\sum_{\indxra}M_{\indxra\indyl}
\left(
\frac{\partial P_{\indxra\indyk}}{\partial \xrond_\indxrc}
+
\frac{\partial P_{\indxra\indxrc}}{\partial y_\indyk}
-
\frac{\partial P_{\indxrc\indyk}}{\partial \xrond_\indxra}
\right)
\right]
\right]
$\\\null \   $\displaystyle  \hphantom{=\; }
+
\sum_\indyl [P^{-1}]_{\indyi\indyl}
\left[
\sum_{\indxrd,\indxre}
\left[
\vphantom{\sum_{\indxrd,\indxre}}
\left(
\frac{\partial P_{\indyl\indyj}}{\partial \xrond_\indxrd}
+
\frac{\partial P_{\indyl\indxrd}}{\partial y_\indyj}
-
\frac{\partial P_{\indxrd\indyj}}{\partial y_\indyl}
\right)
-
\sum_{\indxra}M_{\indxra\indyl}
\left(
\frac{\partial P_{\indxra\indyj}}{\partial \xrond_\indxrd}
+
\frac{\partial P_{\indxra\indxrd}}{\partial y_\indyj}
-
\frac{\partial P_{\indxrd\indyj}}{\partial \xrond_\indxra}
\right)
\right]
M_{\indxrd\indyk}
\right]
$\\\null \   $\displaystyle  \hphantom{=\; }
-
\sum_\indyl [P^{-1}]_{\indyi\indyl}
\left[
\sum_{\indxrc,\indxrd,\indxre}M_{\indxrc\indyj}
\left[
\vphantom{\sum_{\indxrf,\indxra}}
\left(
\frac{\partial P_{\indyl\indxrd}}{\partial \xrond_\indxrc}
+
\frac{\partial P_{\indyl\indxrc}}{\partial \xrond_\indxrd}
-
\frac{\partial P_{\indxrc\indxrd}}{\partial y_\indyl}
\right)
-
\sum_{\indxra}M_{\indxra\indyl}
\left(
\frac{\partial P_{\indxra\indxrd}}{\partial \xrond_\indxrc}
+
\frac{\partial P_{\indxra\indxrc}}{\partial \xrond_\indxrd}
-
\frac{\partial P_{\indxrc\indxrd}}{\partial \xrond_\indxra}
\right)
\right]\right]
M_{\indxrd\indyk}
$\\[1em]
We collect the terms involving a derivative with respect to a $\xrond$ component. Without the factor $\displaystyle \sum_\indyl [P^{-1}]_{\indyi\indyl}$, it is
\\[1em]$\displaystyle  
\mbox{Term}(\partial \xrond)\;=\; 
\sum_{\indxra}M_{\indxra\indyl}
\frac{\partial P_{\indyj\indyk}}{\partial \xrond_\indxra}
$\hfill \null \\\null \   $\displaystyle  \hphantom{=\; }
-
\sum_{\indxrc}M_{\indxrc\indyj}
\left[
\frac{\partial P_{\indyl\indyk}}{\partial \xrond_\indxrc}
-
\sum_{\indxra}M_{\indxra\indyl}
\left(
\frac{\partial P_{\indxra\indyk}}{\partial \xrond_\indxrc}
-
\frac{\partial P_{\indxrc\indyk}}{\partial \xrond_\indxra}
\right)
\right]
$\\\null \   $\displaystyle  \hphantom{=\; }
-
\sum_{\indxrd}
\left[
\frac{\partial P_{\indyl\indyj}}{\partial \xrond_\indxrd}
-
\sum_{\indxra}M_{\indxra\indyl}
\left(
\frac{\partial P_{\indxra\indyj}}{\partial \xrond_\indxrd}
-
\frac{\partial P_{\indxrd\indyj}}{\partial \xrond_\indxra}
\right)
\right]
M_{\indxrd\indyk}
$\\\null \   $\displaystyle  \hphantom{=\; }
+
\sum_{\indxrc,\indxrd}M_{\indxrc\indyj}
\left[
\vphantom{\sum_{\indxrf,\indxra}}
\left(
\frac{\partial P_{\indyl\indxrd}}{\partial \xrond_\indxrc}
+
\frac{\partial P_{\indyl\indxrc}}{\partial \xrond_\indxrd}
\right)
-
\sum_{\indxra}M_{\indxra\indyl}
\left(
\frac{\partial P_{\indxra\indxrd}}{\partial \xrond_\indxrc}
+
\frac{\partial P_{\indxra\indxrc}}{\partial \xrond_\indxrd}
-
\frac{\partial P_{\indxrc\indxrd}}{\partial \xrond_\indxra}
\right)
\right]
M_{\indxrd\indyk}
$\\[1em]
We change the dummy variable to get a derivative with respect to $\xrond_\indxra$
\\[1em]$\displaystyle  
\mbox{Term}(\partial \xrond)\;=\; 
\sum_{\indxra}M_{\indxra\indyl}
\frac{\partial P_{\indyj\indyk}}{\partial \xrond_\indxra}
$\hfill \null \\\null \   $\displaystyle  \hphantom{=\; }
-
\sum_{\indxra}M_{\indxra\indyj}
\frac{\partial P_{\indyl\indyk}}{\partial \xrond_\indxra}
-
\sum_{\indxra}
\frac{\partial P_{\indyl\indyj}}{\partial \xrond_\indxra}
M_{\indxra\indyk}
$\hfill \null \\\null \   $\displaystyle  \hphantom{=\; }
+
\sum_{\indxra}M_{\indxra\indyj}
\sum_{\indxrc}M_{\indxrc\indyl}
\frac{\partial P_{\indxrc\indyk}}{\partial \xrond_\indxra}
-
\sum_{\indxrc}M_{\indxrc\indyj}
\sum_{\indxra}M_{\indxra\indyl}
\frac{\partial P_{\indxrc\indyk}}{\partial \xrond_\indxra}
$\\\null \   $\displaystyle  \hphantom{=\; }
+
\sum_{\indxra}
\sum_{\indxrd}M_{\indxrd\indyl}
\frac{\partial P_{\indxrd\indyj}}{\partial \xrond_\indxra}
M_{\indxra\indyk}
-
\sum_{\indxrd}
\sum_{\indxra}M_{\indxra\indyl}
\frac{\partial P_{\indxrd\indyj}}{\partial \xrond_\indxra}
M_{\indxrd\indyk}
$\\\null \   $\displaystyle  \hphantom{=\; }
+
\sum_{\indxra,\indxrd}M_{\indxra\indyj}
\frac{\partial P_{\indyl\indxrd}}{\partial \xrond_\indxra}
M_{\indxrd\indyk}
+
\sum_{\indxrc,\indxra}M_{\indxrc\indyj}
\frac{\partial P_{\indyl\indxrc}}{\partial \xrond_\indxra}
M_{\indxra\indyk}
$\\\null \   $\displaystyle  \hphantom{=\; }
-
\sum_{\indxra}M_{\indxra\indyj}
\sum_{\indxrd,\indxrc}M_{\indxrc\indyl}
\frac{\partial P_{\indxrc\indxrd}}{\partial \xrond_\indxra}
M_{\indxrd\indyk}
-
\sum_{\indxrc}M_{\indxrc\indyj}
\sum_{\indxrd,\indxra}M_{\indxrd\indyl}
\frac{\partial P_{\indxrd\indxrc}}{\partial \xrond_\indxra}
M_{\indxra\indyk}
+
\sum_{\indxrc}M_{\indxrc\indyj}
\sum_{\indxra,\indxrd}M_{\indxra\indyl}
\frac{\partial P_{\indxrc\indxrd}}{\partial \xrond_\indxra}
M_{\indxrd\indyk}
$\\[1em]

We regroup the terms with the same $M_{\indxra\cdot}$ in factor to obtain
\\[1em]$\displaystyle  
\mbox{Term}(\partial \xrond)
$\hfill \null \\\null \   $\displaystyle =\; 
\sum_{\indxra}M_{\indxra\indyl}
\left[
\frac{\partial P_{\indyj\indyk}}{\partial \xrond_\indxra}
-
\sum_{\indxrc}M_{\indxrc\indyj}
\frac{\partial P_{\indxrc\indyk}}{\partial \xrond_\indxra}
-
\sum_{\indxrd}
\frac{\partial P_{\indxrd\indyj}}{\partial \xrond_\indxra}
M_{\indxrd\indyk}
+
\sum_{\indxrd,\indxrc}M_{\indxrc\indyj}
\frac{\partial P_{\indxrc\indxrd}}{\partial \xrond_\indxra}
M_{\indxrd\indyk}
\right]
$\hfill \null \\\null \   $\displaystyle  \hphantom{=\; }
-
\sum_{\indxra}M_{\indxra\indyj}
\left[
\frac{\partial P_{\indyl\indyk}}{\partial \xrond_\indxra}
-
\sum_{\indxrc}M_{\indxrc\indyl}
\frac{\partial P_{\indxrc\indyk}}{\partial \xrond_\indxra}
-
\sum_{\indxrd}
\frac{\partial P_{\indyl\indxrd}}{\partial \xrond_\indxra}
M_{\indxrd\indyk}
+
\sum_{\indxrd,\indxrc}M_{\indxrc\indyl}
\frac{\partial P_{\indxrc\indxrd}}{\partial \xrond_\indxra}
M_{\indxrd\indyk}
\right]
$\hfill \null \\\null \   $\displaystyle  \hphantom{=\; }
-
\sum_{\indxra}M_{\indxra\indyk}
\left[
\frac{\partial P_{\indyl\indyj}}{\partial \xrond_\indxra}
-
\sum_{\indxrd}M_{\indxrd\indyl}
\frac{\partial P_{\indxrd\indyj}}{\partial \xrond_\indxra}
-
\sum_{\indxrc}M_{\indxrc\indyj}
\frac{\partial P_{\indyl\indxrc}}{\partial \xrond_\indxra}
+
\sum_{\indxrc}M_{\indxrc\indyj}
\sum_{\indxrd}M_{\indxrd\indyl}
\frac{\partial P_{\indxrd\indxrc}}{\partial \xrond_\indxra}
\right]
$\\[1em]
Next we use the following identity written with an appropriate choice of the dummy variables to meet the above expression,
\\[1em]$\displaystyle 
\frac{\partial [P_y]_{\indyj\indyk}}{\partial \xrond_\indxra}\;=\; 
\frac{\partial }{\partial \xrond_\indxra}
\left\{
P_{\indyj\indyk}
-
\sum_{\indxrb,\indxrc}
P_{\indyj\indxrb}
[P_{\xrond\xrond}^{-1}]_{\indxrb\indxrc}
P_{\indxrc\indyk}
\right\}
$\hfill \null \\\null \hfill $\displaystyle 
\renewcommand{\arraystretch}{1.7}
\begin{array}[b]{@{}cl@{}}
=&\displaystyle 
\frac{\partial P_{\indyj\indyk}}{\partial \xrond_\indxra}
-
\sum_{\indxrd,\indxre}
\frac{\partial P_{\indyj\indxrd}}{\partial \xrond_\indxra}
[P_{\xrond\xrond}^{-1}]_{\indxrd\indxre}
P_{\indxre\indyk}
+
\sum_{\indxrb\indxrc\indxrd\indxre}
P_{\indyj\indxrb}
[P_{\xrond\xrond}^{-1}]_{\indxrb\indxrc}
\frac{\partial P_{\indxrc\indxrd}}{\partial \xrond_\indxra}
[P_{\xrond\xrond}^{-1}]_{\indxrd\indxre}
P_{\indxre\indyk}
-
\sum_{\indxrb\indxrc}
P_{\indyj\indxrb}
[P_{\xrond\xrond}^{-1}]_{\indxrb\indxrc}
\frac{\partial P_{\indxrc\indyk}}{\partial \xrond_\indxra}
\\
=&\displaystyle 
\frac{\partial P_{\indyj\indyk}}{\partial \xrond_\indxra}
-
\sum_{\indxrd}
\frac{\partial P_{\indyj\indxrd}}{\partial \xrond_\indxra}
M_{\indxrd\indyk}
+
\sum_{\indxrc,\indxrd}
M_{\indxrc\indyj}
\frac{\partial P_{\indxrc\indxrd}}{\partial \xrond_\indxra}
M_{\indxrd\indyk}
-
\sum_{\indxrb\indxrc}
M_{\indxrc\indyj}
\frac{\partial P_{\indxrc\indyk}}{\partial \xrond_\indxra}
\end{array}
$\refstepcounter{equation}\label{LP202} $(\theequation)$
\\[1em]
We conclude, with the symmetry of $P$,
$$
\mbox{Term}(\partial \xrond)\;=\; 
\sum_{\indxra}
\frac{\partial [P _y]_{\indyj\indyk}}{\partial \xrond_\indxra}
M_{\indxra\indyl} 
-
\sum_{\indxra}
\frac{\partial [P_y]_{\indyl\indyk}}{\partial \xrond_\indxra}
M_{\indxra\indyj}
-\sum_{\indxra}
\frac{\partial [P _y]_{\indyl\indyj}}{\partial \xrond_\indxra}
M_{\indxra\indyk}
$$

We proceed as above with collecting the terms involving a derivative with respect to a $y$ component. This 
gives
\\[1em]$\displaystyle  
\mbox{Term}(\partial y)
\renewcommand{\arraystretch}{1.7}
\begin{array}[t]{@{}rcl@{}}
&=&\displaystyle 
\left[
\left(
\frac{\partial P_{\indyl\indyk}}{\partial y_\indyj}
+
\frac{\partial P_{\indyl\indyj}}{\partial y_\indyk}
-
\frac{\partial P_{\indyj\indyk}}{\partial y_\indyl}
\right)
-
\sum_{\indxra}M_{\indxra\indyl}
\left(
\frac{\partial P_{\indxra\indyk}}{\partial y_\indyj}
+
\frac{\partial P_{\indxra\indyj}}{\partial y_\indyk}
\right)
\right]
\\
&&\displaystyle 
-
\sum_{\indxrc}M_{\indxrc\indyj}
\left[
\vphantom{\sum_{\indxrd,\indxre}}
\left(
\frac{\partial P_{\indyl\indxrc}}{\partial y_\indyk}
-
\frac{\partial P_{\indxrc\indyk}}{\partial y_\indyl}
\right)
-
\sum_{\indxra}M_{\indxra\indyl}
\frac{\partial P_{\indxra\indxrc}}{\partial y_\indyk}
\right]
\\
&&\displaystyle 
-
\sum_{\indxrd,\indxre}
\left[
\vphantom{\sum_{\indxrd,\indxre}}
\left(
\frac{\partial P_{\indyl\indxrd}}{\partial y_\indyj}
-
\frac{\partial P_{\indxrd\indyj}}{\partial y_\indyl}
\right)
-
\sum_{\indxra}M_{\indxra\indyl}
\frac{\partial P_{\indxra\indxrd}}{\partial y_\indyj}
\right]
M_{\indxrd\indyk}
\\
&&\displaystyle 
-
\sum_{\indxrc,\indxrd}M_{\indxrc\indyj}
\frac{\partial P_{\indxrc\indxrd}}{\partial y_\indyl}
M_{\indxrd\indyk}
\\[1em]
&=&\displaystyle 
\frac{\partial P_{\indyl\indyk}}{\partial y_\indyj}
-
\sum_{\indxra}M_{\indxra\indyl}
\frac{\partial P_{\indxra\indyk}}{\partial y_\indyj}
-
\sum_{\indxrd,\indxre}
\frac{\partial P_{\indyl\indxrd}}{\partial y_\indyj}
M_{\indxrd\indyk}
+
\sum_{\indxra,\indxrd}M_{\indxra\indyl}
\frac{\partial P_{\indxra\indxrd}}{\partial y_\indyj}
M_{\indxrd\indyk}
\\
&&\displaystyle 
+\; 
\frac{\partial P_{\indyl\indyj}}{\partial y_\indyk}
-
\sum_{\indxra}M_{\indxra\indyl}
\frac{\partial P_{\indxra\indyj}}{\partial y_\indyk}
-
\sum_{\indxrc}M_{\indxrc\indyj}
\frac{\partial P_{\indyl\indxrc}}{\partial y_\indyk}
+
\sum_{\indxrc,\indxra}M_{\indxrc\indyj}
M_{\indxra\indyl}
\frac{\partial P_{\indxra\indxrc}}{\partial y_\indyk}
\\
&&\displaystyle
-\; 
\frac{\partial P_{\indyj\indyk}}{\partial y_\indyl}
+
\sum_{\indxrc}M_{\indxrc\indyj}
\frac{\partial P_{\indxrc\indyk}}{\partial y_\indyl}
+
\sum_{\indxrd,\indxre}
\frac{\partial P_{\indxrd\indyj}}{\partial y_\indyl}
M_{\indxrd\indyk}
-
\sum_{\indxrc,\indxrd}M_{\indxrc\indyj}
\frac{\partial P_{\indxrc\indxrd}}{\partial y_\indyl}
M_{\indxrd\indyk}

\\[1em]
&=&\displaystyle 
\frac{\partial [P_y]_{\indyl\indyk}}{\partial y_\indyj}
+
\frac{\partial [P_y]_{\indyl\indyj}}{\partial y_\indyk}
-
\frac{\partial [P_y]_{\indyj\indyk}}{\partial y_\indyl}
\end{array}
$\\[1em]
With all this and the identity
$$
 [P^{-1}]_{\indyi\indyl}\;=\; [P_y^{-1}]_{\indyi\indyl}
$$
we have obtained
\\[1em]$\displaystyle 
2\left[\secff _P h^{\ortho,\ortho}\right]_{\indyj\indyk}^\indyi
$\hfill \null \\\null \   $\displaystyle  =\; 
\sum_\indyl \Py ^{-1}_{\indyi\indyl}
\left(
\frac{\partial \Py _{\indyl\indyk}}{\partial y_\indyj}
+
\frac{\partial  \Py _{\indyl\indyj}}{\partial y_\indyk}
-
\frac{\partial  \Py _{\indyj\indyk}}{\partial y_\indyl}
\right)
$\hfill \null \\\null \   $\displaystyle  \hphantom{=\; }
-
\sum_\indyl 
 [P_y^{-1}]_{\indyi\indyl}
\left(
\frac{\partial [P_y]_{\indyl\indyk}}{\partial y_\indyj}
+
\frac{\partial [P_y]_{\indyl\indyj}}{\partial y_\indyk}
-
\frac{\partial [P_y]_{\indyj\indyk}}{\partial y_\indyl}
\right)
$\hfill \null \\\null \   $\displaystyle  \hphantom{=\; }
-
\sum_\indyl 
 [P_y^{-1}]_{\indyi\indyl}
\left(
\sum_{\indxra,\indxrf}
\frac{\partial [P _y]_{\indyj\indyk}}{\partial \xrond_\indxra}
[P_{\xrond\xrond}^{-1}]_{\indxra\indxrf}
P_{\indxrf\indyl} 
-
\sum_{\indxra,\indxrf}
\frac{\partial [P_y]_{\indyl\indyk}}{\partial \xrond_\indxra}
[P_{\xrond\xrond}^{-1}]_{\indxra\indxrf}P_{\indxrf\indyj}
-\sum_{\indxra,\indxrf}
\frac{\partial [P _y]_{\indyl\indyj}}{\partial \xrond_\indxra}
[P_{\xrond\xrond}^{-1}]_{\indxra\indxrf}P_{\indxrf\indyk}
\right)
$
\sousection{Proof of Lemma \ref{lem11}}
\label{complement47}
Let $s_4$ in $(s_1,s_2)$ be such that
$$
\frac{d\bfgamma}{ds}(s_4) \in \bfDistrib _\bfP ^\ortho(\bfgamma(s_4))
$$
Let $\coordyxr$ be a coordinate chart around $\bfgamma(s_4)$. Then, with $P$ and $(\gamma _y,\gamma _\xrond)$ 
being the corresponding expression of the metric and the geodesic, from Lemma \ref{lem6}, there exists a 
vector $w_4$ in $\RR^p$ satisfying
$$
\left(\begin{array}{@{\,  }c@{\,  }}
\displaystyle 
\frac{d\gamma _y}{ds}(s_4)
\\[0.5em]\displaystyle 
\frac{d\gamma _\xrond}{ds}(s_4)
\end{array}\right)\;=\; 
\left(\begin{array}{@{\,  }c@{\,  }}
I_p
\\
P_{\xrond\xrond}(\gamma _y(s_4),\gamma _{\xrond}(s_4))^{-1}
P_{\xrond y}(\gamma _y(s_4),\gamma _{\xrond}(s_4))^{-1}
\end{array}\right)
w_4
$$
Also if, for some $s$ in $(s_1,s_2)$,
there exists a 
vector $w$ in $\RR^p$ satisfying
$$
\left(\begin{array}{@{\,  }c@{\,  }}
\displaystyle
\frac{d\gamma _y}{ds}(s)
\\[0.5em]\displaystyle 
\frac{d\gamma _\xrond}{ds}(s)
\end{array}\right)\;=\; 
\left(\begin{array}{@{\,  }c@{\,  }}
I_p
\\
P_{\xrond\xrond}(\gamma _y(s),\gamma _{\xrond}(s))^{-1}
P_{\xrond y}(\gamma _y(s),\gamma _{\xrond}(s))
\end{array}\right)
w
$$
then we have
$$
\frac{d\gamma}{ds}(s) \in \Distrib _P^\ortho(\gamma(s))
$$
So our result holds if we can show that we have
$$
\frac{d}{ds}\left\{
P_{\xrond\xrond}(\gamma _y(s),\gamma _{\xrond}(s)) 
\frac{d\gamma _\xrond}{ds}(s)
+
P_{\xrond y}(\gamma _y(s),\gamma _{\xrond}(s))
\frac{d\gamma _y}{ds}(s)
\right\}\;=\; 0
\qquad \forall s\in (s_1,s_2)
$$
But from the Euler-Lagrange form of the geodesic equation (\ref{14}), we get
\\[1em]$\displaystyle 
2\frac{d}{ds}\left\{
P_{\xrond\xrond}(\gamma _y(s),\gamma _{\xrond}(s)) 
\frac{d\gamma _\xrond}{ds}(s)
+
P_{\xrond y}(\gamma _y(s),\gamma _{\xrond}(s))
\frac{d\gamma _y}{ds}(s)
\right\}
$\hfill \null \\\null \hfill $\displaystyle
\begin{array}{@{}cl@{}r@{}}
=&\displaystyle 
\left.
\frac{\partial }{\partial \xrond}
\left\{
\left(\begin{array}{@{\,  }cc@{\,  }}
\displaystyle 
\frac{d\gamma _y}{ds}(s)^\top
&\displaystyle 
\frac{d\gamma _\xrond}{ds}(s)^\top
\end{array}\right)
\left(\begin{array}{@{\,  }cc@{\,  }}
P_{yy}(\gamma _y(s),\xrond) 
&
P_{y\xrond}(\gamma _y(s),\xrond) 
\\
P_{\xrond y}(\gamma _y(s),\xrond) 
&
P_{\xrond\xrond}(\gamma _y(s),\xrond) 
\end{array}\right)
\left(\begin{array}{@{\,  }c@{\,  }}
\displaystyle
\frac{d\gamma _y}{ds}(s)
\\\displaystyle 
\frac{d\gamma _\xrond}{ds}(s)
\end{array}\right)
\right\}
\right|_{\xrond=\gamma _\xrond(s)}
\\[1em]
=&\displaystyle 
\left.
\frac{\partial }{\partial \xrond}
\left\{
\frac{d\gamma _y}{ds}(s)^\top
\left(\begin{array}{@{\,  }cc@{\,  }}
\displaystyle 
I_p
&\displaystyle 
-P_{y \xrond }(\gamma _y(s),\gamma _{\xrond}(s))
P_{\xrond\xrond}(\gamma _y(s),\gamma _{\xrond}(s))^{-1}
\end{array}\right) \times
\right.
\right.
\\&\multicolumn{1}{c}{%
\displaystyle 
\times
\left(\begin{array}{@{\,  }cc@{\,  }}
P_{yy}(\gamma _y(s),\xrond) 
&
P_{y\xrond}(\gamma _y(s),\xrond) 
\\
P_{\xrond y}(\gamma _y(s),\xrond) 
&
P_{\xrond\xrond}(\gamma _y(s),\xrond) 
\end{array}\right)
\times
}
\\&\multicolumn{1}{r}{%
\displaystyle 
\left.\left.
\times
\left(\begin{array}{@{\,  }c@{\,  }}
\displaystyle 
I_p
\\
\displaystyle 
-
P_{\xrond\xrond}(\gamma _y(s),\gamma _{\xrond}(s))^{-1}
P_{\xrond y}(\gamma _y(s),\gamma _{\xrond}(s))
\end{array}\right)
\frac{d\gamma _y}{ds}(s)
\right\}
\right|_{\xrond=\gamma _\xrond(s)}
}
\\[1.5em]
&\displaystyle 
\;-\; 
\left.
\frac{\partial }{\partial \xrond}
\left\{
\left(
P_{\xrond\xrond}(\gamma _y(s),\gamma _{\xrond}(s)) 
\frac{d\gamma _\xrond}{ds}(s)
+
P_{\xrond y}(\gamma _y(s),\gamma _{\xrond}(s))
\frac{d\gamma _y}{ds}(s)
\right)^\top
\times
\right.
\right.
&
\refstepcounter{equation}\label{LP209}(\theequation)
\\&\multicolumn{1}{c}{%
\times
P_{\xrond\xrond}(\gamma _y(s),\xrond) ^{-1}
\times
}
\\&\multicolumn{1}{r}{%
\displaystyle 
\left.\left.
\times
\left(
P_{\xrond\xrond}(\gamma _y(s),\gamma _{\xrond}(s)) 
\frac{d\gamma _\xrond}{ds}(s)
+
P_{\xrond y}(\gamma _y(s),\gamma _{\xrond}(s))
\frac{d\gamma _y}{ds}(s)
\right)
\right\}
\right|_{\xrond=\gamma _\xrond(s)}
}
\end{array}
$\\[1em]
Fortunately, we have
\\[1em]$\displaystyle 
\frac{\partial P_y}{\partial \xrond_\indxra}(y,\xrond)
\;=\; 
\left(\begin{array}{@{\,  }cc@{\,  }}
\displaystyle 
I_p
&\displaystyle 
-P_{y \xrond }(y,z)
P_{\xrond\xrond}(y,z)^{-1}
\end{array}\right) 
\left(\begin{array}{@{\,  }cc@{\,  }}
\displaystyle 
\frac{\partial P_{yy}}{\partial \xrond_\indxra}(y,\xrond) 
&\displaystyle 
\frac{\partial P_{y\xrond}}{\partial \xrond_\indxra}(y,\xrond) 
\\[0.5em]\displaystyle 
\frac{\partial P_{\xrond y}}{\partial \xrond_\indxra}(y,\xrond) 
&\displaystyle 
\frac{\partial P_{\xrond\xrond}}{\partial \xrond_\indxra}(y,\xrond) 
\end{array}\right)
\left(\begin{array}{@{\,  }c@{\,  }}
\displaystyle 
I_p
\\
\displaystyle 
-
P_{\xrond\xrond}(y,z)^{-1}
P_{\xrond y}(y,z)
\end{array}\right)
$\\[1em]
with the notation (see (\ref{LP199}))
$$
P_y(y,\xrond)\;=\; P_{yy}(y,\xrond)\;-\; P_{y\xrond}(y,\xrond)P_{\xrond\xrond}(y,\xrond)^{-1}P_{\xrond y}(y,\xrond)
\  .
$$
Since $\bfh$ is a Riemannian submersion, $P_y$ does not depend on $\xrond$. The first claim follows from the 
equation
$$
\frac{dT}{ds}(s)\;=\; 
\;-\; 
\left.\frac{\partial }{\partial \xrond}
\left\{
T(s)^\top
P_{\xrond\xrond}(\gamma _y(s),\xrond) ^{-1}
T(s)
\right\}
\right|_{\xrond=\gamma _\xrond(s)}
$$
obtained from (\ref{LP209}) by letting
$$
T(s)  \;=\; \left(
P_{\xrond\xrond}(\gamma _y(s),\gamma _{\xrond}(s)) 
\frac{d\gamma _\xrond}{ds}(s)
+
P_{\xrond y}(\gamma _y(s),\gamma _{\xrond}(s))
\frac{d\gamma _y}{ds}(s)
\right)
$$

For the second claim, the Euler-Lagrange form of the geodesic equation (\ref{14}) gives also
\\[1em]$\displaystyle 
2\frac{d}{ds}\left\{
P_{y\xrond}(\gamma _y(s),\gamma _{\xrond}(s)) 
\frac{d\gamma _\xrond}{ds}(s)
+
P_{y y}(\gamma _y(s),\gamma _{\xrond}(s))
\frac{d\gamma _y}{ds}(s)
\right\}
$\hfill \null \\\null \hfill $\displaystyle
\begin{array}{@{}cl@{}r@{}}
=&\displaystyle 
\left.
\frac{\partial }{\partial y}
\left\{
\left(\begin{array}{@{\,  }cc@{\,  }}
\displaystyle 
\frac{d\gamma _y}{ds}(s)^\top
&\displaystyle 
\frac{d\gamma _\xrond}{ds}(s)^\top
\end{array}\right)
\left(\begin{array}{@{\,  }cc@{\,  }}
P_{yy}(y,\gamma _\xrond(s)) 
&
P_{y\xrond}(y,\gamma _\xrond(s)) 
\\
P_{\xrond y}(y,\gamma _\xrond(s)) 
&
P_{\xrond\xrond}(y,\gamma _\xrond(s)) 
\end{array}\right)
\left(\begin{array}{@{\,  }c@{\,  }}
\displaystyle
\frac{d\gamma _y}{ds}(s)
\\\displaystyle 
\frac{d\gamma _\xrond}{ds}(s)
\end{array}\right)
\right\}
\right|_{y=\gamma _y(s)}
\\[1em]
\end{array}
$\\[1em]
With what we have above, this reduces to
$$
2 \frac{d}{ds}\left\{
P_y(\gamma _y(s))\frac{d\gamma _y}{ds}(s)
\right\}
\;=\; 
\left.
\frac{\partial }{\partial y}
\left\{
\frac{d\gamma _y}{ds}(s)^\top 
P_y(y)
\frac{d\gamma _y}{ds}(s)
\right\}
\right|_{y=\gamma _y(s)}
$$
Since $\bfh$ is a Riemannian submersion, we have
$$
P_y(y)\;=\; \Py(y)
$$
and the above equation is nothing but the geodesic equation in the $\bfy$-manifold.
%
%

%
%
\sousection{{Details on (\ref{LP165})}}
\label{complement38}
$\displaystyle 
\frac{\partial \changex}{\partial x}(x)^\top
\bar P_{mod}(\bar x)
\frac{\partial \changex}{\partial x}(x)
$\hfill \null \\\null \hfill $
\renewcommand{\arraystretch}{2}
\begin{array}{cl@{}}
=&\displaystyle 
\frac{\partial \changex}{\partial x}(x)^\top
\left(\bar P(\bar x) +
\frac{\partial \bar h}{\partial \bar x}(\bar x)^\top  
\left[\bar \Py (\bar h(\bar x))
-
\left(\frac{\partial \bar h}{\partial \bar x}(\bar x) \bar P(\bar x)^{-1}\frac{\partial \bar h}{\partial \bar 
x}(\bar x)^\top\right)^{-1}
\right]
\frac{\partial \bar h}{\partial \bar x}(\bar x)
\right)
\frac{\partial \changex}{\partial x}(x)
\\
=&\displaystyle 
P(x) + \frac{\partial h}{\partial x}(x)^\top \frac{\partial \changey }{\partial y}(h(x))^\top
\times
\\&\multicolumn{1}{c}{\displaystyle 
\times
\left[\bar \Py (\bar h(\bar x))
-
\left(\frac{\partial \bar h}{\partial \bar x}(\bar x) 
\frac{\partial \changex}{\partial x}(x) P( x)^{-1}\frac{\partial \changex}{\partial x}(x) ^\top\frac{\partial \bar h}{\partial \bar x}(\bar x)^\top
\right)^{-1}\right]
\times
}
\\&\multicolumn{1}{r}{\displaystyle 
\times
\frac{\partial \changey }{\partial y}(h(x))\frac{\partial h}{\partial x}(x)
}
\\
=&\displaystyle 
P(x) + 
\frac{\partial h}{\partial x}(x)^\top 
 \Py ( h( x))
\frac{\partial h}{\partial x}(x)
\\
\multicolumn{2}{@{}r@{}}{\displaystyle %
\qquad \qquad 
-\frac{\partial h}{\partial x}(x)^\top \frac{\partial \changey }{\partial y}(h(x))^\top
\left(
\frac{\partial \changey }{\partial y}(h(x))\frac{\partial h}{\partial x}(x)
P( x)^{-1}
\frac{\partial h}{\partial x}(x)^\top \frac{\partial \changey }{\partial y}(h(x))^\top\right)^{-1}
\frac{\partial \changey }{\partial y}(h(x))\frac{\partial h}{\partial x}(x)
}
\\
=&\displaystyle 
P_{mod}(x)
\end{array}
$
%
%

%
%
\sousection{Direct proof of the property \ref{point4} of Lemma \ref{lem9}}
\label{complement28}
The Hessian is
$$
\left[\Hess _P h \right]_{\indxa\indxb}^\indyi=
\frac{\partial ^2h_\indyi}{\partial x_\indxa\partial x_\indxb}
-
\frac{1}{2}\sum_{\indxc,\indxd}
\left(\frac{\partial P_{\indxa\indxd}}{\partial x_\indxb}
+
\frac{\partial P_{\indxb\indxd}}{\partial x_\indxa}
-
\frac{\partial P_{\indxa\indxb}}{\partial x_\indxd}
\right)
\left[P^{-1}\right]_{\indxc\indxd}
\frac{\partial h_\indyi}{\partial x_\indxc}
\  .
$$
 We note that, with compact notations and
$$
S(x)\;=\; \Py (h(x)) -
\left(\frac{\partial h}{\partial x}(x) P(x)^{-1}\frac{\partial h}{\partial x}(x)^\top
\right)^{-1}
\  ,
$$
 we  have
\\[1em]\vbox{\noindent$\displaystyle 
P_{mod}(x)^{-1}\frac{\partial h}{\partial x}(x)^\top\;=\; 
\left(P(x) + \frac{\partial h}{\partial x}(x)^\top  S(x)\,  \frac{\partial h}{\partial 
x}(x)\right)^{-1}\frac{\partial h}{\partial x}(x)^\top
$\hfill \null \\\null \hfill $\displaystyle 
\begin{array}[b]{@{}c@{\; }l@{}}
=&\displaystyle 
P(x)^{-1}\frac{\partial h}{\partial x}(x)^\top
\left(I_p+
\left(S(x)^{-1}+\frac{\partial h}{\partial x}(x)P(x)^{-1}\frac{\partial h}{\partial x}(x)^\top\right)^{-1}
\frac{\partial h}{\partial x}(x) P(x)^{-1}\,  \frac{\partial h}{\partial x}(x)^\top\right)^{-1}%
\hskip -1em\null 
\\ 
=&\displaystyle
P(x)^{-1}\frac{\partial h}{\partial x}(x)^\top\left(I_p+S(x)\frac{\partial h}{\partial 
x}(x)P(x)^{-1}\frac{\partial h}{\partial x}(x)^\top\right)^{-1}
\end{array}
$}
\\[1em]
Expanded with indices, this equation is
$$
\sum_\indxd
\left[P_{mod}^{-1}\right]_{\indxc\indxd}\frac{\partial h_\indyi}{\partial x_\indxd}
\;=\; 
\sum_{\indxd,\indyj}
\left[P^{-1}\right]_{\indxc\indxd}\frac{\partial h_\indyj}{\partial x_\indxd}
\left[
\left(I_p+S\frac{\partial h}{\partial x}P^{-1}\frac{\partial h}{\partial x}^\top\right)^{-1}
\right]_{\indyj\indyi}
$$
On another hand, from
$$
{P_{mod}}_{\,  \indxa\indxb}\;=\; P_{\indxa\indxb}
+ \sum_{\indyk,\indyl} \frac{\partial h_\indyk}{\partial x_\indxa}S_{\indyk\indyl} \frac{\partial h_\indyl}{\partial x_\indxb}
$$
we obtain
\\[1em]\vbox{\noindent$\displaystyle 
\sum_{\indxa,\indxb}v_\indxa v_\indxb\left(
\frac{\partial {P_{mod}}_{\,\indxa\indxc}}{\partial x_\indxb}
+
\frac{\partial {P_{mod}}_{\,\indxb\indxc}}{\partial x_\indxa}
-
\frac{\partial {P_{mod}}_{\,\indxa\indxb}}{\partial x_\indxc}
\right)
$\hfill \null \\\null \qquad \qquad  $\displaystyle \;=\; 
\sum_{\indxa,\indxb}v_\indxa v_\indxb\left(
\frac{\partial P_{\indxa\indxc}}{\partial x_\indxb}
+
\frac{\partial P_{\indxb\indxc}}{\partial x_\indxa}
-
\frac{\partial P_{\indxa\indxb}}{\partial x_\indxc}
\right)
$\hfill \null \\\null \hfill $\displaystyle 
\;+\; 
\sum_{\indxa,\indxb,\indyk,\indyl}v_\indxa v_\indxb
\left(
\frac{\partial }{\partial x_\indxb}
\left(
\frac{\partial h_\indyk}{\partial x_\indxa}S_{\indyk\indyl} \frac{\partial h_\indyl}{\partial x_\indxc}
\right)
+
\frac{\partial }{\partial x_\indxa}
\left(
\frac{\partial h_\indyk}{\partial x_\indxb}S_{\indyk\indyl} \frac{\partial h_\indyl}{\partial x_\indxc}
\right)
-
\frac{\partial }{\partial x_\indxc}
\left(
\frac{\partial h_\indyk}{\partial x_\indxa}S_{\indyk\indyl} \frac{\partial h_\indyl}{\partial x_\indxb}
\right)
\right)
$}\\[1em]
But when
\begin{equation}
\label{trav26}
\sum_\indxd \frac{\partial h_\indyk}{\partial x_\indxd} v_\indxd\;=\; 0
\qquad \forall \indyk
\end{equation}
we have
\\[1em]\vbox{\noindent$\displaystyle 
\sum_{\indxa,\indxb,\indyk,\indyl}v_\indxa v_\indxb
\left(
\frac{\partial }{\partial x_\indxb}
\left(
\frac{\partial h_\indyk}{\partial x_\indxa}S_{\indyk\indyl} \frac{\partial h_\indyl}{\partial x_\indxc}
\right)
+
\frac{\partial }{\partial x_\indxa}
\left(
\frac{\partial h_\indyk}{\partial x_\indxb}S_{\indyk\indyl} \frac{\partial h_\indyl}{\partial x_\indxc}
\right)
-
\frac{\partial }{\partial x_\indxc}
\left(
\frac{\partial h_\indyk}{\partial x_\indxa}S_{\indyk\indyl} \frac{\partial h_\indyl}{\partial x_\indxb}
\right)
\right)
$\hfill \null \\\null \hfill $\displaystyle 
\;=\; 2\sum_{\indxa,\indxb,\indyk,\indyl}v_\indxa v_\indxb
\frac{\partial ^2h_\indyk}{\partial x_\indxa\partial x_\indxb}S_{\indyk\indyl} \frac{\partial h_\indyl}{\partial x_\indxc}
\  .
$}\\[1em]
This yields
\\[1em]\vbox{\noindent$\displaystyle 
\sum_{\indxa,\indxb}v_\indxa v_\indxb\left(
\frac{\partial {P_{mod}}_{\,\indxa \indxc}}{\partial x_\indxb}
+
\frac{\partial {P_{mod}}_{\,\indxb\indxc}}{\partial x_\indxa}
-
\frac{\partial {P_{mod}}_{\,\indxa\indxb}}{\partial x_\indxc}
\right)
$\hfill \null \\\null \qquad \qquad  $\displaystyle \;=\; 
\sum_{\indxa\indxb}v_\indxa v_\indxb
\left(\frac{\partial P_{\indxa\indxc}}{\partial x_\indxb}
+
\frac{\partial P_{\indxb\indxc}}{\partial x_\indxa}
-
\frac{\partial P_{\indxa\indxb}}{\partial x_\indxc}
+
2\sum_{\indyk,\indyl}
\frac{\partial ^2h_\indyk}{\partial x_ \indxa \partial x_ \indxb }S_{\indyk\indyl} \frac{\partial h_\indyl}{\partial x_ \indxc }
\right)
$}\\[1em]
and therefore
\\[1em]\vbox{\noindent$\displaystyle 
\sum_{\indxc,\indxd}\sum_{\indxa,\indxb}v_\indxa v_\indxb\left(\frac{\partial {P_{mod}}_{\,\indxa\indxc}}{\partial x_ \indxb }
+
\frac{\partial {P_{mod}}_{\,\indxb\indxc}}{\partial x_ \indxa }
-
\frac{\partial {P_{mod}}_{\,\indxa\indxb}}{\partial x_ \indxc }
\right)\left[P_{mod}^{-1}\right]_{\indxc\indxd}
\frac{\partial h_\indyi}{\partial x_\indxd }
$\hfill \null \\\null \qquad   $\displaystyle \;=\; 
\sum_{\indxc,\indxd,\indyj}
\left[\sum_{\indxa,\indxb}v_\indxa v_\indxb
\left(\frac{\partial P_{\indxa\indxc}}{\partial x_ \indxb }
+
\frac{\partial P_{\indxb\indxc}}{\partial x_ \indxa }
-
\frac{\partial P_{\indxa\indxb}}{\partial x_ \indxc }
+
2\sum_{\indyk,\indyl}
\frac{\partial ^2h_\indyk}{\partial x_ \indxa \partial x_ \indxb }S_{\indyk\indyl} \frac{\partial h_\indyl}{\partial x_ \indxc }
\right)
\right]\times
$\hfill \null \\\null \hfill $\displaystyle 
\times
\left[P^{-1}\right]_{\indxc\indxd}\frac{\partial h_\indyj}{\partial x_\indxd }
\left[
\left(I_p+S \frac{\partial h}{\partial x}P^{-1}\frac{\partial h}{\partial x}^\top\right)^{-1}
\right]_{\indyj\indyi}
$}\\[1em]
But, when the property (\ref{LP151}) holds, (\ref{trav26}) implies
$$
\sum_{\indxa,\indxb}
v_\indxa v_\indxb
\frac{1}{2}\sum_{\indxc,\indxd}
\left(\frac{\partial P_{\indxa\indxd}}{\partial x_\indxb}
+
\frac{\partial P_{\indxb\indxd}}{\partial x_\indxa}
-
\frac{\partial P_{\indxa\indxb}}{\partial x_\indxd}
\right)
\left[P^{-1}\right]_{\indxc\indxd}
\frac{\partial h_\indyj}{\partial x_\indxc}
\;=\; 
\sum_{\indxa,\indxb}
v_\indxa v_\indxb
\frac{\partial ^2h_\indyj}{\partial x_\indxa\partial x_\indxb}
\  .
$$
So we have
\\[1em]\vbox{\noindent$\displaystyle 
\sum_{\indxc\indxd}\sum_{\indxa\indxb}v_\indxa v_\indxb\left(\frac{\partial {P_{mod}}_{\,\indxa\indxc}}{\partial x_ \indxb }
+
\frac{\partial {P_{mod}}_{\,\indxb\indxc}}{\partial x_ \indxa }
-
\frac{\partial {P_{mod}}_{\,\indxa\indxb}}{\partial x_ \indxc }
\right)\left[P_{mod}^{-1}\right]_{\indxc\indxd}
\frac{\partial h_\indyi}{\partial x_\indxd }
$\hfill \null \\\null \qquad \qquad  $\displaystyle \;=\; 
2\sum_{\indyj}
\left[\sum_{\indxa\indxb}v_\indxa v_\indxb
\left(
\frac{\partial ^2h_\indyj}{\partial x_ \indxa \partial x_ \indxb }
+
\sum_{\indxc,\indxd,\indyk,\indyl}
\frac{\partial ^2h_\indyk}{\partial x_ \indxa \partial x_ \indxb }S_{\indyk\indyl} \frac{\partial h_\indyl}{\partial x_ \indxc }
\left[P^{-1}\right]_{\indxc\indxd}\frac{\partial h_\indyj}{\partial x_\indxd }\right)
\right]\times
$\hfill \null \\\null \hfill $\displaystyle 
\times
\left[
\left(I_p+S \frac{\partial h}{\partial x}P^{-1}\frac{\partial h}{\partial x}^\top\right)^{-1}
\right]_{\indyj\indyi}
$}\\\vbox{\noindent
\null \qquad \qquad  $\displaystyle \;=\; 
2\sum_{\indyj}
\left[\sum_{\indxa,\indxb}v_\indxa v_\indxb
\left(
\frac{\partial ^2h_\indyj}{\partial x_ \indxa \partial x_ \indxb }
+
\sum_{\indyk}
\frac{\partial ^2h_\indyk}{\partial x_ \indxa \partial x_ \indxb }
\left[S \frac{\partial h}{\partial x}P^{-1}\frac{\partial h}{\partial x}\right]_{\indyk\indyj}\right)
\right]\times
$\hfill \null \\\null \hfill $\displaystyle 
\times
\left[
\left(I_p+S \frac{\partial h}{\partial x}P^{-1}\frac{\partial h}{\partial x}^\top\right)^{-1}
\right]_{\indyj\indyi}
$}\\\vbox{\noindent
\null \qquad \qquad  $\displaystyle \;=\; 
2\sum_{\indxa,\indxb}v_\indxa v_\indxb
\sum_{\indyk}
\frac{\partial ^2h_\indyk}{\partial x_ \indxa \partial x_ \indxb }
\sum_{\indyj}\left[
I_p+S \frac{\partial h}{\partial x}P^{-1}\frac{\partial h}{\partial x}^\top
\right]_{\indyk\indyj}
\left[
\left(I_p+S \frac{\partial h}{\partial x}P^{-1}\frac{\partial h}{\partial x}^\top\right)^{-1}
\right]_{\indyj\indyi}
$}\\\null \qquad \qquad  $\displaystyle \;=\; 
2\sum_{\indxa,\indxb}v_\indxa v_\indxb
\frac{\partial ^2h_\indyi}{\partial x_ \indxa \partial x_ \indxb }
$\\[1em]
This establishes that, if we have
$$
v^\top \Hess_Ph (x) v\;=\; 0\qquad \forall (x,v):\,  L_vh(x)\;=\; 0
\  ,
$$
then we have also~:
$$
v^\top \Hess_{P_{mod}}h (x) v\;=\; 0\qquad \forall (x,v):\,  L_vh(x)\;=\; 0
$$
%
%

%
%
\sousection{Proof of Lemma \ref{lem3}}
\label{complement40}
\subsubsection{A first necessary and sufficient condition for the integrability of the orthogonal distribution.}
\label{complement48}
\begin{lemma}
The orthogonal distribution $\bfDistrib _\bfP^\ortho$ is integrable on $\Ouv$ if and only if, for any $\bfx_0$ in $\Ouv$, 
there exists a coordinate chart 
$\coordx$ around $\bfx_0$ such that the expressions of $\bfh$ and $\bfP$ in these coordinates satisfy the following 
(coordinate independent) condition, for all 
$\coordxp$ in $\coordxm(\coordxd \cap \Ouv)$ and all $\indyi$ and $\indyj$ in $\{1,\ldots,p\}$,
\IfReport{%
\begin{equation}
\label{LP191}
0= 
\left(I_n - P(x)^{-1}\frac{\partial h}{\partial x}(x)^\top 
\left(\frac{\partial h}{\partial x}(x) 
P(x)^{-1}\frac{\partial h}{\partial x}(x)^\top
\right)^{-1}
\frac{\partial h}{\partial x}(x)\right)
\left[P(x)^{-1}
\frac{\partial h_\indyi}{\partial x}(x)^\top
 ,\,  P(x)^{-1}\frac{\partial 
h_\indyj}{\partial x}(x)^\top\right]
\end{equation} 
}{%
\\[1em]$\displaystyle 
0\;=\; 
$\refstepcounter{equation}\label{LP191}\hfill$(\theequation)$
\\\null \hfill $\displaystyle 
\left(I_n - P(x)^{-1}\frac{\partial h}{\partial x}(x)^\top 
\left(\frac{\partial h}{\partial x}(x) 
P(x)^{-1}\frac{\partial h}{\partial x}(x)^\top
\right)^{-1}
\frac{\partial h}{\partial x}(x)\right)
\left[P(x)^{-1}
\frac{\partial h_\indyi}{\partial x}(x)^\top
 ,\,  P(x)^{-1}\frac{\partial 
h_\indyj}{\partial x}(x)^\top\right]
$\\[1em]
}
where $[u,v]$ denotes the Lie bracket of $u$ and $v$
\end{lemma}

\begin{proof}
From Frobenius Theorem, the orthogonal distribution $\bfDistrib _\bfP^\ortho$ is integrable on $\Ouv$
if and only if, for any $\bfx_0$ in $\Ouv$, there exists a coordinate chart 
$\coordx$ around $\bfx_0$ such that
the expression $[u^\ortho(x),v^\ortho(x)]$ of the Lie bracket of two arbitrary vectors $u^\ortho(x)$ and $v^\ortho(x)$
of $\Distrib _P^\ortho(x)$ belongs to $\Distrib _P^\ortho(x)$, 
or , equivalently, (see\cite[p. 213]{ONeill.83})
if and only if the $P$-orthogonal projection of $[u^\ortho(x),v^\ortho(x)]$ onto the tangent 
distribution $\Distrib ^\tangent(x)$  is zero.

So the claim follows from
\begin{itemize}
\item\null \quad 
$\displaystyle 
\left(I_n - P(x)^{-1}\frac{\partial h}{\partial x}(x)^\top  \left(\frac{\partial h}{\partial x}(x) 
P(x)^{-1}\frac{\partial h}{\partial x}(x)^\top \right)^{-1}\frac{\partial h}{\partial x}(x)\right)
$\quad 
is a $P$-orthogonal projector onto the tangent distribution.
\item
Since the orthogonal distribution $\Distrib _P^\ortho(x)$ is spanned by the columns of
$P(x)^{-1}\frac{\partial h}{\partial x}(x)^\top$, it is necessary and sufficient to restrict our attention  
to these columns.
\end{itemize}
\end{proof}

\subsubsection{A second necessary and sufficient condition for the integrability of the orthogonal distribution.}
\label{complement32}
We exploit the properties
\begin{enumerate}
\item
By definition, the integrability of the distribution $\bfDistrib _\bfP ^\ortho  $ everywhere 
locally on $\Ouv$ means:
\\
For any $\bfx_0$ in $\Ouv$ and any coordinate chart pair $\coordx$ around $\bfx_0$
and $\coordy$ around $\bfh(\bfx_0)$
there exists
a $C^s$ function $h^\ortho:\coordxm(\coordxd)\to \RR^{n-p}$ satisfying
the properties listed in Lemma \ref{lem3}.
\item
For any coordinate chart pair $\coordx$ and $\coordy$,
by letting $h$, $P$ and $\Distrib _P ^\ortho  (x)$ be the corresponding expressions of 
$\bfh $, $\bfP $ and $\bfDistrib _\bfP ^\ortho  (\bfx)$,  we have that
the orthogonal distribution  $\Distrib _P ^\ortho  (x)$ is spanned by the columns
$P(x)^{-1}\frac{\partial h}{\partial x}(x)^\top $.
\item
The integrability of a distribution is a coordinate invariant. So it is sufficient to establish this 
property in our ``preferred'' coordinates. 
\end{enumerate}

From Frobenius Theorem, the orthogonal distribution $\bfDistrib _\bfP^\ortho$ is integrable
if and only if it is involutive, and therefore if and only if, for any $\bfx_0$ in $\Ouv$, there exists
a coordinate chart $\coordyxr$ and, for any 
pair $(\indyi,\indyj)$, there exist $p$ functions $\mu _\indyl:\coordyxrd\to \RR$ satisfying
$$
\left[
\left(\begin{array}{@{\,  }c@{\,  }}
\delta _{\indyk\indyi} \\ M(y,z)_{\indxrb\indyi}
\end{array}\right)
,
\left(\begin{array}{@{\,  }c@{\,  }}
\delta _{\indyk\indyj} \\ M(y,z)_{\indxrb\indyj}
\end{array}\right)
\right]
\;=\; \sum_\indyl
\left(\begin{array}{@{\,  }c@{\,  }}
\delta _{\indyk\indyl} \\ M(y,z)_{\indxrb\indyl}
\end{array}\right)\mu _\indyl(y,z)
\qquad \forall (y,z)\in \coordyxrm(\coordyxrd)
\  ,
$$
with the notation
$$
M_{\indxrb\indyi}\;=\; \sum_{\indxrc}\left[P_{\xrond\xrond}^{-1}\right]_{\indxrb \indxrc}P_{\indxrc\indyi}
$$
or
$$
\sum_\indxrb P_{\indxra\indxrb}M_{\indxrb\indyi}\;=\; P_{\indxra\indyi}
\  .
$$
The functions $\mu _\indyl$ must be zero and the necessary and sufficient condition reduces to
$$
-
\frac{\partial M_{\indxrb\indyj}}{\partial y_\indyi}
+
\sum_{\indxrd} M_{\indxrd\indyi}
\frac{\partial M_{\indxrb\indyj}}{\partial \xrond_\indxrd}
+
\frac{\partial M_{\indxrb\indyi}}{\partial y_\indyj}
-
\sum_{\indxrd} M_{\indxrd\indyj}
\frac{\partial M_{\indxrb\indyi}}{\partial \xrond_\indxrd}
\;=\; 0
\qquad \forall (\indyi,\indyj,\indxrb)
$$
After multiplication by $P_{\indxra\indxrb}$ and summation in $\indxrb$, this gives
$$
\sum_\indxrb P_{\indxra\indxrb}
\left[\frac{\partial M_{\indxrb\indyi}}{\partial y_\indyj}-\frac{\partial M_{\indxrb\indyj}}{\partial 
y_\indyi}\right]
\;=\; 
\sum_{\indxrd,\indxrb} 
\left[M_{\indxrd\indyj}
P_{\indxra\indxrb}\frac{\partial M_{\indxrb\indyi}}{\partial \xrond_\indxrd}
-
M_{\indxrd\indyi}
P_{\indxra\indxrb}\frac{\partial M_{\indxrb\indyj}}{\partial \xrond_\indxrd}\right]
$$
With
$$
\sum_\indxrb P_{\indxra\indxrb} \frac{\partial M_{\indxrb\indyi}}{\partial \mkern 2mu\raise -0.25em\hbox{\huge$\cdot$}}
\;=\; 
\frac{\partial P_{\indxra\indyi}}{\partial \mkern 2mu\raise -0.25em\hbox{\huge$\cdot$}}-
\sum_\indxrb M_{\indxrb\indyi}\frac{\partial P_{\indxra \indxrb}}{\partial \mkern 2mu\raise -0.25em\hbox{\huge$\cdot$}}
$$
this yields
\\[1em]$\displaystyle 
\frac{\partial P_{\indxra\indyi}}{\partial y_\indyj }-\frac{\partial P_{\indxra\indyj}}{\partial y_\indyi }
-\sum_\indxrb \left[M_{\indxrb\indyi}
\frac{\partial P_{\indxra\indxrb}}{\partial y_\indyj}-
M_{\indxrb\indyj}\frac{\partial P_{\indxra\indxrb}}{\partial 
y_\indyi}\right]
$\hfill \null \\\null \hfill $\displaystyle
\;=\; 
\sum_{\indxrd} 
M_{\indxrd\indyj}
\frac{\partial P_{\indxra\indyi}}{\partial \xrond_\indxrd}
-
\sum_{\indxrd} 
M_{\indxrd\indyi}
\frac{\partial P_{\indxra\indyj}}{\partial \xrond_\indxrd}
+
\sum_{\indxrd,\indxrb} 
M_{\indxrd\indyi}M_{\indxrb\indyj}
\frac{\partial P_{\indxra\indxrb}}{\partial \xrond_\indxrd}
-
\sum_{\indxrd,\indxrb} 
M_{\indxrd\indyj}M_{\indxrb\indyi}
\frac{\partial P_{\indxra\indxrb}}{\partial \xrond_\indxrd}
$\\[1em]
which reduces to
%
\begin{equation}
\label{LP201}
\frac{\partial P_{\indxra\indyi}}{\partial y_\indyj }-\frac{\partial P_{\indxra\indyj}}{\partial y_\indyi }
\;=\; 
\sum_\indxrd M_{\indxrd\indyi}
\left[
\frac{\partial P_{\indxra\indxrd}}{\partial y_\indyj}
-
\frac{\partial P_{\indxra\indyj}}{\partial \xrond_\indxrd}
+
M_{\indxrb\indyj}
\frac{\partial P_{\indxra\indxrb}}{\partial \xrond_\indxrd}
\right]
\;-\; 
M_{\indxrd\indyj}
\left[
\frac{\partial P_{\indxra\indxrd}}{\partial y_\indyi}
-
\frac{\partial P_{\indxra\indyi}}{\partial \xrond_\indxrd}
+
M_{\indxrb\indyi}
\frac{\partial P_{\indxra\indxrb}}{\partial \xrond_\indxrd}
\right]
\end{equation}
and therefore to
\\[1em]$\displaystyle 
\frac{\partial P_{\indxra\indyi}}{\partial y_\indyj }-\frac{\partial P_{\indxra\indyj}}{\partial y_\indyi }
\; 
=\; 
\sum_{\indxrc,\indxrd}
\left[P_{\xrond\xrond}^{-1}\right]_{\indxrd \indxrc}
P_{\indxrc\indyi}
\left[
\frac{\partial P_{\indxra\indxrd}}{\partial y_\indyj}
-
\frac{\partial P_{\indxra\indyj}}{\partial \xrond_\indxrd}
+
\sum_{\indxre} \left[P_{\xrond\xrond}^{-1}\right]_{\indxrb \indxre}
P_{\indxre\indyj}
\frac{\partial P_{\indxra\indxrb}}{\partial \xrond_\indxrd}
\right]
$\refstepcounter{equation}\label{LP195}\hfill$(\theequation)$\\\null \hfill$\displaystyle 
\;-\; 
\sum_{\indxrc,\indxrd}
\left[P_{\xrond\xrond}^{-1}\right]_{\indxrd \indxrc}
P_{\indxrc\indyj}
\left[
\frac{\partial P_{\indxra\indxrd}}{\partial y_\indyi}
-
\frac{\partial P_{\indxra\indyi}}{\partial \xrond_\indxrd}
+
\sum_{\indxre} \left[P_{\xrond\xrond}^{-1}\right]_{\indxrb \indxre}
P_{\indxre\indyi}
\frac{\partial P_{\indxra\indxrb}}{\partial \xrond_\indxrd}
\right]
$\\[1em]
We have established that the orthogonal distribution $\bfDistrib _\bfP^\ortho$ is integrable on $\Ouv$ if and 
only if,
for any $\bfx_0$ in $\Ouv$, we there exists a coordinate chart 
$\coordyxr$ around $\bfx_0$ such that the above equation holds for all $(\indyi,\indyj,\indxra)$ and all $(y,z)$ in
$\coordyxrm(\coordyxrd)$.

\subsubsection{(\ref{LP161}) implies (\ref{LP195}) when $\bfh$ is a Riemannian submersion}

Let $\coordyxr$ and $\coordy$ be an arbitrary coordinate chart pair.
$v^\tangent$ is a tangent vector if and only if there exists a vector $v$ in $\RR^{n-p} $ satisfying
$$
v^\tangent\;=\; \left(\begin{array}{cc}
0
\\
I_{n-p}
\end{array}\right) v
$$
From Lemma \ref{lem6}, $u^\ortho $ is an orthogonal vector if and only if there exists a 
vector $u$ in $\RR^p$ satisfying
\begin{equation}
\label{LP203}
u^\ortho\;=\; \left(\begin{array}{cc}
I_p
\\
-P_{\xrond\xrond}^{-1}P_{\xrond y}
\end{array}\right) u
\end{equation}
So, since we have
$$
\left(\begin{array}{cc}
I_p
&
-
P_{y\xrond}P_{\xrond\xrond}^{-1}
\end{array}\right) 
\left(\begin{array}{cc}
\Gammay _{yy}^\indyi
-\Gamma _{yy}^\indyi
&
-\Gamma _{y\xrond}^\indyi
\\
- \Gamma _{\xrond y}^\indyi
&
-\Gamma _{\xrond\xrond}^\indyi
\end{array}\right)
\left(\begin{array}{cc}
0
\\
I_{n-p}
\end{array}\right) 
\;=\; 
-
\Gamma _{y\xrond}^\indyi
\;+\; 
P_{y\xrond}P_{\xrond\xrond}^{-1}
\Gamma _{\xrond\xrond}^\indyi
\  ,
$$
the equation (\ref{LP161}) is equivalent to
$$ 
\Gamma _{\indxra\indyk}^\indyi 
-
\sum_{\indxrb}\Gamma _{\indxra\indxrb}^\indyi M_{\indxrb\indyk}
\;=\; 0
\  .
$$
We expand to obtain
\\[1em]$\displaystyle 
\sum_\indyj [P^{-1}]_{\indyi\indyj}
\left(
\frac{\partial P_{\indyj\indxra}}{\partial y_\indyk}
+
\frac{\partial P_{\indyj\indyk}}{\partial \xrond_\indxra}
-
\frac{\partial P_{\indyk\indxra}}{\partial y_\indyj}
\right)
+
\sum_\indxrc
[P^{-1}]_{\indyi\indxrc}
\left(
\frac{\partial P_{\indxrc\indxra}}{\partial y_\indyk}
+
\frac{\partial P_{\indxrc\indyk}}{\partial \xrond_\indxra}
-
\frac{\partial P_{\indyk\indxra}}{\partial \xrond_\indxrc}
\right)\;=\; 
$\hfill \null \\\null \hfill $\displaystyle 
\sum_{\indxrb}
\left[
\sum_\indyj
[P^{-1}]_{\indyi\indyj}
\left(
\frac{\partial P_{\indyj\indxra}}{\partial \xrond_\indxrb}
+
\frac{\partial P_{\indyj\indxrb}}{\partial \xrond_\indxra}
-
\frac{\partial P_{\indxrb\indxra}}{\partial y_\indyj}
\right)
+
\sum_\indxrc
[P^{-1}]_{\indyi\indxrc}
\left(
\frac{\partial P_{\indxrc\indxra}}{\partial \xrond_\indxrb}
+
\frac{\partial P_{\indxrc\indxrb}}{\partial \xrond_\indxra}
-
\frac{\partial P_{\indxrb\indxra}}{\partial \xrond_\indxrc}
\right)
\right]
M_{\indxrb\indyk}
$\\[1em]
With (\ref{LP200}) and the fact that, $P^{-1}$ being symmetric positive definite, its $P_{yy}^{-1}$ block is 
invertible, this yields
\\[1em]$\displaystyle 
\left(
\frac{\partial P_{\indyj\indxra}}{\partial y_\indyk}
+
\frac{\partial P_{\indyj\indyk}}{\partial \xrond_\indxra}
-
\frac{\partial P_{\indyk\indxra}}{\partial y_\indyj}
\right)
-
\sum_\indxrc
M_{\indxrc\indyj}
\left(
\frac{\partial P_{\indxrc\indxra}}{\partial y_\indyk}
+
\frac{\partial P_{\indxrc\indyk}}{\partial \xrond_\indxra}
-
\frac{\partial P_{\indyk\indxra}}{\partial \xrond_\indxrc}
\right)\;=\; 
$\hfill \null \\\null \hfill $\displaystyle 
\sum_{\indxrb}
\left[
\left(
\frac{\partial P_{\indyj\indxra}}{\partial \xrond_\indxrb}
+
\frac{\partial P_{\indyj\indxrb}}{\partial \xrond_\indxra}
-
\frac{\partial P_{\indxrb\indxra}}{\partial y_\indyj}
\right)
-
\sum_\indxrc
M_{\indxrc\indyj}
\left(
\frac{\partial P_{\indxrc\indxra}}{\partial \xrond_\indxrb}
+
\frac{\partial P_{\indxrc\indxrb}}{\partial \xrond_\indxra}
-
\frac{\partial P_{\indxrb\indxra}}{\partial \xrond_\indxrc}
\right)
\right]
M_{\indxrb\indyk}
$\\[1em]
By re-ordering, we obtain
\\[1em]$\displaystyle 
\left(
\frac{\partial P_{\indyj\indxra}}{\partial y_\indyk}
-
\frac{\partial P_{\indyk\indxra}}{\partial y_\indyj}
\right)
-
\sum_\indxrc
M_{\indxrc\indyj}
\left(
\frac{\partial P_{\indxrc\indxra}}{\partial y_\indyk}
-
\frac{\partial P_{\indyk\indxra}}{\partial \xrond_\indxrc}
\right)
-
\sum_{\indxrb}
\left(
\frac{\partial P_{\indyj\indxra}}{\partial \xrond_\indxrb}
-
\frac{\partial P_{\indxrb\indxra}}{\partial y_\indyj}
\right)
M_{\indxrb\indyk}
$\hfill \null \\\null \hfill $\displaystyle 
\;=\; -\sum_{\indxrb,\indxrc}
M_{\indxrc\indyj}
\left(
\frac{\partial P_{\indxrc\indxra}}{\partial \xrond_\indxrb}
-
\frac{\partial P_{\indxrb\indxra}}{\partial \xrond_\indxrc}
\right)
M_{\indxrb\indyk}
-
\frac{\partial [P_y]_{\indyj\indyk}}{\partial \xrond_\indxra}
$\\[1em]
where we have used (\ref{LP202}).
This is nothing but the sufficient condition for integrability (\ref{LP201}) 
if
$$
\frac{\partial [P_y]_{\indyj\indyk}}{\partial \xrond_\indxra}\;=\; 0
\  .
$$
But this identity is a consequence of the fact that $\bfh$ is a Riemannian submersion.
\sousection{The  function $\bar h^\ortho$ given by (\ref{LP190}) satisfies (\ref{LP149})}
\label{complement46}
The data are
$$
h(\euX(x,\tau (x)))\;=\; h(\coordxm(\bfx_0))
\quad ,\qquad 
\bar h^\ortho(x)\;=\; h^\ortho(\euX(x,\tau  (x)))
\  .
$$
where $\euX(x,t)$ is the solution of
$$
\dot x = P(x)^{-1}\frac{\partial h}{\partial x}(x)^\top
\  .
$$
So along the solutions we have
$$
\dot{\overparen{h(x)}}
\;=\; \frac{\partial h}{\partial x}(x) P(x)^{-1}\frac{\partial h}{\partial x}(x)^\top
$$
The right hand side being strictly positive, $h(x)$ is strictly  monotonic along the solution, hence the 
local existence and uniqueness of $\tau (x)$. Also the implicit function theorem guarantees its smoothness.

On another hand the semigroup property of the flow gives
$$
\frac{\partial \euX}{\partial x}(x,t) P(x)^{-1}\frac{\partial h}{\partial x}(x)^\top
\;=\; 
P(\euX(x,t))^{-1}\frac{\partial h}{\partial x}(\euX(x,t))^\top
$$

Then we compute
\\[1em]$\displaystyle 
\frac{\partial \bar h^\ortho}{\partial x}(x)
P(x)^{-1}\frac{\partial h}{\partial x}(x)^\top
$\hfill \null \\\null \hfill $
\renewcommand{\arraystretch}{2}
\begin{array}{cl}
=&\displaystyle 
\frac{\partial h^\ortho}{\partial x}(\euX(x,\tau (x)))
\left[
\frac{\partial \euX}{\partial x}(x,\tau (x))P(x)^{-1}\frac{\partial h}{\partial x}(x)^\top
 +
\frac{\partial \euX}{\partial t}(x,\tau (x)) \frac{\partial \tau}{\partial x}(x)
P(x)^{-1}\frac{\partial h}{\partial x}(x)^\top
\right]
\  ,
\\
=&\displaystyle 
\frac{\partial h^\ortho}{\partial x}(\euX(x,\tau (x)))
P(\euX(x,\tau (x)))^{-1}\frac{\partial h}{\partial x}(\euX(x,\tau (x)))^\top
\left[1+\frac{\partial \tau}{\partial x}(x)P(x)^{-1}\frac{\partial h}{\partial x}(x)^\top\right]
\end{array}
$\\[1em]
But we have
$$
\frac{\partial h}{\partial x}(\euX(x,\tau (x)))
\left[
\frac{\partial \euX}{\partial x}(x,\tau (x))
 +
\frac{\partial \euX}{\partial t}(x,\tau (x)) \frac{\partial \tau}{\partial x}(x)
\right]
=0
$$
and therefore
$$
\frac{\partial h}{\partial x}(\euX(x,\tau (x)))
P(\euX(x,\tau (x)))^{-1}\frac{\partial h}{\partial x}(\euX(x,\tau (x)))^\top
\left[1+\frac{\partial \tau}{\partial x}(x)P(x)^{-1}\frac{\partial h}{\partial x}(x)^\top\right]
=0
$$
Since $\frac{\partial h}{\partial x}P^{-1}\frac{\partial h}{\partial x}^\top$ is not zero, this implies
$$
\left[1+\frac{\partial \tau}{\partial x}(x)P(x)^{-1}\frac{\partial h}{\partial x}(x)^\top\right]\;=\; 0
$$
So we do have
$$
\frac{\partial \bar h^\ortho}{\partial x}(x)
P(x)^{-1}\frac{\partial h}{\partial x}(x)^\top\;=\; 0
$$
%
%

%
%
\sousection{{About (\ref{LP126})}}
\label{complement25}
We have
\begin{eqnarray*}
P_{i\beta }(q,v)
&=&
-c g_{i\beta }(q)
+bg_{\beta {\indxrb }}(q)\mathfrak{C} _{{\indxra }i}^{\indxrb }(q) v_{\indxra }
\  ,
\\[0.5em]
P_{\alpha j}(q,v) 
& =&
-cg_{\alpha j}(q)
+b g_{\alpha {\indxrd }}(q)\mathfrak{C} _{{\indxrb }j}^{\indxrd }(q) v_{\indxrb }
\  ,\\[0.5em]
P_{ij}(q,v) &=&
ag_{ij}(q)
-c\left(
 g_{i{\indxrb }}(q)\mathfrak{C} _{{\indxra }j}^{\indxrb }(q) v_{\indxra } 
+
g_{{\indxra }j}(q)\mathfrak{C} _{{\indxrb }i}^{\indxra }(q) v_{\indxrb } 
\right)
+ bg_{{\indxrc }{\indxrd }}(q)
\mathfrak{C} _{{\indxra }i}^{\indxrc }(q)
\mathfrak{C} _{{\indxrb }j}^{\indxrd }(q) v_{\indxra } v_{\indxrb }
\  ,\\[0.5em]
&=&
ag_{ij}(q)
-c\left(
 g_{{\indxrc}j}(q)\mathfrak{C} _{{\indxra }i}^{\indxrc }(q) v_{\indxra } 
+
g_{i{\indxrd }}(q)\mathfrak{C} _{{\indxrb }j}^{\indxrd }(q) v_{\indxrb } 
\right)
+ bg_{{\indxrc }{\indxrd }}(q)
\mathfrak{C} _{{\indxra }i}^{\indxrc }(q)
\mathfrak{C} _{{\indxrb }j}^{\indxrd }(q) v_{\indxra } v_{\indxrb }
\  ,\\[0.5em]
&=&
ag_{ij}(q)
+\left[
-c
 g_{{\indxrc}j}(q)
+ bg_{{\indxrc }{\indxrd }}(q)
\mathfrak{C} _{{\indxrb }j}^{\indxrd }(q) v_{\indxrb }
\right]
\mathfrak{C} _{{\indxra }i}^{\indxrc }(q) v_{\indxra } 
\\&&\qquad \qquad \qquad \displaystyle
+\left[
-c
g_{i{\indxrd }}(q)
+ bg_{{\indxrc }{\indxrd }}(q)
\mathfrak{C} _{{\indxra }i}^{\indxrc }(q) v_{\indxra } 
\right] \mathfrak{C} _{{\indxrb }j}^{\indxrd }(q) v_{\indxrb } 
-bg_{{\indxrc }{\indxrd }}(q)
\mathfrak{C} _{{\indxra }i}^{\indxrc }(q)
\mathfrak{C} _{{\indxrb }j}^{\indxrd }(q) v_{\indxra } v_{\indxrb }
\  ,\\[0.5em]
P_{ij}(q,v) &=&
ag_{ij}(q)
+\sum_{\indxrf,\indxrc }P_{\indxrc  j}(q,v) 
\mathfrak{C} _{{\indxrf }i}^{\indxrc  }(q) v_{\indxrf } 
+\sum_{\indxre,\indxrd }P_{i\indxrd }(q,v)\mathfrak{C} _{{\indxre }j}^{\indxrd  }(q) v_{\indxre } 
-b\sum_{\indxre,\indxrf,\indxrc ,\indxrd }g_{{\indxrc  }{\indxrd  }}(q)
\mathfrak{C} _{{\indxrf }i}^{\indxrc  }(q)
\mathfrak{C} _{{\indxre }j}^{\indxrd  }(q) v_{\indxrf } v_{\indxre }
\end{eqnarray*}
and
\begin{eqnarray*}
P_{i\indxrd }(q,v)
-\sum_{\indxrf,\indxrc }\mathfrak{C} _{{\indxrf }i}^{\indxrc  }(q) v_{\indxrf } P_{\indxrc \indxrd }(q,v)
&=&-c g_{i\indxrd  }(q)
\end{eqnarray*}
By construction, for any positive definite matrix $S$, we have
\\[1em]$\displaystyle 
\left(\begin{array}{cc}
\delta _{\indyi\indyk}
&\displaystyle 
-\sum_{\indxrf}\mathfrak{C} _{{\indxrf }\indyi}^{\indxrc  }(q) v_{\indxrf } 
\\
0
&
S_{\indxra  \indxrc }(q)
\end{array}\right)
\left(\begin{array}{cc}
P_{kl}(q,v) &P_{k\indxrd }(q,v) 
\\
P_{\indxrc \indyl} (q,v) & P_{\indxrc \indxrd }(q,v) 
\end{array}\right)
\left(\begin{array}{cc}
\delta _{\indyl \indyj}
&
0
\\
-\displaystyle \sum_{\indxre}\mathfrak{C} _{{\indxre }\indyj}^{\indxrd  }(q) v_{\indxre } 
&
S_{\indxrd  \indxrb }(q)
\end{array}\right)
$
%
\\[1em]\null \quad $\displaystyle 
\;=\; 
\left(\begin{array}{cc}
P _{il}
-\sum_{\indxrf,\indxrc }\mathfrak{C} _{{\indxrf }i}^{\indxrc  }(q) v_{\indxrf } P_{\indxrc l}
&
P_{i\indxrd }
-\sum_{\indxrf,\indxrc }\mathfrak{C} _{{\indxrf }i}^{\indxrc  }(q) v_{\indxrf } P_{\indxrc \indxrd }
\\{}
\sum_{\indxrc }S_{\indxra  \indxrc }(q)P_{\indxrc l}
&
\sum_{\indxrc }S_{\indxra  \indxrc }(q)P_{\indxrc \indxrd }
\end{array}\right)
\left(\begin{array}{cc}
\delta _{l j}
&
0
\\
-\sum_{\indxre}\mathfrak{C} _{{\indxre }j}^{\indxrd  }(q) v_{\indxre } 
&
S_{\indxrd  \indxrb }(q)
\end{array}\right)
$\\[1em]\null \quad $\displaystyle 
\;=\; 
\left(\begin{array}{cc}
P _{ij}-\sum_{\indxrf,\indxrc }\mathfrak{C} _{{\indxrf }i}^{\indxrc  }(q) v_{\indxrf } P_{\indxrc j}
-\sum_{\indxre,\indxrd }P_{i\indxrd }\mathfrak{C} _{{\indxre }j}^{\indxrd  }(q) v_{\indxre } 
+
\sum_{\indxre,\indxrf,\indxrc ,\indxrd }\mathfrak{C} _{{\indxrf }i}^{\indxrc  }(q) v_{\indxrf } P_{\indxrc \indxrd }\mathfrak{C} _{{\indxre }j}^{\indxrd  }(q) v_{\indxre } 
\quad \null 
\\[-0.7em]
&\multicolumn{1}{@{}l@{}}{\nearrow\quad }
\\[-0.5em]\multicolumn{1}{@{}r@{}}{
\left[P_{i\indxrd }
-\sum_{\indxrf,\indxrc ,\indxrd }\mathfrak{C} _{{\indxrf }i}^{\indxrc  }(q) v_{\indxrf } P_{\indxrc \indxrd }\right]
S_{\indxrd \indxrb }(q)
}
\\[1em]
\sum_{\indxrc ,\indxrd }
S_{\indxra  \indxrc }(q)\left[P_{\indxrc j}
-\sum_{\indxre}P_{\indxrc \indxrd }
\mathfrak{C} _{{\indxre }j}^{\indxrd  }(q) v_{\indxre } \right]
\\[-0.7em]
&\multicolumn{1}{@{}l@{}}{\nearrow}
\\[-0.5em]\multicolumn{1}{@{}r@{}}{
\sum_{\indxrc ,\indxrd }S_{\indxra  \indxrc }(q)P_{\indxrc \indxrd }S_{\indxrd  \indxrb }(q)
}
\end{array}\right)
$
\\[1em]\null \quad $\displaystyle 
\;=\; 
\left(\begin{array}{cc}
a g_{\indyi\indyj}(q) & \displaystyle 
-c \sum_{\indxrd }g_{\indyi\indxrd  }(q) S_{\indxrd \indxrb }(q)\\ 
\displaystyle -c \sum_{\indxrc }S_{\indxra  \indxrc }(q)g_{\indxrc  j} &
\displaystyle 
b  \sum_{\indxrc ,\indxrd }S_{\indxra  \indxrc }(q)g_{\indxrc \indxrd }(q)S_{\indxrd \indxrb }(q)
\end{array}\right)
$
%
%

%
%
\sousection{{Details on (\ref{LP167})}}
\label{complement39}
We obtain
$$
\left[P_{y\xrond}P_{\xrond\xrond}^{-1} P_{\xrond y}\right]_{\indyi\indyj}
\;=\; \frac{1}{b}\,  \sum_{\indxra,\indxrb}
\left[-c g_{\indyi\indxrb } 
+b\sum_{\indxrf,\indxrg}g_{\indxrb \indxrg} \mathfrak{C} _{{\indxrf }i}^{\indxrg}  \xrond _{\indxrf }
\right]
[g^{-1}]_{\indxrb\indxra}
\left[-cg_{\indxra  j} 
+b \sum_{\indxre,\indxrh}g_{\indxra  {\indxrh}} \mathfrak{C} _{{\indxre }j}^{\indxrh }  \xrond _{\indxre }
\right]
$$
where
$$
\sum_\indxrb
g_{\indyi\indxrb }[g^{-1}]_{\indxrb\indxra}\;=\; \delta _{\indyi\indxra}
\quad ,\qquad 
\sum_\indxrb
g_{\indxrg\indxrb }[g^{-1}]_{\indxrb\indxra}\;=\; \delta _{\indxrg\indxra}
\  .
$$
So we are left with
\begin{eqnarray*}
\left[P_{y\xrond}P_{\xrond\xrond}^{-1} P_{\xrond y}\right]_{\indyi\indyj}
&=&
\frac{1}{b}
\sum_{\indxra}
\left[
-c\delta _{\indyi\indxra}
+ b\sum_{\indxrf}
\mathfrak{C} _{{\indxrf }\indyi}^{\indxra}  \xrond _{\indxrf }
\right]
\left[-cg_{\indxra  \indyj} 
+b \sum_{\indxre,\indxrh}g_{\indxra\indxrh}\mathfrak{C} _{{\indxre }j}^{\indxrh }  \xrond _{\indxre }
\right]
\\
&=&
\frac{c^2}{b} g_{\indyi\indyj}
-c\sum_{\indxra,\indxrf}
\mathfrak{C} _{{\indxrf }\indyi}^{\indxra}  \xrond _{\indxrf }g_{\indxra  j} 
-c\sum_{\indxre,\indxrh}g_{\indyi\indxrh}\mathfrak{C} _{{\indxre }j}^{\indxrh }  \xrond _{\indxre }
+b 
\sum_{\indxrg,\indxrf,\indxre,\indxrh}
g_{\indxrg  {\indxrh}} 
\mathfrak{C} _{{\indxrf }\indyi}^{\indxrg} 
\mathfrak{C} _{{\indxre }\indyj}^{\indxrh } 
\xrond _{\indxrf }
\xrond _{\indxre }
\\
&=&\frac{c^2}{b} g_{\indyi\indyj} + P_{\indyi\indyj}-ag_{\indyi\indyj}
\end{eqnarray*}
%
%

%
%
\sousection{Integrability of the orthogonal distribution implies the metric is flat}
\label{complement27}
Following (\ref{LP195}) in the \complement \ref{complement32}, we  want to know whether or not we have
\\[1em]$\displaystyle 
\frac{\partial P_{\indxra\indyi}}{\partial y_\indyj }-\frac{\partial P_{\indxra\indyj}}{\partial y_\indyi }
\; 
=\; 
\sum_{\indxrc,\indxrd}
\left[P_{\xrond\xrond}^{-1}\right]_{\indxrd \indxrc}
P_{\indxrc\indyi}
\left[
\frac{\partial P_{\indxra\indxrd}}{\partial y_\indyj}
-
\frac{\partial P_{\indxra\indyj}}{\partial \xrond_\indxrd}
+
\sum_{\indxre} \left[P_{\xrond\xrond}^{-1}\right]_{\indxrb \indxre}
P_{\indxre\indyj}
\frac{\partial P_{\indxra\indxrb}}{\partial \xrond_\indxrd}
\right]
$\refstepcounter{equation}\label{labequ14}\hfill$(\theequation)$
\\\null \hfill$\displaystyle 
\;-\; 
\sum_{\indxrc,\indxrd}
\left[P_{\xrond\xrond}^{-1}\right]_{\indxrd \indxrc}
P_{\indxrc\indyj}
\left[
\frac{\partial P_{\indxra\indxrd}}{\partial y_\indyi}
-
\frac{\partial P_{\indxra\indyi}}{\partial \xrond_\indxrd}
+
\sum_{\indxre} \left[P_{\xrond\xrond}^{-1}\right]_{\indxrb \indxre}
P_{\indxre\indyi}
\frac{\partial P_{\indxra\indxrb}}{\partial \xrond_\indxrd}
\right]
$\\[1em]
when $P$ is defined as
\begin{eqnarray*}
P_{\indyi\indyj}(y,\xrond) &=&
ag_{\indyi\indyj}(y)
-c\left(
\sum_{\indxrg,\indxrf} g_{j{\indxrg }}(y)\mathfrak{C} _{{\indxrf }i}^{\indxrg }(y) \xrond_{\indxrf } 
+
\sum_{\indxrh,\indxre}g_{i{\indxrh }}(y)\mathfrak{C} _{{\indxre }j}^{\indxrh }(y) \xrond_{\indxre } 
\right)
\\  
&&%
\setbox0=\hbox{$%
\null \hskip 1cm P_{\indyi\indyj}(y,\xrond)%
\hskip 2\arraycolsep%
=
\hskip 2\arraycolsep \null $}%
\setlength{\longueur}{\linewidth}
\addtolength{\longueur}{-\wd0}
\setbox0=\hbox{$%
\displaystyle + b\sum_{\indxrg,\indxrh,\indxre,\indxrf}g_{{\indxrg}{\indxrh}}(y)
\mathfrak{C} _{{\indxrf }i}^{\indxrg}(y)
\mathfrak{C} _{{\indxre }j}^{\indxrh}(y) \xrond_{\indxrf } \xrond_{\indxre }
\  ,\hskip 1cm\null 
$}
\addtolength{\longueur}{-\wd0}
\hskip \longueur \box0
\\[0.5em]
P_{\indyi\indxrb }(y,\xrond)&=&
-c g_{\indyi\indxrb }(y)
+b\sum_{\indxrf,\indxrg}g_{\indxrb \indxrg}(y)\mathfrak{C} _{{\indxrf }i}^{\indxrg}(y) \xrond_{\indxrf }
\  ,
\\[0.5em]
P_{\indxra  \indyj}(y,\xrond) & =&
-cg_{\indxra  j}(y)
+b \sum_{\indxre,\indxrh}g_{\indxra  {\indxrh}}(y)\mathfrak{C} _{{\indxre }j}^{\indxrh }(y) \xrond_{\indxre }
\  ,\\[0.5em]
P_{\indxra  \indxrb }(y,\xrond)&=&
bg_{\indxra  \indxrb }(y)
\  ,
\end{eqnarray*}
We get
\begin{eqnarray*}
\frac{\partial P_{ \indxra \indyj} }{\partial y _\indyi  }&=&
-c \frac{\partial g_{\indxra \indyj}}{\partial y_\indyi }(y)
+b \sum_{\indxre,\indxrh}
\frac{\partial g_{\indxra  {\indxrh}}}{\partial y_\indyi}(y)
\mathfrak{C} _{{\indxre }j}^{\indxrh }(y) \xrond_{\indxre }
+b \sum_{\indxre,\indxrh}g_{\indxra  {\indxrh}}(y)\frac{\partial \mathfrak{C} _{{\indxre }j}^{\indxrh }}{\partial y_\indyi}(y) \xrond_{\indxre }
\\
\frac{\partial P_{ \indxra \indyi} }{\partial y _\indyj  }
-\frac{\partial P_{ \indxra \indyj} }{\partial y _\indyi  }&=&
-c \left[\frac{\partial g_{\indxra \indyi}}{\partial y_\indyj }(y)
-\frac{\partial g_{\indxra \indyj}}{\partial y_\indyi }(y)
\right]
\\&&+b \sum_{\indxre,\indxrh}\left[
\frac{\partial g_{\indxra  {\indxrh}}}{\partial y_\indyj}(y)
\mathfrak{C} _{{\indxre }\indyi}^{\indxrh }(y) 
-
\frac{\partial g_{\indxra  {\indxrh}}}{\partial y_\indyi}(y)
\mathfrak{C} _{{\indxre }\indyj}^{\indxrh }(y) 
\right]\xrond_{\indxre }
\\&&
+b \sum_{\indxre,\indxrh}g_{\indxra  {\indxrh}}(y)
\left[\frac{\partial \mathfrak{C} _{{\indxre }\indyi}^{\indxrh }}{\partial y_\indyj}(y) 
-
\frac{\partial \mathfrak{C} _{{\indxre }j}^{\indxrh }}{\partial y_\indyi}(y)
\right]\xrond_{\indxre }
\end{eqnarray*}
then
\begin{eqnarray*}
\left[P_{\xrond\xrond}^{-1}\right]_{\indxrc \indxrd}&=& \frac{1}{b}[g^{-1}]_{\indxrc \indxrd}
\\
\sum_{\indxrc}\left[P_{\xrond\xrond}^{-1}\right]_{\indxrc \indxrd}
P_{\indxrc \indyi} 
&=&\frac{1}{b}\sum_{\indxrc}[g^{-1}]_{\indxrc \indxrd}
\left[-cg_{\indxrc  \indyi}(y)
+b \sum_{\indxre,\indxrh}g_{\indxrc  {\indxrh}}(y)\mathfrak{C} _{{\indxre }\indyi}^{\indxrh }(y) \xrond_{\indxre }
\right]
\\
&=&-\frac{c}{b} \delta _{\indxrd\indyi}+
\sum_{\indxre}\mathfrak{C} _{{\indxre }\indyi}^{\indxrd }(y) \xrond_{\indxre }
\\
\frac{\partial P_{ \indxra \indxrd} }{\partial y _\indyj  }
&=&b \frac{\partial g_{ \indxra \indxrd} }{\partial y_\indyj  }(y)
\\
\frac{\partial P_{ \indxra \indyj} }{\partial \xrond  _\indxrd }
&=&
b \sum_{\indxrh}g_{\indxra  {\indxrh}}(y)\mathfrak{C} _{{\indxrd }\indyj}^{\indxrh }(y)
\\
\frac{\partial P_{ \indxra \indxrd} }{\partial y _\indyj  }
-\frac{\partial P_{ \indxra \indyj} }{\partial \xrond  _\indxrd }
&=&b \frac{\partial g_{ \indxra \indxrd} }{\partial y_\indyj  }(y)
-
b \sum_{\indxrh}g_{\indxra  {\indxrh}}(y)\mathfrak{C} _{{\indxrd }\indyj}^{\indxrh }(y)
\\
\sum_{\indxrd,\indxrc}\left[P_{\xrond\xrond}^{-1}\right]_{\indxrc \indxrd}
\left[
\frac{\partial P_{\indxra\indxrd }}{\partial y _\indyj}
-
\frac{\partial P_{ \indxra \indyj} }{\partial \xrond  _\indxrd  }
\right]P_{\indxrc \indyi} 
&=&
\sum_{\indxrd}
\left[-c \delta _{\indxrd\indyi}
+b\sum_{\indxre}\mathfrak{C} _{{\indxre }\indyi}^{\indxrd }(y) \xrond_{\indxre }
\right]
\left[
\frac{\partial g_{ \indxra \indxrd} }{\partial y_\indyj  }(y)
-
\sum_{\indxrh}g_{\indxra  {\indxrh}}(y)\mathfrak{C} _{{\indxrd }\indyj}^{\indxrh }(y)
\right]
\\
&=&
-c \frac{\partial g_{ \indxra \indyi} }{\partial y_\indyj  }(y)
+ c \sum_{\indxrh}g_{\indxra  {\indxrh}}(y)\mathfrak{C} _{{\indyi }\indyj}^{\indxrh }(y)
\\&&
+ b\sum_{\indxrd,\indxre}\frac{\partial g_{ \indxra \indxrd} }{\partial y_\indyj  }(y)
\mathfrak{C} _{{\indxre }\indyi}^{\indxrd }(y) \xrond_{\indxre }
-b \sum_{\indxrd,\indxre,\indxrh}g_{\indxra  {\indxrh}}(y)\mathfrak{C} _{{\indxrd }\indyj}^{\indxrh }(y)\mathfrak{C} _{{\indxre }\indyi}^{\indxrd }(y) \xrond_{\indxre }
\end{eqnarray*}
\\[1em]$\displaystyle 
\sum_{\indxrd,\indxrc}\left[P_{\xrond\xrond}^{-1}\right]_{\indxrc \indxrd}
\left(
\left[
\frac{\partial P_{\indxra\indxrd }}{\partial y _\indyj}
-
\frac{\partial P_{ \indxra \indyj} }{\partial \xrond  _\indxrd  }
\right]P_{\indxrc \indyi} 
-
\left[
\frac{\partial P_{\indxra\indxrd }}{\partial y _\indyi}
-
\frac{\partial P_{ \indxra \indyi} }{\partial \xrond  _\indxrd  }
\right]P_{\indxrc \indyj} 
\right)
$\null \hfill \\\null \quad  $\displaystyle 
=\; -c \left[\frac{\partial g_{ \indxra \indyi} }{\partial y_\indyj  }(y)
-
\frac{\partial g_{ \indxra \indyj} }{\partial y_\indyi }(y)
\right]
+
c \sum_{\indxrh}g_{\indxra  {\indxrh}}(y)\left[\mathfrak{C} _{{\indyi }\indyj}^{\indxrh }(y)
-
\mathfrak{C} _{{\indyj }\indyi}^{\indxrh }(y)\right]
$\null \hfill \\\null \hfill $\displaystyle 
+
b\sum_{\indxrd,\indxre}
\left[
\frac{\partial g_{ \indxra \indxrd} }{\partial y_\indyj  }(y)
\mathfrak{C} _{{\indxre }\indyi}^{\indxrd }(y) -
\frac{\partial g_{ \indxra \indxrd} }{\partial y_\indyi  }(y)
\mathfrak{C} _{{\indxre }\indyj}^{\indxrd }(y) 
\right]\xrond_{\indxre }
-b \sum_{\indxrd,\indxre,\indxrh}g_{\indxra  {\indxrh}}(y)
\left[
\mathfrak{C} _{{\indxrd }\indyj}^{\indxrh }(y)\mathfrak{C} _{{\indxre }\indyi}^{\indxrd }(y) 
-
\mathfrak{C} _{{\indxrd }\indyi}^{\indxrh }(y)\mathfrak{C} _{{\indxre }\indyj}^{\indxrd }(y) 
\right]\xrond_{\indxre }
$\\[1em]\null \quad  $\displaystyle 
=\; -c \left[\frac{\partial g_{ \indxra \indyi} }{\partial y_\indyj  }(y)
-
\frac{\partial g_{ \indxra \indyj} }{\partial y_\indyi }(y)
\right]
$\null \hfill \\\null \hfill $\displaystyle 
+
b\sum_{\indxrd,\indxre}
\left[
\frac{\partial g_{ \indxra \indxrd} }{\partial y_\indyj  }(y)
\mathfrak{C} _{{\indxre }\indyi}^{\indxrd }(y) -
\frac{\partial g_{ \indxra \indxrd} }{\partial y_\indyi  }(y)
\mathfrak{C} _{{\indxre }\indyj}^{\indxrd }(y) 
\right]\xrond_{\indxre }
-b \sum_{\indxrd,\indxre,\indxrh}g_{\indxra  {\indxrh}}(y)
\left[
\mathfrak{C} _{{\indxrd }\indyj}^{\indxrh }(y)\mathfrak{C} _{{\indxre }\indyi}^{\indxrd }(y) 
-
\mathfrak{C} _{{\indxrd }\indyi}^{\indxrh }(y)\mathfrak{C} _{{\indxre }\indyj}^{\indxrd }(y) 
\right]\xrond_{\indxre }
$\\[1em]
with using $\mathfrak{C} _{{\indyi }\indyj}^{\indxrh }=\mathfrak{C} _{{\indyj }\indyi}^{\indxrh }$.
And finally
$$
\frac{\partial P_{\indxra\indxrb}}{\partial \xrond_\indxrd}\;=\; 0
\  .
$$
Hence (\ref{labequ14})  holds if we have
$$
\sum_{\indxre,\indxrh}g_{\indxra  {\indxrh}}(y)
\left(
\left[\frac{\partial \mathfrak{C} _{{\indxre }\indyi}^{\indxrh }}{\partial y_\indyj}(y) 
-
\frac{\partial \mathfrak{C} _{{\indxre }j}^{\indxrh }}{\partial y_\indyi}(y)
\right]
+
\sum_{\indxrd}
\left[
\mathfrak{C} _{{\indxrd }\indyj}^{\indxrh }(y)\mathfrak{C} _{{\indxre }\indyi}^{\indxrd }(y) 
-
\mathfrak{C} _{{\indxrd }\indyi}^{\indxrh }(y)\mathfrak{C} _{{\indxre }\indyj}^{\indxrd }(y) 
\right]
\right)\xrond_{\indxre }
\;=\; 
0
$$
Note that the  Riemann curvature tensor $\mathfrak{R} $ is
(see \cite[Proposition 7.4]{Lee-RM}
$$
\mathfrak{R} _{\indyi\indyj\indxre}^\indxrh
\;=\; 
\left[\frac{\partial \mathfrak{C} _{{\indxre }\indyi}^{\indxrh }}{\partial y_\indyj}(y) 
-
\frac{\partial \mathfrak{C} _{{\indxre }j}^{\indxrh }}{\partial y_\indyi}(y)
\right]
-
\sum_{\indxrd}
\left[
\mathfrak{C} _{{\indxrd }\indyj}^{\indxrh }(y)\mathfrak{C} _{{\indxre }\indyi}^{\indxrd }(y) 
-
\mathfrak{C} _{{\indxrd }\indyi}^{\indxrh }(y)\mathfrak{C} _{{\indxre }\indyj}^{\indxrd }(y) 
\right]
$$
Since $g$ is invertible and $\xrond_\indxre$ is arbitrary, the condition is that the Riemann curvature tensor is 
zero and therefore
(see \cite[Theorem 7.10]{Lee-RM}) the configuration space is flat, i.e.
every point has a neighborhood that is isometric 
to an open set in $\RR^p$ with its Euclidean metric.
%
%

%
%
\sousection{If $\bfXi$ is connected and 1-dimensional it is $\bfSS^1$}
\label{complement49}
Let the $\bfx$-manifold be $\bfRR^2$ equipped with globally defined coordinates $(x_1,x_2)$. Let $\bfRR$ be 
the $\bfy$-manifold equipped with a globally defined coordinate $y$. Let $\Ouv$ be the open set 
$\bfRR^2\setminus\{0\}$.
The output function $\bfh:\Ouv\to \RR$ we consider is, when expressed in these coordinates,
$$
y\;=\; h(x) := x_1^2 + x_2^2
\  .
$$

To possibly satisfy the assumptions of Theorem~\ref{prop17}, we select $\bfXi$ as a connected 1-dimensional manifold.
Then, we look for a function $\bfh^\ortho: \RR^2\setminus\{0\} \to \bfXi$,
such that the rank of $\bfhhperp=(\bfh,\bfh^\ortho)$ is $2$ on $\RR^2\setminus\{0\}$.

Specifically, let $\coordxi$ be a coordinate chart around
$\bfh^\ortho (1,0)$ in $\bfXi$. We want
\begin{eqnarray*}
\renewcommand{\arraystretch}{1.5}
\textsf{Rank}\left(\begin{array}{@{\,  }c@{\,  }}
\frac{\partial h}{\partial x}(x)
\\
\frac{\partial h^\ortho}{\partial x}(x)
\end{array}\right)
&=&
\renewcommand{\arraystretch}{1.5}
\textsf{Rank}\left(\begin{array}{@{\,  }cc@{\,  }}
x_1 & x_2
\\
\frac{\partial h^\ortho}{\partial x_1}(x_1,x_2)
&
\frac{\partial h^\ortho}{\partial x_2}(x_1,x_2)
\end{array}\right)\;=\; 2
\qquad \forall (x_1,x_2)\in [h^\ortho]^{-1}(\coordxid)
\end{eqnarray*}
This is possible only if we have
$$
x_1\frac{\partial h^\ortho}{\partial x_2}(x_1,x_2)
-
x_2\frac{\partial h^\ortho}{\partial x_1}(x_1,x_2) 
\neq 0
\qquad \forall (x_1,x_2)\in [h^\ortho]^{-1}(\coordxid)\  .
$$
Now note that
\\[1em]$\displaystyle 
\frac{d}{dt}\left\{h^\ortho(\cos(t),\sin(t))\right\}
\;=\; \cos(t)\frac{\partial h^\ortho}{\partial x_2}(\cos(t),\sin(t))
-
\sin(t)\frac{\partial h^\ortho}{\partial x_1}(\cos(t),\sin(t)) 
$\hfill \null \\\null \hfill $\displaystyle 
\qquad \forall t:\,  h^\ortho(\cos(t),\sin(t)) \in \coordxim(\coordxid)
$\\[1em]
But the right hand side having a constant sign, the integral of the left hand side on $(0,2\pi )$
cannot be zero as would imply the periodicity of the function 
$t\mapsto h^\ortho(\cos(t),\sin(t)) $. We conclude that the closed interval $(0,2\pi )$
cannot be contained in $h^\ortho(\cos(t),\sin(t)) \in \coordxim(\coordxid)$. This implies that $\bfXi$ cannot 
be covered by a single coordinate domain. So it cannot be diffeomorphic to an interval of the form $(0,1)$, 
$(0,1]$ or $[0,1]$.
It is therefore diffeomorphic to $\bfSS^1$
(see \cite[Appendix]{Milnor}).
%
%

%
%
\sousection{Proof of (\ref{LP78})}
\label{complement41}
With introducing indices to make more explicit the matrix products, we have
$$
\left(\frac{\partial }{\partial x}\left(
\vrule height 0.5em depth 0.5em width 0pt
v^\top P(x) v \right) \grad_P h   (x)
\;+\; 2
v^\top P(x)\left[ \frac{\partial \grad_P h  }{\partial x}(x) v\right]
\right)
\; =\;  
\sum_{i,j}
v_iv_j
\left[\sum_k \frac{\partial P_{ij}}{\partial x_k}\grad_P h   _k
\;+\; 2
\sum_k
P_{ik}
\frac{\partial \grad_P h   _k}{\partial x_j}
\right]
$$
But, with the expression of $\bfgrad_\bfP \bfh   $, we have
$$
\grad_P h   _k=\sum _l [P^{-1}]_{kl}\frac{\partial h}{\partial x_l}
\quad ,\qquad 
\sum_k
P_{ik} \grad_P h   _k\;=\; \frac{\partial h}{\partial x_i}
$$
We have also
$$
\frac{\partial [P^{-1}]_{kl}}{\partial x _j}
\;=\; -\sum_{m,n} [P^{-1}]_{km}
\frac{\partial P_{mn}}{\partial x_j} [P^{-1}]_{nl}
$$
This implies
$$
\frac{\partial \grad_P h  _k}{\partial x_j}
\;=\; \sum_l
\left[[P^{-1}]_{kl}\frac{\partial ^2h}{\partial x_l\partial x_j}
-\sum_{m,n} [P^{-1}]_{km}
\frac{\partial P_{mn}}{\partial x_j} [P^{-1}]_{nl}
\frac{\partial h}{\partial x_l}
\right]
$$
and
\begin{eqnarray*}
\sum_k
P_{ik}\frac{\partial \grad_P h   _k}{\partial x_j}
&=&
\sum_k
P_{ik}
\sum_l
\left[[P^{-1}]_{kl}\frac{\partial ^2h}{\partial x_l\partial x_j}
-\sum_{m,n} [P^{-1}]_{km}
\frac{\partial P_{mn}}{\partial x_j} [P^{-1}]_{nl}
\frac{\partial h}{\partial x_l}
\right]
\\
&=& 
\frac{\partial ^2h}{\partial x_i\partial x_j}-
\sum_l\sum_n
\frac{\partial P_{in}}{\partial x_j} [P^{-1}]_{nl}
\frac{\partial h}{\partial x_l}
\end{eqnarray*}
All this gives
\\[1em]$\displaystyle 
\left(\frac{\partial }{\partial x}\left(
\vrule height 0.5em depth 0.5em width 0pt
v^\top P(x) v \right) \grad_P h   (x)
\;+\; 2
v^\top P(x)\left[ \frac{\partial \grad_P h  }{\partial x}(x) v\right]
\right)\; =
$\hfill \null \\\null \hfill $\displaystyle  
\sum_{i,j}
v_iv_j
\left(2\frac{\partial ^2h}{\partial x_i\partial x_j}
\;+\; 
\sum_l\left[\sum_k \frac{\partial P_{ij}}{\partial x_k}
[P^{-1}]_{kl}
\;-\; 2
\sum_n
\frac{\partial P_{in}}{\partial x_j} [P^{-1}]_{nl}
\right]\frac{\partial h}{\partial x_l}\right)
$\\[1em]
But
\begin{eqnarray*}
-\frac{1}{2}\sum_k \frac{\partial P_{ij}}{\partial x_k}
[P^{-1}]_{kl}
\;+\; 
\sum_n
\frac{\partial P_{in}}{\partial x_j} [P^{-1}]_{nl}
&=&
\sum_k [P^{-1}]_{kl}
\left(-\frac{1}{2}
\frac{\partial P_{ij}}{\partial x_k}
\;+\; 
\frac{\partial P_{ik}}{\partial x_j} [P^{-1}]_{kl}
\right)
\\
&=&
\frac{1}{2} \sum_k
[P^{-1}]_{lk} \displaystyle
\left(
\frac{\partial P_{ki}}{\partial x_j}(x) 
+
\frac{\partial P_{kj}}{\partial x_i}(x)
-
\frac{\partial P_{ij}}{\partial x_k}(x)
\right)
\\
&=&\Gamma^l_{ij}
\end{eqnarray*}
is a Christoffel symbol. So we have finally
\begin{eqnarray*}
\left(\frac{\partial }{\partial x}\left(
\vrule height 0.5em depth 0.5em width 0pt
v^\top P(x) v \right) \grad_P h   (x)
\;+\; 2
v^\top P(x)\left[ \frac{\partial \grad_P h  }{\partial x}(x) v\right]
\right)
&=& 2
\sum_{i,j}
v_iv_j
\left(\frac{\partial ^2h}{\partial x_i\partial 
x_j}-\Gamma^l_{ij}\frac{\partial h}{\partial x_l}\right)
\\
&=& 2\sum_{i,j} v_iv_j \Hess_P h_{ij}
\  .
\end{eqnarray*}

%
%
%


%
%
\sousection{Input-dependent systems}
\label{complement29}
Let $\bfRR^n$ be equipped with a single coordinate chart with coordinates $\coordxp$ and
coordinate map $\coordxm:\bfRR^n\to \RR^n$, an 
homeomorphism, satisfying
$$
\coordxp\;=\; \coordxm(\bfx)\quad \forall \bfx\in\bfRR^n
$$
See \cite[p. 40]{Lee.13}.
Let $\Ouv$ be an open subset of $\RR ^n$.

Either because the given system is input-dependent or because we immerse our given system into an 
input-dependent one, 
we consider the case where everything depends on a, possibly time varying, input vector $\entree  $. 
In what follows, we consider this input vector $\entree $ as well as its time derivative $\dot \entree$ as 
being given and fixed and we denote $\mathfrak{U}$ the set of values taken by the pair $(\entree,\dot\entree)$. The system is then
\begin{equation}
\label{u-eqcomp1}
\dot x \;=\;  f(x,\entree  )
\quad ,\qquad 
y\;=\; h(x)
\end{equation}
and we denote $X(x,t;\entree  )$ its solution. Let also 
the metric be input-dependent as $(x,\entree  ) \to P(x,\entree  )$. It gives rise to a continuous family of 
Riemannian spaces.

\begin{description}
\item[\it{Attached to $P$,}] we have the Christoffel symbols
$$
\Gamma_{ij}^l(x,\entree  )\!=\! \frac{1}{2}\sum_{k}
(P^{-1})_{kl}(x,\entree  )
\left[\frac{\partial P_{ik}(x,\entree  )}{\partial x_j}
+
\frac{\partial P_{jk}(x,\entree  )}{\partial x_i}
-
\frac{\partial P_{ij}(x,\entree  )}{\partial x_k}
\right]
\  .
$$
and the geodesic equation
\begin{equation}
\label{u-14}
2\,  \frac{d}{ds}\left(\sum_j P_{kj}(\gamma (s),\entree  )\frac{d\gamma _j}{ds}(s)\right)\;=\; 
\sum_{ij}\frac{d\gamma _{i}}{ds}(s)\frac{\partial P_{ij}}{\partial x_k}(\gamma (s),\entree  )\frac{d\gamma _j}{ds}(s)
\end{equation}
\item[\it{Attached to $P$ and $f$,}] we have the Lie derivative of $P$
$$
L_fP(x,\entree  )
\;=\; 
P(x,\entree  )\frac{\partial f}{\partial x}(x,\entree  )
\;+\; 
\frac{\partial f}{\partial x}(x,\entree  )^\top P(x,\entree  )
\; +\; \sum_k
\frac{\partial P}{\partial x_k}(x,\entree ) f_k(x,\entree  )
$$
\item[\it{Attached to $P$ and $h$,}] we have the gradient
$$
\grad_P h(x,\entree  )\;=\; P(x,\entree  )^{-1}\frac{\partial h}{\partial x}(x)^\top
$$
and the Hessian
$$
\Hess _P h (x,\entree  ) \;=\; \frac{1}{2}\,  \mathcal{L}_{\grad _f h} P(x,\entree  )
$$
or, with coordinates
$$
(\Hess _P h (x,\entree  ))_{ij}=\frac{\partial ^2h}{\partial x_i\partial x_j}(x)
-
\sum_{l}\Gamma _{ij}^l(x,\entree  )
\frac{\partial h}{\partial x_l}(x)
\  .
$$
\item[\quad \it{For each possible (fixed) $\entree  $, the length of a $C^1$ path}] $\gamma $ between points $x_1$ and $x_2$ 
is given by
$$
\left.\vrule height 1em depth 0.5em width 0pt
L(\gamma ,\entree  )\right|_{s_1}^{s_{b}}\;=\; \int_{s_1}^{s_2}
\sqrt{\sum_{ij}
\frac{d\gamma _i}{ds}(s) P(\gamma (s),\entree ) _{ij}\frac{d\gamma _j}{ds}(s)
} \,  ds,
$$
where
$
\gamma (s_1)\;=\; x_1$ and $\gamma (s_2)\;=\; x_2$.
\end{description}

There is nothing special occurring as long as $\entree  $ is 
constant. But if $\entree  $ is a time function which therefore varies as the flow generated by (\ref{u-eqcomp1}), then 
care must be taken.

Let $\gamma $ be a geodesic between $x_1$ and $x_2$ at time say $0$, i.e. obtained with the metric
$x\mapsto P( x,\entree  (0))$.
Let
$t\mapsto\Gamma (s,t)$ be a $C^1$ function satisfying
$$
\frac{\partial \Gamma }{\partial t}(s,t)\;=\; f(\Gamma (s,t),\entree  (t))
\  ,\quad  
\Gamma (s,0)\;=\; \gamma (s)
\  .
$$
At time $t$, $s\mapsto \Gamma(s,t)$ is a path between $X(x_1,t;\entree  )$ and $X(x_2,t;\entree  )$. 
Its length is
$$
\left.\vrule height 1em depth 0.5em width 0pt
L(\Gamma ,\entree  )\right|_{s_1}^{s_{b}}\;=\; \int_{s_1}^{s_2}
\sqrt{\sum_{ij}
\frac{\partial \Gamma _i}{\partial s}(s,t) P_{ij}(\Gamma (s,t),\entree (t)) \frac{\partial \Gamma _j}{\partial 
s}(s,t)
} \,  ds 
\  .
$$
Its derivative with respect to $t$ is
$$
\frac{d}{dt}\!\!\left(\left.\!\!\vrule height 1em depth 0.5em width 0pt
L(\Gamma ,\entree  )\right|_{s_1}^{s_{b}}\right)
=
\int_{s_1}^{s_2}
\frac{
\begin{array}{r}
\multicolumn{1}{l}{%
2\sum_{jk}\frac{\partial ^2\Gamma _k}{\partial s\partial t}(s,t) P_{kj}(\Gamma (s,t),\entree (t))\frac{\partial \Gamma _j}{\partial 
s}(s,t)
}
\\\quad \quad 
+\; \sum_{ijk}
\frac{\partial \Gamma _i}{\partial s}(s,t) 
\left[\frac{\partial P_{ij}}{\partial x_k}(\Gamma (s,t),\entree (t)) \frac{\partial \Gamma _k}{\partial 
t}(s,t)\right]
\frac{\partial \Gamma _j}{\partial s}(s,t)
\quad \quad 
\\\multicolumn{1}{r}{
+\; \sum_{ijl}
\frac{\partial \Gamma _i}{\partial s}(s,t) 
\left[\frac{\partial P_{ij}}{\partial \entree  _l}(\Gamma (s,t),\entree (t))\dot \entree  _l(t)\right]
\frac{\partial \Gamma _j}{\partial s}(s,t)
}
\end{array}
}{2
\sqrt{\sum_{ij}
\frac{\partial \Gamma _i}{\partial s}(s,t) P_{ij}(\Gamma (s,t),\entree (t))\frac{\partial \Gamma _j}{\partial 
s}(s,t)
}}\,  ds 
\  ,
$$
where we have
$$
\frac{\partial ^2\Gamma _k}{\partial s\partial t}(s,t) 
\;=\; \sum_i
\frac{\partial f_k}{\partial x_i}(\Gamma (s,t),\entree  (t))
\frac{\partial \Gamma _i}{\partial s}(s,t)
$$
So, at $t=0$, we have~:
$$
\left.\frac{d}{dt}\!\left(\left.\!\vrule height 1em depth 0.5em width 0pt
L(\Gamma ,\entree  )\right|_{s_1}^{s_{b}}\right)\right|_{t=0}
=\; \int_{s_1}^{s_2}
\frac{
\begin{array}{r}
\multicolumn{1}{l}{%
\sum_{ijk}
2\frac{d \gamma _i}{d s}(s)^\top
\frac{\partial f_k}{\partial x_i}(\gamma (s),\entree  )
P_{kj}(\gamma (s),\entree  )\frac{d \gamma _j}{ds}(s)
}
\\\quad \quad 
+\; \sum_{ijk}
\frac{d \gamma _i}{d s}(s) 
\left[\frac{\partial P_{ij}}{\partial x_k}(\gamma (s),\entree  ) f_k(\gamma (s),\entree  )\right]
\frac{d \gamma _j}{d s}(s)
\quad \quad 
\\\multicolumn{1}{r}{
+\; \sum_{ijl}
\frac{d \gamma _i}{d s}(s) 
\left[\frac{\partial P_{ij}}{\partial \entree  _l}(\gamma (s),\entree  )\dot \entree  _l \right]
\frac{d \gamma _j}{d s}(s)
}
\end{array}
}{2
\sqrt{
\frac{d \gamma _i}{d s}(s) P_{ij}(\gamma (s),\entree  )\frac{d \gamma _j}{d s}(s)
}}\,  ds 
\  ,
$$
But the geodesic equation above gives
\\[1em]$\displaystyle 
\sum_{ijk}
\frac{d \gamma _i}{d s}(s)^\top
\left(
2\frac{\partial f_k}{\partial x_i}(\gamma (s),\entree  )
P_{kj}(\gamma (s),\entree (t))
\;+\;  
\frac{\partial P_{ij}}{\partial x_k}(\gamma (s),\entree  ) f_k(\gamma (s),\entree  )
\right)\frac{d \gamma _j}{d s}(s)
$\hfill \null \\\null \hfill $\displaystyle 
\;=\; 
2\sum_{jk}\frac{d}{ds}\left(f_k(\gamma (s),\entree  )\right)
P_{kj}(\gamma (s),\entree (t))\frac{d \gamma _j}{ds}(s)
\; +\; 
2\sum_{jk}\frac{d}{ds}\left(P_{kj}(\gamma (s),\entree  )\frac{d\gamma _j}{ds}(s)\right)
 f_k(\gamma (s),\entree  )
$\\[1em]
So, if $\gamma $ is normalized, i.e. satisfies
$$
\sum_{ij}\frac{d \gamma _i}{d s}(s) P_{ij}(\gamma (s),\entree  )\frac{d \gamma _j}{d s}(s)\;=\; 1
\qquad \forall s
$$
we have finally
\\[1em]$\displaystyle 
\left.\frac{d}{dt}\left(\left.\vrule height 1em depth 0.5em width 0pt
L(\Gamma ,\entree  )\right|_{s_1}^{s_{b}}\right)\right|_{t=0}
\;=\; 
f(\gamma (s_2))^\top P (\gamma (s_2),\entree  ) \frac{d\gamma }{ds}(s_2)
-
f(\gamma (s_1))^\top P (\gamma (s_1),\entree  ) \frac{d\gamma }{ds}(s_1)
$\refstepcounter{equation}\label{u-10b}\hfill$(\theequation)$\\\null \hfill $\displaystyle
\;+\;
\frac{1}{2}\int_{s_1}^{s_2}
\frac{\partial }{\partial \entree  }\left(\frac{d \gamma }{d s}(s) ^\top P(\gamma (s),\entree )
\frac{d \gamma }{d s}(s)\right)\dot \entree  
\,  ds 
$\\[1em]
So, compared with the usual first variation formula, we have on the right hand side the extra term
$$
\frac{1}{2}\int_{s_1}^{s_2}
\frac{\partial }{\partial \entree  }\left(\frac{d \gamma }{d s}(s) ^\top P(\gamma (s),\entree )
\frac{d \gamma }{d s}(s)\right)\dot \entree  
\;=\; 
\frac{\partial }{\partial \entree  }\left(\left.\vrule height 1em depth 0.5em width 0pt
L(\gamma ,\entree  )\right|_{s_1}^{s_{b}}\right) \dot \entree  
\  .
$$
The above establishes also that lengths are contracted by the flow if we have~:
$$
L_fP (x,\entree  )+ \frac{\partial P}{\partial \entree  }(x,\entree  )\dot \entree  
\;=\; 
P(x,\entree  )\frac{\partial f}{\partial x}(x,\entree  )
\;+\; 
\frac{\partial f}{\partial x}(x,\entree  )^\top P(x,\entree  )
\; +\; \sum_k
\frac{\partial P}{\partial x_k}(x,\entree ) f_k(x,\entree  )
\; +\; 
\frac{\partial P}{\partial \entree  }(x,\entree )\dot \entree  
 < 0
$$
So again, compared with the usual condition, we have the extra term
$\frac{\partial P}{\partial \entree  }\dot \entree  $.
\par\vspace{1em}
For this framework, the definition of strong differential detectability becomes
\\[1em]
\textbf{Strong differential detectability:}
For each possible (fixed) $\entree$, there exist a continuous function $\rho :\Ouv\to [0,+\infty[ $,
and a strictly positive real number $q$
satisfying
\begin{equation}
\label{u-LP10}
\mathcal{L}_f P(x,\entree  )
\;+\; 
\frac{\partial P}{\partial \entree  }(x,\entree ) \dot \entree  
\; \leq  \; 
\rho (x,\entree  )\,  \frac{\partial h}{\partial x}(x)^\top \frac{\partial h}{\partial x}(x)
\;-\; \qlower\,  P(x,\entree  )
\qquad \forall x\in\Ouv\  ,\quad \forall (\entree,\dot \entree)\in \mathfrak{U}
\  .
\end{equation}

Consider the case where $\entree  $ is actually a function of $x$, i.e.
$$
\entree  \;=\; \varpi (x)
\  ,
$$
where $f$ and $\varpi $ are such that there exists $M$ satisfying
$$
\frac{\partial f}{\partial \entree  }(x,\varpi  (x) )\frac{\partial \varpi  }{\partial x}(x) \;=\; M(x)\,  \frac{\partial h}{\partial x}(x)
\  .
$$
The nominal Strong differential detectability property makes sense in this case. It is
\\[1em]\vbox{\noindent$\displaystyle 
\mathcal{L}_f P(x,\varpi  (x) )\;+\; 
\frac{\partial P}{\partial \entree  }(x,\varpi  (x)) \frac{\partial \varpi  }{\partial x}(x) f(x,\varpi  (x))
$\hfill \null \\[0.3em]\null \hfill $\displaystyle 
\;+\; 
P(x,\varpi  (x) )\frac{\partial f}{\partial \entree  }(x,\varpi  (x) )\frac{\partial \varpi  }{\partial x}(x) 
\;+\; 
\frac{\partial \varpi  }{\partial x}(x) ^\top\frac{\partial f}{\partial \entree  }(x,\varpi  (x) )^\top P(x,\varpi  (x) )
$\hfill \null \\[0.3em]\null \hfill $\displaystyle \; \leq  \; 
\rho (x)\,  \frac{\partial h}{\partial x}(x)^\top \frac{\partial h}{\partial x}(x)
\;-\; \qlower\,  P(x,\varpi  (x) )
$}\\[1em]
or
\\[1em]\vbox{\noindent$\displaystyle 
\mathcal{L}_f P(x,\varpi  (x) )\;+\; 
\frac{\partial P}{\partial \entree  }(x,\varpi  (x)) \frac{\partial \varpi  }{\partial x}(x) f(x,\varpi  (x))
$\hfill \null \\[0.3em]\null \hfill $\displaystyle 
\;+\; 
P(x,\varpi  (x) )\,  M(x)\,  \frac{\partial h}{\partial x}(x)
\;+\; 
\frac{\partial h}{\partial x}(x)^\top M(x)^\top P(x,\varpi  (x) )
$\hfill \null \\[0.3em]\null \hfill $\displaystyle \; \leq  \; 
\rho (x)\,  \frac{\partial h}{\partial x}(x)^\top \frac{\partial h}{\partial x}(x)
\;-\; \qlower\,  P(x,\varpi  (x) )
$}\\[1em]
Instead the new Strong differential detectability property (\ref{u-LP10}) is
$$ 
\mathcal{L}_f P(x,\varpi  (x) )
\;+\; 
\frac{\partial P}{\partial \entree  }(x,\varpi  (x)) \frac{\partial \varpi  }{\partial x}(x) f(x,\varpi  (x))
\; \leq  \; 
\rho _{new}(x,\varpi  (x))\,  \frac{\partial h }{\partial x}(x )^\top \frac{\partial h }{\partial x}(x)
\;-\; \qlower _{new}\,  P(x,\varpi  (x) )
$$
It is implied by the nominal one as can be seen by picking
$$
\qlower _{new}\;=\; \frac{\qlower}{2}
\quad ,\qquad 
\rho _{new} (x,h(x))\;=\; 
\rho (x) + \frac{1}{2\qlower}\sup_{v:|v|=1}\left|v^\top P(x,h(x) )^{1/2}M(x)\right|^2
$$

\begin{theorem}
\label{thm5}
Assume the function $h$ is $C^3$ and there exists a $\entree$-dependent $C^3$ Riemannian metric $P$ defined on $\RR^n$ such that
\begin{description}
\item[\normalfont A1$_u$~:]
For each possible (fixed) $\entree$, the metric $P$ is complete.
\item[\normalfont A2$_u$~:]
We have the strong differential detectability.
\item[\normalfont A3$_u$~:]
There exists a $C^3$ function
$
\wpunbf : (y_1,y_2)\in \RR^p\times\RR^p  \mapsto \wp(y_1,y_2)\in[0,+\infty[ 
$
satisfying
\begin{equation}
\label{u-LP65}
\wpunbf (y,y)= 0
\  ,\quad
\frac{\partial ^2\wpunbf}{\partial y_1\partial y_1}(y,y) > 0
\qquad \forall y\in h(\Ouv)
\end{equation}
and, for each possible (fixed) $\entree$ and any  geodesic $\gamma ^*$, taking values in $\Ouv$ and 
minimizing on the maximal interval $[s_1,s_2]$,
we have
\\[1em]$\displaystyle 
\frac{d}{ds}\left\{\wpunbf  (h(\gamma ^*(s)),h(\gamma ^*(s_3)))\right\}\; >\; 0
\qquad \forall s\in ]s_3,s_4],
$\refstepcounter{equation}\label{u-LP99}\hfill$(\theequation)$
\\\null\hfill$\displaystyle
\forall s_3,s_4 \in [s_1,s_2]:\  
s_3<s_4 \quad \& \quad 
h(\gamma ^*(s_3)) \neq h(\gamma ^*(s_4))
\  .
$
\end{description}
Under these conditions,
for any strictly positive real number $E$ and any
closed subset $\mathcal{C}$ of $\Ouv$ with a non empty 
interior, for each possible (fixed) $\entree$ ,
there exists a continuous function 
$k _E^* :\Close\to \RR_{>0} $ such that
\begin{list}{}{%
\parskip 0pt plus 0pt minus 0pt%
\topsep 0.5ex plus 0pt minus 0pt%
\parsep 0pt plus 0pt minus 0pt%
\partopsep 0pt plus 0pt minus 0pt%
\itemsep 0.5ex plus 0pt minus 0pt
\settowidth{\labelwidth}{--}%
\setlength{\labelsep}{0.5em}%
\setlength{\leftmargin}{\labelwidth}%
\addtolength{\leftmargin}{\labelsep}%
}
\item[--]
for any continuous function 
$k _E:\Close\to \RR$
satisfying
$$
k _E(\hat x ,\entree )\geq k _E^*(\hat x ,\entree )
\qquad \forall \hat x \in \Close
\  ,
$$
\item[--]
for the observer given by
\begin{equation}
\label{u-eqn:GeodesicObserverVectorField}
\dot{\hat{x}}\;=\; F({\hat{x}},y,\entree  )\;=\; f({\hat{x}},\entree  )
- k_E({\hat{x}},\entree  ) P({\hat{x}},\entree  )^{-1}\frac{\partial h}{\partial x}({\hat{x}})^\top
\frac{\partial \wpunbf }{\partial y_1}(h({\hat{x}}),-y)
\  ,
\end{equation}
\item[--]
and for all $x$ and $\hat x $ in $\Close$ such that, for the given (fixed) $\entree$, there exists a minimizing normalized geodesic $\gamma ^*$ satisfying
$$
x=\gamma ^*(0)
\quad ,\qquad
\hat x =\gamma ^*(\hat s)
\quad ,\qquad 
\gamma ^*(s)\in\Close\quad \forall s\in [0,\hat s]
\  ,
$$
and satisfying
\begin{equation}
\label{u-eqn:Basin}
d(\hat x ,x) \; <\;   E\  ,
\end{equation}
\end{list}
we have\footnote{
$\Did d(\hat x ,x)$ is the upper right Dini derivative along the solution
}
\begin{equation}
\label{u-LP63}
\Did d(\hat x ,x)
\; \leq \;  \displaystyle -
\frac{\qlower}{4} 
\,  d(\hat x ,x)\end{equation} 
where $d$ denotes the Riemannian distance induced by $P$ for the given (fixed) $\entree$.
\end{theorem}

\begin{IEEEproof}
Since we have
$$
\hat x\;=\; x
\qquad \Rightarrow\qquad 
F(\hat x,y,\entree )\;=\; f(\hat x,\entree )\;=\; f(x,\entree ),
$$
the result already holds when $d(x,\hat x)$ is zero. 
Therefore, the remainder of the proof 
only considers pairs $x$ and $\hat x$ of distinct points in $\Close$.

Let $x$ and $\hat x$ be distinct points in $\Close$ such that there exists a minimizing normalized geodesic $\gamma ^*$ 
satisfying
$$
x=\gamma ^*(0)
\quad ,\qquad
\hat x =\gamma ^*(\hat s)
\quad ,\qquad 
\gamma ^*(s)\in\Close\quad \forall s\in [0,\hat s]
\  .
$$
We have
\begin{equation}
\label{u-eqn:Rdistance}
d(\hat x,x)\;=\; \left.
\vrule height 1em depth 0.5em width 0pt
L(\gamma^*)\right|_{0}^{\hat s}\;=\;
\hat s
\  .
\end{equation}
The  $\entree$-dependent first order variation formula (\ref{u-10b}) gives~:
\begin{eqnarray}
\nonumber
\Did d(\hat x,x) &\leq &
\left.\frac{d}{dt} \left.
\vrule height 1em depth 0.5em width 0pt
L(\Gamma (\mbox{\LARGE .}\, ,t))\right|_{0}^{\hat s}\right|_{t=0}
\\ \label{u-LP7}
&\leq &
\frac{d\gamma ^*}{ds}(\hat s)^\top P(\gamma ^*(\hat s),\entree )  \,  F(\gamma ^*(\hat s),y,\entree )
- \frac{d\gamma ^*}{ds}(0)^\top P(\gamma ^*(0),\entree )\,  f(\gamma ^*(0),\entree )
\\\nonumber
&&\qquad \qquad \qquad \qquad \qquad \qquad \;+\;
\frac{1}{2}\int_{s_1}^{s_2}
\frac{\partial }{\partial \entree  }\left(\frac{d \gamma }{d s}(s) ^\top P(\gamma (s),\entree )
\frac{d \gamma }{d s}(s)\right)\dot \entree  
\,  ds
\  ,
\end{eqnarray}
and the observer and system dynamics give~:
\\[1em]\vbox{\noindent$\displaystyle 
\frac{d\gamma ^*}{ds}(\hat s)^\top P(\gamma ^*(\hat s),\entree )  
\left[F(\gamma ^*(\hat s),y,\entree ) - f(\gamma ^*(\hat s),\entree )\right]\;-\;  \frac{d\gamma ^*}{ds}(0)^\top P(\gamma ^*(0),\entree )
\left[F(\gamma ^*(0),y,\entree )-f(\gamma ^*(0),\entree )\right]
$\hfill\null\\\null\hfill$\displaystyle
=\;
-k_E(\gamma ^*(\hat s))\, 
\frac{d\,  h\circ \gamma ^*}{ds}(\hat s)^{\top}
\frac{\partial \wpunbf  }{\partial y_1}(h(\gamma ^*(\hat s)),y)
\  .
$\refstepcounter{equation}\label{u-LP35}\qquad $(\theequation)$}\\[1em]
On the other hand, we have
\\[1em]$\displaystyle 
\frac{d\gamma ^*}{ds}(\hat s)^\top P(\gamma ^*(\hat s),\entree )  \,  f(\gamma ^*(\hat s),\entree )
- \frac{d\gamma ^*}{ds}(0)^\top P(\gamma ^*(0),\entree )\,  f(\gamma ^*(0),\entree )
$\refstepcounter{equation}\label{u-Aux1}\hfill$(\theequation)$
\\\null\hfill$\displaystyle
\;=\;  
\int_0^{\hat s}\frac{d}{ds} \left(\frac{d\gamma ^*}{ds}(s)^\top P(\gamma ^*(s),\entree )  
\,  f(\gamma ^*(s))\right) ds
\  .
$\\[1em]
With the Euler-Lagrange form of the geodesic equation (\ref{u-14}),
the definition of the Lie 
derivative $\mathcal{L}_fP$ and (\ref{u-LP10}),
we get
\\[1em]$\displaystyle 
\frac{d}{ds} \left(\frac{d\gamma ^*}{ds}(s)^\top P(\gamma ^*(s),\entree )  
\,  f(\gamma ^*(s),\entree )\right)
$\hfill\null\\$\displaystyle
\renewcommand{\arraystretch}{1.5}
\begin{array}[b]{@{}c@{\; }l@{}}
=& \displaystyle 
\frac{1}{2}\,  
\frac{d\gamma ^*}{ds}(s)^\top \mathcal{L}_fP(\gamma ^*(s),\entree )  
\frac{d\gamma ^*}{ds}(s)
\  ,
\\
\leq &\displaystyle 
\frac{\rho (\gamma ^*(s))}{2}
\left|
\frac{\partial h}{\partial x}(\gamma ^*(s)) 
\frac{d \gamma ^*}{ds}(s)\right|^2
\;-\; 
\frac{\qlower}{2}
\, 
\frac{d\gamma ^*}{ds}(s)^\top  P(\gamma ^*(s),\entree ) 
\frac{d\gamma ^*}{ds}(s)
\;-\;
\frac{1}{2}\frac{\partial }{\partial \entree  }
\left(\frac{d\gamma ^*}{ds}(s)^\top
P(\gamma ^*(s),\entree ) \frac{d\gamma ^*}{ds}(s)\right) \dot \entree
\  ,
\\
\leq& \displaystyle 
\frac{\rho (\gamma ^*(s)),\entree }{2}
\left|\frac{d\,  h \circ \gamma ^*}{ds}(s)\right|^2
\;-\;
\frac{\qlower}{2}
\;-\;
\frac{1}{2}\frac{\partial }{\partial \entree  }
\left(\frac{d\gamma ^*}{ds}(s)^\top
P(\gamma ^*(s),\entree ) \frac{d\gamma ^*}{ds}(s)\right) \dot \entree
\  ,
\end{array}
$\refstepcounter{equation}\label{u-Aux2}\hfill$(\theequation)$\\[1em]
where, in the last inequality, we have used
$$
\frac{d\gamma ^*}{ds}(s)^\top  P(\gamma ^*(s),\entree ) 
\frac{d\gamma ^*}{ds}(s)=1
$$
since $\gamma^*$ is normalized.
With
$
d(\hat x,x)\;=\; 
\hat s
$
as given in \eqref{u-eqn:Rdistance},
replacing \eqref{u-Aux2} into \eqref{u-Aux1}
yields
\\[0.7em]$\displaystyle 
\label{u-Aux01}
\frac{d\gamma ^*}{ds}(\hat s)^\top P(\gamma ^*(\hat s),\entree )  \,  f(\gamma ^*(\hat s),\entree )
- \frac{d\gamma ^*}{ds}(0)^\top P(\gamma ^*(0),\entree )\,  f(\gamma ^*(0),\entree )
$\refstepcounter{equation}\label{u-Aux02}\hfill$(\theequation)$
\\[0.7em]\null\hfill  $\displaystyle 
\leq \; 
\int_0^{\hat s}
\frac{\rho (\gamma ^*(s),\entree  )}{2}
\left|\frac{d\,  h \circ \gamma ^*}{ds}(s)\right|^2 ds
\;-\;
\frac{\qlower}{2}\,  d({\hat{x}},x)
\;-\; \frac{1}{2}\int_0^{\hat s}\frac{1}{2}\frac{\partial }{\partial \entree  }
\left(\frac{d\gamma ^*}{ds}(s)^\top
P(\gamma ^*(s),\entree ) \frac{d\gamma ^*}{ds}(s)\right) \dot \entree \,  ds
\,  .
$\\[0.7em]
Then, from (\ref{u-LP7}),
using (\ref{u-LP35}) and \eqref{u-Aux01}, we obtain
\\[1em]$\displaystyle 
\Did d(\hat x,x)
$\hfill\null\\\null\qquad $\displaystyle
\leq \; 
\left[\frac{d\gamma ^*}{ds}(\hat s)^\top P(\gamma ^*(\hat s),\entree )  
\left(F(\gamma ^*(\hat s),y) - f(\gamma ^*(\hat s),\entree )\right)
- \frac{d\gamma ^*}{ds}(0)^\top P(\gamma ^*(0),\entree )
\left(F(\gamma ^*(0),y,\entree )-f(\gamma ^*(0),\entree )\right)
\right]
$\hfill\null\\\null\hfill$\displaystyle
+ \left[\frac{d\gamma ^*}{ds}(\hat s)^\top P(\gamma ^*(\hat s),\entree )  
f(\gamma ^*(\hat s),\entree )
- \frac{d\gamma ^*}{ds}(0)^\top P(\gamma ^*(0),\entree )
f(\gamma ^*(0),\entree )\right]
$\hfill\null\\\null\hfill$\displaystyle 
\;+\; 
\frac{1}{2}\int_0^{\hat s}\frac{1}{2}\frac{\partial }{\partial \entree  }
\left(\frac{d\gamma ^*}{ds}(s)^\top
P(\gamma ^*(s),\entree ) \frac{d\gamma ^*}{ds}(s)\right) \dot \entree 
\  ,
$\\[0.7em]\null\qquad $\displaystyle
\leq \; \displaystyle 
-\;  k_E(\hat x)\,  
\frac{d\,  h\circ \gamma ^*}{ds}(\hat s)^{\top}
\frac{\partial \wpunbf  }{\partial y_1}(h(\gamma ^*(\hat s)),y)^\top
+ \int_0^{\hat s}
\frac{\rho (\gamma ^*(s),\entree )}{2}
\left|\frac{d\,  h \circ \gamma ^*}{ds}(s)\right|^2
ds
\;-\; 
\frac{\qlower}{2}\,  d(\hat x,x)
\  .
$\refstepcounter{equation}\label{u-LP59}\hfill$(\theequation)$\\[0.5em]

To proceed it is appropriate to associate two functions $a_{(\hat x,x,\gamma ^*)}$ and
$b_{(\hat x,x,\gamma ^*,\entree )}$ 
to any triple $(\hat x,x,\gamma ^*)$ with
$\hat x$ and $x$  distinct points in $\Close$
and
$\gamma ^*$ a minimizing normalized geodesic
between $x=\gamma ^*(0)$ and $\hat x=\gamma ^*(\hat s)$
satisfying
$
\gamma ^*(s) \in \Close $ for all $s \in [0,\hat s]$. These 
functions are defined on $[0,\hat s]$ as follows:
\begin{eqnarray*}
a_{(\hat x,x,\gamma ^*)}(r)&=&\frac{1}{r}\,  
\frac{d\,  h\circ \gamma ^*}{ds}(r)^{\top}
\frac{\partial \wpunbf }{\partial y_1}(h(\gamma ^*(r)),h(\gamma ^*(0)))^\top
\end{eqnarray*}
if  $0 < r \leq \hat s$, and
\begin{eqnarray*}
a_{(\hat x,x,\gamma ^*)}(0) &=& \frac{d\,  h\circ \gamma ^*}{ds}(0)^\top
\frac{\partial ^2 \wpunbf }{\partial y_1\partial y_1}(h(\gamma ^*(0)),h(\gamma ^*(0)))\frac{d\,  h\circ \gamma ^*}{ds}(0)
\  ;
\end{eqnarray*}
and
\begin{eqnarray*}
b_{(\hat x,x,\gamma ^*),\entree }(r)&=&
\frac{1}{r}\,  
\int_0^{r}
\frac{\rho (\gamma ^*(s),\entree )}{2}
\left|\frac{d\,  h \circ \gamma ^*}{ds}(s)\right|^2
ds
\end{eqnarray*}
if  $0 < r \leq \hat s$, and
\begin{eqnarray*}
b_{(\hat x,x,\gamma ^*,\entree }(0)&=&
\frac{\rho (\gamma ^*(0),\entree )}{2}
\left|\frac{d\,  h \circ \gamma ^*}{ds}(0)\right|^2.
\end{eqnarray*}
We remark with (\ref{u-LP65}) that $\wpunbf $ reaches its global minimum 
at $y_1=y_2=h(x)$. This implies
\begin{equation}
\label{u-extra6}
\frac{\partial \wpunbf }{\partial y_1}(h(\gamma ^*(r)),h(\gamma ^*(0)))
\;=\; 
\left[\int_0^1
\left(\frac{\partial ^2\wpunbf }{\partial y_1^2}(h(\gamma ^*(\sigma r)),h(\gamma ^*(0)))
\frac{d\,  h \circ \gamma ^*}{ds}(\sigma r)\right)
d\sigma \right] r
\qquad 
\forall r\in [0,\hat s]
\  .
\end{equation}
As a consequence, the functions $a_{(\hat x,x,\gamma ^*)}$ and $b_{(\hat x,x,\gamma ^*,\entree )}$ are continuous on
$[0,\hat s]$.

To investigate further the properties of these functions, we
distinguish $2$ cases~:
\begin{list}{}{%
\parskip 0pt plus 0pt minus 0pt%
\topsep 1ex plus 0pt minus 0pt%
\parsep 0pt plus 0pt minus 0pt%
\partopsep 0pt plus 0pt minus 0pt%
\itemsep 1ex plus 0pt minus 0pt
\settowidth{\labelwidth}{$\  h(x)\neq  h(\hat x)$}%
\setlength{\labelsep}{0.5em}%
\setlength{\leftmargin}{\labelwidth}%
\addtolength{\leftmargin}{\labelsep}%
}
\item[$h(x)\neq  h(\hat x).$]
Condition A3$_u$ gives the implication
$$
h(x)\;\neq \; h(\hat x)\qquad \Longrightarrow\qquad 
a_{(\hat x,x,\gamma ^*)}(r)\; >\; 0\quad \forall r\in (0,\hat  s]
\  .
$$
\item[$h(x)=  h(\hat x).$]
In this case, there are only the following two 
possibilities:
\begin{enumerate}
\item
$h\circ \gamma ^*$ is constant on $[0,\hat s]$. Then we have
$$
\frac{d h\circ \gamma ^*}{ds}(s)\;=\; 0\qquad \forall s\in[0,\hat s]
$$
and therefore
$$
a_{(\hat x,x,\gamma ^*)}(r)\;=\; b_{(\hat x,x,\gamma ^*,\entree )} (r)\;=\; 0\qquad \forall r\in[0,\hat s]
\  .
$$
\item
$h\circ \gamma ^*$ is not constant on $[0,\hat s]$. Then, there exists 
some $s_1$ in $]0, \hat s]$ such that
$$
h(\gamma (s_1))\; \neq \; h(\gamma ^*(0))=h(x)
\  .
$$
With Condition A3$_u$, this implies that the function $s\mapsto \wpunbf  (h(\gamma ^*(s)),h(\gamma ^*(0)))$ 
is not constant on $[0, \hat s]$. But since, we have
$$
0\;=\; \wpunbf  (h(\hat x),h(x))\;=\; \wpunbf  (h(\gamma ^*(\hat s)),h(\gamma ^*(0)))\;=\;  
\wpunbf  (h(\gamma ^*(0)),h(\gamma ^*(0)))
\  ,
$$
this function must
reach a maximum at some point $s_m$ in $]0,\hat s[$ where we have
$$
\wpunbf  (h(\gamma ^*(s_m)),h(\gamma ^*(0))) >0
\quad ,\qquad 
\frac{d}{ds}\wpunbf  (h(\gamma ^*(s_m)),h(\gamma ^*(0)))= 0,
$$
and therefore
$$
h(\gamma ^*(s_m))\; \neq \; h(\gamma ^*(0))
\  .
$$
But this contradicts Condition A3$_u$. So this case is impossible.
\end{enumerate}
\end{list}
In any case, we have established that $a_{(\hat x,x,\gamma ^*)}(\hat s)$ is non negative and if it is zero then
$b_{(\hat x,x,\gamma ^*,\entree )}(r)$ is zero for all $r\in [0,\hat s].$

Let $x_0$ be a point in $\Close$. Call it origin.
For each integer $i$, we introduce the set
$$
\mathcal{C}_i
\;=\;
 \{\xdummy  \in \Close
\,  :\:  i\leq 
d(x_0 ,\xdummy )
\leq i+1\}
\  .
$$
It is $\entree$-dependent via $d$. But for the given (fixed) $\entree$, the metric $P$ is complete and the  Hopf-Rinow Theorem 
\cite[Theorem II.1.1]{Sakai.96} implies
$\mathcal{C}_i$ is compact.
We consider the case where $\hat x$ is in $\mathcal{C}_i$ and $x$ distinct from $\hat x$ satisfies
$$
d(\hat x,x)\;=\; \hat s\; \leq \; E
\  .
$$
Invoking again the completeness of $P$,
the Hopf-Rinow Theorem guarantees the set
$$
\mathcal{K}_i\;=\; \{\xdummy  \in \RR^n\,:\: 
d(x_0 ,\xdummy )
\leq i+1 + E \}
$$
is compact and the domain of definition of
$\gamma ^*$ contains at least $[0,E]$. Then since,  $\gamma ^*$ being normalized, we have
$$
\begin{array}{rcll}
 d(\gamma ^*(\hat s),\gamma ^*(s))&=&\hat s -s\qquad \forall s\in [0,\hat s]\  ,
 \\
 &\leq &E\qquad \forall s\in [\hat s,E]\  
\end{array}
$$
and therefore
$$
d(x_0 ,\gamma ^*(s) )\;\leq \; d(x_0 ,\gamma ^*(\hat s) )+ d(\gamma ^*(\hat s),\gamma ^*(s))
\; \leq \;  i+1 + E\qquad \forall s\in [0,E]
\  .
$$
It follows that
$x$ and $\gamma ^*(s)$ for $s$ in $[0,E]$ are in $\mathcal{K}_i$ and
the domain of definition of the functions $a_{(\hat x,x,\gamma ^*)}$ and $b_{(\hat x,x,\gamma ^*,\entree )}$
can be extended to $[0,E]$ while preserving at least their continuity.

To conclude it is sufficient to
prove the existence of a real number $k_i(\entree )$ such 
that, for the given (fixed) $\entree $, 
we have
\\[0.5em]\null \hfill $\displaystyle 
\frac{q}{4}\;+\; k_i(\entree )\,  a_{(\hat x,x,\gamma ^*)}(\hat s) > 
b_{(\hat x,x,\gamma ^*,\entree )}(\hat s).
$\hfill \null \\[0.5em]
Indeed, with this inequality, the definitions of the functions $a_{(\hat x,x,\gamma ^*)}$ and $b_{(\hat x,x,\gamma ^*,\entree )}$ and
(\ref{u-LP59}) where $d(\hat x,x)=\hat s$,
we obtain the result  
provided the function $k_E$ satisfies
$$
k_E(\hat x,\entree)\; \geq \; k_i(\entree )
\qquad \forall \hat x
\in\mathcal{C}_i\  .
$$

We prove the existence of $k_i$ by contradiction. Suppose that such $k_i$ does not exist.
Then, there exists a sequence $(\hat s_n, x_n,\hat x_n, \gamma _n^*)$, with 
$\hat s_n $ in $]0,E]$,
$\hat x_n$ in $\mathcal{C}_i$, $x_n$ in $\mathcal{K}_i$,
$\gamma _n^*$ a normalized geodesic satisfying
$$
x_n=\gamma _n^*(0)
\quad ,\qquad
\hat x _n=\gamma _n^*(\hat s _n)
\quad ,\qquad 
\gamma _n^*(s)\in\Close\quad \forall s\in [0,\hat s_n]
\  ,
$$
minimizing on $[0,\hat s_n]$, defined on $[0,E]$ where it takes its values in the compact set $\mathcal{K}_i$ 
independent of $n$. And we have
\begin{equation}
\label{u-LP16}
\frac{q}{4}\;+\; n\,  a_{(\hat x _n,x_n,\gamma _n^*)}(\hat s_n)\; \leq \; 
b_{(\hat x _n,x_n,\gamma _n^*,\entree )}(\hat s_n)
\  .
\end{equation}
The functions $h$,
$\rho $
and $\frac{\partial  h}{\partial x}$, restricted to 
the compact set $\mathcal{K}_i$ where the functions $\gamma _n^*$ take their values,
are continuous and bounded.  Also, from the geodesic equation, 
the same holds for $\gamma _n^*$, $\frac{d\gamma _n^*}{ds}$ and 
$\frac{d^2\gamma _n^*}{ds^2}$ restricted to
$[0, E]$.
With the definition of $b_{(\hat x _n,x_n,\gamma _n^*,\entree )}$, this 
implies that the right-hand side of (\ref{u-LP16}) is upper bounded, say by 
$B$. Consequently, we have
$$
\frac{q}{4}\;+\; n\,  a_{(\hat x _n,x_n,\gamma _n^*)}(\hat s_n)\; \leq 
\; B \qquad \forall n.
$$
Since $a_{(\hat x _n,x_n,\gamma _n^*)}(\hat s_n)$ is nonnegative, 
we obtain
\begin{equation}
\label{u-extra3}
\lim_{n\to \infty }a_{(\hat x _n,x_n,\gamma _n^*)}(\hat s_n)\;=\; 0
\end{equation}
and
\begin{equation}
\label{u-extra2}
\liminf_{n\to \infty }b_{(\hat x _n,x_n,\gamma _n^*,\entree )}(\hat s_n)\; \geq \; \frac{q}{4}
\  .
\end{equation}

We  have the following claim.

\noindent
{\it Claim: There exists a subsequence $(\hat s_{n_1},x_{n_1},\hat x _{n_1},\gamma _{n_1}^*)$ of $(\hat s_{n},x_{n},\hat x _{n},\gamma _{n}^*)$ such that
\begin{eqnarray}\label{u-eqn:LimitingSolPoints}
\lim_{n_1\to \infty } (x_{n_1},\hat x _{n_1})
&=&
( x_\omega ,\hat x _\omega )
\  ,\\\label{u-extra4}
\lim_{n_1\to \infty }\hat s_{n_1}\;=\;\lim_{n_1\to \infty } d(\hat x_{n_1},x_{n_1})
&=& d(\hat x_\omega ,x_\omega )\;=\; \hat s_\omega 
\  ,
\\
\label{u-eqn:LimitingSol}
\lim_{n_1\to \infty }\gamma _{n_1}^*(s)
&=&
\gamma _\omega (s)
\ \textrm{uniformly in} \  s \in [0, E]
\  ,\\\label{u-extra5}
\lim_{n_1\to \infty }\frac{d\gamma _{n_1}^*}{ds}(s)&=&\frac{d\gamma _\omega }{ds}(s)
\ \textrm{uniformly in} \  s \in [0, E]
\  , \
\end{eqnarray}
where $\gamma_\omega:[0,\hat s_\omega]\to \Close$ is  a minimizing normalized geodesic between 
$x_\omega$ and $\hat x_\omega$.}

\noindent
\par\vspace{1em}\noindent
\textit{Proof~:}
From the properties listed above of the sequence $(\hat s_n, x_n,\hat x_n, \gamma _n^*)$
where $\gamma _n^*$ is a solution of the geodesic 
equation, we deduce the existence of a subsequence with index 
$n_1$ and a quadruple $(\hat s_\omega , x_\omega ,\hat x_\omega , \gamma _\omega )$ such that 
\eqref{u-eqn:LimitingSolPoints}-\eqref{u-extra5} hold (see, for instance,  \cite[Theorem 5, \S 1]{Filippov.88} and
\cite[Proposition II.2.1]{Sakai.96}),
where $\gamma _\omega $ is a solution of the geodesic equation and,
since $\Close $ is closed, it satisfies
$$
\gamma _\omega (0) \;=\;  x_\omega
\  ,\quad 
\gamma _\omega (\hat s_\omega ) \;=\;  \hat x _\omega
\  ,\quad 
\gamma _\omega (s)\in \Close \quad  \forall s\in [0,\hat s_\omega ]
\  .
$$
Finally, according to \cite[Lemma III.4.2]{Sakai.96}, it is minimizing 
between $x_\omega$ and $\hat x _\omega$.
\qquad $\square$

\par\vspace{1em}

If $\hat x _\omega \neq x_\omega$, the functions
$a_{(\hat x_\omega ,x_\omega ,\gamma _\omega ^*)}$ 
and
$b_{(\hat x_\omega ,x_\omega ,\gamma _\omega ^*,\entree )}$ are well defined and we have
$$
\lim_{n_1\to \infty }
a_{(\hat x_{n_1},x_{n_1},\gamma _{n_1}^*)}(s)
\;=\; 
a_{(\hat x_\omega ,x_\omega ,\gamma _\omega ^*)}(s)
\quad ,\qquad 
\lim_{n_1\to \infty }
b_{(\hat x_{n_1},x_{n_1},\gamma _{n_1}^*,\entree )}(s)
\;=\; 
b_{(\hat x_\omega ,x_\omega ,\gamma _\omega ^*,\entree }(s)
\qquad \forall s\in [0,E]
\  .
$$
It follows from (\ref{u-extra3}) that
$a_{(\hat x _\omega ,x_\omega ,\gamma _\omega )}(\hat s_\omega )$ 
is zero and we have seen that this implies the same holds for
$b_{(\hat x _\omega ,x_\omega ,\gamma _\omega ,\entree )}(r)$
for all $r\in [0,\hat s_\omega ]$. This contradicts (\ref{u-extra2}).

\par\vspace{1em}

If $\hat x _\omega = x_\omega$, the functions
$a_{(\hat x_\omega ,x_\omega ,\gamma _\omega ^*)}$ 
and
$b_{(\hat x_\omega ,x_\omega ,\gamma _\omega ^*,\entree )}$ are not defined.
Nevertheless (\ref{u-extra4}) and the uniformity in (\ref{u-eqn:LimitingSol}) and 
(\ref{u-extra5}) give~:
\begin{eqnarray*}
\lim_{n_1 \to \infty }\hat s_{n_1}&=&0
\  ,
\\
\lim_{n_1\to \infty }
\frac{1}{\hat s_{n_1}}\,  
\int_0^{\hat s_{n_1}}
\frac{\rho (\gamma _{n_1}^*(s),\entree )}{2}
\left|\frac{d\,  h \circ \gamma _{n_1}^*}{ds}(s)\right|^2
&=&
\frac{\rho (\gamma _\omega ^*(0),\entree )}{2}
\left|\frac{\partial h}{\partial x}(\gamma _\omega ^*(0))\frac{d\gamma _\omega }{ds}(0)\right|^2
\  ,
\\
\lim_{n_1\to \infty }
\frac{d\,  h\circ \gamma _{n_1}^*}{ds}(\hat s_{n_1})
\;=\; 
\lim_{n_1\to \infty }
\frac{\partial h}{\partial x}(\hat x_{n_1})
\frac{d\gamma _{n_1}}{ds}(\hat s_{n_1})
&=&
\frac{\partial h}{\partial x}(\hat x_\omega )
\frac{d\gamma _\omega }{ds}(0)
\  .
\end{eqnarray*}
Also we obtain with (\ref{u-extra6})
\begin{eqnarray*}
\lim_{n_1 \to \infty }\frac{1}{\hat s_{n_1}}  
\frac{\frac{\partial \wpunbf }{\partial y_1}(h(\gamma _{n_1}^*(\hat s_{n_1})),h(\gamma ^*(0)))^\top}{s_{n_1}}
&\hskip -0.7em = & \hskip -0.7em 
\frac{\partial ^2\wpunbf }{\partial y_1^2}(h(\gamma _\omega ^*(0)),h(\gamma _\omega ^*(0)))
\frac{\partial h}{\partial x}(\gamma _\omega ^*(0))\frac{d\gamma _\omega }{ds}(0)
\  ,
\\
&\hskip -0.7em = & \hskip -0.7em
\frac{\partial ^2\wpunbf }{\partial y_1^2}(h(\hat x_\omega ),h(\hat x_\omega ))
\frac{\partial h}{\partial x}(\hat x_\omega )\frac{d\gamma _\omega }{ds}(0)
\  ,
\\
\lim_{n_1 \to \infty }\frac{1}{\hat s_{n_1}}  
\frac{d  h\circ \gamma _{n_1}^*}{ds}(\hat s_{n_1})^{\top}
\frac{\partial \wpunbf }{\partial y_1}(h(\gamma _{n_1}^*(\hat s_{n_1})),h(\gamma ^*(0)))^\top
&\hskip -0.7em = & \hskip -0.7em
\frac{d\gamma _\omega }{ds}(0)^\top
\frac{\partial h}{\partial x}(\hat x_\omega )^\top
\frac{\partial ^2\wpunbf }{\partial y_1^2}(h(\hat x_\omega ),h(\hat x_\omega ))
\frac{\partial h}{\partial x}(\hat x_\omega )
\frac{d\gamma _\omega }{ds}(0)
\, .
\end{eqnarray*}
So (\ref{u-extra3}) and (\ref{u-extra2}) give~:
\begin{eqnarray*}
\frac{\rho (\hat x_\omega ,\entree )}{2}
\left|\frac{\partial h}{\partial x}(\hat x_\omega )\frac{d\gamma _\omega }{ds}(0)\right|^2
&\geq & \frac{q}{4}
\  ,\\
\frac{d\gamma _\omega }{ds}(0)^\top
\frac{\partial h}{\partial x}(\hat x_\omega )^\top
\frac{\partial ^2\wpunbf }{\partial y_1^2}(h(\hat x_\omega ),h(\hat x_\omega ))
\frac{\partial h}{\partial x}(\hat x_\omega )
\frac{d\gamma _\omega }{ds}(0)
&=& 0
\end{eqnarray*}
Since the matrix $\frac{\partial ^2 \wpunbf  }{\partial y_1 \partial y_1}(x_\omega ,x_\omega )$ is positive 
definite by assumption, we have a contradiction.
So we have established the existence of $k_i$.
\end{IEEEproof}
\sousection{Construction of the input-dependent system and the metric for the harmonic oscillator}
\label{complement30}
\subsubsection{A reduced order observer in other coordinates}
Consider the harmonic oscillator with unknown frequency (\ref{LP132})
the state of which evolves in the invariant set $\Ouv_\varepsilon $ defined in (\ref{LP205}).
Expressed in the other coordinates
$$
\barxrond=
\left(\begin{array}{@{\,  }c@{\,  }}
\barxrond_\indxra
\\
\barxrond_\indxrb
\end{array}\right)=\varphi (y,\xrond)=
\left(\begin{array}{@{\,  }c@{\,  }}
\xrond_\indxra - y
\\
\xrond_\indxrb+\frac{y^2}{2}
\end{array}\right)
$$
its dynamics are
$$
\dot \barxrond=
\left(\begin{array}{@{\,  }c@{\,  }}
\dot{\barxrond}_\indxra 
\\
\dot{\barxrond}_\indxrb
\end{array}\right)
=
\left(\begin{array}{@{\,  }c@{\,  }}
 -y \xrond_\indxrb-\xrond_\indxra
\\
y \xrond_\indxra
\end{array}\right)
=
\left(\begin{array}{@{\,  }c@{\,  }}
-y\barxrond_\indxrb-\barxrond_\indxra +\frac{y^3}{2}-y
\\
y \barxrond_\indxra+ y^2
\end{array}\right)
=
\bar f(y,\barxrond)
$$
Let an observer be given by a copy
$$
\dot{{\hat \barxrond}}_\indxra 
 = -y{\hat \barxrond}_\indxrb-{\hat \barxrond}_\indxra
+\frac{y^3}{2}-y
\quad ,\qquad 
\dot{{\hat \barxrond}}_\indxrb 
= y {\hat \barxrond}_\indxra+ y^2
$$

\subsubsection{Convergence and construction of a metric}
We study its convergence by considering the estimation errors
$$
{\bar e}_\indxra=\hat{\barxrond}_\indxra-\barxrond_\indxra
\quad ,\qquad 
{\bar e}_\indxrb=\hat \barxrond_\indxrb-\barxrond_\indxrb
$$
\startmodif
which
\stopmodif
satisfy
$$
\dot {\bar e}_\indxra = -y {\bar e}_\indxrb - {\bar e}_\indxra
\quad ,\qquad 
\dot {\bar e}_\indxrb = y {\bar e}_\indxra
$$
In these equations $y$ plays the role of an exogenous input. We keep the notation $y$ however. But later we 
shall use $u_y$ instead.

We obtain
\begin{equation}
\label{LP237}
\dot{\overparen{{\bar e}_\indxra^2 + {\bar e}_\indxrb^2}}\;=\; -2{\bar e}_\indxra^2
\end{equation}
The right hand side is negative but not negative definite.
To build a strict Lyapunov function, we introduce the auxiliary variables
\begin{equation}
\label{LP207}
z_1={\bar e}_\indxra
\quad ,\qquad 
z_2=-y {\bar e}_\indxrb
\quad ,\qquad 
z_3=-\dot y {\bar e}_\indxrb
\end{equation}
They satisfy
\begin{eqnarray*}
\dot z_1&=& z_2 - z_1
\\
\dot z_2&=&z_3-y^2 z_1
\\
\dot z_3&=&-\ddot y {\bar e}_\indxrb - \dot y y z_1
\end{eqnarray*}
where $y$ satisfies
$$
\dot y = \xrond_\indxra
\quad ,\qquad 
\ddot y = - y \xrond_\indxrb
$$
Again here, as $y$ is, $\xrond _\indxra$ and $\xrond _\indxrb$ are exogenous inputs.

We obtain
\begin{eqnarray*}
\dot z_1&=& z_2 - z_1
\\
\dot z_2&=&z_3-y^2 z_1
\\
\dot z_3&=&- \xrond_\indxrb z_2 - \xrond_\indxra y z_1
\end{eqnarray*}
This defines a dynamic system with inputs $(y,\xrond_\indxra,\xrond_\indxrb)$, state $(z_1,z_2,z_3)$ and output $z_1$.
Our next step is to build an output injection term making the origin asymptotically stable.
For this, we consider a first intermediate system
\begin{eqnarray*}
\dot{\zeta }_2 &=&\zeta _3 - 2 d{} \zeta _2
\\
\dot{\zeta }_3 &=& -  \xrond_\indxrb \zeta _2 -d{} \zeta _2
\end{eqnarray*}
\startmodif
where $d{}$ is some real number.
\stopmodif
It satisfies
\begin{eqnarray*}
\frac{1}{2}\dot{\overparen{(\zeta _2 -\zeta _3)^2 + \zeta _3^2}}
&=&
 (\zeta _2 -\zeta _3)(\zeta _3- 2 d{} \zeta _2 + [\xrond_\indxrb +d{}]  \zeta _2)
- \zeta _3[\xrond_\indxrb +d{}] \zeta _2
\\
&=&
 (\zeta _2 -\zeta _3)(\zeta _3 - [d{} -\xrond_\indxrb ]  \zeta _2)
-\zeta _3 [\xrond_\indxrb +d{}] \zeta _2
\\
&=&
- [d{} - \xrond_\indxrb]  \zeta _2^2
-\zeta _3^2
+[1 -2\xrond_\indxrb ]\zeta _2\zeta _3
\end{eqnarray*}
which is negative definite if
$$
d{}-\xrond_\indxrb >0\quad ,\qquad [1 -2\xrond_\indxrb ]^2 < 4[d{}-\xrond_\indxrb] 
$$
or
$$
\xrond_\indxrb  <\frac{\sqrt{4d{}-1}}{2}
\  .
$$
Consider now a second intermediate system
$$\begin{array}{rcl|l}
\dot{\zeta }_1&=&\multicolumn{2}{l}{[\zeta _2 - b{} \zeta _1]}
\\[0.5em]
\cline{1-3}\vrule height 1.3em depth 0pt width 0pt
\dot{\zeta }_2 &=&\zeta _3 - 2 d{} \zeta _2&+2d{} [\zeta _2-b{}\zeta _1]
\\
\dot{\zeta }_3 &=& - \xrond_\indxrb  \zeta _2 -d{} \zeta _2&+d{} [\zeta _2-b{} \zeta _1]
\end{array}
$$
It satisfies
\\[1em]$\displaystyle 
\frac{1}{2}\dot{\overparen{(\zeta _2 - b{} \zeta _1)^2+(\zeta _2 -\zeta _3)^2 + \zeta _3^2}}
$\hfill \null \\\null \qquad  $\displaystyle 
=
(\zeta _2 - b{} \zeta _1)\left(\zeta _3 - 2 d{} b{}\zeta _1-b{}(\zeta _2 - b{} \zeta _1)\right)
$\hfill \null \\\null \hfill $\displaystyle 
- [d{} - \xrond_\indxrb]  \zeta _2^2
-\zeta _3^2
+[1 -2\xrond_\indxrb ]\zeta _2\zeta _3
$\hfill \null \\\null \hfill $\displaystyle 
+ d{} (\zeta _2 -\zeta _3)(\zeta _2-b{}\zeta _1)
+d{} \zeta _3(\zeta _2-b{}\zeta _1)
$\\\null \qquad  $\displaystyle 
=
-(b{}-2d{})(\zeta _2 - b{} \zeta _1)^2
- [d{} - \xrond_\indxrb]  \zeta _2^2
-\zeta _3^2
$\hfill \null \\\null \hfill $\displaystyle 
+[1 -2\xrond_\indxrb ]\zeta _2\zeta _3
+ (\zeta _2 - b{} \zeta _1)(
[\zeta _3-2d{}\zeta_2]
+d{} [\zeta _2-\zeta_3]
+ d{} \zeta _3
)
$\\[0.7em]\null \qquad  $\displaystyle 
=
-(b{}-2d{})(\zeta _2 - b{} \zeta _1)^2
- [d{} - \xrond_\indxrb]  \zeta _2^2
-\zeta _3^2
+[1 -2\xrond_\indxrb ]\zeta _2\zeta _3
- d{} (\zeta _2-b{}\zeta _1)\zeta _2
+(\zeta _2-b{}\zeta _1)\zeta _3
$\\[1em]
It remains to find $b{}$ and $d{}$ to satisfy
\startmodif
$$
\left(\begin{array}{@{\,  }ccc@{\,  }}
2(b{}-2d{})
&
d{}
&
-1
\\
d{}
&
2 [d{} - \xrond_\indxrb] 
&
- [1 -2\xrond_\indxrb ]
\\
-1
&
- [1 -2\xrond_\indxrb ]
&
2
\end{array}\right)
\geq S
$$
\stopmodif
where $S$ is a constant symmetric positive definite matrix. We have seen above the block
\startmodif
$$
\left(\begin{array}{@{\,  }cc@{\,  }}
2 [d{} - \xrond_\indxrb] 
&
- [1 -2\xrond_\indxrb ]
\\{}
- [1 -2\xrond_\indxrb ]
&
2
\end{array}\right)
$$
\stopmodif
is positive definite if $d{}$ is large enough to satisfy
$$
\xrond_\indxrb  <\frac{\sqrt{4d{}-1}}{2}
\  .
$$
So it remains to pick $b{}$ large enough to satisfy
$$
2(b{}-2d{}) >
\left(\begin{array}{@{\,  }cc@{\,  }}
d{}  & - 1
\end{array}\right)
\left(\begin{array}{@{\,  }cc@{\,  }}
2 [d{} - \xrond_\indxrb] 
&
- [1 -2\xrond_\indxrb ]
\\{}
- [1 -2\xrond_\indxrb ]
&
2
\end{array}\right)^{-1}
\left(\begin{array}{@{\,  }c@{\,  }}
d{} \\ - 1
\end{array}\right)
$$
Now we can see the $z$ system as the following perturbation of the second intermediate system
$$\begin{array}{rcl|l}
\dot z_1
&=&[z_2  - b{} z_1] &- \left( 1 - b{} \right) z_1
\\
\dot z_2
&=&z_3 -2d{} z_2  + 2d{} [z_2-b{} z_1] &-\left( y^2 - 2 d{} b{} \right)z_1
\\
\dot z_3
&=&- \xrond_\indxrb z_2 -d{} z_2 + d{}[z_2- b{} z_1] &- \left(\xrond_\indxra y - d{} b{}\right) z_1
\end{array}
$$
It satisfies
\\[1em]$\displaystyle 
\frac{1}{2}\dot{\overparen{( z_2 - b{}  z_1)^2+( z_2 - z_3)^2 +  z_3^2}}
$\hfill \null \\\null \qquad $\displaystyle 
\leq 
-\left(\begin{array}{@{\,  }ccc@{\,  }}
[z_2-b{} z_1]
&
z_2
&
z_3
\end{array}\right) S
\left(\begin{array}{@{\,  }c@{\,  }}
[z_2-b{} z_1]
\\
z_2
\\
z_3
\end{array}\right)
$\hfill \null \\\null \hfill $\displaystyle 
- [z_2-b{}z_1][y^2 - 2 d{} b{} -b{}(1-b{})]z_1
- z_2 [y^2 -\xrond_\indxra y- d{} b{}]z_1
+ z_3 \left(y^2-2\xrond_\indxra y \right) z_1
$ \\[1em]

By completing the squares, we can find a real number $a{}$, depending ona bound on $|y|$, to obtain
$$
\frac{1}{2}\dot{\overparen{( z_2 - b{}  z_1)^2+( z_2 - z_3)^2 +  z_3^2}}
\leq 
-\frac{1}{2}\left(\begin{array}{@{\,  }ccc@{\,  }}
[z_2-b{} z_1]
&
z_2
&
z_3
\end{array}\right) S
\left(\begin{array}{@{\,  }c@{\,  }}
[z_2-b{} z_1]
\\
z_2
\\
z_3
\end{array}\right) + a{} z_1^2
$$
Substituting the $z_i$ by their definition (\ref{LP207}) in terms of the $e$ gives
\startmodif
$$
\frac{1}{2}\dot{\overparen{( y {\bar e}_\indxrb + b{}  {\bar e}_\indxra)^2+(
y {\bar e}_\indxrb - \xrond_\indxra {\bar e}_\indxrb
)^2 +  
\xrond_\indxra^2 {\bar e}_\indxrb^2}}
\leq 
-\frac{1}{2}\left(\begin{array}{@{\,  }ccc@{\,  }}
[y {\bar e}_\indxrb+b{} {\bar e}_\indxra]
&
y {\bar e}_\indxrb
&
\xrond_\indxra {\bar e}_\indxrb
\end{array}\right) S
\left(\begin{array}{@{\,  }c@{\,  }}
[y {\bar e}_\indxrb+b{} {\bar e}_\indxra]
\\
y {\bar e}_\indxrb
\\
\xrond_\indxra {\bar e}_\indxrb
\end{array}\right) + a{} {\bar e}_\indxra^2
\  .
$$
With (\ref{LP237}), this yields finally
\\[1em]$\displaystyle 
\frac{1}{2}\dot{\overparen{2a{}({\bar e}_\indxra^2+{\bar e}_\indxrb^2)+( y {\bar e}_\indxrb + b{}  {\bar e}_\indxra)^2
+( y  - \xrond_\indxra )^2{\bar e}_\indxrb^2 +  \xrond_\indxra ^2 {\bar e}_\indxrb^2}}
$\hfill \null \\[-0.7em]\null \hfill  $\displaystyle 
\leq 
-\frac{1}{2}\left(\begin{array}{@{\,  }ccc@{\,  }}
[y {\bar e}_\indxrb+b{} {\bar e}_\indxra]
&
y {\bar e}_\indxrb
&
\xrond_\indxra {\bar e}_\indxrb
\end{array}\right) S
\left(\begin{array}{@{\,  }c@{\,  }}
[y {\bar e}_\indxrb+b{} {\bar e}_\indxra]
\\
y {\bar e}_\indxrb
\\
\xrond_\indxra {\bar e}_\indxrb
\end{array}\right) - a{} {\bar e}_\indxra^2
\  .
$\\[1em]
We have obtained a quadratic form in the errors $({\bar e}_\indxra,{\bar e}_\indxrb)$ given by the matrix
$$
\bar P(y,\barxrond_\indxra)
\;=\;
\left(\begin{array}{@{\,  }cc@{\,  }}
2a{} +b{}^2 
&
yb{}
\\
yb{}
&
2a{} + 2y^2 + 2 \barxrond_\indxra^2 + 2 y \barxrond_\indxra 
\end{array}\right)
$$
\stopmodif
knowing that  the ``coefficients'' $(y,\xrond_\indxra)$ in $\bar P$ satisfy
$$
\dot y = \xrond_\indxra
\quad ,\qquad 
\dot \xrond_\indxra\;=\; -y \xrond_\indxrb
$$

\subsubsection{Return to the given coordinates}
The above
\startmodif
matrix
\stopmodif
has been constructed to show the flow generated by the reduced order observer
$$
\dot{{\hat \barxrond}}_\indxra 
 = -y[{\hat \barxrond}_\indxrb-\frac{y^2}{2}]-[{\hat \barxrond}_\indxra + y]
\quad ,\qquad 
\dot{{\hat \barxrond}}_\indxrb 
= y {\hat \barxrond}_\indxra+ y^2
$$
is contracting. In doing so, we have seen that $y$, and also $\barxrond_\indxra$ and $\barxrond_\indxrb$
may play the role of exogenous inputs. This observer is simply a copy of the given system in other 
coordinates. To insist on the role played by these exogenous inputs, we prefer to write
the given system is these other coordinates as~:
$$
\dot y\;=\; \barxrond_\indxra+ y
\quad ,\qquad 
\dot{{ \barxrond}}_\indxra 
 = -\entree _ y[\barxrond_\indxrb-\frac{y^2}{2}]-[\barxrond _\indxra+y]
\quad ,\qquad 
\dot{{ \barxrond}}_\indxrb 
= \entree _y [\barxrond_\indxra+ y]
$$
where $y$ as such is a state component, whereas $u_y$ is seen as an input.
Correspondingly, the metric is
\startmodif
$$
\bar P(\entree _y,\bar\entree _\indxra)
\;=\;
\left(\begin{array}{@{\,  }cc@{\,  }}
2a{} +b{}^2 
&
\entree _yb{}
\\
\entree _yb{}
&
2a{} + 2\entree _y^2 + 2 \bar \entree _\indxra^2 + 2 \entree _y \bar \entree _\indxra 
\end{array}\right)
$$
\stopmodif

Since the given coordinates are
$$
y= y
\quad ,\qquad 
\xrond_\indxra\;=\; \barxrond_\indxra + y
\quad ,\qquad 
\xrond_\indxrb\;=\; \barxrond_\indxrb -\frac{y^2}{2}
$$
we modify accordingly the input in
$$
\entree _y=\entree _y
\quad ,\qquad 
\entree _\indxra=\bar \entree _\indxra + \entree _y
$$
So coming back to the given coordinates, the expression of the augmented system is
$$
\dot y = \xrond_\indxra
\quad ,\qquad 
\dot{\xrond}_\indxra= 
-\entree _y \xrond_\indxrb
\quad ,\qquad 
\dot{\xrond}_\indxrb = 
[\entree _y  -y]\xrond_\indxra
$$
\startmodif
The one of the metric is
\\[1em]$\displaystyle  
P(y,\entree  _y,\entree  _\indxra)
$\hfill\null\\\null\hfill$\displaystyle
\begin{array}{rcl}
&=&
\left(\begin{array}{c}
1 
\\
0 
\\
0 
\end{array}\right)
c{}
\left(\begin{array}{@{\,  }ccc@{\,  }}
1 & 0 & 0
\end{array}\right)
\\[-0.5em]&&\qquad \qquad +
\left(\begin{array}{cc}
-1 & y
\\
1 & 0
\\
0 & 1
\end{array}\right)
\left(\begin{array}{@{\,  }ccc@{\,  }}
2a{} +b{}^2 
&
\entree _yb{}
\\
\entree _yb{}
&
2a{} + 2\entree _y^2 + 2 \entree _\indxra^2 - 2 \entree _y \entree _\indxra 
\end{array}\right)
\left(\begin{array}{ccc}
-1 &1 & 0
\\
y & 0 & 1
\end{array}\right)
\  ,
\\
&=&
\left(\begin{array}{c}
1 
\\
0 
\\
0 
\end{array}\right)
c{}
\left(\begin{array}{@{\,  }ccc@{\,  }}
1 & 0 & 0
\end{array}\right)
\\[-0.5em]&&\qquad \qquad +
\left(\begin{array}{@{\,  }cc@{\,  }}
-(2a{} +b{}^2 ) +  y\entree  _y b{}& 
-\entree  _y b{} + y(2a{} + 2\entree _y^2 + 2 \entree _\indxra^2 - 2 \entree _y \entree _\indxra )
\\
2a{} +b{}^2 
&
\entree _yb{}
\\
\entree _yb{}
&
2a{} + 2\entree _y^2 + 2 \entree _\indxra^2 - 2 \entree _y \entree _\indxra 
\end{array}\right)
\left(\begin{array}{ccc}
-1 &1 & 0
\\
y & 0 & 1
\end{array}\right)
\  ,
\\[2em]
&=&
\left(
\renewcommand{\arraystretch}{1.5}
\begin{array}{@{\,  }c|c|c@{\,  }}
\begin{array}{@{}l@{\quad }}
c{}+(2a{} +b{}^2 ) -2y\entree  _yb{}
\\\multicolumn{1}{@{\quad }r@{}}{
+ y^2(2a{} + 2\entree _y^2 + 2 \entree _\indxra^2 - 2 \entree _y \entree _\indxra )
}
\end{array}&
\star
& 
\star
\\\hline
-(2a{} +b{}^2) + y \entree _yb{}
&
2a{} +b{}^2 
&
\star
\\\hline
- \entree _yb{}
+y\left(2a{} + 2\entree _y^2 + 2 \entree _\indxra^2 - 2 \entree _y \entree _\indxra \right)
&
\entree _yb{}
&
2a{} + 2\entree _y^2 + 2 \entree _\indxra^2 - 2 \entree _y \entree _\indxra 
\end{array}\right)
\  .
\end{array}
$\\[1em]
\stopmodif
%
%

%
%
\sousection{Proof of convergence of the observer (\ref{LP206})}
\label{complement42}
Let us show that the flow generated by the system
\begin{eqnarray*}
\dot{\hat y}&=&\hat \xrond_\indxra - \frac{k_E}{c{}}[\hat y -y]
\\
\dot{\hat{\xrond}}_\indxra&=&
-y \hat{\xrond}_\indxrb
- \frac{k_E}{c{}}[\hat y -y]
\\
\dot{\hat{\xrond}}_\indxrb &=&
\left[\frac{k_E}{c{}}\hat y - \hat \xrond_\indxra\right]
[\hat y-y]
\end{eqnarray*} 
where $y$ is an exogenous input satisfying
$$
\dot y \;=\; \xrond_\indxra
\  ,
$$
\startmodif
is contracting with respect to the metric
\stopmodif
\\[1em]$\displaystyle 
P(\hat y,y,\xrond_\indxra)
$\hfill \null \\[1em] \null \hfill $
= \renewcommand{\arraystretch}{1.5}
\left(\begin{array}{@{\,  }c|c|c@{\,  }}
\begin{array}{@{}l@{\quad }}
c{}+(2a{} +b{}^2 ) -2\hat yyb{}
\\\multicolumn{1}{@{\quad }r@{}}{
+ \hat y^2(2a{} + 2y^2 + 2 \xrond _\indxra ^2 - 2 y \xrond _\indxra )
}
\end{array}&
-(2a{} +b{}^2)+\hat y yb{}
& 
-yb{}
+\hat y\left(2a{} + 2y^2 + 2 \xrond _\indxra ^2 - 2 y \xrond _\indxra  \right)
\\\hline
-(2a{} +b{}^2)+\hat y yb{}
&
2a{} +b{}^2 
&
yb{}
\\\hline
-yb{}
+\hat y\left(2a{} + 2y^2 + 2 \xrond _\indxra ^2 - 2 y \xrond _\indxra  \right)
&
yb{}
&
2a{} + 2y^2 + 2 \xrond _\indxra ^2 - 2 y \xrond _\indxra 
\end{array}\right)
\:  .
$\\[1em]
The Lie derivative of the metric along the vector field given by the system is
\\[1em]$\displaystyle 
\left(\begin{array}{@{\,  }ccc@{\,  }}
v_y & v_\indxra & v_\indxrb
\end{array}\right)
\left[
\frac{\partial P}{\partial \hat y}(\hat y,y,\xrond_\indxra)
\left(\hat \xrond_\indxra-\frac{k_E}{c{}}[\hat y-y]\right)
+
\frac{\partial P}{\partial y}(\hat y,y,\xrond_\indxra)
\xrond_\indxra
-
\frac{\partial P}{\partial \xrond  _\indxra}(\hat y,y,\xrond_\indxra)
y \xrond_\indxrb
\right]
\left(\begin{array}{@{\,  }c@{\,  }}
v_y \\ v_\indxra \\  v_\indxrb
\end{array}\right)
$\hfill \null \\[0.5em]\null \hfill $\displaystyle 
\;+\; 2 \left(\begin{array}{@{\,  }ccc@{\,  }}
v_y & v_\indxra & v_\indxrb
\end{array}\right)P(\hat y,y,\xrond _\indxra)
\left(
\begin{array}{@{\,  }ccc@{\,  }}
-\frac{k_E}{c{}} & 1 & 0
\\
-\frac{k_E}{c{}} & 0 & -y
\\
\frac{k_E}{c{}} [\hat y-y]+\frac{k_E}{c{}}\hat y - \hat \xrond_\indxra
&
y-\hat y
&
0
\end{array}\right)\left(\begin{array}{@{\,  }c@{\,  }}
v_y \\ v_\indxra \\ v_\indxrb
\end{array}\right)
\  .
$\\[1em]
We develop
\\[1em]$\displaystyle 
\left[
\frac{\partial P}{\partial \hat y}(\hat y,y,\xrond_\indxra)
\left(\hat \xrond_\indxra-\frac{k_E}{c{}}[\hat y-y]\right)
+
\frac{\partial P}{\partial y}(\hat y,y,\xrond_\indxra)
\xrond_\indxra
-
\frac{\partial P}{\partial \xrond  _\indxra}(\hat y,y,\xrond_\indxra)
y \xrond_\indxrb
\right]
$\hfill \null \\[1em]\null \hfill $\displaystyle 
=\; 
\left(\begin{array}{@{\,  }ccc@{\,  }}
-2yb{} + 4 \hat y(a{} + y^2 +  \xrond _\indxra ^2 -  y \xrond _\indxra )
&
y b{}
& 
\startmodif
2(a{} + y^2 +  \xrond _\indxra ^2 -  y \xrond _\indxra )
\stopmodif
\\
yb{}
&
0
&
0
\\
2\left(a{} +  y^2 +  \xrond _\indxra ^2 -  y \xrond _\indxra  \right)
&
0
&
0
\end{array}\right)
\left(\hat \xrond_\indxra-\frac{k_E}{c{}}[\hat y-y]\right)
$\hfill \null \\[1em]\null \hfill $\displaystyle 
\;+\; 
\left(\begin{array}{@{\,  }ccc@{\,  }}
-2\hat yb{} + 2\hat y^2(2 y - \xrond _\indxra )
&
\hat y b{}
& 
-b{}
+
2\hat y\left(2y - \xrond _\indxra  \right)
\\
\hat y b{}
&
0
&
b
\\
-b{}
+
2\hat y\left(2y - \xrond _\indxra  \right)
&
b{}
&
4y - 2 \xrond _\indxra 
\end{array}\right)
\xrond _\indxra
$\hfill \null \\[1em]\null \hfill $\displaystyle 
\;-\; 
\left(\begin{array}{@{\,  }ccc@{\,  }}
2\hat y^2(2 \xrond _\indxra  -  y  )
&
0
& 
2\hat y\left(2\xrond _\indxra  -  y  \right)
\\
0
&
0
&
0
\\
2\hat y\left(2\xrond _\indxra  -  y  \right)
&
0
&
4 \xrond _\indxra  - 2 y 
\end{array}\right) y \xrond_\indxrb
$\\[2em]$\displaystyle 
P(\hat y,y,\xrond _\indxra)
\left(
\begin{array}{@{\,  }ccc@{\,  }}
-\frac{k_E}{c{}} & 1 & 0
\\
-\frac{k_E}{c{}} & 0 & -y
\\
\frac{k_E}{c{}} [\hat y-y]+\frac{k_E}{c{}}\hat y - \hat \xrond_\indxra
&
y-\hat y
&
0
\end{array}\right)\;=\; 
$\\[1em]\null \hfill $\displaystyle 
\renewcommand{\arraystretch}{1.5}
\left(\begin{array}{@{\,  }c|c|c@{\,  }}
\begin{array}{@{}l@{\quad }}
c{}+(2a{} +b{}^2 ) -2\hat yyb{}
\\\multicolumn{1}{@{\quad }r@{}}{
+ \hat y^2(2a{} + 2y^2 + 2 \xrond _\indxra ^2 - 2 y \xrond _\indxra )
}
\end{array}&
-(2a{} +b{}^2)+\hat y yb{}
& 
-yb{}
+\hat y\left(2a{} + 2y^2 + 2 \xrond _\indxra ^2 - 2 y \xrond _\indxra  \right)
\\\hline
-(2a{} +b{}^2)+\hat y yb{}
&
2a{} +b{}^2 
&
yb{}
\\\hline
-yb{}
+\hat y\left(2a{} + 2y^2 + 2 \xrond _\indxra ^2 - 2 y \xrond _\indxra  \right)
&
yb{}
&
2a{} + 2y^2 + 2 \xrond _\indxra ^2 - 2 y \xrond _\indxra 
\end{array}\right)\times
$\hfill \null \\[1em]\null \hfill $
\times
\left(
\begin{array}{@{\,  }ccc@{\,  }}
-\frac{k_E}{c{}} & 1 & 0
\\
-\frac{k_E}{c{}} & 0 & -y
\\
\frac{k_E}{c{}} [\hat y-y]+\frac{k_E}{c{}}\hat y - \hat \xrond_\indxra
&
y-\hat y
&
0
\end{array}\right)
$\\[1em]
We get
\begin{eqnarray*}
\mbox{column 1}&=& \renewcommand{\arraystretch}{1.5}
\begin{array}{@{\,  }|c|@{\,  }} 
-k_E
+
\left[
-yb{}
+\hat y\left(2a{} + 2y^2 + 2 \xrond _\indxra ^2 - 2 y \xrond _\indxra  \right)
\right]
\left[
\frac{k_E}{c{}} [\hat y-y]
- \hat \xrond_\indxra
\right]
\\
+ y b{}\left[\frac{k_E}{c{}} [\hat y-y]
- \hat \xrond_\indxra\right]
\\
\left(2a{} + 2y^2 + 2 \xrond _\indxra ^2 - 2 y \xrond _\indxra  \right)
\left[
\frac{k_E}{c{}} [\hat y-y]
- \hat \xrond_\indxra
\right]
\end{array}
\\
\mbox{column 2}&=&\renewcommand{\arraystretch}{1.5}
\begin{array}{@{\,  }|c|@{\,  }} 
c{}+(2a{} +b{}^2 ) -\hat yyb{}
+
y
\left[-yb{}
+\hat y\left(2a{} + 2y^2 + 2 \xrond _\indxra ^2 - 2 y \xrond _\indxra  \right)
\right]
\\
-(2a{} +b{}^2)
+ b{} y^2
\\
-yb{}
+ y\left(2a{} + 2y^2 + 2 \xrond _\indxra ^2 - 2 y \xrond _\indxra  \right)
\end{array}
\\
\mbox{column 3}&=&\renewcommand{\arraystretch}{1.5}
\begin{array}{@{\,  }|c|@{\,  }} 
-y\left[-(2a{} +b{}^2)+\hat y yb{}\right]
\\
-y\left(2a{} +b{}^2 \right)
\\
- y^2 b{}
\end{array}
\end{eqnarray*}
It follows from these expressions that, when $v_y$ is zero, the Lie derivative is
\\[1em]$\displaystyle 
2\left[-(2a{}+b{}^2)+b{}y^2\right]v_\indxra^2
+
2\left[(2y-\xrond_\indxra)\xrond_\indxra
-
(2\xrond_\indxra-y)y\xrond_\indxrb-y^2 b{}\right]v_\indxrb^2
$\hfill \null \\\null \hfill $\displaystyle 
+
2\left[b{}\xrond_\indxra
-yb{}
+ y\left(2a{} + 2y^2 + 2 \xrond _\indxra ^2 - 2 y \xrond _\indxra  \right)
-y(2a{}+b{}^2)
\right]v_\indxra v_\indxrb
$\\[1em]$\displaystyle 
\;=\; 
\startmodif
-2\left[(2a{}+b{}^2)-b{}y^2\right]v_\indxra^2
\stopmodif
-
2\left[\xrond_\indxra^2+y^2 (b{}-\xrond_\indxrb)-2y\xrond_\indxra(1-\xrond_\indxrb)\right]v_\indxrb^2
$\hfill \null \\\null \hfill $\displaystyle 
+
2\left[b{}\xrond_\indxra
-yb{}-yb{}^2
+ y\left(2y^2 + 2 \xrond _\indxra ^2 - 2 y \xrond _\indxra  \right)
\right]v_\indxra v_\indxrb
$\\[1em]
To have a better view on the left hand side, we let~:
$$
\rho{} \;=\; v_\indxra
\quad ,\qquad 
r{} _y\;=\; y v_\indxrb
\quad ,\qquad 
r{} _\indxra\;=\; \xrond _\indxra v_\indxrb
$$
This yields
\\[1em]$\displaystyle 
\;=\; 
\startmodif
-2\left[(2a{}+b{}^2)-b{}y^2\right]\rho{} ^2
\stopmodif
-
2\left[r{} _\indxra^2+r{} _y^2 (b{}-\xrond_\indxrb)-2r{} _y r{} _\indxra(1-\xrond_\indxrb)\right]
$\hfill \null \\\null \hfill $\displaystyle 
+
2\left[
\startmodif
b{}r{} _\indxra
\stopmodif
-r{} _yb{}(b{}+1)
+ r{} _y\left(2y^2 + 2 \xrond _\indxra ^2 - 2 y \xrond _\indxra  \right)
\right]\rho{}  
$\\[1em]
We have recovered (\ref{LP173}). We know that, $(y,\xrond_\indxra,\xrond_\indxrb)$ being in 
$\Ouv_\varepsilon $, we can choose $b{}$ and then $a{}$ to make this expression negative definite in $(\rho{} 
,r{} _\indxra,r{} _\indxrb)$ and therefore in $(v_\indxra,v_\indxrb)$ when 
$(y,\xrond_\indxra,\xrond_\indxrb)$ is in $\Ouv_\varepsilon $.
Then with, say $c=k_E$, for all $k_E$ large enough, the full expression is negative 
definite in $(v ,v_\indxra,v_\indxrb)$.
%
%

%
%
\sousection{More comments on Dynamic Extension}
\label{complement31}
\subsubsection{Choose properly the output.}
Instead of
\begin{equation}
\label{trav19}
\ya=\ha(x,\xe))
\end{equation}
with $\ha$ satisfying (\ref{LP179}), we could take the output as
\begin{equation}
\label{trav20}
\ya\;=\; \left(\begin{array}{@{}c@{}}
\ha(x,\xe) \\ \xe
\end{array}\right)
\end{equation}
since $\xe$ is known to be constantly zero along any solution of the given system.
We show in the following
example that, if a metric $\bfP$, for the original $\bfx$-manifold $\RR^n$, does not satisfy Condition A3,
dynamic extension is of no help with this choice of augmented output.

\begin{example}\bgroup\normalfont
\label{travex1}
Consider $\RR^2$ equipped with a metric $\bfP$ which possibly satisfies Condition A2 but not Condition A3. 
For this reason we try to go via a dynamic extension.

According to \cite[Theorem 2.6]{127}), there exist local coordinates $(y,\xrond)$ for $\bfx$ in which the expression $P$ of 
$\bfP$ takes the diagonal form~:
\begin{equation}
\label{LP208}
P(y,\xrond)\;=\; \left(\begin{array}{cc}
P_{y}(y,\xrond) & 0 
\\
0 & P_{\xrond}(y,\xrond) 
\end{array}\right)
\  .
\end{equation}
\startmodif
Assume that the level sets of the output function are not 
\stopmodif
totally geodesic. Then, since
the level sets of the output function
are totally geodesic for all $\bfy$ in $\bfh(\Ouv)$ if and only if
the block $P_{\xrond}$ does not depend on $y$,
$\frac{\partial P_{\xrond}}{\partial y}(y,\xrond)$ is not zero everywhere.


We add an extra dimension and consider the augmented state $\bfxa$ in $\RR^3$
with coordinates
$$
\xa=(y,\xrond,\xe)
\  .
$$

If we choose $y$ and $\xe$ as outputs, the problem we are interested in is to find a metric $\bfPa$ in $\RR^3$
such that
\begin{list}{}{%
\parskip 0pt plus 0pt minus 0pt%
\topsep 0.5ex plus 0pt minus 0pt%
\parsep 0pt plus 0pt minus 0pt%
\partopsep 0pt plus 0pt minus 0pt%
\itemsep 0.5ex plus 0pt minus 0pt
\settowidth{\labelwidth}{\quad a)}%
\setlength{\labelsep}{0.5em}%
\setlength{\leftmargin}{\labelwidth}%
\addtolength{\leftmargin}{\labelsep}%
}
\item[a)]
Condition A3 holds for $\bfPa$ and therefore
the sets where $y$ and $\xe$ are constant are totally geodesic, 
\item[b)]
the expression of the pull-back metric of $\bfPa$ by
the projection $\pi : (y,\xrond,\xe) \mapsto (y,\xrond)$
is (\ref{LP208}).
\end{list} 
The requirement b) means that if $\Pa$ is the expression of $\bfPa$ in the coordinates
$(y,\xrond,\xe)$
$$
\Pa(y,\xrond,\xe) 
\;=\; \left(\begin{array}{ccc}
\Pa_{yy}(y,\xrond,\xe) & \Pa_{y \xrond }(y,\xrond,\xe) & \Pa_{y \xe}(y,\xrond,\xe)
\\
\Pa_{y \xrond }(y,\xrond,\xe) & \Pa_{\xrond\xrond}(y,\xrond,\xe) & \Pa_{\xrond \xe}(y,\xrond,\xe)
\\
\Pa_{y \xe}(y,\xrond,\xe) & \Pa_{\xrond \xe}(y,\xrond,\xe) & \Pa_{\xe \xe}(y,\xrond,\xe)
\end{array}\right)
$$
then we have 
\begin{eqnarray*}
\left(\begin{array}{cc}
P_{y}(y,\xrond) & 0 
\\
0 & P_{\xrond}(y,\xrond) 
\end{array}\right) 
&=&
\left(\begin{array}{@{}c@{}}
\displaystyle
\frac{\partial  \pi }{\partial y}(y,\xrond,0) ^\top
\\[01em]
\displaystyle 
\frac{\partial   \pi }{\partial \xrond}(y,\xrond,0) ^\top
\end{array}\right)
\Pa(y,\xrond,0) 
\left(\begin{array}{@{}cc@{}}
\displaystyle 
\frac{\partial   \pi }{\partial y}(y,\xrond,0) 
&
\displaystyle 
\frac{\partial \pi }{\partial \xrond}(y,\xrond,0) 
\end{array}\right)
\  ,
\\
&=&
 \left(\begin{array}{@{}ccc@{}}
1 & 0 & 0 \\ 0 & 1 & 0
\end{array}\right)
\Pa(y,\xrond,0) 
\left(\begin{array}{@{}cc@{}}
1 & 0  \\ 0  & 1  \\ 0 & 0
\end{array}\right)
\  .
\end{eqnarray*}
This gives
$$
\Pa_{yy}(y,\xrond,0) \;=\; 
P_{y}(y,\xrond) 
\quad ,\qquad 
\Pa_{y \xrond }(y,\xrond,0)\;=\; 0 
\quad ,\qquad 
\Pa_{\xrond\xrond}(y,\xrond,0) \;=\; P_{\xrond}(y,\xrond) 
$$
and yields
$$
\frac{\partial \Pa_{y \xrond }}{\partial \xrond}(y,\xrond,0)=0
\quad ,\qquad \frac{\partial \Pa_{\xrond \xrond }}{\partial y}(y,\xrond,0)= 
\frac{\partial P_{\xrond }}{\partial y}(y,\xrond)
\ne 0
\qquad \forall (y,\xrond)
\  .
$$

On another hand if the sets where $y$ and $\xe$ are constant
are totally geodesic,
the Christoffel symbols
$\Gamma_{a\xrond\xrond}^y$ and $\Gamma_{a\xrond\xrond}^{\xe}$ of $\Pa$ are zero. 
With evaluating at the point $(y,\xrond,0)$, we obtain
$$\displaystyle 
\left(\begin{array}{c@{\qquad }c@{\qquad }c}\displaystyle
-\frac{\partial \Pa_{\xrond\xrond}}{\partial y}
& \displaystyle
\frac{\partial \Pa_{\xrond\xrond}}{\partial \xrond}
& \displaystyle
2\frac{\partial \Pa_{\xrond \xe}}{\partial \xrond}
-
\frac{\partial \Pa_{\xrond\xrond}}{\partial \xe}
\end{array}\right)
\renewcommand{\arraystretch}{1.5}
\left(\begin{array}{c@{\quad }c}
\Pa_{\xrond\xrond}\Pa_{\xe\xe}-\Pa_{\xrond \xe}^2
&
-\Pa_{\xrond\xrond}\Pa_{y \xe}
\\
\Pa_{\xrond \xe} \Pa_{y \xe} 
&
-\Pa_{yy}\Pa_{\xrond \xe}
\\
-\Pa_{\xrond\xrond}\Pa_{y \xe}
&
\Pa_{yy}\Pa_{\xrond\xrond}
\end{array}\right)\;=\; 0
\  .
$$
We can view this as a system of $2$ linear equations in the $2$ unknowns 
$\Pa_{y \xe}$ and $\Pa_{\xe \xe}$, i.e.,
$$\displaystyle 
\renewcommand{\arraystretch}{1.5}
\left(\begin{array}{c@{\quad}c}
c 
&\displaystyle 
-\frac{\partial \Pa_{\xrond\xrond}}{\partial y}
\Pa_{\xrond\xrond}
\\[0.5em]\displaystyle 
\frac{\partial \Pa_{\xrond\xrond}}{\partial y}
\Pa_{\xrond\xrond}
&
0
\end{array}\right)
\left(\begin{array}{c}
\Pa_{y \xe}
\\[0.5em]
\Pa_{\xe \xe}
\end{array}\right)
=
\left(\begin{array}{@{}c@{}}
\displaystyle 
-\frac{\partial \Pa_{\xrond\xrond}}{\partial y}
\Pa_{\xrond \xe}^2
\\[0.5em]
c
\Pa_{yy}
\end{array}\right)
$$
with the notation
$$
c (y,\xrond)= 
\frac{\partial \Pa_{\xrond\xrond}}{\partial \xrond}(y,\xrond,0)
\Pa_{\xrond \xe}(y,\xrond,0)
-
\left[2\frac{\partial \Pa_{\xrond \xe}}{\partial \xrond}(y,\xrond,0)
-\frac{\partial \Pa_{\xrond\xrond}}{\partial \xe}(y,\xrond,0)
\right] \Pa_{\xrond\xrond}(y,\xrond,0)
\  .
$$
But there exist pairs $(y,\xrond)$ where $\frac{\partial \Pa_{\xrond\xrond}}{\partial y}(y,\xrond,0)
\Pa_{\xrond\xrond}(y,\xrond,0)$ is nonzero. For these, we can solve for 
to obtain
$$
\left(\begin{array}{c}
\Pa_{y \xe}
\\[1em]
\Pa_{\xe \xe}
\end{array}\right)
= \left(\begin{array}{c}
\displaystyle
\frac{c\Pa_{yy}}{
\frac{\partial \Pa_{\xrond\xrond}}{\partial y}
\Pa_{\xrond\xrond}
}
\\[1em]\displaystyle 
\frac{\Pa_{\xrond \xe}^2}{\Pa_{\xrond \xrond}}
+ \frac{c^2\Pa_{y y}}{\frac{\partial \Pa_{\xrond\xrond}}{\partial y}^2 \Pa_{\xrond \xrond}^2}
\end{array}\right)
=
\left(\begin{array}{c}
\displaystyle 
\frac{c\Pa_{yy}}{
\frac{\partial \Pa_{\xrond\xrond}}{\partial y}
\Pa_{\xrond\xrond}
}
\\[1em]\displaystyle 
\frac{\Pa_{\xrond \xe}^2}{\Pa_{\xrond \xrond}} 
+ \frac{\Pa_{y \xe}^2}{\Pa_{yy}}
\end{array}\right)
$$
With these expressions the determinant of $\Pa$ evaluated at $(y,\xrond,0)$ is
\begin{eqnarray*}
\Pa_{yy} \Pa_{\xrond\xrond} \Pa_{\xe\xe} - \Pa_{y\xe}^2\Pa_{\xrond\xrond} -
\Pa_{yy} \Pa_{\xrond\xe}^2
& = & 
\Pa_{yy} \Pa_{\xrond\xrond} 
\left(
\frac{\Pa_{\xrond \xe}^2}{\Pa_{\xrond \xrond}} 
+ \frac{\Pa_{y \xe}^2}{\Pa_{yy}}
\right)
 - \Pa_{y\xe}^2\Pa_{\xrond\xrond} -
\Pa_{yy} \Pa_{\xrond\xe}^2
\\
&=& 0
\  .
\end{eqnarray*}
This contradicts the fact that $\Pa$ is positive definite.
We conclude that, when the output is (\ref{trav20}), there is no metric $\bfPa$ for $\RR^3$
satisfying Condition A3 and the pull-back of which in the initial 
$\bfx$-manifold does not satisfy this condition.
\egroup\end{example}

\subsubsection{Possible procedures for designing an observer via a dynamic extension}
In the standard procedures for designing an observer via a dynamic extension, the choice of the function
$\fa$ and $\ha$ is actually indirect.  The first step consists in finding an injective immersion
$\bfimmer:\bfRR^n\to \bfRR^m$
mapping the given state $\bfx$ into an augmented state $\bfxa$ for which
there exist special coordinates $\xa$ the dynamics of which can be written in a so called normal form.
For example,
when such a form is a triangular, the high gain observer paradigm applies.
When it is a Hurwitz form, the
Luenberger paradigm applies.

\paragraph{\textcolor{my-violet}{Triangular normal form.}}~\\
\label{sec2}
For systems (\ref{eqn:Plant1}) with $p=1$ and which are strongly differentially observable with an 
order $m$, strictly larger than $n$,
on an open set $\Ouv $
(compare with Example \ref{ex4}), 
the function 
$\bfimmer_{m}:\Ouv  \to \RR^{m}$~:
$$
\bfimmer_{m}(\bfx)
\;=\; \left(\begin{array}{c}
\bfh(\bfx) \\ 
L_\bff \bfh(\bfx) \\ \vdots \\ 
L_\bff^{m-1}\bfh(\bfx)
\end{array}\right)
$$
is (by definition) an injective immersion from $\Ouv$ to $\RR^m$.
Let $x$ be coordinates for $\bfx$ in $\bfRR^n$, $y$ for $\bfy$ in $\RR^p$, $\xe$ for $\bfxe$ in 
$\bfRR^{m-n}$ and $\xa$ for $\bfxa$ in 
$\bfRR^m$ such that there is a diffeomorphism
$\oha$ satisfying
$$
\xa=\oha(x,\xe)
\quad ,\qquad 
\oha(x,0)\;=\; \immer_{m}(x)
\  .
$$
Then, with the notations
\begin{eqnarray}
\nonumber
&\displaystyle 
C\;=\; \left(\begin{array}{ccccc}
1&0&\ldots & \ldots& 0
\end{array}\right)
\quad ,\quad 
A\;=\; \left(\begin{array}{ccccc}
0 & 1 & 0 & \ldots & 0
\\
0 & 0 & 1 &\ddots & \vdots
\\
\vdots && \ddots&\ddots&0
\\
\vdots&&&\ddots&1
\\
0 & \ldots& \ldots& \ldots & 0
\end{array}\right)
\quad ,\quad 
B\;=\; \left(\begin{array}{c}
0 \\ \vdots\\\vdots\\0 \\ 1
\end{array}\right)
\  ,
\end{eqnarray}
the expression
$\immer_{m}(x)$ of $\bfimmer_{m}(\bfx)$ satisfies
$$
\dot{\overparen{\immer_{m}(x)}}
\;=\; 
A\,  \immer_{m}(x)\;+\; B\,  L_f^{m}h(x)
\quad ,\qquad 
y\;=\; C\,  \immer_{m}(x) \;=\; h(x)
$$
In this case, the  augmented dynamics are naturally
$$
\vardot{\xa}{x}\;=\; A \xa + B [L_f^m h(x)]
\quad ,\qquad 
y\;=\; C \xa
\  .
$$

From what is known on high gain observer, Condition A2 holds for these augmented dynamics with a metric $\bfPa$ the expression $\Pa$ of which, in the 
coordinates $\xa$, is a constant matrix. Then, because for these 
augmented dynamics, the output function is linear, condition A3 is also satisfied. 

\paragraph{\textcolor{my-violet}{Hurwitz form.}}~\\
With referring to \cite[Chapter 3]{Bernard}, we know that, under weak smoothness 
and observability assumptions, for any $m\times m$ 
Hurwitz matrix $A$ and any $m\times p$ matrix $B$ such that the pair $(A,B)$ is controllable, 
there exists $\immer:\RR^n\to \RR^m$ satisfying
$$
\dot{\overparen{\immer (x)}}\;=\; A\,\immer(x) + B\,  h(x)
$$
which is an injective immersion at least when $m\geq 2n+1$.
In this case, when expressed in the same coordinates as
$\immer$, the  augmented dynamics are
$$
\vardot{\xa}{x}\;=\; A\,  \xa + B\,  y
\  .
$$
Because $A$ is Hurwitz, the flow generated by this system
is strongly contracting for an Euclidean metric (constant in the coordinates $\barxa$). In such a case we 
do not even need condition A3 to be satisfied by the augmented dynamics.

\sousection{A reduced order observer}
\label{complement50}
In \cite{127}, we have shown that  strong differential detectability implies the existence of a locally 
convergent reduced order observer. Now, we elaborate on the fact that Condition A3 together with Condition A2 implies
the existence of a convergent reduced order observer of a special form.
This is a consequence of the fact that, when the correction term is a gradient like, then Condition A3 is 
equivalent to the infinite gain margin property. But by increasing the gain, we introduce a time scale 
separation.

Specifically, under Assumption~\ref{H1}, the function $\bfh:\bfRR^n\to \bfRR^p$ is a submersion on
$\Ouv\subset\bfRR^n$. For any $\bfx_0$ in $\Ouv$, let $\coordyxr$ be a 
coordinate chart around $\bfx_0$.
We write the expression of the metric $\bfP$ in these coordinates as~:
$$
P(y,\xrond)\;=\; \left(\begin{array}{cc}
P_{yy}(y,\xrond) & P_{y\xrond }(y,\xrond)
\\
P_{\xrond y}(y,\xrond) & P_{\xrond\xrond}(y,\xrond) 
\end{array}\right)
\  .
$$
The expression of the system is
\begin{equation}
\label{LP108}
\dot y\;=\; f_y(y,\xrond)
\quad ,\qquad 
\dot \xrond\;=\; f_\xrond(y,\xrond)
\end{equation}
and the one for the observer  is
\begin{eqnarray}\label{eqn:hatyNEW}
\dot {\hat y}&=&f_y(\hat y,\hat \xrond) - \ell\,  k_E^*(\hat y,\hat \xrond)
\left[\grad_Ph\right]_y(y,\xrond)\frac{\partial \wp }{\partial y_a}(\hat y,y)^\top
\\ \label{eqn:hatxrondNEW}
\dot{\hat \xrond}&=&f_\xrond(\hat y,\hat \xrond)
-\ell\,  k_E^*(\hat y,\hat \xrond) \left[\grad_Ph\right]_\xrond(y,\xrond)\frac{\partial \wp }{\partial y_a}(\hat y,y)^\top
\  ,
\end{eqnarray}
where we have chosen $\ell k_E^*$, with $\ell$ any real number in $[0,+\infty)$, as function $k_E$ in (\ref{eqn:GeodesicObserverVectorField})
and with the notations
\begin{equation}
\label{LP70}
\begin{array}{@{}r@{\; }c@{\; }l@{}}
\left[\grad_Ph\right]_y(y,\xrond)&=&P_y(y,\xrond)^{-1}
\\[0.5em]
\!\left[\grad_Ph\right]_\xrond(y,\xrond)
&=&
-P_{\xrond\xrond}(y,\xrond)^{-1}
P_{\xrond y}(y,\xrond)
P_y(y,\xrond)^{-1}
\\[0.5em]
P_y(y,\xrond)&=&
P_{yy}(y,\xrond)-P_{y\xrond }(y,\xrond) P_{\xrond\xrond}(y,\xrond)^{-1}P_{\xrond y}(y,\xrond)\  .
\end{array}
\end{equation}
where the function $\wp $ is given by condition A3.

To go on, we assume the function
$
(y_1,y_2)\mapsto \left(
\frac{\partial \wp }{\partial y_1} (y_1,y_2),y_2\right)
$
is a diffeomorphism\footnote{
This is the case when $\wp$ is square of distance given by a Riemannian metric on the $\bfy$-manifold
such that we have uniqueness of minimizing geodesics. Indeed, for $y_2$ fixed, the unique minimizing
geodesic $\gammay $ between $y_1$ and 
$y_2$ parameterized in such a way that $y_2=\gammay (0)$ and $y_1=\gammay (1)$ is such that $y_1\mapsto 
v_2=\frac{d\gamma }{ds}(0)=\phi_e(y_1)$ is a diffeomorphism, given by the exponential map. On another hand 
the geodesic flow gives a diffeomorphism $(y_1,v_1)=\phi_{gy}(y_2,v_2),\phi_{gv}(y_2,v_2))$. So with $y_2$ 
fixed, we have $v_1=\frac{d\gamma }{ds}(1)=\phi_{gv}(y_2,\phi_e(y_1))$. Since we have
$\frac{\partial \phi_e}{\partial y_1}$ invertible and $\phi_{gy}(y_2,\phi_e(y_1))=y_1$, we can conclude that
$y_1\mapsto v_1$ is a diffeomorphism on $h(\coordyxrm(\coordyxrd))$, maybe after reducing $\coordyxrd$. Finally, we have
$
\frac{\partial \wp}{\partial y_1}=v_1^\top P_{yy}(y_1)
$.
}   on 
$h(\coordyxrm(\coordyxrd))\times h(\coordyxrm(\coordyxrd)))$.
With (\ref{LP65}) and continuity, this implies that the matrix
$
\frac{\partial ^2\wp 
}{\partial y_1^2}(\hat y,y)
$
is positive definite.
In this case we have
\begin{equation}
\label{LP74}
\wp (y,y)\;=\; 0
\quad ,\qquad 
\frac{\partial \wp }{\partial y_1}(y,y)\;=\; 0
\quad ,\qquad 
\frac{\partial ^2\wp }{\partial y_1^2}(y,y)\;+\; 
\frac{\partial ^2 \wp }{\partial y_1\partial y_2}(y,y)\;=\; 0
\  
\end{equation}
and the partial inverse of the diffeomorphism denoted~:
$$
y_1\;=\; \wp _1^{-1}\left(
\frac{\partial \wp }{\partial y_1} (y_1,y_2),y_2\right)
$$
satisfies~:
$$
y\;=\; \wp _1^{-1}\left(0,y\right)
\  .
$$

To express the observer in the standard form of the singular perturbation theory, we
replace the coordinate $\hat y$ by $\yrond$ defined as
$$
\yrond =\ell{}\,  \frac{\partial \wp }{\partial y_1} (\hat y,y)^\top
\  .
$$
The expression of the system-observer dynamics are then~:
\begin{equation}
\label{LP73}
\begin{array}{@{}rcl@{}}
\frac{1}{\ell{}}\,  \dot \yrond
&=&\displaystyle 
\frac{\partial ^2\wp }{\partial y_1^2}(\wp _1^{-1}\left(\frac{\yrond  }{\ell},y\right),y)
\left[
\vrule height 1.2em depth 1.2em width 0pt
f_y\left(\wp _1^{-1}\left(\frac{\yrond   }{\ell},y\right),\hat \xrond   \right)
\right.
\hskip 6.5cm\null 
\\\multicolumn{3}{r}{\displaystyle 
\left.\vrule height 1.2em depth 1.2em width 0pt
\;-\;
\bar k_E^*\left(\wp _1^{-1}\left(\frac{\yrond   }{\ell},y\right),\hat \xrond \right)
\left[\grad _Ph\right]_y 
\!\left(\wp _1^{-1}\left(\frac{\yrond   }{\ell},y\right),\hat \xrond \right)
\: \yrond   \right]\qquad 
}
\\\multicolumn{3}{r}{\displaystyle 
\;+\; 
\frac{\partial ^2\wp }{\partial y_1 \partial y_2}\left(\wp _1^{-1}\left(\frac{\yrond   }{\ell},y\right),y\right)f_y( y, \xrond   )
\  ,
}
\\[1.5em]
\dot{\hat \xrond    }&=&\displaystyle 
f_\xrond   \left(\wp _1^{-1}\left(\frac{\yrond   }{\ell},y\right),\hat \xrond    \right)\;-\; 
\bar k_E^*\left(\wp _1^{-1}\left(\frac{\yrond   }{\ell},y\right),\hat \xrond \right)
\left[\grad _Ph\right]_\xrond \!\left(\wp _1^{-1}\left(\frac{\yrond   }{\ell},y\right),\hat \xrond \right)\:   \yrond   
\  ,
\\[0.5em]
\dot y &=& f_y(y,\xrond    )
\  ,
\\[0.5em]
\dot \xrond    &=&f_\xrond    (y,\xrond    )
\end{array}
\end{equation}

Since $k_E^*$ has strictly positive values and $\left[\grad_Ph\right]_y(y,\xrond)$ is an invertible matrix (see (\ref{LP70})), with (\ref{LP74}), the algebraic equation
associated with the singularly perturbed system
is given by the $\dot \yrond$ equation when $\ell$ 
is infinity as
\begin{equation}
\label{LP72}
\yrond  \;-\; 
\frac{\left[\grad _Ph\right]_y \!(y,\hat \xrond )^{-1}
}{
\bar k_E^*\left(\wp _1^{-1}\left(\frac{\yrond   }{\ell},y\right),\hat \xrond \right)
}
\left[\vrule height 0.6em depth 0.6em width 0pt
f_y(y,\hat\xrond)-f_y(y,\xrond)\right]
\;=\; 0
\  .
\end{equation}
It has trivially a single root (in $\yrond   $) which  is also non singular.
The slow system we obtain is~:
\begin{equation}
\label{LP71}
\begin{array}{rcl}
\dot{\hat \xrond    }&=&f_\xrond   (y,\hat \xrond    ) 
\;+\; 
P_{\xrond\xrond}(y,\hat \xrond)^{-1}
P_{\xrond y}(y,\hat \xrond)
\left[\vrule height 0.6em depth 0.6em width 0pt
f_y(y,\hat\xrond)-f_y(y,\xrond)\right]
\  ,
\\[0.5em]
\dot y &=& f_y(y,\xrond    )
\  ,
\\[0.5em]
\dot \xrond    &=&f_\xrond    (y,\xrond    )
\end{array}
\end{equation}
where we have used (\ref{LP70}), 
and the above root $\yrond   $ to express  the $\dot {\hat \xrond}$ equation.
The interesting feature of (\ref{LP71}) is that its
$\dot {\hat \xrond}$ component
\begin{equation}
\label{2}
\dot{\hat \xrond    }\;=\; f_\xrond   (y,\hat \xrond    ) 
\;+\; 
P_{\xrond\xrond}(y,\hat \xrond)^{-1}
P_{\xrond y}(y,\hat \xrond)
\left[\vrule height 0.6em depth 0.6em width 0pt
f_y(y,\hat\xrond)-f_y(y,\xrond)\right]
\end{equation}
appears as a full order observer for 
the reduced order (fictitious) system
\begin{equation}
\label{LP90}
\dot{\hat \xrond    }\;=\; f_\xrond   \left(y(t),\hat \xrond    \right)
\quad ,\qquad 
{y\!_f}\;=\; f_y(y(t),\xrond    )
\end{equation}
with ${y\!_f}$ as measurement and $y$ as an exogenous input 
given by the ``dynamic output feedback''~:
\begin{equation}
\label{LP91}
\dot y(t)\;=\; {y\!_f}(t)
\  .
\end{equation}
Its correction term is proportional to the difference $f_y(y,\hat \xrond  )
\;-\;
f_y( y, \xrond   )$. This appears naturally as a consequence of our various restrictions.
Also, it does not have quite a gradient like correction term. We do have the inverse of the 
matrix $P_{\xrond\xrond}$. 
But, instead of $\frac{\partial f_y}{\partial \xrond}$, we have $-P_{\xrond y}$.

If Conditions A2 and A3 are satisfied, the original observer (\ref{eqn:hatyNEW})-(\ref{eqn:hatxrondNEW}
 is convergent whatever $\ell$ is in $[1,+\infty)$. So we expect the observer
in \eqref{LP71} with state in $\RR^{n-p}$ to also be convergent  under the same assumptions.

\begin{proposition}
\label{prop18}
Assume $\bfh$ is a submersion and there exist a complete Riemannian metric $\bfP$,
and an open $\Ouv$ subset of $\RR^n$, such that Condition A2 holds and
\begin{description}
\item[\normalfont\textit{A3.1~:}]
\startmodif
All sets $\mathfrak{H}(\bfy) \cap \Ouv$ for $\bfy$ in $\bfh(\Ouv)$ are strongly geodesically convex.
\stopmodif
\end{description}
Then, for any point $\bfx_0$ in $\Ouv$,
there exists a coordinate chart $\coordyxr$ around $\bfx_0$ such that, along the solutions of (\ref{LP71}), we have~:
\begin{equation}
\label{eqn:ReducedContraction}
\Did d((y,\hat \xrond),(y,\xrond))
\; \leq \;  \displaystyle -
\frac{\qlower}{2} 
\,  d((y,\hat \xrond),(y,\xrond))
\qquad \forall ((y,\hat \xrond),(y,\xrond))\in \coordyxrm(\coordyxrd)^2
\  .
\end{equation}
\end{proposition}

\par\vspace{1em}
We remind the reader that $A3$ implies $A3.1$.
\startmodif
See \cite[Proposition A3 1)b)]{57}.
\stopmodif
\par\vspace{1em}
\begin{proof}
\startmodif
Let $\coordyxr$ be an arbitrary a coordinate chart  around $\bfx_0$. May be with reducing $\coordyxrd$,
we can assume this set is weakly geodesically convex. It follows that, for any
$(y,\hat \xrond)$ and $(y,\xrond)$ in $\coordyxrm(\coordyxrd)$ and therefore in the same level set
$\mathfrak{H}(y) \cap \Ouv$, there there exists a minimizing geodesic $\gamma ^*$
joining these points and remaining in $\coordyxrm(\coordyxrd) \cap \mathfrak{H}(y) \cap \Ouv$. Hence
$\gamma ^*$ can be decomposed as
\stopmodif
$$
\gamma ^*(s)\;=\; \left(\gammay ^*(s),\gamma _\xrond^*(s)\right)
$$
where
$$
\gammay ^*(s)\;=\; y\qquad \forall s\in [s,\hat s]
\  .
$$
This yields
\\[1em]$\displaystyle 
\left(\begin{array}{@{}cc@{}}
\displaystyle \frac{d\gammay ^*}{ds}(s)^\top
&\displaystyle \frac{ d \gamma _\xrond^*}{ds}(s)^\top
\end{array}\right)
\renewcommand{\arraystretch}{1.5}
\left(\begin{array}{@{}cc@{}}
P_{yy}\left(\gammay ^*(s),\gamma _\xrond^*(s)\right)
&
P_{y\xrond}\left(\gammay ^*(s),\gamma _\xrond^*(s)\right)
\\
P_{\xrond y}\left(\gammay ^*(s),\gamma _\xrond^*(s)\right)
&
P_{\xrond\xrond}\left(\gammay ^*(s),\gamma _\xrond^*(s)\right)
\end{array}\right)
\left(\begin{array}{@{}c@{}}
\displaystyle \frac{d\gammay ^*}{ds}(s)
\\
\displaystyle \frac{ d \gamma _\xrond^*}{ds}(s)
\end{array}\right)
$\hfill \null \\[0.5em]\null \hfill $\displaystyle 
\;=\; 
\frac{ d \gamma _\xrond^*}{ds}(s)^\top
P_{\xrond\xrond}\left(\gammay ^*(s),\gamma _\xrond^*(s)\right)
\frac{ d \gamma _\xrond^*}{ds}(s)
$\\[1em]
and
$$
d((y,\hat \xrond),(y,\xrond))
\; = \; 
\int_{s}^{\hat s} \sqrt{
\frac{d \gamma _\xrond ^*}{d s}(s)^\top
P_{\xrond\xrond}(y,\gamma _\xrond ^*(s))\,  
\frac{d \gamma _\xrond ^*}{ds}(s)}\:  ds
$$
Also with $\frac{d \gammay ^{}}{ds}(s)=0$, the geodesic equation gives
\begin{eqnarray}
\label{LP87}
2\frac{d}{ds}\left\{
P_{y\xrond}(y,\gamma _\xrond ^*(s))\,  \frac{d\gamma _\xrond ^*}{ds}(s)\right\}
&=&\frac{\partial}{ \partial y}\left\{\frac{d\gamma _\xrond ^*}{ds}(s)^\top P_{\xrond\xrond}(y , 
\gamma _\xrond ^*(s))\frac{d\gamma _\xrond ^*}{ds}(s)\right\}^\top
\\
\label{LP86}
2\frac{d}{ds}\left\{
P_{\xrond\xrond}(y,\gamma _\xrond ^*(s))\,  \frac{d\gamma _\xrond ^*}{ds}(s)\right\}
&=&\left.\frac{\partial}{ \partial \xrond}\left\{\frac{d\gamma _\xrond ^*}{ds}(s)^\top P_{\xrond\xrond}(y , 
\xrond ^*)\frac{d\gamma _\xrond ^*}{ds}(s)\right\}\right|_{\xrond=\gamma _\xrond ^*(s)}^\top
\end{eqnarray}

Let $(Y(y,\xrond,t),\mathcal{X}(y,\xrond,t)$ be the solution of (\ref{LP108}) issued from $(y,\xrond)$ and 
taking values in the open set $\coordyxrm(\coordyxrd)$ for $t$ in $[0,T)$.
Let  $(t,s)\mapsto \Gamma _\xrond(t,s)$ be solution of
\\[1em]$\displaystyle 
\frac{\partial \Gamma _\xrond }{\partial t}(t,s)\;=\;  f_\xrond   (Y(y,\xrond,t),\Gamma _\xrond   (t,s)  ) 
$\hfill\null \\\null \hfill $\displaystyle
+
P_{\xrond\xrond}(Y(y,\xrond,t),\Gamma _\xrond (t,s))^{-1}
\!
P_{\xrond y}(Y(y,\xrond,t),\Gamma _\xrond )
\!
\left[
f_y(Y(y,\xrond,t),\Gamma _\xrond  (t,s) )
-
f_y(Y(y,\xrond,t),\mathcal{X}(y,\xrond,t) )
\right]
$\\[1em]
with initial condition
$$
\Gamma _\xrond(0,s)\;=\; \gamma _\xrond^*(s)
\  .
$$
From the definition of Dini derivative, we have
$$
\Did d((y,\hat \xrond),(y,\xrond))
\; \leq \; 
\left.\frac{d}{dt}\left\{\int_{s}^{\hat s} \sqrt{
\frac{\partial \Gamma _\xrond }{\partial s}(t,s)^\top
P_{\xrond\xrond}(Y(y,\xrond,t),\Gamma _\xrond (t,s))\,  
\frac{\partial \Gamma _\xrond }{\partial s}(t,s)}\:  ds\right\}\right|_{t=0}
$$
But (\ref{LP86}) gives
\\[1em]$\displaystyle 
\left.\frac{d}{dt}\left\{
\frac{\partial \Gamma _\xrond }{\partial s}(t,s)^\top
P_{\xrond\xrond}(Y(y,\xrond,t),\Gamma _\xrond (t,s))\,  
\frac{\partial \Gamma _\xrond }{\partial s}(t,s)\right\}\right|_{t=0}
$\hfill \null \\$\displaystyle 
=\; 
\frac{d \gamma _\xrond ^*}{d s}(s)^\top\left(
\vrule height 1.3em depth 1.3em width 0pt
\right.
\left.\frac{\partial }{\partial t}
\left\{
P_{\xrond\xrond}(Y(y,\xrond,t),\Gamma _\xrond (t,s))
\right\}\right|_{t=0}
$\hfill \null \\\null \qquad \qquad \qquad  $\displaystyle
+ 2
P_{\xrond\xrond}(y,\gamma _\xrond^*(s))\,  
\left[
\vrule height 1.3em depth 1.3em width 0pt
\frac{\partial f_\xrond}{\partial \xrond}(y ,\hat \xrond )
\right.
$\hfill \null \\\null \hfill $\displaystyle 
\left.+
\left.\frac{\partial }{\partial \xrond}
\left\{
P_{\xrond\xrond}(y ,\hat \xrond)^{-1} P_{\xrond y}(y ,\hat \xrond)
\left[\vrule height 0.6em depth 0.6em width 0pt
f_y(y,\hat\xrond)-f_y(y,\xrond)\right]
\right\} \right|_{\xrond=\gamma _\xrond^*(s)}
\right]
\left.
\vrule height 1.3em depth 1.3em width 0pt
\right)
\frac{d \gamma _\xrond ^*}{d s}(s)
$\\[1em]
where the partial derivative $\frac{\partial \hfill\null }{\partial\hat \xrond_\indxra}$ with respect to the
component $\indxra$ of $\hat \xrond$ is
\\[1em]\vbox{\noindent$\displaystyle 
\frac{\partial }{\partial \hat \xrond_\indxra  }
\left\{
P_{\xrond\xrond}(y ,\hat \xrond)^{-1} P_{\xrond y}(y ,\hat \xrond)
\left[\vrule height 0.6em depth 0.6em width 0pt
f_y(y,\hat\xrond)-f_y(y,\xrond)\right]
\right\} 
$\hfill \null \\\null \hfill $
\renewcommand{\arraystretch}{1.9}
\begin{array}{@{}r@{\qquad }r}
\displaystyle 
\;=\; 
-
P_{\xrond\xrond}(y ,\hat \xrond)^{-1} 
\frac{\partial P_{\xrond\xrond}}{\partial \xrond_\indxra  }(y ,\hat \xrond)
P_{\xrond\xrond}(y ,\hat \xrond)^{-1} 
P_{\xrond y}(y ,\hat \xrond)
\left[\vrule height 0.6em depth 0.6em width 0pt
f_y(y,\hat\xrond)-f_y(y,\xrond)\right]
&
(\star 1)
\\
\displaystyle 
\;+\; 
P_{\xrond\xrond}(y ,\hat \xrond)^{-1} 
\frac{\partial P_{\xrond y}}{\partial \xrond_\indxra  }(y ,\hat \xrond)
\left[\vrule height 0.6em depth 0.6em width 0pt
f_y(y,\hat\xrond)-f_y(y,\xrond)\right]
&
(\star 2)
\\
\displaystyle 
\;+\; 
P_{\xrond\xrond}(y ,\hat \xrond)^{-1} 
P_{\xrond y}(y ,\hat \xrond)
\frac{\partial f_y}{\partial \xrond_\indxra  }(y ,\hat \xrond)
\end{array}$}\\[1em]
We have also, with $v_\xrond=\frac{d\gamma _\xrond^*}{ds}(s)$ and with adding and subtracting 
$f_y(y,\hat\xrond)$ in the first 
line of the right member,
\\[1em]\vbox{\noindent$\displaystyle  
v_\xrond^\top \left.\frac{\partial }{\partial t}
\left\{
P_{\xrond\xrond}(Y(y,\xrond,t),\Gamma _\xrond (t,s))
\right\}\right|_{t=0}
v_\xrond
$\hfill \null \\\null\hfill $\displaystyle 
\begin{array}{@{}c@{\; }l@{\; }c@{}}
=&\displaystyle 
\left.\frac{\partial}{ \partial y}\left\{v_\xrond^\top P_{\xrond\xrond}(y ,\hat \xrond)v_\xrond\right\}f_y(y,\hat\xrond)
\;-\; 
\frac{\partial}{ \partial y}\left\{v_\xrond^\top P_{\xrond\xrond}(y ,\hat \xrond)v_\xrond\right\}
\left[\vrule height 0.6em depth 0.6em width 0pt
f_y(y,\hat\xrond)-f_y(y,\xrond)\right]
\right|_{\xrond=\gamma _\xrond^*(s)}
&(\star 3)
\\ 
&\displaystyle 
\left.+\; 
\frac{\partial}{\partial \xrond}\left\{v_\xrond^\top P_{\xrond\xrond}(y ,\hat \xrond)v_\xrond\right\}
\left(
f_\xrond   (y ,\hat \xrond    ) 
+ 
P_{\xrond\xrond}(y ,\hat \xrond)^{-1}P_{\xrond y}(y,\hat \xrond)
\left[\vrule height 0.6em depth 0.6em width 0pt
f_y(y,\hat\xrond)-f_y(y,\xrond)\right]
\right)
\right|_{\xrond=\gamma _\xrond^*(s)}
&(\star 4)
\end{array}
$}\\[1em]
We collect all the terms in the lines above labeled with $(\star)$ which have $\left[ f_y(y ,\hat \xrond  ) 
- f_y(y,\xrond) \right]$ in factor. They give a term like 
$$
(\star 3) + (\star 4)
+
2P_{\xrond\xrond}(y ,\hat \xrond)\left[(\star 1) + (\star 2)\right]
\;=\; \frac{d \gamma _\xrond^*}{ds}(s)^\top S(y,\xrond,\gamma _\xrond^*(s))\frac{d \gamma _\xrond^*}{ds}(s)
$$
which
is a row vector with the same dimension as $y$. To express it
we let $\frac{d \gamma _\indxra^*}{ds}$ denote the $\indxra$ component of the vector $\frac{d \gamma 
_\xrond^*}{ds}$. We get
\\[1em]$\displaystyle  
\frac{d \gamma _\xrond^*}{ds}(s)^\top S(y,\xrond,\gamma _\xrond^*(s))\frac{d \gamma _\xrond^*}{ds}(s)
$\hfill \null \\[0.5em] \null   \hfill    $\displaystyle
\begin{array}{@{}c@{\; }l@{}}
=&\displaystyle 
-\frac{\partial}{ \partial y}\left\{\frac{d \gamma _\xrond^*}{ds}(s)^\top
P_{\xrond\xrond}(y ,\gamma _\xrond^*(s))
\frac{d \gamma _\xrond^*}{ds}(s)\right\}
\\&\displaystyle 
\;+\; \left.
\frac{\partial}{ \partial \xrond}\left\{\frac{d \gamma _\xrond^*}{ds}(s)^\top
P_{\xrond\xrond}(y ,\xrond)
\frac{d \gamma _\xrond^*}{ds}(s)\right\}
\right|_{\xrond=\gamma _\xrond^*(s)}
P_{\xrond\xrond}(y ,\gamma _\xrond^*(s))^{-1}P_{\xrond y}(y,\gamma _\xrond^*(s))
\\&\displaystyle 
\;+\; 2
\sum_\indxra  
\frac{d \gamma _\indxra^*}{ds}(s)  \left[
-
P_{y\xrond }(y ,\gamma _\xrond^*(s))
P_{\xrond\xrond}(y ,\gamma _\xrond^*(s))^{-1} 
\frac{\partial P_{\xrond\xrond}}{\partial \xrond_\indxra  }(y ,\gamma _\xrond^*(s)) 
+
\frac{\partial P_{y\xrond }}{\partial \xrond_\indxra  }(y ,\gamma _\xrond^*(s))
\right]
\frac{d \gamma _\xrond^*}{ds}(s) 
\end{array}
$\\[1em]
where we have
\begin{eqnarray*}
\sum_\indxra \frac{d \gamma _\indxra^*}{ds}(s) 
\frac{\partial P_{\xrond\xrond}}{\partial \xrond_\indxra  }(y ,\gamma _\xrond^*(s))  
\frac{d \gamma _\xrond^*}{ds}(s)
&=&
\left.\frac{d}{ds}\left\{ P_{\xrond\xrond}(y , \xrond) 
\frac{d \gamma _\xrond^*}{ds}(s)  \right\}\right|_{ \xrond=\gamma _\xrond^*(s)}
\hskip -1em
\;-\; P_{\xrond\xrond}(y ,\gamma _\xrond^*(s))
\frac{d ^2\gamma _\xrond^*}{ds^2}(s)
\\
\sum_\indxra
\frac{d \gamma _\indxra^*}{ds}(s) 
\frac{\partial P_{y\xrond }}{\partial \xrond_\indxra  }(y ,\gamma _\xrond^*(s))
\frac{d \gamma _\xrond^*}{ds}(s)
&=&
\left.\frac{d}{ds}\left\{ P_{y\xrond }(y , \xrond) 
\frac{d \gamma _\xrond^*}{ds}(s)  \right\}\right|_{ \xrond=\gamma _\xrond^*(s)}
\hskip -1em
\;-\; P_{y\xrond }(y ,\gamma _\xrond^*(s))
\frac{d ^2\gamma _\xrond^*}{ds^2}(s)
\end{eqnarray*}
So, with (\ref{LP87}) and (\ref{LP86}), we get
$$
\frac{d \gamma _\xrond^*}{ds}(s)^\top S(y,\xrond,\gamma _\xrond^*(s))\frac{d \gamma _\xrond^*}{ds}(s)
\;=\; 0
\  .
$$
With all this, we are left with~:
\\[1em]$\displaystyle  
\left.\frac{d}{dt}\left\{
\frac{\partial \Gamma _\xrond }{\partial s}(t,s)^\top
P_{\xrond\xrond}(Y(y,\xrond,t),\Gamma _\xrond (t,s))\,  
\frac{\partial \Gamma _\xrond }{\partial s}(t,s)\right\}\right|_{t=0}
$\hfill \null \\\null\qquad  $\displaystyle 
\;=\; 
\frac{\partial}{ \partial y}\left\{\frac{d \gamma _\xrond^*}{ds}(s)^\top P_{\xrond\xrond}(y ,\gamma _\xrond^*(s))\frac{d \gamma _\xrond^*}{ds}(s)\right\}f_y(y,\gamma _\xrond^*(s))
$\hfill \null \\\null\hfill   $\displaystyle
\; +\; 
\left.\frac{\partial}{ \partial \xrond}\left\{
\frac{d \gamma _\xrond^*}{ds}(s)^\top
P_{\xrond\xrond}(y , \xrond)\frac{d \gamma _\xrond^*}{ds}(s)
\right\}\right|_{ \xrond =\gamma _\xrond^*(s)}
f_\xrond   (y ,\gamma _\xrond^*(s)    ) 
$  \\\null \hfill $\displaystyle 
\; +\; 2
\frac{d \gamma _\xrond^*}{ds}(s)^\top \left(
P_{\xrond\xrond}(y ,\gamma _\xrond^*(s))
\frac{\partial f_\xrond}{\partial \xrond}(y ,\gamma _\xrond^*(s) )
+
P_{\xrond y}(y ,\gamma _\xrond^*(s))
\frac{\partial f_y}{\partial \xrond}(y ,\gamma _\xrond^*(s))
\right)\frac{d \gamma _\xrond^*}{ds}(s)
$\\[1em]
But (\ref{3}) in Condition A2 expressed sin the $(y,\xrond)$ coordinates gives~:
\\[1em]$\displaystyle  
-q \: v_\xrond^\top P_{\xrond\xrond}(y,\xrond)  v_\xrond
%
\;\geq \; 
\frac{\partial}{ \partial y}\left\{v_\xrond^\top P_{\xrond\xrond}(y , \xrond)v_\xrond\right\}
f_y(y,\xrond)
\;+\; 
\frac{\partial}{ \partial \xrond}\left\{v_\xrond^\top P_{\xrond\xrond}(y , \xrond)v_\xrond\right\}
f_\xrond   (y , \xrond    ) 
$\hfill \null \\ \null \hfill $\displaystyle 
\;+\; 2
v_\xrond^\top \left[ P_{\xrond\xrond}(y , \xrond)
\frac{\partial f_\xrond}{\partial \xrond}(y , \xrond )
+
P_{\xrond y}(y , \xrond)
\frac{\partial f_y}{\partial \xrond }(y , \xrond)
\right] v_\xrond
\qquad \forall (y,\xrond)\in \coordyxrm(\coordyxrd)
\  .
$\\[1em]
So we have finally~:
$$
\left.\frac{d}{dt}\left\{
\frac{\partial \Gamma _\xrond }{\partial s}(t,s)^\top
P_{\xrond\xrond}(Y(y,\xrond,t),\Gamma _\xrond (t,s))\,  
\frac{\partial \Gamma _\xrond }{\partial s}(t,s)\right\}\right|_{t=0}
\; \leq \; -q \: 
\frac{d \gamma _\xrond^*}{ds}(s)^\top P_{\xrond\xrond}(y,\gamma _\xrond^*(s))  \frac{d \gamma _\xrond^*}{ds}(s)
$$
and therefore
$$
\Did d((y,\hat \xrond),(y,\xrond))
\; \leq \; -\frac{q}{2}
\left\{\int_{s}^{\hat s} \sqrt{
\frac{d \gamma _\xrond^*}{ds}(s)^\top P_{\xrond\xrond}(y,\gamma _\xrond^*(s))  \frac{d \gamma _\xrond^*}{ds}(s)
}\:  ds\right\}
\  .
$$
This gives (\ref{eqn:ReducedContraction}).
\end{proof}

\begin{remark}
\normalfont
It follows from this proof that besides Condition A2, the key point for obtaining the convergence of the 
reduced order observer is the equation
\\[1em]$\displaystyle  
\frac{\partial}{ \partial y}\left\{v_\xrond^\top
P_{\xrond\xrond}(y ,\xrond )
v_\xrond\right\}
\;-\; 
\frac{\partial}{ \partial \xrond}\left\{v_\xrond^\top
P_{\xrond\xrond}(y , \xrond)
v_\xrond\right\}
P_{\xrond\xrond}(y ,\xrond )^{-1}P_{\xrond y}(y,\xrond )
$\hfill \null \\\null \hfill $\displaystyle 
\;=\; 2
\sum_\indxra  
v_{\xrond \indxra} \left[
-
P_{y\xrond }(y ,\xrond )
P_{\xrond\xrond}(y ,\xrond )^{-1} 
\frac{\partial P_{\xrond\xrond}}{\partial \xrond_\indxra  }(y ,\xrond ) 
+
\frac{\partial P_{y\xrond }}{\partial \xrond_\indxra  }(y ,\xrond )
\right]
v_\xrond 
$\\[1em]
i.e.
$$
\frac{\partial P_{\indxra\indxrb}}{\partial y_\indyi}
-
\sum_{\indxrc,\indxrd}
\frac{\partial P_{\indxra\indxrb}}{\partial \xrond_\indxrc}
[P_{\xrond\xrond}]^{-1}_{\indxrc\indxrd}P_{\indxrd\indyi}\;=\; 
\frac{\partial P_{\indyi \indxrb}}{\partial \xrond_\indxra}
+
\frac{\partial P_{\indyi \indxra}}{\partial \xrond_\indxrb}
-
\sum_{\indxrc,\indxrd}
P_{\indyi\indxrd}
[P_{\xrond\xrond}]^{-1}_{\indxrd\indxrc}
\left[\frac{\partial P_{\indxrc\indxrb}}{\partial \xrond_\indxra}+
\frac{\partial P_{\indxrc\indxra}}{\partial \xrond_\indxrb}
\right]
$$
or
$$
\sum_{\indxrc,\indxrd}
P_{\indyi\indxrd}
[P_{\xrond\xrond}]^{-1}_{\indxrd\indxrc}
\left[\frac{\partial P_{\indxrc\indxrb}}{\partial \xrond_\indxra}+
\frac{\partial P_{\indxrc\indxra}}{\partial \xrond_\indxrb}
-
\frac{\partial P_{\indxra\indxrb}}{\partial \xrond_\indxrc}
\right]\;=\; 
\frac{\partial P_{\indyi \indxrb}}{\partial \xrond_\indxra}
+
\frac{\partial P_{\indyi \indxra}}{\partial \xrond_\indxrb}
-
\frac{\partial P_{\indxra\indxrb}}{\partial y_\indyi}
$$
With multiplication by $[P^{-1}]_{\indyj\indyi}$ and summation in $\indyi$
and since
$$
\sum_{\indyi}[P^{-1}]_{\indyj\indyi}
\sum_{\indxrd}
P_{\indyi\indxrd}
[P_{\xrond\xrond}]^{-1}_{\indxrd\indxrc}
\;=\; 
-[P^{-1}]_{\indyj\indxrc}
$$
the condition above is equivalent to the condition
$$
\Gamma ^\indyj_{\indxra\indxrb}\;=\; 0
$$
known to be equivalent to the fact that the level sets of the output function are totally geodesic.
From the proof we know this is implied by Condition A3.1.
\end{remark}

\par\vspace{1em}

The observer (\ref{2}) can be also written as
$$
\dot{\hat \xrond    } + P_{\xrond    \xrond    }(y,\hat \xrond    )^{-1}
P_{y\xrond    }(y,\hat \xrond    )^{\top} \dot y
\;=\; 
f_\xrond   \left(y ,\hat \xrond    \right)\;+\;   
P_{\xrond    \xrond    }(y,\hat \xrond    )^{-1}
P_{y\xrond    }(y,\hat \xrond    )^{\top} f_y( y,\hat \xrond    )
$$
Assume
\textit{\begin{description}
\item[\normalfont\textit{A3.2~:}] 
The number of 
outputs $p$ is one or the orthogonal distribution is integrable.
\end{description}}
In this case, there exists a function $h^\ortho:\coordyxrm(\coordyxrd)\to\RR^{n-p} $ such that
$(y,\xrond)\mapsto (y,{\barxrond})=(y,h^\ortho (y,\xrond))$
is a diffeomorphism satisfying
\begin{equation}
\label{6}
P_{y \xrond   }(y,\hat \xrond     )^\top
\;=\; 
P_{\xrond     \xrond     }(y,\hat \xrond     )
\frac{\partial h^\ortho}{\partial \xrond     }(y,\hat \xrond     )^{-1}
\frac{\partial h^\ortho}{\partial y} (y,\hat \xrond     )
\  .
\end{equation}
In the case, we can get successively
\begin{eqnarray*}
\dot{\hat \xrond    } 
\;+\;  \frac{\partial h^\ortho }{\partial \xrond     }(y,\hat \xrond     )^{-1}
\frac{\partial h^\ortho }{\partial y} (y,\hat \xrond     ) \dot y
&=&
f_\xrond   \left(y ,\hat \xrond    \right)\;+\;   
\frac{\partial h^\ortho }{\partial \xrond     }(y,\hat \xrond     )^{-1}
\frac{\partial h^\ortho }{\partial y} (y,\hat \xrond     ) f_y( y,\hat \xrond    )
\\
\frac{\partial h^\ortho }{\partial \xrond     }(y,\hat \xrond     ) \dot{\hat \xrond    } 
\;+\; 
\frac{\partial h^\ortho }{\partial y} (y,\hat \xrond     ) \dot y
&=&
\frac{\partial h^\ortho }{\partial \xrond     }(y,\hat \xrond     ) f_\xrond   \left(y ,\hat \xrond    \right)
\;+\;   
\frac{\partial h^\ortho }{\partial y} (y,\hat \xrond     ) f_y( y,\hat \xrond    )
\end{eqnarray*}
Hence by letting
$$
f_{\barxrond}(y,{\barxrond}) \;=\; \dot {\barxrond}\;=\; \frac{\partial \varphi}{\partial \xrond     }(y, \xrond     ) f_\xrond   \left(y , \xrond    \right)
\;+\;   
\frac{\partial \varphi}{\partial y} (y, \xrond     ) f_y( y, \xrond   )
$$
we see that the dynamics of $\hat {\barxrond}$ 
is simply
$$
\dot {\hat {\barxrond}}\;=\; f_{\barxrond}(y,\hat {\barxrond}) 
$$
In this way we recover exactly  \cite[Proposition 
2.8]{127} under the Conditions A1, A2, A3.1, and A3.2.
%
%
\stoparchive


\begin{IEEEbiography}%
{Ricardo G. Sanfelice} received the B.S. degree in Electronics Engineering from the Universidad de Mar del
Plata, Buenos Aires, Argentina, in 2001, and the M.S. and Ph.D. degrees in Electrical and Computer Engineering
from the University of California, Santa Barbara, CA, USA, in 2004 and 2007, respectively.  In 2007 and 2008,
he held postdoctoral positions at the Laboratory for Information and Decision Systems at the Massachusetts
Institute of Technology and at the Centre Automatique et Syst\`{e}mes at the \'Ecole des Mines de Paris.  In
2009, he joined the faculty of the Department of Aerospace and Mechanical Engineering at the University of
Arizona, Tucson, AZ, USA, where he was an Assistant Professor.  In 2014, he joined the University of
California, Santa Cruz, CA, USA, where he is currently Professor in the Department of Electrical and Computer
Engineering.  Prof.  Sanfelice is the recipient of the 2013 SIAM Control and Systems Theory Prize, the
National Science Foundation CAREER award, the Air Force Young Investigator Research Award, the 2010 IEEE
Control Systems Magazine Outstanding Paper Award, and the 2020 Test-of-Time Award from the Hybrid Systems:
Computation and Control Conference.  His research interests are in modeling, stability, robust control,
observer design, and simulation of nonlinear and hybrid systems with applications to power systems, aerospace,
and biology.
\end{IEEEbiography}

\begin{IEEEbiography}%
{Laurent Praly}
graduated as an engineer from \'EcoleNationale Sup\'{e}rieure des Mines de Paris in 1976 and got his PhD in
Automatic Control and Mathematics in 1988 from Universit\'{e} Paris IX Dauphine.  After working in industry
for three years, in 1980 he joined the Centre Automatique et Syst\`{e}mes at \'Ecole des Mines de Paris.  From
July 1984 to June 1985, he spent a sabbatical year as a visiting assistant professor in the Department of
Electrical and Computer Engineering at the University of Illinois at Urbana-Champaign.  Since 1985 he has
continued at the Centre Automatique et Syst\`{e}mes where he served as director for two years.  He has made
several long term visits to various institutions (Institute for Mathematics and its Applications at the
University of Minnesota, University of Sydney, University of Melbourne, Institut Mittag-Leffler, University of
Bologna).  His main interest is in feedback stabilization of controlled dynamical systems under various
aspects -- linear and nonlinear, dynamic, output, under constraints, with parametric or dynamic uncertainty,
disturbance attenuation or rejection.  On these topics he is contributing both on the theoretical aspect with
many academic publications and the practical aspect with applications in power systems, mechanical systems,
aerodynamical and space vehicles.
\end{IEEEbiography}

\end{document}